%% file: main.tex
\documentclass[11pt,a4paper,twoside]{report}
\usepackage{tikz}
\usepackage{tabularx, xltabular, multirow}
  \keepXColumns
  \newcolumntype{R}{X}
  \newcolumntype{L}{>{\hsize=.3\hsize}X}
  \newcolumntype{Y}{>{\centering\arraybackslash}X}
  \newcolumntype{Z}{>{\hsize=.5\hsize}X}
  \usetikzlibrary{arrows,chains,matrix,positioning,scopes,shapes}
\makeatletter
\tikzset{join/.code=\tikzset{after node path={%
\ifx\tikzchainprevious\pgfutil@empty\else(\tikzchainprevious)%
edge[every join]#1(\tikzchaincurrent)\fi}}}
\makeatother

\tikzset{>=stealth',every on chain/.append style={join},
         every join/.style={->}}
\tikzstyle{labeled}=[execute at begin node=$\scriptstyle,
   execute at end node=$]
 \pgfdeclarelayer{background}
\pgfdeclarelayer{behind}
\pgfdeclarelayer{above}
\pgfdeclarelayer{glass}
\pgfsetlayers{background,behind,main,above,glass}
\tikzstyle{box}=[rectangle, draw=black, minimum width=2cm, minimum height=1cm, align=center]
\tikzstyle{line} = [draw=black, ->]
\tikzstyle{arrow} = [thick,->,>=stealth]
\usepackage{tabularx}
\usepackage{booktabs}
\usepackage[titletoc]{appendix}
\usepackage{makeidx}
\makeindex
\usepackage[totoc]{idxlayout}

\usepackage{paralist} % very flexible & customisable lists (eg. enumerate/itemize, etc.)
\usepackage{verbatim} % adds environment for commenting out blocks of text & for better verbatim

\usepackage{fancyhdr} % This should be set AFTER setting up the page geometry
\usepackage{refcount}

\usepackage{caption}

\RenewDocumentCommand{\caption}{+o+m}{%
  \def\@figcaptype{figure}
  \ifx\@captype\@figcaptype
  \ifallcaptionsshort
  \IfValueTF{#1}{% 
    \latex@@caption[#1]{#2}%
  }{%
    \latex@@caption[#2]{#2}% No [#1] given, use the long caption then!
  }
  \else
  \latex@@caption[#2]{#2}%
  \fi
  \else
  \IfValueTF{#1}{% 
    \latex@@caption[#1]{#2}%
  }{%
    \latex@@caption[#2]{#2}% No [#1] given, use the long caption then!
  }
  \fi
}
\makeatother

\usepackage{graphicx}
\usepackage[export]{adjustbox}
\usepackage{graphics}
\usepackage{thesis}
\usepackage[a4paper, twoside, bindingoffset=1.5cm, inner=2.5cm, outer=2.5cm, top=2.5cm, bottom=2.5cm, headsep=1cm]{geometry}
\usetikzlibrary{shapes.geometric, arrows, arrows.meta,positioning}

\usepackage{amsmath}
\allowdisplaybreaks
\usepackage{amsthm}
\renewcommand\qedsymbol{$\blacksquare$}
\allowdisplaybreaks[4]
\usepackage{array} 
\usepackage{amssymb}

\usepackage{stfloats}
\usepackage{footnote}
\usepackage{booktabs}
\usepackage{array}
\usepackage{algorithm}
\usepackage{algpseudocode}
\usepackage{subeqnarray}
\usepackage{cases}
\usepackage{threeparttable}
\usepackage{color}
\usepackage{url}
\usepackage{bm}
\usepackage{physics}
\usepackage{pifont}
\usepackage{CJKutf8}
\usepackage{hyperref}
\usepackage{indentfirst}

\usepackage[style=ieee-alphabetic,doi=false, isbn=false]{biblatex}
\usepackage{subfig}
\input{mathcommand.tex}

\begin{document}

%set double line spacing
\linespread{1.5}

%cover page information

\thesistitle{Finite-Graph-Cover-Based \\ Analysis of Factor Graphs in Classical and Quantum Information Processing Systems}
\authorname{HUANG, Yuwen}
\degree{Doctor of Philosophy}
\programme{Information Engineering}
\supervisor{Pascal O. Vontobel}
\submitdate{September 2024}
\coverpage

%=======================================================

\pagenumbering{roman}

%abstract

	\vspace*{2cm}
	\large \noindent {\bf
	Abstract of thesis entitled: \thesistitle \\
	Submitted by \authorname \\
	for the degree of \degree \\
	at \institution~in September 2024
					}
	\vskip 1cm \noindent
	\addcontentsline{toc}{chapter}{Abstract}
	\input{abstract.tex}
	\newpage

%=======================================================

%acknowledgement

%
        \chapter*{Acknowledgement}

        \addcontentsline{toc}{chapter}{Acknowledgement}
        \input{acknowledgement.tex}
	\newpage

%=======================================================

%dedication

% \dedicationpage

%=======================================================

%preamble sections

\tableofcontents

%----------------------------------------------------------------------------

% \cleardoublepage
% \phantomsection
\listoffigures
\addcontentsline{toc}{chapter}{\listfigurename}

%----------------------------------------------------------------------------

%
        \chapter*{List of Publications}
        \addcontentsline{toc}{chapter}{List of Publications}
        \input{listofpubs.tex}
         
	\newpage

% \listoftables

    \chapter*{Notations\markboth{NOTATIONS}{}}
    \addcontentsline{toc}{chapter}{Notations}
	\vskip 1cm \noindent
	\input{notation.tex}
	\newpage

    \chapter*{Abbreviations\markboth{ABBREBIATIONS}{}}
    \addcontentsline{toc}{chapter}{Abbreviations}
	\vskip 1cm \noindent
	\input{abbreviations.tex}
	\newpage

% =======================================================

% initialization

\setcounter{page}{0}
\pagenumbering{arabic}
\pagestyle{headings}

%=======================================================

%insert the chapters

\chapter{Introduction}

Factor graphs are used for representing both classical statistical models~\cite{Kschischang2001, Forney2001,Loeliger2004} and quantum information processing systems~\cite{Cao2017,Alkabetz2021}. For classical statistical models, we consider standard factor graphs (S-FGs)~\cite{Kschischang2001, Forney2001,Loeliger2004}, \ie, factor graphs with local functions taking on
non-negative real values. S-FGs have a wide range of applications across various fields, including statistical mechanics (see, \eg,~\cite{Mezard2009}), coding theory (see, \eg, \cite{Richardson2008}), and communications (see, \eg, \cite{Wymeersch2007}). For representing quantities of interests in quantum information processing systems, we consider double-edges factor graphs (DE-FGs)~\cite{Cao2017,Cao2021}, 
where local functions take complex values and have to satisfy some positive semi-definiteness constraints. DE-FGs extend the concept of S-FGs and are closely related to tensor networks formalism~\cite{Cirac2009,Coecke2010,Robeva2019,Alkabetz2021}, which is a key research area in theoretical physics~\cite{Cowen2015}.

Many typical inference problems can be formulated as computing the marginals and the partition function for a suitably defined multivariate function whose factorization is depicted by an S-FG~\cite{Mezard2009} or DE-FG~\cite{Alkabetz2021}. Exact computation of these quantities is, in general, challenging for both S-FGs and DE-FGs. In order to efficiently approximate these quantities, a popular approach is to apply the sum-product algorithm (SPA), also known as loopy belief propagation (LBP), on S-FGs or DE-FGs. The SPA is a heuristic algorithm that has been successfully used for approximating the marginal functions and the partition functions for various classes of S-FGs and DE-FGs.

There are various characterizations of the pseudo-marginal functions obtained via an SPA fixed point. When an S-FG or DE-FG is cycle-free (\ie, tree-structured), the pseudo-marginal functions evaluated at the SPA fixed point are the exact marginal functions of the multivariate function represented by the S-FG or DE-FG (see, \eg,~\cite{Mezard2009}). For S-FGs or DE-FGs with cycles, the SPA often provides surprisingly good approximations of the  marginals and partition functions. However, there exists instances of cyclic S-FGs and DE-FGs where the SPA gives a poor approximations of these quantities~\cite{Murphy1999,Heskes2003,Weller2014}.
%***************************************************************************

For S-FGs with cycles, Yedidia \etal~in~\cite{Yedidia2005} defined the Bethe free energy function and also the Bethe approximation of the partition function, also known as the Bethe partition function, in terms of the minimum of the Bethe free energy function. The main idea of the Bethe free energy function dates back to the work of Bethe in 1935~\cite{Bethe1935}. Then Yedidia \etal~\cite{Yedidia2005} provided a theoretical justification of the SPA by showing that the fixed points of the SPA on an S-FG correspond to the stationary points of the Bethe free energy function. Therefore, the Bethe partition function can be computed with the help of the SPA.

For S-FGs with cycles, Chertkov and Chernyak~\cite{Chertkov2006, Chernyak2007} 
introduced a concept called loop calculus, which provides an expression
that relates the partition function to the Bethe partition function. 
This expression yields an exact representation of the partition function as a finite sum of terms, with the first term being the Bethe partition function and higher-order
terms obtained by adding so-called loop corrections. Mori in~\cite{Mori2015} extended the results by Chertkov and Chernyak to S-FGs with non-binary alphabets, using concepts from information geometry.

%***************************************************************************
Finite graph covers of an S-FG or DE-FG offer a different perspective to understand the SPA on these graphs. A notable observation about finite graph covers is that each cover looks locally the same as the base graph, \ie, the original S-FG or DE-FG~\cite{Koetter2003}. This implies that the SPA, a locally operating algorithm, ``cannot distinguish'' whether it is running on the base graph, or implicitly on any of the finite graph covers.

Finite graph covers serves as a useful theoretical tool for analyzing the SPA and the Bethe partition functions for S-FGs. Specifically:
%----------------------------------------------------------------------------
\begin{itemize}
  \item In order to understand the performance of the SPA for decoding low-density parity-check (LDPC) codes, Koetter \etal~\cite{Koetter2003,Koetter2007}~studied finite graph covers of S-FGs representing LDPC codes. Ruozzi \etal~\cite{Ruozzi2009} utilized finite graph covers to characterize the behavior of the max-product algorithm for Gaussian graphical models. (For Gaussian graphical models, the max-product algorithm is essentially equivalent to the SPA.)

  \item Vontobel~\cite{Vontobel2013} provided a combinatorial characterization of the Bethe partition function in terms of finite graph covers of an S-FG, as illustrated in Fig.~\ref{fig: combinatorial chara for sfg}. In Fig.~\ref{fig: combinatorial chara for sfg}, $ \sfN $ denotes an arbitrary S-FG, $ Z(\sfN) $ is the partition function, $ \ZB(\sfN) $ is the Bethe partition function, and $ \ZBM(\sfN) $, the degree-$M$ Bethe partition function, is defined based on the degree-$M$ covers of $ \sfN $.
  Vontobel~\cite{Vontobel2013} showed that
  \begin{align*}
    \ZBM(\sfN) &= Z(\sfN), \qquad M = 1, \nonumber\\
    \limsup_{ M \to \infty} \ZBM(\sfN) &= \ZB(\sfN).
  \end{align*} 
  Similar characterizations were discussed in \cite{Altieri2017,Angelini2022}.

  \item Based on Vontobel's findings~\cite{Vontobel2013}, 
  Ruozzi \etal~in~\cite{Ruozzi2012} and Csikv\'{a}ri \etal~in~\cite{Csikvari2022} 
  proved that the partition functions are lower bounded by the Bethe partition functions
  for certain classes of S-FGs, particularly resolving a conjecture by Sudderth
  \etal~\cite{E.B.Sudderth2007} regarding log-supermodular graphical models. 
\end{itemize}
%----------------------------------------------------------------------------

\begin{figure}[t]
  \centering
  \begin{tikzpicture}
    \pgfmathsetmacro{\Ws}{1.3};
    % Expression on the LHS and RHS
    %----------------------------------------------------------------------------
    \begin{pgfonlayer}{main}
      \node (ZB1)     at (0,0) [] 
      {$\qquad\qquad\qquad\quad\left. \ZBM(\sfN) \right|_{M = 1} = Z(\sfN)$};
      \node (ZBM)     at (0,\Ws) [] {$\hspace{0.5 cm}\ZBM(\sfN)$};
      \node (ZBinfty) at (0,2*\Ws) {$\qquad\qquad\qquad\qquad
      \left. \ZBM(\sfN) \right|_{M \to \infty} 
      = \ZB(\sfN)$};
    %----------------------------------------------------------------------------
      \draw[]
        (ZB1) -- (ZBM) -- (ZBinfty);
    \end{pgfonlayer}
    % \begin{pgfonlayer}{glass}
    %   \node (1) at (0.25, \Ws) [] {$\ZBM(\sfN)$};
    % \end{pgfonlayer}
    %\vspace{0.3cm}
  \end{tikzpicture}
  \caption[A combinatorial characterization of the Bethe partition function.]{A combinatorial characterization of the Bethe partition function for an S-FG.
  \label{fig: combinatorial chara for sfg}}
\end{figure}

In this thesis, we also leverage finite graph covers to analyze the SPA and the relationship between the Bethe partition functions and the partition functions for both S-FGs and DE-FGs.
Our investigation yields two main contributions:
%----------------------------------------------------------------------------
\begin{enumerate}
  \item \textbf{Degree-$M$-Bethe-permanent-based bounds for the permanent \\ of a non-negative square matrix}

  Vontobel in~\cite{Vontobel2013a} introduced a class of S-FGs such that the partition function of each S-FG is equal to the permanent of a non-negative square matrix. The combinatorial characterization of the Bethe partition function for such an S-FG is given by the degree-$M$ Bethe permanent, which is defined based on the degree-$M$ covers of the S-FG. In this thesis, we present degree-$M$-Bethe-permanent-based bounds on the permanent of a non-negative square matrix, resolving a conjecture proposed by Vontobel in~\cite[Conjecture 51]{Vontobel2013a} Detailed discussion can be found in Section~\ref{sec: degree M charactertion of Bethe and Sinkhorn permanent}.

  \item \textbf{Combinatorial characterization of Bethe partition function \\ for DE-FGs}

  For an arbitrary DE-FG, a natural question is whether
  there is a combinatorial characterization of its Bethe partition function in terms of its finite graph covers. In this thesis, we prove this characterization for DE-FGs that satisfy an easily checkable condition. However, based on the obtained numerical results, we suspect that this characterization holds more broadly.
  The detailed discussion is in Section~\ref{intro: characterize ZB for DE-FG}.
\end{enumerate}
%----------------------------------------------------------------------------

\section[Finite-Graph-Cover-Based Bounds on the Permanent]{\texorpdfstring{Degree-$M$}{} Bethe and Sinkhorn Permanent Based Bounds on the Permanent of a Non-negative square Matrix}\label{sec: degree M charactertion of Bethe and Sinkhorn permanent}

This section establishes finite-graph-cover-based bounds for the permanent of a nonnegative square matrix. In particular, Sections~\ref{chapt:intro:permanent} and~\ref{chapt:intro:Bethe permanent} introduce the permanent of a non-negative square matrix and the Bethe permanent, respectively, with the latter being an S-FG-based approximation of the matrix permanent. Section~\ref{chapt: permanent: combinatorial chara of ZB and ZscC} presents the degree-$M$ Bethe permanent, which is defined based on degree-$M$ covers of the underlying S-FG. Section~\ref{sec:main:constributions:1} then establishes degree-$M$-Bethe-permanent-based bounds on the matrix permanent. Moreover, Sections~\ref{chapt:intro:Bethe permanent}--\ref{sec:main:constributions:1} show that similar results holds for the (scaled) Sinkhorn permanent~\cite{N.Anari2021}, which is another S-FG-based approximation of the matrix permanent.

\subsection{The Permanent of a Square Matrix}
\label{chapt:intro:permanent}

For any $n \in \sZpp$, define $[n] \defeq \{ 1, 2, \ldots, n \}$ and let
$\setS_{[n]}$ be the set of all $n!$ bijections from $[n]$ to $[n]$. Recall
the definition of the permanent of a real square matrix (see, \eg,
\cite{Minc1984}).

%***************************************************************************

\begin{definition}\label{sec:1:def:13}
  Let $n \in \sZpp$ and let
  $ \mtheta \defeq \bigl( \theta(i,j) \bigr)_{\! i,j\in [n]}$ be a real matrix of
  size $n \times n$. The permanent of $ \mtheta $ is defined to be
  \begin{align}
    \perm(\mtheta)
      &\defeq
         \sum\limits_{\sigma \in \setS_{[n]}} 
           \prod\limits_{i \in [n]}
             \theta\bigl( i, \sigma(i) \bigr) .
               \label{eq:def:permanent:1} \\[-1.00cm] \nonumber
  \end{align}
  \edefinition
\end{definition}

%***************************************************************************

In this thesis, we consider only the matrices $\mtheta$ such that there is at least one
$\sigma \in \setS_{[n]}$ such that
$\prod\limits_{i \in [n]} \theta\bigl( i, \sigma(i) \bigr) \neq 0$.

%In this thesis, we consider only matrices $\mtheta$ where
%$\perm(\mtheta) \neq 0$.

%***************************************************************************

Computing the exact permanent is in the complexity class \#P, where \#P is the
set of the counting problems associated with the decision problems in the
class NP. Notably, already the computation of the permanent
of matrices that contain only zeros and ones is \#P-complete~\cite{Valiant1979}.

%***************************************************************************

\subsection[The Bethe Permanent and the Sinkhorn Permanent]{The Bethe Permanent and the (Scaled) Sinkhorn Permanent of a Non-negative Square Matrix}
\label{chapt:intro:Bethe permanent}

While being able to compute the exact permanent of a matrix might be desirable
in some circumstances, many applications require only a good approximation of
the permanent. In this thesis, we focus on the important special
case of approximating the permanent of a non-negative real matrix, \ie, a
matrix where all entries are non-negative real numbers. In particular, we
study two graphical-model-based methods for approximating the permanent of a
non-negative square matrix. The first method is motivated by Bethe-approximation /
SPA based methods proposed in~\cite{Chertkov2008,
  Huang2009} and studied in more detail in~\cite{Vontobel2013a}. The second
method is motivated by Sinkhorn's matrix scaling algorithm~\cite{Linial2000}.

%***************************************************************************

The main idea of these methods is to construct an S-FG~\cite{Forney2001, Loeliger2004} whose partition
function equals the permanent.
Consequently, the permanent can be obtained by
finding the minimum of the Gibbs free energy function associated with the
factor graph. However, because finding the minimum of the Gibbs free energy
function is in general intractable unless $n$ is small, one considers the
minimization of some functions that suitably approximate the Gibbs free energy
function.

%***************************************************************************

A first approximation of the Gibbs free energy function is given by the Bethe
free energy function~\cite{Yedidia2005}; the resulting approximation of
$\perm(\mtheta)$ is called the Bethe permanent and denoted by
$ \permb(\mtheta)$. Vontobel~\cite{Vontobel2013a} showed that for the relevant
factor graph, the Bethe free energy function is convex and that the SPA
converges to its minimum.
(This is in contrast to the situation for factor
graphs in general, where the Bethe free energy function associated with a
factor graph is not convex and the SPA is not guaranteed to converge to its
minimum.) The Bethe permanent $\permb(\mtheta)$ satisfies
\begin{align}
  1 
    &\leq 
       \frac{\perm(\mtheta)}{ \permb(\mtheta)} 
     \leq 
        2^{n/2}.
        \label{eq:ratio:permanent:bethe:permanent:1}
\end{align}
The first inequality in~\eqref{eq:ratio:permanent:bethe:permanent:1} was
proven by Gurvits~\cite{Gurvits2011} with the help of an inequality of
Schrijver~\cite{Schrijver1998}. Later on, an alternative proof based on
results from real stable polynomials was given by Straszak and Vishnoi
in~\cite{Straszak2019} and Anari and Gharan in~\cite{Anari2021}. The second
inequality in~\eqref{eq:ratio:permanent:bethe:permanent:1} was conjectured by
Gurvits~\cite{Gurvits2011} and proven by Anari and
Rezaei~\cite{Anari2019}. Both inequalities
in~\eqref{eq:ratio:permanent:bethe:permanent:1} are the best possible given
only the knowledge of the size of the matrix~$\mtheta$.

A second approximation of the Gibbs free energy function is given by the
Sinkhorn free energy function; the resulting approximation of $\perm(\mtheta)$
is called the Sinkhorn permanent and denoted by $\perms(\mtheta)$. (The
Sinkhorn free energy function was formally defined in~\cite{Vontobel2014}
based on some considerations in~\cite{Huang2009}. Note that
  $\perms(\mtheta)$ appears already in~\cite{Linial2000} as an approximation
  of $\perm(\mtheta)$ without being called the Sinkhorn permanent.) The
Sinkhorn free energy function is easily seen to be convex and it can be
efficiently minimized with the help of Sinkhorn's matrix scaling
algorithm~\cite{Linial2000}, thereby motivating its name. In this thesis, we
will actually work with the scaled Sinkhorn permanent
$\permscs(\mtheta) \defeq \e^{-n} \cdot \perms(\mtheta)$, a variant of the
Sinkhorn permanent $\perms(\mtheta)$ that was introduced
in~\cite{Anari2021}. The scaled Sinkhorn permanent $\permscs(\mtheta)$
satisfies
\begin{align}
  \e^n \cdot \frac{n!}{n^n}
    &\leq 
       \frac{\perm(\mtheta)}{\permscs(\mtheta)}
     \leq 
       \e^n .
         \label{eq:ratio:permanent:scaled:Sinkhorn:permanent:1}
\end{align}
The first inequality in~\eqref{eq:ratio:permanent:scaled:Sinkhorn:permanent:1}
follows from van der Waerden's inequality proven by
Egorychev~\cite{Egorychev1981} and Falikman~\cite{Falikman1981}, whereas the
second equality follows from a relatively simple to prove inequality. Both
inequalities in~\eqref{eq:ratio:permanent:scaled:Sinkhorn:permanent:1} are the
best possible given only the knowledge of the size of the
matrix $\mtheta$. Note that the value on the LHS
of~\eqref{eq:ratio:permanent:scaled:Sinkhorn:permanent:1} is approximately
$\sqrt{2 \pi n}$.

%***************************************************************************

\subsection[A Combinatorial Characterization]{A Combinatorial Characterization of the Bethe Permanent and the (Scaled) Sinkhorn Permanent}
\label{chapt: permanent: combinatorial chara of ZB and ZscC}

While the definition of $\perm(\mtheta)$ in~\eqref{eq:def:permanent:1} clearly
has a combinatorial flavor (\eg, counting weighted perfect matchings in a
complete bipartite graph with two times $n$ vertices), it is, a priori, not
clear if a combinatorial characterization can be given for $\permb(\mtheta)$
and $\permscs(\mtheta)$. However, this is indeed the case.

%***************************************************************************

Namely, using the results in~\cite{Vontobel2013}, a finite-graph-cover-based
combinatorial characterization was given for $\permb(\mtheta)$
in~\cite{Vontobel2013a}. Namely,
\begin{align*}
  \permb(\mtheta)
    &= \limsup_{M \to \infty} \,
         \permbM{M}(\mtheta) ,
           % \label{eq:bethe:perm:as:limit:1}
\end{align*}
where the degree-$M$ Bethe permanent is defined to be
\begin{align*}
  \permbM{M}(\mtheta)
    &\defeq
      \sqrt[M]{
        \bigl\langle
          \perm( \mtheta^{\uparrow \mP_{M}} )
        \bigr\rangle_{\mP_{M} \in \tPsi_{M} }
     }  , 
       % \label{eq:def:degree:M:Bethe:permanent:1}
\end{align*}
where the angular brackets represent an arithmetic average and where
$\mtheta^{\uparrow \mP_{M}}$ represents a degree-$M$ cover of $\mtheta$
defined by a collection of permutation matrices $\mP_{M}$. (See
Definition~\ref{def:matrix:degree:M:cover:1} for the details.)

%***************************************************************************

In this thesis we offer a combinatorial characterization of the scaled Sinkhorn
permanent $\permscs(\mtheta)$. Namely,
\begin{align*}
  \permscs(\mtheta)
    &= \limsup_{M \to \infty} \, 
         \permscsM{M}(\mtheta) , 
           % \label{eq:scaled:Sinkhorn:permanent:as:limit:1}
\end{align*}
where $\permscs(\mtheta)$ was defined
in~\cite{N.Anari2021}, and where
the degree-$M$ scaled Sinkhorn permanent is defined to be \\[-0.35cm]
\begin{align}
  \permscsM{M}(\mtheta)
    &\defeq 
       \sqrt[M]{
         \perm
           \bigl(
             \langle
               \mtheta^{\uparrow \mP_{M}}
             \rangle_{\mP_{M} \in \tPsi_{M} }
           \bigr)
       } \nonumber \\
    &= \sqrt[M]{
         \perm( \mtheta \otimes \mU_{M,M})
       }  ,\nonumber
\end{align}
where $\otimes$ denotes the Kronecker product of two matrices and where
$ \mU_{M,M} $ is the matrix of size $M \times M$ with all entries equal to
$1/M$.  (See Definition~\ref{def:Sinkhorn:degree:M:cover:1} and
Proposition~\ref{sec:1:prop:14} for the details.)~\footnote{Note that the quantity
  $\sqrt[M]{ \perm( \mtheta \otimes \mU_{M,M}) }$ appears in the literature at
  least as early as~\cite{Bang1976,Friedland1979}. These papers show that for
  $M = 2^s$, $s \in \sZp$, it holds that
  $\sqrt[M]{ \perm( \mtheta \otimes \mU_{M,M}) } \leq \perm(\mtheta)$, which,
  in our notation, implies $1 \leq \perm(\mtheta) / \permscsM{M}(\mtheta)$ for
  $M = 2^s$, $s \in \sZp$.}%
~\footnote{Barvinok~\cite{Barvinok2010} considered combinatorial expressions
  in the same spirit as $\sqrt[M]{ \perm( \mtheta \otimes \mU_{M,M}) }$, but
  the details and the limit are different.}

%***************************************************************************

\begin{figure}[t]
  \begin{center}
    \hspace{-2.25cm}\input{figures/relate_z_zbm_zb.tex}
    \hspace{-1.75cm}\input{figures/relate_z_permscsm_permscs.tex}
  \end{center}
  \caption[Combinatorial characterizations of the 
    Bethe and Sinkhorn permanents]{Combinatorial characterizations of the 
    Bethe and scaled Sinkhorn permanents.}
  \label{fig:combinatorial:characterization:1}
  \vspace{0.00cm} % see if this space is needed in the end
\end{figure}

%***************************************************************************

%***************************************************************************

These combinatorial characterizations of the Bethe and scaled Sinkhorn
permanent are summarized in Fig.~\ref{fig:combinatorial:characterization:1}.

%***************************************************************************

Besides the quantities $\permbM{M}(\mtheta)$ and $\permscsM{M}(\mtheta)$ being
of inherent interest, they are particularly interesting toward understanding
the ratios 
\begin{align*}
  \frac{ \perm(\mtheta) }{ \permb(\mtheta) }, \qquad 
  \frac{ \perm(\mtheta) }{ \permscs(\mtheta) }
\end{align*}
for general non-negative square matrices or for
special classes of non-negative square matrices.

%***************************************************************************

For example, as discussed in~\cite[Sec.~VII]{Vontobel2013}, one can consider
the equation
\begin{align*}
  \hspace{-0.25cm}
  \frac{\perm(\mtheta)}{\permb(\mtheta)}
    &= \frac{\perm(\mtheta)}{\permbM{2}(\mtheta)}
       \!\cdot\!
       \frac{\permbM{2}(\mtheta) }{\permbM{3}(\mtheta)}
       \!\cdot\!
       \frac{\permbM{3}(\mtheta) }{\permbM{4}(\mtheta)}
       \!\cdots .
         % \label{eq:Z:ratios:long:1}
\end{align*}
If one can give bounds on the ratios appearing on the RHS of this equation,
then one obtains bounds on the ratio on the LHS of this equation.

% In fact, one of the main results of the present thesis is in this spirit. In
% particular, from the results in this thesis it follows that XXXXX
% \begin{align}
%   LBXXXXX
%     &\leq 
%        \frac{\permbM{M}(\mtheta)}{\permbM{M+1}(\mtheta)}
%      \leq
%        UBXXXXX .
% \end{align}
% XXXXX Also $(M \! - \! 1)/M$ ranges from $0$ to $1$, with $1/2$ achieved for
% $M = 2$. XXXXX also similar strategy for other graphical models XXXXX

%***************************************************************************

Alternatively, one can consider the equation
\begin{align}
  \underbrace{
    \frac{\perm(\mtheta)}
         {\permb(\mtheta)}
  }_{\text{\ding{192}}}
    &= \underbrace{
         \frac{\perm(\mtheta)}
              {\permbM{2}(\mtheta)}
       }_{\text{\ding{193}}}
       \cdot
       \underbrace{
         \frac{\permbM{2}(\mtheta)}
              {\permb(\mtheta)}
       }_{\text{\ding{194}}} .
         \label{eq:perm:ratios:1}
\end{align}
Based on numerical and analytical results, Ng and Vontobel~\cite{KitShing2022}
concluded that it appears that for many classes of matrices of interest, the
ratio~\ding{193} behaves similarly to the ratio~\ding{194}. Therefore,
understanding the ratio~\ding{193} goes a long way toward understanding the
ratio~\ding{192}. (Note that the ratio~\ding{193} is easier to analyze than
the ratio~\ding{192} for some classes of matrices of interest.)

%***************************************************************************

Results similar in spirit to the above results were given for the Bethe
approximation of the partition function of other graphical models (see, \eg,
\cite{Ruozzi2012,Vontobel2016,Csikvari2022}).

\subsection{Main Contributions}
\label{sec:main:constributions:1}

%***************************************************************************

Let $n, M \in \sZpp$, let $\mtheta$ be a fixed non-negative matrix of size
$n \times n$, and let $\GamMn$ be the set of doubly stochastic matrices of
size $n \times n$ where all entries are integer multiples of $1/M$. Moreover,
for $\vgam \in \GamMn$, let
$\mtheta^{ M \cdot \vgam } \defeq \prod\limits_{i,j \in [n]} \bigl( \theta(i,j)
\bigr)^{\! M \cdot \gamma(i,j)}$. (Note that in this expression, all exponents
are non-negative integers.) Throughout this thesis, we consider $ \mtheta $
such that $ \perm(\mtheta) $, $ \permbM{M}(\mtheta) $, and
$ \permscsM{M}(\mtheta) $ are positive real-valued.

%***************************************************************************

\begin{lemma}
  \label{lem: expression of permanents w.r.t. C}

  There are collections of non-negative real numbers \\
  \mbox{$\bigl\{ \CM{n}( \vgam ) \bigr\}_{\vgam \in \GamMn}$,\!
  $\bigl\{ \CBM{n}( \vgam ) \bigr\}_{\vgam \in \GamMn}$,\!
  $\bigl\{ \CscSgen{M}{n}( \vgam ) \bigr\}_{\vgam \in \GamMn}$} such that
  \begin{align}
    \bigl( \perm(\mtheta) \bigr)^{\! M} 
      &= \sum\limits_{\vgam \in \GamMn} 
           \mtheta^{ M \cdot \vgam }
           \cdot
           \CM{n}( \vgam ) , 
             \label{sec:1:eqn:43} \\
    \bigl( \permbM{M}(\mtheta) \bigr)^{\! M} 
      &= \sum\limits_{\vgam \in \GamMn} 
           \mtheta^{ M \cdot \vgam }
           \cdot
           \CBM{n}( \vgam ) ,
             \label{sec:1:eqn:190}  \\
    \bigl( \permscsM{M}(\mtheta) \bigr)^{\! M}
      &= \sum\limits_{\vgam \in \GamMn}
           \mtheta^{ M \cdot \vgam } \cdot \CscSgen{M}{n}( \vgam ).
             \label{sec:1:eqn:168} \\[-0.75cm] \nonumber
  \end{align}
  \elemma
\end{lemma}

%***************************************************************************

With the help of the inequalities
  in~\eqref{eq:ratio:permanent:bethe:permanent:1}
  and~\eqref{eq:ratio:permanent:scaled:Sinkhorn:permanent:1}, along with some
  results that are established in this thesis, we can make the following
  statement about the coefficients $\CM{n}( \vgam )$, $\CBM{n}( \vgam )$, and
  $\CscSgen{M}{n}( \vgam )$ in Lemma~\ref{lem: expression of permanents
    w.r.t. C}.

%***************************************************************************

\begin{theorem}
  \label{thm: inequalities for the coefficients}

  For every $\vgam \in \GamMn$, the coefficients $\CM{n}( \vgam )$,
  $\CBM{n}( \vgam )$, and $\CscSgen{M}{n}( \vgam )$ satisfy
  \begin{align}
    1 
      &\leq 
         \frac{ \CM{n}(\vgam) }{\CBM{n}(\vgam) }
       \leq
         \Bigl( 2^{n/2} \Bigr)^{\! M-1} ,
           \label{sec:1:eqn:58} \\
    \Biggl( \frac{M^{M}}{M!} \Biggr)^{\!\!\! n}
    \cdot
    \biggl(
      \frac{ n! }{ n^{n} }
    \biggr)^{\!\! M-1} 
      &\leq
         \frac{\CM{n}(\vgam)}{\CscSgen{M}{n}(\vgam)} 
       \leq 
         \Biggl( \frac{M^{M}}{M!} \Biggr)^{\!\! n}.
           \label{sec:1:eqn:180}
  \end{align}
  \etheorem
\end{theorem}

%***************************************************************************
Combining Lemma~\ref{lem: expression of permanents w.r.t. C} and
Theorem~\ref{thm: inequalities for the coefficients}, we can make the
following statement.

%***************************************************************************

\begin{theorem}
  \label{th:main:permanent:inequalities:1}
  \label{TH:MAIN:PERMANENT:INEQUALITIES:1}
  It holds that
  \begin{align}
    1 
      &\leq
         \frac{\perm(\mtheta)}{\permbM{M}(\mtheta)}
       \leq
         \Bigl( 2^{n/2} \Bigr)^{\!\! \frac{M-1}{M}}
           \label{SEC:1:EQN:147}, \\
    \frac{M^{n}}{(M!)^{n/M}}
    \cdot
    \biggl( \frac{n!}{n^{n}} \biggr)^{\!\! \frac{M-1}{M}}
      &\leq
         \frac{ \perm(\mtheta) }{ \permscsM{M}(\mtheta) }
       \leq
         \frac{M^{n}}{(M!)^{n/M}}.
           \label{sec:1:eqn:200}
  \end{align}
  Note that in the limit $M \to \infty$ we recover the inequalities
  in~\eqref{eq:ratio:permanent:bethe:permanent:1}
  and~\eqref{eq:ratio:permanent:scaled:Sinkhorn:permanent:1}.
  \etheorem
\end{theorem}

%***************************************************************************
Finally, we discuss the following asymptotic characterization of the
coefficients $\CM{n}( \vgam )$, $\CBM{n}( \vgam )$, and
$\CscSgen{M}{n}( \vgam )$.

%***************************************************************************

\begin{proposition}
  \label{prop:coefficient:asymptotitic:characterization:1}
  
  Let $\vgam \in \GamMn$. It holds that
  \begin{align}
    \CM{n}( \vgam )
      &= \exp\bigl( M \cdot \HG'(\vgam) + o(M) \bigr) ,
           \label{sec:1:eqn:125} \\
    \CBM{n}( \vgam )
      &= \exp\bigl( M \cdot \HBthe(\vgam) + o(M) \bigr) ,
           \label{sec:1:eqn:206}  \\
    \CscSgen{M}{n}( \vgam )
      &= \exp\bigl( M \cdot \HscSthe(\vgam) + o(M) \bigr) ,
           \label{sec:1:eqn:207}
  \end{align}
  where $\HG'(\vgam)$, $\HBthe(\vgam)$, and $\HscSthe(\vgam)$ are the
  modified Gibbs entropy function, the Bethe entropy function, and the
  scaled Sinkhorn entropy function, respectively, as defined in Section~\ref{sec:free:energy:functions:1}.

  \vspace{-0.02 cm}
  \eproposition
\end{proposition}

%***************************************************************************

Let us put the above results into some context:
\begin{itemize}

  \item The first inequality in~\eqref{sec:1:eqn:58} proves the first (weaker)
    conjecture in~\cite[Conjecture 52]{Vontobel2013a}.

  \item Smarandache and Haenggi~\cite{Smarandache2015} claimed a proof of the
    second (stronger) conjecture
    in~\cite[Conjecture~52]{Vontobel2013a},\footnote{The second
        (stronger) conjecture in~\cite[Conjecture~52]{Vontobel2013a} stated that
        for each $ \mP_{M} \in \tPsi_{M} $, there is a collection of
        non-negative real numbers
        $\bigl\{ C_{\mathrm{B},M,n,\mP_{M}}( \vgam ) \bigr\}_{\vgam \in \GamMn}$
        such that
        \begin{align*}
          \perm( \mtheta^{\uparrow \mP_{M}} ) = \sum\limits_{\vgam \in \GamMn}
        \mtheta^{ M \cdot \vgam } \cdot C_{\mathrm{B},M,n,\mP_{M}}( \vgam ) ,
        \end{align*}
        and that for every $ \vgam \in \GamMn $, it holds that
        $ \CM{n}(\vgam) / C_{\mathrm{B},M,n,\mP_{M}}( \vgam ) \geq 1.  $}
    which would imply the first (weaker) conjecture in~\cite[Conjecture
    52]{Vontobel2013a}. However, the proof in~\cite{Smarandache2015} is
    incomplete.

  \item The first inequality in~\eqref{SEC:1:EQN:147} proves the first (weaker)
    conjecture in~\cite[Conjecture 51]{Vontobel2013a}.

  \item The proof of the first inequality in~\eqref{sec:1:eqn:58}, and with that
    the proof of the first inequality in~\eqref{SEC:1:EQN:147}, uses an
    inequality of 
    Schrijver~\cite{Schrijver1998}. By now, there are various
    proofs of Schrijver's inequality and also various proofs of the first
    inequality in~\eqref{eq:ratio:permanent:bethe:permanent:1} available in the
    literature, however, the ones that we are aware of are either quite involved
    (like Schrijver's original proof of his inequality) or use reasonably
    sophisticated machinery like real stable polynomials. Given the established
    veracity of the first inequality in~\eqref{SEC:1:EQN:147}, it is hoped that
    a more basic proof for this inequality, and with that of the first
    inequality in~\eqref{eq:ratio:permanent:bethe:permanent:1}, can be
    found. Similar statements can be made for the other inequalities
    in~\eqref{SEC:1:EQN:147} and~\eqref{sec:1:eqn:200}.

    In fact, for $M = 2$, a basic proof of the inequalities
    in~\eqref{SEC:1:EQN:147} is available: it is based
    on~\cite[Proposition~2]{KitShing2022} (which appears in this thesis as
    Proposition~\ref{prop:ratio:perm:permbM:2:1}), along with the observation
    that the quantity $c(\sigma_1,\sigma_2)$ appearing therein is bounded as
    $0 \leq c(\sigma_1,\sigma_2) \leq n/2$. Another basic proof of the first
    inequality in~\eqref{SEC:1:EQN:147} is based on a straightforward
    generalization of the result in~\cite[Lemma 4.2]{Csikvari2017} from
    $\{0,1\}$-valued matrices to arbitrary non-negative square matrices and then
    averaging over $2$-covers.

    Moreover, for the case where $\mtheta$ is proportional to the all-one matrix
    and where $M \in \sZpp$ is arbitrary, a basic proof of the first inequality
    in~\eqref{SEC:1:EQN:147} was given in~\cite[Appendix~I]{Vontobel2013a}.
    However, presently it is not clear how to generalize this proof to other
    matrices $\mtheta$.

  \item For $M = 2$, the lower and upper bounds in~\eqref{SEC:1:EQN:147} are
    exactly the square root of the lower and upper bounds
    in~\eqref{eq:ratio:permanent:bethe:permanent:1}, respectively, thereby
    corroborating some of the statements made after~\eqref{eq:perm:ratios:1}.

  \item Although we take advantage of special properties of the factor graphs
    under consideration, potentially the techniques in this thesis can be
    extended toward a better understanding of the Bethe approximation of the
    partition function of other factor graphs. (For more on this topic, see the discussion in Chapter~\ref{chapt: summary and outlook}.)

\end{itemize}

\section[The Graph Covers for DE-FGs]{Characterizing the Bethe Partition Function of \\Double-Edges Factor Graphs via Graph Covers}
\label{intro: characterize ZB for DE-FG}

The results in Section~\ref{sec: degree M charactertion of Bethe and Sinkhorn permanent} are for factor graphs called S-FGs, \ie, for factor graphs with local functions taking on non-negative real values. Recently, several papers~\cite{Loeliger2012, Loeliger2017, Loeliger2020, Mori2015a} have explored more general factor graphs, in particular those where the local functions take on complex values. This exploration is particularly relevant for representing quantities of interest in quantum information processing. (For connections of these factor graphs to other graphical notations in physics, see~\cite[Appendix~A]{Loeliger2017}.) The structure of these factor graphs is not completely arbitrary. In order to formalize such factor graphs, Cao and Vontobel~\cite{Cao2017} introduced DE-FGs, where local functions take on complex values and have to satisfy certain positive semi-definiteness constraints. They~\cite{Cao2017} further showed that the partition function of a DE-FG is a non-negative real number. However, because the partition function is the sum of complex numbers (notably, the real and the imaginary parts of these complex numbers can have a positive or negative sign), approximating the partition function is, in general, much more challenging for DE-FGs than for S-FGs. This problem is known as the numerical sign problem in applied mathematics and theoretical physics~(see, \eg, \cite{LohJr.1990}). 

Based on the positive semi-definiteness constraints posed for DE-FGs, Cao and Vontobel~\cite{Cao2021} developed the SPA for DE-FGs, where the SPA messages also satisfy these constraints, and further defined the Bethe partition function based on the SPA fixed-point messages. \footnote{Similar to the definition of the Bethe partition function in~\cite[Def.~8]{Cao2017}, in this thesis, the Bethe partition function is formulated based on the so-called pseudo-dual Bethe free energy function for S-FGs. Although
generalizing the primal Bethe free energy function from S-FGs to DE-FGs is
formally straightforward, it poses challenges because of the
multi-valuedness of the complex logarithm and is left for future research.} 
Then they showed that the Bethe partition function is a non-negative real number.

DE-FGs have a strong connection to tensor networks formulism~\cite{Cirac2009,Coecke2010,Robeva2019,Alkabetz2021}, which has become a central topic in theoretical physics~\cite{Cowen2015}. In~\cite{Alkabetz2021}, Alkabetz and Arad showed how to map a projected-entangled-pair-state (PEPS) tensor network to a DE-FG. 
Alkabetz and Arad in~\cite{Alkabetz2021} connected the SPA for DE-FGs to tensor network contractions. Motivated by this connection, Tindall and Fishman~\cite{Tindall2023} related the SPA for DE-FGs to known tensor network gauging methods. Moreover, Tindall \etal~\cite{Tindall2024} showed that for the kicked Ising quantum system studied in~\cite{Kim2023}, the SPA-based approximation results are more accurate and precise than those obtained from the quantum processor and many other classical methods. These promising results motivate us to study the SPA for DE-FGs.

Given that both the partition function and the Bethe partition function are non-negative real numbers for an arbitrary DE-FG~\cite{Cao2017}, the existence of a combinatorial characterization of the Bethe partition function in terms of finite graph covers for DE-FGs remains an open question. Numerical results on small DE-FGs suggested that such a characterization might exist. However, because in general, the partition function of a DE-FG is a sum of complex numbers, proving the combinatorial characterization for DE-FGs turns out to be much more challenging than for S-FGs. The method of types, which was the main ingredient in the proof of the combinatorial characterization for S-FGs~\cite{Vontobel2013}, cannot be directly adapted to the proof for DE-FGs. The challenges of approximating the sum of complex numbers are illustrated by the following example.
\begin{example}
  For any positive-valued integer $ M $ and any complex number $ \alpha $, we define
  \begin{align}
    Z_{M} \defeq 
    \sum\limits_{\ell = 0}^{M} c_{M,\ell},
    \label{eqn: expression of ZM}
  \end{align}
  where
  \begin{align*}
    c_{M,\ell} = \binom{M}{ \ell } 
    \cdot ( 1 - \alpha )^{M - \ell}
    \cdot \alpha^{ \ell }.
  \end{align*}
  Using the binomial theorem, one can verify that $ Z_{M} = 1 $ for any positive-valued integer $ M $ and any complex number $ \alpha $.

  If $ \alpha $ is a non-negative number satisfying $ 0 \leq \alpha \leq 1 $, then $ c_{M,\ell} $ is non-negative for all $ \ell \in \{0,\ldots,M\} $, and we obtain
  \begin{align*}
    \max\limits_{ \ell \in \{0,\ldots,M\} } c_{M,\ell} 
    \leq Z_{M}
    \leq (M+1) \cdot \max\limits_{ \ell \in \{0,\ldots,M\} } c_{M,\ell}, 
  \end{align*}
  which implies
  \begin{align*}
    \frac{1}{M} \cdot 
    \log( \max\limits_{ \ell \in \{0,\ldots,M\} } c_{M,\ell} )
    &\leq 
    \frac{1}{M} \cdot \log(Z_{M})
    \nonumber\\
    &\leq 
    \frac{1}{M} 
    \cdot \log(M+1)
    +
    \frac{1}{M} 
    \cdot \log( \max\limits_{ \ell \in \{0,\ldots,M\} } c_{M,\ell} ).
  \end{align*}
  In the limit $ M \to \infty $, we obtain
  \begin{align}
    \lim\limits_{ M \to \infty }
    \frac{1}{M} \cdot \log(Z_{M})
    = 
    \lim\limits_{ M \to \infty }
    \frac{1}{M} \cdot 
    \log( \max\limits_{ \ell \in \{0,\ldots,M\} } c_{M,\ell} ).
    \label{eqn: limit of 1 M log ZM}
  \end{align}
  For simplicity of exposition, we suppose that there exists a positive-valued integer $ M $ such that $ \alpha \cdot M \in \sZpp $. Then for such an $ \alpha $, 
  we have
  \begin{align}
    \max\limits_{ \ell \in \{0,\ldots,M\} } c_{M,\ell}
    = c_{M,\alpha \cdot M}
    = 
    \hspace{-0.5 cm}
    \underbrace{
      \binom{M}{\alpha \cdot M} 
    }_{ \overset{(a)}{=} \exp( M \cdot h(\alpha) + o(M) ) }
    \hspace{-0.7 cm}
    \cdot
    \underbrace{ (1-\alpha)^{ M (1-\alpha) }
    \cdot \alpha^{ M \alpha }
    }_{ = \exp( -M \cdot h(\alpha) ) }
    = \exp( o(M) ),
    \label{eqn: Cml is a little o function wrt M}
  \end{align}
  where $ h(\alpha) \defeq -\alpha \cdot \log( \alpha ) 
  - (1-\alpha) \cdot \log( 1 - \alpha ) $,
  where $ o(M) $ is the usual little-$o$ notation for a function in $M$,
  and where step $(a)$ follows from Stirling's approximation (see, \eg,~\cite{Robbins1955}).
  Plugging~\eqref{eqn: Cml is a little o function wrt M} into~\eqref{eqn: limit of 1 M log ZM}, we obtain
  \begin{align*}
    \lim\limits_{ M \to \infty }
    \frac{1}{M} \cdot \log(Z_{M})
    =
    \lim\limits_{ M \to \infty }
    \frac{1}{M} \cdot 
    \log( \max\limits_{ \ell \in \{0,\ldots,M\} } c_{M,\ell} )
    =\lim\limits_{ M \to \infty }
    \frac{o(M)}{M} = 0.
  \end{align*}
  We can see that for the case $ 0 \leq \alpha \leq 1 $, all the terms in $ Z_{M} $ are non-negative and are added up constructively, allowing us to approximate $ Z_{M} $ by the term with the largest magnitude. Following a similar idea, Vontobel~\cite{Vontobel2013}
  used the method of types to characterize the valid configurations in finite graph covers of an S-FG, and he further gave a combinatorial characterization of the Bethe partition function of the S-FG.

  However, if $ \alpha $ is a complex-valued number with non-zero imaginary part, the terms in the expression of $ Z_{M} $ in~\eqref{eqn: expression of ZM}, might take values in the field of complex numbers, and these terms are added up constructively and destructively. In this case,
  we cannot approximate $ \lim\limits_{ M \to \infty } M^{-1} \cdot \log(Z_{M}) $ using the same idea as in the case $ 0 \leq \alpha \leq 1 $, as the term with the largest magnitude gives a poor estimate of $ Z_{M} $. 
\end{example}
% and the $M$-th power of $ (a-b) $ are given by
% \begin{align*}
%   (a+b)^{M} &= 
%   \sum\limits_{ k = 0 }^{M} \binom{ M }{ k } \cdot a^{k} \cdot b^{M-k}, 
%   \nonumber\\
%   (a-b)^{M} &= \sum\limits_{ k = 0 }^{M} \binom{ M }{ k } \cdot a^{k} \cdot (-b)^{M-k}, 
% \end{align*}

%  we need to develop a novel proof approach for the case
% of DE-FGs; this novel proof approach is the main result of this thesis besides
% the statement itself. (Currently, the characterization of a DE-FG's Bethe partition function in terms of its finite graph covers is for the strict-sense DE-FGs
% satisfying an (easily checkable) condition. However, based on the numerical
% results, we suspect that the characterization holds more
% broadly.)

\subsection{Main Contributions}

We develop a novel approach to give a combinatorial characterization of the Bethe partition functions for a certain class of DE-FGs. This novel proof approach is one of the main results in this thesis, besides the statement itself. 

% Currently, the characterization is for DE-FGs satisfying an (easily checkable) condition. However, based on the obtained numerical results, we suspect that the characterization holds more broadly.

%***************************************************************************

Our proof approach consists of three main ingredients:
\begin{enumerate}

  \item At the first step, we apply the loop-calculus transform (LCT)\footnote{The loop-calculus transform can be seen as a particular instance of a holographic
  transform~\cite{AlBashabsheh2011}.} to a DE-FG based on a
  fixed point of the SPA. (If the SPA has multiple fixed points,
  we assume that it is the fixed point with the largest Bethe partition
  function value.) The benefit of the LCT is that the
  transformed SPA fixed-point messages have an elementary form, which is independent of the edge and the direction of a message. Note that the LCT leaves the partition function of the DE-FG unchanged.

  \item At the second step, for any positive-valued integer $M$, we construct the average degree-$M$ cover for the transformed DE-FG. Then we express the average of the partition functions of the degree-$M$ covers in terms of the partition function of this average degree-$M$ cover. 

  \item At the third step, we evaluate the partition function of this average degree-$M$ cover in the limit $M \to \infty$.
  So far, we can relate the limit of this partition function to the Bethe partition function when some (easily checkable) conditions are satisfied. However, based on the obtained numerical results, we suspect that this combinatorial characterization holds more generally. We leave it as an open problem to broadly characterize when the DE-FGs do not satisfy the (easily checkable) conditions. 

\end{enumerate}

Based on finite graph covers of a DE-FG, we further develop the symmetric-subspace transform (SST), which allows one to express the partition function of the average degree-$M$ cover of the DE-FG in terms of some integral.
Both the LCT and the SST should be of interest beyond proving the main results in this thesis. The reasons are listed as follows.
%----------------------------------------------------------------------------
\begin{enumerate}
  \item The LCT introduced in this thesis is based on SPA fixed-point messages consisting of complex-valued components, which is applicable for both S-FGs and DE-FGs. This contrasts with the LCT proposed in~\cite{Mori2015}, which requires SPA fixed-point messages to be non-zero valued. Therefore, our proposed transform generalizes the one presented in~\cite{Mori2015}.
  
  \item The SST defined in this thesis is applicable for both S-FGs and DE-FGs. To the best of our knowledge, it is the first time that this transform is introduced in factor-graph literature, although it has wide applications 
  in quantum physics~\cite{Wood2015,Harrow2013}.
\end{enumerate}
%----------------------------------------------------------------------------

\section{Structure of this Thesis}

This thesis is structured as follows. Chapter~\ref{chapt: preliminaries} reviews the basic definitions and notations that are frequently used in this thesis. In particular, Chapter~\ref{chapt: preliminaries} introduces S-FGs, a special class of S-FGs such that the partition function of each S-FG is equal to the permanent of a non-negative square matrix, DE-FGs, and finite graph covers of these factor graphs. Chapter~\ref{chapt: finite graph cover based bound for matrix permanent} proves degree-$M$-Bethe-permanent-based bounds on the permanent of a non-negative square matrix and also similar results for another approximation of the permanent known as the (scaled) Sinkhorn permanent. Chapters~\ref{chapt: LCT} and~\ref{chapt:SST} introduce the LCT and the SST, respectively, for both S-FGs and DE-FGs. Chapter~\ref{chapt: graph cover thm for DENFG} proves the graph-cover theorem, \ie, a combinatorial characterization of the Bethe partition function in terms of finite graph covers, for a class of DE-FGs satisfying an easily checkable condition. Various (longer) proofs have been collected in the appendices.

%***************************************************************************

\chapter{Preliminaries}
\label{chapt: preliminaries}
\label{CHAPT: PRELIMINARIES}

In this chapter, we introduce the preliminaries of this thesis. We begin by introducing the basic definitions and notations used throughout this thesis in Section~\ref{sec: basic definition}. Next, we review the fundamental concepts of standard normal factor graphs (S-NFGs) in Section~\ref{sec:SNFG}. In Section~\ref{sec: snfg for permanent of nonnagtive square matrix}, we study a class of S-NFGs such that the partition function of each of these S-NFGs equals the permanent of a non-negative square matrix. Moving to the quantum case, we introduce double-edge normal factor graphs (DE-NFGs) in Section~\ref{sec: basic of denfg}. Finally, in Section~\ref{sec: finite graph covers}, we present the finite graph covers of the factor graphs introduced in this chapter.

\section{Basic Definition and Notations}
\label{sec: basic definition}

If not mentioned otherwise, all variable alphabets are assumed to be finite.
We will use, without essential loss of generality, normal
factor graphs (NFGs), \ie, factor graphs where variables are associated with
edges~\cite{Forney2001, Loeliger2004}.

Let $n$ be a positive integer and let $\mtheta \in \sR_{\geq 0}^{n \times
  n}$. We define $\suppmtheta$ to be the support of $\mtheta$, \ie,
\begin{align*}
  \suppmtheta 
    &\defeq
       \bigl\{ 
         (i,j) \in [n] \times [n] 
       \bigm|
         \theta(i,j) > 0
       \bigr\}.
\end{align*}

%***************************************************************************

Consider two finite subsets $ \set{I} $ and $ \set{J} $ of $[n]$ such that
$ |\set{I}| = | \set{J} | $. The set $ \setS_{\set{I} \to \set{J}} $ is
defined to be the set of all bijections from $ \set{I} $ to $ \set{J} $. In
particular, if $ \set{I} = \set{J} $, then the set
$ \setS_{\set{I} \to \set{J}} $ represents the set of all permutations of the
elements in $ \set{I} $, and we simply write $ \setS_{\set{I}} $ instead of
$\setS_{\set{I} \to \set{J}}$. Given
$ \sigma \in \setS_{\set{I} \to \set{J}} $, we define the matrix
$ \mP_{\sigma} \in \sZp^{\set{I} \times \set{J}} $ to be
\begin{align*}
  \mP_{\sigma} 
    &\defeq 
       \bigl(
         P_{\sigma}( i, j )
       \bigr)_{\! (i,j) \in \set{I} \times \set{J}}  , 
       \ \ \ 
  P_{\sigma}( i, j )
     \defeq 
       \begin{cases}
         1 & \sigma(i) = j \\
         0 & \text{otherwise}
     \end{cases}.
     \nonumber
       % \label{sec:1:eqn:209}
\end{align*}
Note that $\mP_{\sigma}$ is a permutation matrix, \ie, a matrix where all
entries are equal to $0$ except for exactly one $1$ per row and exactly one
$1$ per column. Moreover, given a set $ \setS_{\set{I} \to \set{J}} $, the set
$\setP_{\set{I} \to \set{J}}$ is defined to be the set of the associated
permutation matrices, \ie,
\begin{align}
  \setP_{\set{I} \to \set{J}}
    &\defeq
       \bigl\{ 
         \mP_{\sigma} 
       \bigm|
         \sigma \in \setS_{\set{I} \to \set{J}}
       \bigr\}. 
         \label{sec:1:eqn:181}
\end{align}
If $\set{I} = \set{J}$, we simply write $\setP_{\set{I}}$ instead of
$\setP_{\set{I} \to \set{J}}$. A particularly important special case for this
thesis is the choice $\set{I} = \set{J} = [n]$: the set $\setP_{[n]}$
represents the set of all permutation matrices of size $n \times n$ where the
rows and columns are indexed by $[n]$.

%***************************************************************************

We define $\Gamma_n$ to be the set of all doubly stochastic
matrices of size $n \times n$, \ie,
\begin{align*}
  \Gamma_n
    &\defeq
       \left\{ 
         \vgam
         = \bigl(
             \gamma(i,j)
           \bigr)_{\! i,j \in [n]}
      \ \middle| \ 
        \begin{array}{l}
          \gamma(i,j) \in \sR_{\geq 0}, \, \forall i,j \in [n] \\
          \sum\limits_{j \in [n]} \gamma(i,j) = 1, \, \forall i \in [n] \\
          \sum\limits_{i \in [n]} \gamma(i,j) = 1, \, \forall j \in [n]
        \end{array}
      \right\}.
\end{align*}
Moreover, for $M \in \sZpp$, we define $\Gamma_{M,n}$ to be the subset of
$\Gamma_{n}$ that contains all matrices where the entries are multiples of
$1/M$, \ie,
\begin{align*}
  \Gamma_{M,n} 
    &\defeq 
       \bigl\{ 
         \vgam \in \Gamma_n
       \bigm|
         M \cdot \gamma(i,j) \in \sZp, \ \forall i,j \in [n]
       \bigr\}.
\end{align*}
Observe that $\Gamma_{1,n} = \setP_{[n]}$.

%***************************************************************************

All logarithms in this thesis are natural logarithms and the value of
$0 \cdot \log(0)$ is defined to be $ 0 $.

%***************************************************************************

In this thesis, given an arbitrary $ n \in \sZpp $, if there is no ambiguity,
we use $ \prod\limits_{i,j} $, $ \prod\limits_{i} $, $ \prod\limits_{j} $, $ \sum\limits_{i,j} $,
$ \sum\limits_{i} $, and $ \sum\limits_{j} $ for $ \prod\limits_{i,j \in [n]} $, $ \prod\limits_{i \in [n]} $,
$ \prod\limits_{j \in [n]} $, $ \sum\limits_{i,j \in [n]} $, $ \sum\limits_{i \in [n]} $, and
$ \sum\limits_{j \in [n]} $, respectively.

\section{Standard Normal Factor Graphs (S-NFGs)}
\label{sec:SNFG}

%***************************************************************************

In this section, we review some basic concepts and properties of an S-NFG. We use an example to introduce the key concepts of an S-NFG.

\input{figures/snfg_example.tex}

\begin{example}\label{sec:SNFG:exp:1}

  Consider the multivariate function
  \begin{align*}
    g(x_{1},\ldots,x_{5}) \defeq f_{1}(x_{1},x_{2},x_{3}) \cdot f_{2}(x_{1},x_{4})
    \cdot f_{3}(x_{2},x_{5})\cdot f_{4}(x_{3},x_{4},x_{5}).
  \end{align*}
  where $g$, the
  so-called global function, is defined to be the product of the so-called
  local functions $f_{1}$, $f_{2}$, $f_{3}$ and $f_{4}$. We can visualize the factorization of $g$ with the help of the S-NFG $\graphN$ in
  Fig.~\ref{sec:SNFG:fig:1}. Note that the S-NFG $\graphN$ shown in
  Fig.~\ref{sec:SNFG:fig:1} consists of four function nodes
  $f_1, \ldots, f_4$ and five (full) edges with associated variables
  $x_1, \ldots, x_5$.
  \eexample
\end{example}

%***************************************************************************

%***************************************************************************

An edge that connects to two function nodes is called a full edge, whereas an edge
that is connected to only one function node is called a half edge. For the
purposes of this thesis, it is sufficient to consider only S-NFGs with full
edges, because S-NFGs with half edges can be turned into S-NFGs with only full
edges by attaching a dummy $1$-valued function node to the other end of every half edge without changing any marginals or the partition function.

%***************************************************************************

\begin{definition}
  \label{def: def of snfg}
  \index{Normal factor graph!S-NFG}
  An S-NFG $\graphN(\setF, \setEfull, \set{X})$ consists of the following
  objects:
  \begin{enumerate}
    
    \item The graph $(\setF,\setEfull)$ with vertex set $\setF$ and edge set
      $\setEfull$, where $\setF$ is also known as the set of function nodes. (As
      mentioned above, every edge $e \in \setEfull$ will be assumed to be a full
      edge connecting two function nodes.)

    \item The alphabet $\set{X} \defeq \prod\limits_{e \in \setEfull} \setxe$, where
      $\setxe$ is the alphabet of the variable $\xe$ associated with the edge
      $e \in \setEfull$. In this thesis, without loss of generality, we suppose that $ 0 \in \setxe$ for all $ e \in \setEfull $.

  \end{enumerate} 
  For a given S-NFG $\graphN(\setF,\setEfull,\set{X})$, we fix the order of the function nodes by $ \set{F} = \{ f_{1},\ldots,f_{|\set{F}|} \}. $ Then we make the following definitions.
  %------------------------------------------------------------------------
  % \begin{align*}
  %   \set{F} = \{ f_{1},\ldots,f_{|\set{F}|} \}.
  %   % \qquad
  %   % \setEfull = \bigl[ |\setEfull| \bigr] = \{0,1,\ldots,|\setEfull|-1\} .
  % \end{align*}
  %------------------------------------------------------------------------
  
  \begin{enumerate}
    \setcounter{enumi}{2}

  \item For every function node $f \in \setF$, the set $\setpf$ is defined to be the set of edges incident on $f$.

  \item For every edge $e = (f_{i}, f_{j}) \in \setEfull$, the pair $(f_{i}, f_{j})$ is defined to be the pair of function nodes that are connected to edge $e$ such that $ i < j $. The set $ \partial e $ is defined to be the set of function nodes incident on $e$, \ie, $ \partial e \defeq \{f_{i}, f_{j}\} $. Note that for notational convenience, here we impose a direction on every edge $(f_{i}, f_{j})$, \ie, the inequality $ i < j $, that is irrelevant for our results.

  \item For any finite set $ \set{I} \subseteq \setEfull $, we define
   %-----------------------------------------------------------------------
   \begin{align*}
         \setx{ \set{I} } \defeq 
         \prod\limits_{ e \in \set{I} } \setxe, \qquad 
         \vx_{\set{I}} \defeq (\xe)_{e \in \set{I}}
         \in \setx{ \set{I} }.
   \end{align*}
   %-----------------------------------------------------------------------

  \item An assignment $\vx \defeq (\xe)_{e \in \setEfull} \in \set{X}$ is called a
    configuration of the S-NFG $ \sfN $. For each $f \in \setF$, a configuration
    $\vx \in \set{X}$ induces the vector $\xf \defeq (x_{e})_{e \in \setpf} \in \setxf$.

  \item For every $f \in \setF$, the local function associated with $f$ is,
    with some slight abuse of notation, also called $f$. Here, the local
    function $f$ is an arbitrary mapping from $\setxf$ to
    $\sRp$.

  \item The global function $g$ is defined to be the mapping
   \begin{align*}
    g:\set{X} &\to \sRp, \qquad
    \vx \mapsto \prod\limits_{f \in \setF} f(\xf).
   \end{align*}

  \item A configuration $\vx \in \set{X}$ satisfying $g(\vx) \neq 0$ is called a
    valid configuration. The set of valid configurations is defined to be
    \begin{align*}
      \set{C} \defeq \{ \vx \in \set{X} \ | \ g(\vx) \neq 0 \}.
    \end{align*}

  \item The partition function is defined to be
    $Z(\graphN) \defeq \sum\limits_{\vx \in \set{C}}g(\vx)$.\footnote{In this thesis, the
      partition function $Z(\graphN)$ of $\graphN$ is a scalar, \ie, it
      is not really a function. If $\graphN$ depends on some parameter
      (say, some temperature parameter), then $Z(\graphN)$ is a function of
      that parameter.} Clearly, the partition function satisfies
    $Z(\graphN) \in \sRp$.
    \label{partition function of snfg}
  \end{enumerate}
  \edefinition
\end{definition}

%***************************************************************************

In this thesis, we will only consider the S-NFG $\graphN$ for which
$Z(\graphN) \in \sRpp$.

%***************************************************************************

If there is no ambiguity, when we consider S-NFG, we will use the short-hands
$\sum\limits_{\vx}$, 
$\sum\limits_{\xe}$, 
$\sum\limits_{e}$, 
$\sum\limits_{f}$, 
$\prod\limits_{\vx}$,  
$\prod\limits_{\xe}$,
$\prod\limits_{e}$, 
$\prod\limits_{f}$,
$ (\cdot)_{\xe} $,
and
$ (\cdot)_{e} $ for
$\sum\limits_{\vx \in \set{X}}$, 
$\sum\limits_{\xe \in \setxe}$, 
$\sum\limits_{e \in \setEfull}$, 
$\sum\limits_{f \in \setF}$,
$\prod\limits_{\vx \in \set{X}}$, 
$\prod\limits_{\xe \in \setxe}$, 
$\prod\limits_{e \in \setEfull}$, 
$\prod\limits_{f \in \setF}$, 
$ (\cdot)_{\xe \in \setxe} $,
and 
$ (\cdot)_{e\in \setEfull} $,
respectively. 
For any set $ \set{I} \subseteq \setEfull $, we will use the short-hands 
$\sum\limits_{\vx_{\set{I}}}$,
$\prod\limits_{\vx_{\set{I}}}$, 
and $ (\cdot)_{\vx_{\set{I}}} $
for 
$\sum\limits_{ \vx_{\set{I}} \in \setx{\set{I}} }$,
$\prod\limits_{ \vx_{\set{I}} \in \setx{\set{I}} }$, 
and $ (\cdot)_{ \vx_{\set{I}} \in \setx{\set{I}} } $.
Moreover, $\setpf \setminus e$ will be short-hand notation for
$\setpf \setminus \{ e \}$.

Now we introduce the Gibbs free energy function, which provides an alternative expression of the partition function.
%----------------------------------------------------------------------------
\begin{definition}
  (The Gibbs free energy function)
  \index{Gibbs free energy function!for S-NFG}
  Consider some S-NFG $\sfN$ and $ \bm{p} \defeq \bigl( p(\vx) \bigr)_{\vx \in \set{C}} \in \Pi_{\set{C}} $.
  The Gibbs free energy function is defined to be the mapping
  (see, \eg, \cite[Definition 9]{Vontobel2013})
  \begin{align*}
    F_{\mathrm{G}}: \Pi_{\set{C}} &\to \sR, \nonumber\\
    \qquad \bm{p} &\mapsto U_{\mathrm{G}}(\bm{p}) - H_{\mathrm{G}}(\bm{p}),
  \end{align*}
  where
  \begin{alignat*}{3}
    U_{\mathrm{G}}: \Pi_{\set{C}} &\to \sR, \qquad 
    &\bm{p} &\mapsto -\sum\limits_{ \vx \in \set{C} } 
    p(\vx) \cdot \log\bigl( g(\vx) \bigr),
    \nonumber\\
    H_{\mathrm{G}}: \Pi_{\set{C}} &\to \sR, \qquad 
    &\bm{p} &\mapsto -\sum\limits_{ \vx \in \set{C} } 
    p(\vx) \cdot \log\bigl( p(\vx) \bigr).
  \end{alignat*}
  \edefinition
\end{definition}
%----------------------------------------------------------------------------
%----------------------------------------------------------------------------
\begin{proposition}
  \label{prop: minimum of FG relates to Z}
  It holds that $  F_{\mathrm{G}}(\bm{p}) $ is a convex function w.r.t. 
  $ \bm{p} \in \Pi_{\set{C}} $ and
  \begin{align*}
    Z(\sfN) = \exp( - \min_{ \bm{p} \in \Pi_{\set{C}} } F_{\mathrm{G}}(\bm{p}) ).
  \end{align*}
\end{proposition}
%----------------------------------------------------------------------------
%----------------------------------------------------------------------------
\begin{proof}
  See, \eg,~\cite[Lemma 8]{Vontobel2013}.
\end{proof}
%----------------------------------------------------------------------------

In general, neither the direct evaluation of the partition function nor solving the convex optimization problem in Proposition~\ref{prop: minimum of FG relates to Z} leads to a tractable computational problem, except for small-scaled S-NFG. 
This limitation motivates the study of alternative approaches, specifically approximations of the partition function and the Gibbs free energy function.

In the following, we introduce the SPA, whose fixed points often provide decent approximation to the partition function. Note that we only give the technical details of the SPA on an S-NFG; for a general discussion with respect to (w.r.t.) the motivations behind the SPA, see, \eg,~\cite{Kschischang2001,Loeliger2004}.
\begin{definition}\label{sec:SNFG:def:4}
  (The SPA on S-NFG)
  \index{SPA!on S-NFG}
  Consider some S-NFG $\graphN$. The SPA is an iterative algorithm that
  sends messages along edges, where messages are functions over the alphabet
  associated with an edge. (Note that for every iteration and every edge, two
  messages are sent along that edge, one in both directions.) More precisely,
  the SPA consists of the following steps:
  \begin{enumerate}

  \item (Initialization) For every $e \in \setEfull$, $f \in \setF$, the
    message $\mu_{\etof}^{(0)}: \setxe \to \sR_{\geq 0} $, is defined to be some
    arbitrary function. (Typically, $\mu_{\etof}^{(0)}(\xe) \defeq 1 / |\setxe|$
    for all $\xe \in \setxe$.)

  \item (Iteration) For $t = 1, 2, 3, \ldots$, do the following calculations
    until some termination criterion is met:\footnote{The termination
      criterion is typically a combination of numerical convergence and an
      upper bound on the number of iterations.}
    \begin{enumerate}
        \item  For every $e = (f_{i}, f_{j}) \in \setEfull$, we define the scaling factor $ \kappa_{\efj}^{(t)} $ to be
        %--------------------------------------------------------------------
        \begin{align*}
            \kappa_{\efj}^{(t)}
            &\defeq
            \sum\limits_{\xfi}
            f_{i}(\xfi)
            \cdot
            \prod\limits_{e' \in \setpfi \setminus e}
            \mu_{\epfi}^{(t-1)}(x_{e'}).
            % \nonumber\\
            % \kappa_{\efi}^{(t)} &\defeq
            % \sum\limits_{\xfj}
            % f_{j}(\xfj)
            % \cdot
            % \prod\limits_{e' \in \setpfj \setminus e}
            % \mu_{\epfj}^{(t-1)}(x_{e'}).
        \end{align*}
        %--------------------------------------------------------------------
        One can obtain the definition of $ \kappa_{\efi}^{(t)} $ by switching $ f_{i} $ with $ f_{j} $ in the above definition.
        
        \item If 
        $ 
            \prod\limits_{e \in \setEfull} 
            \prod\limits_{f \in \setpe}
            \kappa_{\ef}^{(t)} = 0 
        $,
        we randomly generate each element in $ \vmu^{(t)} $ following uniform distribution in $ [0,1] $.

        \item  If 
        $ 
            \prod\limits_{e \in \setEfull} 
            \prod\limits_{f \in \setpe}
            \kappa_{\ef}^{(t)} \neq 0 
        $, for every $e \in \setEfull$, 
        we define the messages from the edge $ e = (f_{i}, f_{j}) $ to the function node 
        $ f_{j} $ via the mapping:\footnote{If desired or necessary, other generalization procedures is allowed.}
        \begin{align}
          \mu_{\efj}^{(t)}:\, \set{X}_{e} &\to \sR_{\geq 0},
          \nonumber\\
          z_{e} &\mapsto 
          \bigl( \kappa_{\efj}^{(t)} \bigr)^{\! -1}
          \cdot
          \sum\limits_{\xfi: \, x_{e} = z_{e}}
            f_{i}(\xfi)
            \cdot
            \prod\limits_{e' \in \setpfi \setminus e}
              \mu_{\epfi}^{(t-1)}(x_{e'}). \label{sec:SNFG:eqn:8}
        % \\
        % \mu_{\efj}^{(t)}:\, \set{X}_{e} &\to \sR_{\geq 0}, \quad
        %   z_{e} &\mapsto&
        %   \bigl( \kappa_{\efi}^{(t)} \bigr)^{\! -1}
        %   \cdot
        %   \sum\limits_{\xfj: \, x_{e} = z_{e}}
        %     f_{j}(\xfj)
        %     \cdot
        %     \prod\limits_{e' \in \setpfj \setminus e}
        %       \mu_{\epfj}^{(t-1)}(x_{e'}). \label{sec:SNFG:eqn:7}
        \end{align}
        One can obtain the definition of $ \mu_{\efi}^{(t)} $ by switching $ f_{i} $ with $ f_{j} $ in~\eqref{sec:SNFG:eqn:8}.

    \item For every $f \in \setF$, we define
    %------------------------------------------------------------------------
    \begin{align*}
        \kappa_f^{(t)} \defeq 
            \sum\limits_{ \xf }
            f(\xf) \cdot \prod\limits_{e \in \setpf} \mu_{\etof}^{(t)}(\xe).
    \end{align*}
    %------------------------------------------------------------------------
    If $ \kappa_f^{(t)} \neq 0 $, then we define the belief function at function node $f$ to be the mapping
      \begin{align}
        \beli_f^{(t)}: \setxf \to \sR_{\geq 0}, \quad
          \xf
          &\mapsto 
            \bigl( \kappa_f^{(t)} \bigr)^{\! -1}
              \cdot
              f(\xf)
              \cdot
              \prod\limits_{e \in \setpf}
                \mu_{\etof}^{(t)}(\xe).
                    \label{sec:SNFG:eqn:1}
      \end{align}
      The associated vector is defined to be
      $ \vbeli_{f}^{(t)} \defeq \bigl( \beli_f^{(t)}(\xf) \bigr)_{ \! \xf \in \set{X}_{\setpf} } $.
    
    \item For every $e \in \setEfull$, we define
    %------------------------------------------------------------------------
    \begin{align*}
        \kappa_e^{(t)} \defeq \sum\limits_{\xe} 
        \prod\limits_{f \in \setpe}
        \mu_{\ef}^{(t)}(\xe).
    \end{align*}
    %------------------------------------------------------------------------
    If $  \kappa_e^{(t)}\neq 0 $, then we define the belief function at edge $e$ to be the mapping
      \begin{align}
        \beli_e^{(t)}: \setxe \to \sR_{\geq 0}, \quad
        \xe
        &\mapsto
        \bigl( \kappa_e^{(t)} \bigr)^{-1}
        \cdot
        \mu_{\efi}^{(t)}(\xe)
        \cdot
        \mu_{\efj}^{(t)}(\xe).
        \label{sec:SNFG:eqn:2}
      \end{align}
      The vector $ \vbeli_e^{(t)} $ is defined to be $ \vbeli_e^{(t)} \defeq \bigl( \beli_e^{(t)}(\xe) \bigr)_{ \! \xe \in \setxe } $.
    \end{enumerate}
    
    \item At iteration $t \in \sZpp$, the SPA messages can be collected as the SPA message vector:
    %--------------------------------------------------------------------
    \begin{align*}
      \vmu^{(t)} \defeq 
      \bigl( \vmu_{\etof}^{(t)} \bigr)_{\! e \in \setpf,\, f \in \setF},
    \end{align*}
    %--------------------------------------------------------------------
    where $ \vmu_{\ef}^{(t)} \defeq \bigl( \mu_{\ef}^{(t)}(\xe) \bigr)_{ \! \xe } $. 
    With this, the SPA message update rules can be written as
    %--------------------------------------------------------------------
    \begin{align*}
      \vmu^{(t)} = \fSPAN\bigl( \vmu^{(t-1)} \bigr)
    \end{align*}
    %--------------------------------------------------------------------
    for some suitably defined function $\fSPAN$. 

    \item The beliefs at function nodes and edges can be collected as the following vector:
    %-----------------------------------------------------------------------
    \begin{align*}
       \vbeli^{(t)} \defeq \Bigl( 
       ( \vbeli_f^{(t)} )_{f \in \setF}, \,
      ( \vbeli_e^{(t)} )_{e \in \setE} \Bigr).
    \end{align*}
    %-----------------------------------------------------------------------

    \item \index{SPA fixed point!for S-NFG} 
    An SPA message vector
    %--------------------------------------------------------------------
    \begin{align}
      \vmu \defeq ( \vmu_{\etof} )_{e \in \setpf,\, f \in \setF}
      \label{sec:SNFG:eqn:10}
    \end{align}
    %--------------------------------------------------------------------
    with 
    $ \vmu_{\etof} \defeq \bigl( \mu_{\etof}(\xe) \bigr)_{\! \xe} $
    is called an SPA fixed-point message vector if 
    %--------------------------------------------------------------------
    \begin{align}
      \vmu = \fSPAN(\vmu). \label{sec:SNFG:eqn:11}
    \end{align}
    %-------------------------------------------------------------------- 
    % where $ \mu_{\etof} : \setxe \to \sR $ is an arbitrary mapping and $ \vmu_{\etof} = ( \mu_{\etof}(\xe) )_{\xe} $.

    % \item The beliefs evaluated at an SPA message vector $ \vmu $ are given by 
    % %-----------------------------------------------------------------------
    % \begin{align}
    %     \beli_{f,\vmu}(\xf)
    %     &\defeq \kappa_f
    %       \cdot
    %       f(\xf)
    %       \cdot
    %       \prod\limits_{e \in \setpf}
    %         \mu_{\etof}(\xe),
    %           \quad \xf \in \setxf,\, \kappa_f \in \sR_{>0}, 
    %     \label{sec:DENFG:eqn:12} \\
    %     \beli_{e,\vmu}(\xe)
    %     &\defeq \kappa_e
    %     \cdot
    %     \mu_{\efi}(\xe)
    %     \cdot
    %     \mu_{\efj}(\xe),
    %     \quad \xe \in \setxe,\, \kappa_e \in \sR_{>0}, 
    %     \label{sec:DENFG:eqn:13}.
    % \end{align}
    % %-----------------------------------------------------------------------

  \end{enumerate}
  (If desired or necessary, other normalization procedures for messages and
  beliefs can be specified.)
  \edefinition
\end{definition}

%***************************************************************************

\begin{definition}
  \label{def:belief at SPA fixed point for S-NFG}
  Consider some S-NFG $\graphN$ and let
  $\vmu $ be some SPA fixed-point
  message vector for $\graphN$ as defined in~\eqref{sec:SNFG:eqn:10} and~\eqref{sec:SNFG:eqn:11}. The collection of the beliefs 
  %-----------------------------------------------------------------------
  \begin{align}
     \vbeli \defeq \Bigl( ( \vbeli_f )_{f \in \setF}, \,
    ( \vbeli_e )_{e \in \setE} \Bigr)
    \label{eqn: def of beliefs}
  \end{align}
  %-----------------------------------------------------------------------
  induced by the SPA fixed-point messages is computed as follows with 
  $ \vmu_{\etof}^{(t)} = \vmu_{\etof} $ for all $ e \in \setEfull $ and $ f \in \setF $:
  %----------------------------------------------------------------------------
  \begin{enumerate}
    \item for every $f \in \setF$, if 
    \begin{align*}
        \sum\limits_{ \xf }
        f(\xf) \cdot \prod\limits_{e \in \setpf} \mu_{\etof}(\xe)
        >0, 
        %\label{eqn: positive of belief at f}
    \end{align*}
      then the elements in the vector $\vbeli_f$ are 
      obtained by letting $ \vbeli_{f} = \vbeli_{f}^{(t)} $, where the elements in $ \vbeli_{f}^{(t)} $ are given in~\eqref{sec:SNFG:eqn:1};

    \item for every $e \in \setEfull$, if
    \begin{align}
        \sum\limits_{\xe} 
        \prod\limits_{f \in \setpe}
        \mu_{\ef}(\xe) > 0,
        \label{eqn: positive of belief at e}
     \end{align} 
    then the elements in the vector $\vbeli_e$ 
    are obtained by letting $ \vbeli_{e} = \vbeli_{e}^{(t)} $, where the elements in $ \vbeli_{e}^{(t)} $ are given in~\eqref{sec:SNFG:eqn:2}.
  \end{enumerate}
  %----------------------------------------------------------------------------
  \edefinition
\end{definition}

%**************************************************************************

%***************************************************************************

%***************************************************************************

\begin{assumption}
  \label{asmp: assume messages are non-negative}
  In the context of S-NFG, if not mentioned otherwise, we will only consider the following normalization procedure for the SPA 
  message vectors $\vmu^{(t)} $ and the collection of the beliefs $ \vbeli^{(t)} $:
  \begin{align*}
    \vmu_{\etof}^{(t)} &\in \Pi_{\setxe}, \qquad e \in \setpf,\, f \in \setF,
    \nonumber\\
    \vbeli^{(t)}_{e} &\in \Pi_{\setxe}, \qquad e \in \setEfull, \nonumber\\
    \vbeli^{(t)}_{f} &\in \Pi_{\setxf}, \qquad f \in \setF.
  \end{align*}
  
  Furthermore, we only consider the SPA fixed-point message vector $ \vmu $ such that
  the strict inequality in~\eqref{eqn: positive of belief at e} 
  holds for all $ e \in \setEfull $.
  % %----------------------------------------------------------------------------
  % \begin{enumerate}

  %   \item the strict inequality in~\eqref{eqn: positive of belief at f}
  %    holds for all $ f \in \setF $;

  %   \item 
  % \end{enumerate}
  % %----------------------------------------------------------------------------
  (See Item~\ref{sec:SNFG:remk:1:item:2} in Remark~\ref{sec:SNFG:remk:1} for a non-trivial S-NFG with an SPA fixed-point message vector such that the inequality in~\eqref{eqn: positive of belief at e} does not satisfy for some $ e \in \setEfull $.)
  \eassumption
\end{assumption}

%***************************************************************************

%----------------------------------------------------------------------------
\begin{proposition}\hspace{-0.1 cm}\cite[Proposition 5]{Yedidia2005}
  For the S-NFG $ \sfN $ such that $ f(\xf) \in \sR_{>0} $ for all $ \xf \in \setx{\setpf} $ and $ f \in \setF $, each SPA fixed-point message vector $ \vmu $ satisfies the inequality~\eqref{eqn: positive of belief at e} for all $ f \in \setF $ and the inequality~\eqref{eqn: positive of belief at e} for all $ e \in \setEfull $.
  \eproposition
\end{proposition}
%----------------------------------------------------------------------------

%***************************************************************************

% %----------------------------------------------------------------------------
% \begin{proposition}
%   Consider running the SPA defined in Definition~\ref{sec:SNFG:def:4} on some S-NFG $\graphN$. 
%    It holds that $\vmu_{\etof}^{(t)} \in \sRp^{|\setxe|}$ for all $e \in \setpf$,
%       $f \in \setF$, and $t \geq 1$. Moreover, it holds that 
%       $\vbeli_f^{(t)} \in \Pi_{\setxf}$ for all
%       $f \in \setF$ and $\vbeli_e^{(t)} \in \Pi_{\setxe}$ for all $e \in \setEfull$.
% \end{proposition}
% %----------------------------------------------------------------------------
% %----------------------------------------------------------------------------
% \begin{proof}
%    This is a straightforward consequence of Assumption~\ref{asmp: assume messages are non-negative}, the SPA message update rules in Definition~\ref{sec:SNFG:def:4} and the fact that the local functions of an S-NFG take on only non-negative real values.
% \end{proof}
% %----------------------------------------------------------------------------

Now we introduce the Bethe free energy function, which gives an S-FG-based approximation to the Gibbs free energy function. Proposition~\ref{prop: minimum of FG relates to Z} states that the value of the partition function can be evaluated by the minimum of the Gibbs free energy function. Therefore, the minimum of the Bethe free energy function yields an approximation to the partition function. The following definitions are adapted from~\cite[Section V]{Yedidia2005} to S-NFGs.

%***************************************************************************

\begin{definition}\label{sec:SNFG:def:3}
  \!\cite[Section V]{Yedidia2005}
  \index{Local marginal polytope!for S-NFG}
  Let $\graphN$ be some S-NFG. Recall that 
  the collection of the beliefs $  \vbeli $ is defined in~\eqref{eqn: def of beliefs}.
  The local marginal polytope $\LMP(\graphN)$
  associated with $\graphN$ is defined to be the following set:
  \begin{align*}
    \LMP(\graphN)
      &\defeq \left\{ 
        \vbeli
      \ \middle| \ 
        \begin{array}{c}
          \vbeli_f \in \Pi_{\setxf} \ \text{for all} \ f \in \setF \\
          \vbeli_e \in \Pi_{\setxe} \ \text{for all} \ e \in \setE \\[0.10cm]
          \sum\limits_{\xf: \, \xe = \ze}
            \beli_f(\xf)
              = \beli_e(\ze) \\
                  \ \text{for all} \ 
                  e \in \setpf, f \in \setF, \ze \in \setxe
        \end{array}
      \right\}. 
      % \label{sec:SNFG:eqn:12}
  \end{align*}

  \edefinition
\end{definition}

%***************************************************************************

\begin{definition}
  \label{sec:SNFG:def:1}
  \index{Bethe free energy function for S-NFG}
  \!\cite[Section V]{Yedidia2005}
  Let $\graphN$ be some S-NFG. The Bethe free energy function
  is defined to be the following mapping:
  \begin{align*}
    \FB(\vbeli): \LMP(\graphN) &\to \sR, \quad
      \vbeli \mapsto  U_{\mathrm{B}}(\vbeli) -  H_{\mathrm{B}}(\vbeli),
                 % \label{sec:SNFG:eqn:3}
  \end{align*}
  where
  %------------------------------------------------------------------------
  \begin{align*}
    U_{\mathrm{B}}(\vbeli): \LMP(\graphN) &\to \sR,\quad
    \vbeli \mapsto 
    \sum\limits_{f} U_{\mathrm{B},f}(\vbeli_{f}), 
    \nonumber\\
    H_{\mathrm{B}}(\vbeli): \LMP(\graphN) &\to \sR,\quad
    \vbeli \mapsto \sum\limits_{f}
    H_{\mathrm{B},f}(\vbeli_{f})
    - \sum\limits_{e} H_{\mathrm{B},e}(\vbeli_{e}), 
    % \label{sec:SNFG:eqn:9}
  \end{align*}
  %------------------------------------------------------------------------
  with
  %------------------------------------------------------------------------
  \begin{align*}
    U_{\mathrm{B},f}(\vbeli_{f}): \Pi_{\setxf} \to \sR,\quad
    &\vbeli_{f} \mapsto
    -\sum\limits_{\xf} 
      \beli_f(\xf)
      \cdot
        \log\bigl( f(\xf) \bigr), \nonumber\\
    H_{\mathrm{B},f}(\vbeli_{f}): \Pi_{\setxf} \to \sR,\quad
    &\vbeli_{f} \mapsto
    -\sum\limits_{\xf} 
      \beli_f(\xf)
      \cdot
        \log\bigl( \beli_f(\xf) \bigr), \nonumber\\
    H_{\mathrm{B},e}(\vbeli_{e}): \Pi_{\setxe} \to \sR,\quad
    & \vbeli_{e} \mapsto -\sum\limits_{\xe} 
      \beli_{e}(\xe)
      \cdot
      \log\bigl( \beli_{e}(\xe) \bigr).
  \end{align*}
  %------------------------------------------------------------------------
  Here, $ U_{\mathrm{B}} $ is the Bethe average energy function and $ H_{\mathrm{B}} $ is the Bethe entropy function.
  (Note that for NFGs with full and half edges, the sum in $ \sum\limits_{e} H_{\mathrm{B},e}(\vbeli_{e}) $ is only over all full
  edges.)  The Bethe approximation of the partition function, henceforth
  called the Bethe partition function, is defined to be
  \begin{align}
    \ZB(\graphN)
      &\defeq 
      \exp
            \left( 
              -
              \min_{\vbeli \in \LMP(\graphN)}
                \FB(\vbeli)
            \right).
              \label{sec:SNFG:eqn:4}
  \end{align}
  \index{Bethe partition function for S-NFG}This definition is often viewed as the primal formulation of the Bethe partition function.
  \edefinition
\end{definition}

% \subsection{Cycle-free S-NFG}
% %***************************************************************************

% \begin{proof}
%   The statement in Item~\ref{sec:SNFG:lem:1:item:1} is a straightforward consequence of Assumption~\ref{asmp: assume messages are non-negative}, the SPA message update rules in Definition~\ref{sec:SNFG:def:4} and the fact that the local functions of an S-NFG take on only non-negative real values.

%   The statement in Item~\ref{sec:SNFG:lem:1:item:2} is a consequence of the SPA message update rules in Definition~\ref{sec:SNFG:def:4} and the definition of the SPA fixed-point message vector.
% \end{proof}

%***************************************************************************
% In this section, we present properties for a cycle-free S-NFG.
%----------------------------------------------------------------------------
\begin{proposition}
  \label{prop: property of the SPA and ZB for cycle-free SNFG}
  Consider $\graphN$ to be cycle-free, \ie, tree-structured.
  %----------------------------------------------------------------------------
  \begin{enumerate}
    \item\label{item: cycle-free SNFG ZB equals Z}\cite[Proposition 3]{Yedidia2005} It holds that $ \ZB(\graphN) = Z(\graphN) $.

    \item (See, \eg,~\cite[Section 2.5]{JWainwright2008}.) The SPA on $ \graphN $ converges to a unique fixed point and the beliefs evaluated by the SPA fixed-point messages, as defined in Definition~\ref{def:belief at SPA fixed point for S-NFG}, is at the location of the minimum of the Bethe free energy function, \ie, the value of the Bethe partition function can be evaluated by the fixed-point messages. 
  \end{enumerate}
  %----------------------------------------------------------------------------
  By these two properties, we know that for a cycle-free S-NFG $\graphN$, 
  the value of the partition function can be evaluated by the SPA fixed-point messages.
  \eproposition
\end{proposition}
%----------------------------------------------------------------------------

\begin{remark}
  Consider some S-NFG $ \sfN $ with cycles. 
  Suppose that we get an SPA fixed-point message vector $ \vmu $ and obtain a collection of the beliefs $ \vbeli $ based on $ \vmu $, following Definition~\ref{def:belief at SPA fixed point for S-NFG}. 
  One can verify that 
  \begin{align*}
    \vbeli \in \LMP(\graphN).
  \end{align*}
  Therefore, the value of the function $ \exp \bigl( -\FB(\vbeli) \bigr) $, which is obtained by $ \vmu $, can be viewed as an approximation of the value of the Bethe partition function.
\end{remark}

For an S-NFG with cycles, the behavior of the SPA varies and the Bethe partition function provides an approximation of the partition function. One of the main results in the paper by Yedidia \etal~\cite[Section V]{Yedidia2005} is a characterization of the stationary points of the Bethe free energy function in terms of the SPA fixed-point messages. 
%---------------------------------------------------------------------------
\begin{theorem}
  \label{thm: Yedidia results on SPA and FB}
  Consider some S-NFG $ \sfN $.
  % %-----------------------------------------------------------------------
  \begin{enumerate}

    \item\hspace{-0.01cm}\cite[Theorem 2]{Yedidia2005} If the stationary point of the Bethe free energy function is in the interior of the LMP, \ie, 
    \begin{align*}
      \vbeli \in \LMP(\graphN),\quad \vbeli_{f} \in \sR_{>0}^{|\setx{\setpf}|}, \qquad 
      f \in \setF,
    \end{align*}
    then this stationary point corresponds to an SPA fixed point of $ \sfN $.

    \item\hspace{-0.01cm}\cite[Theorem 3]{Yedidia2005} If all the local functions in $ \sfN $ are positive-valued, \ie,  
    \begin{align*}
        f(\xf) \in \sR_{>0}, \qquad
        \xf \in \setxf,\, f \in \setF ,
    \end{align*}
    then each of the local minima of the Bethe free energy function corresponds to an SPA fixed point of $ \graphN $. 
  \end{enumerate}
  %-----------------------------------------------------------------------
  \etheorem
\end{theorem}
%---------------------------------------------------------------------------

For S-NFG $ \sfN $ with cycles, 
the SPA on $ \sfN $ might converge to a fixed point, or it might fail to converge.
The Bethe partition function $\ZB(\graphN)$ can be smaller,
equal, or larger than the partition function $\ZB(\graphN)$.
%***************************************************************************

%---------------------------------------------------------------------------
\begin{enumerate}

    \item In~\cite[Corollary 4.2]{Ruozzi2012}, Ruozzi proved that the Bethe partition function lower bounds the partition function for S-NFGs whose global functions are a binary and log-supermodular function. In subsequent work~\cite{Ruozzi2013}, Ruozzi extended this lower bound to certain classes of S-NFGs whose global functions are not necessarily binary or log-supermodular, including S-NFGs for ferromagnetic Potts models with a uniform external field~\cite[Theorems 5.2 and 5.4]{Ruozzi2013} and its generalizations~\cite[Theorem 5.6]{Ruozzi2013}, as well as a class of S-NFGs related to the problem of counting weighted graph homomorphisms.~\cite[Theorem 6.2]{Ruozzi2013}.

    \item In~\cite{Vontobel2013a}, Vontobel investigated the S-NFG whose partition function equals the permanent of a non-negative square matrix. For such S-NFG, he showed that the Bethe free energy function is a convex function and the SPA finds the minimum of the Bethe free energy function exponentially fast. Huang and Vontobel~\cite{Huang2024} extended Vontobel's results to a broader class of bipartite S-FGs where each local function is defined based on a (possibly different) multi-affine homogeneous real stable (MAHRS) polynomial. They further showed that a certain projection of the local marginal polytope equals the convex hull of the set of valid configurations, the Bethe free energy function possesses some convexity properties, and, for the typical case where the S-FG has an SPA fixed point consisting of positive-valued messages only, the SPA finds the value of the Bethe partition function exponentially fast.

    \item Gurvits in~\cite[Theorem 2.2]{Gurvits2011} proved that for the S-NFG studied in~\cite{Vontobel2013a}, the Bethe partition function lower bounds the partition function,
    \ie, he proved the lower bound in~\eqref{eq:ratio:permanent:bethe:permanent:1}. Subsequent work by Straszak and Vishnoi~\cite[Theorem 3.2]{Straszak2019} and Anari and Gharan~\cite[Theorem 1.1]{Anari2021} extended this lower bound to the S-NFGs where each local function corresponds to a (possibly different) multi-affine real stable polynomial. (Note that they did not assume homogenity.) 
    The class of the S-NFGs considered in~\cite{Straszak2019,Anari2021} is a generalization of the class of the S-NFGs studied in~\cite{Vontobel2013a,Huang2024}.
    
\end{enumerate}
%---------------------------------------------------------------------------

%***************************************************************************
Besides bounding the partition function using the Bethe partition function, we can also connect these two functions via finite graph covers.  
Unlike the analytical definition of the Bethe partition function in~\eqref{sec:SNFG:eqn:4}, Vontobel in~\cite{Vontobel2013} provided a combinatorial characterization of the Bethe partition function in terms of finite graph covers. A key step in achieving this characterization was to use the method of types to characterize valid configurations in $M$-covers. This characterization is illustrated in Fig.~\ref{fig: combinatorial chara for sfg}.
For $ M$ being one, the degree-$M$ Bethe partition function is equal to the partition function. As $ M $ grows to infinity, the limit superior of the degree-$M$ Bethe partition function is equal to the Bethe partition function. Section~\ref{sec:GraCov} provides technical details about $M$-covers and  the degree-$M$ Bethe partition function. 

%***************************************************************************

Given that the Bethe partition function is closely related to the SPA, in the following, we give a pseudo-dual formulation of the Bethe partition function based on the SPA fixed-point message vectors.

%***************************************************************************

\begin{definition}(The pseudo-dual formulation of the Bethe partition function for S-NFGs)
  \label{sec:SNFG:def:2}
  \index{Pseudo-dual Bethe partition function!for S-NFG}
  Consider some S-NFG $\graphN$ and let
  $\vmu $ be some SPA
  message vector for $\graphN$ as defined in~\eqref{sec:SNFG:eqn:10}.\footnote{The SPA message vector $\vmu$ does
    not necessarily have to be an SPA fixed-point message vector.} 
  If 
  \begin{align}
    Z_e(\vmu) > 0, \qquad e \in \setEfull,
    \label{eqn:positivity of SNFG Ze}
  \end{align}
  then we define the
  $\vmu$-based Bethe partition function to be
  \begin{align}
    \ZBSPA(\graphN,\vmu)
    &\defeq 
    \frac{\prod_f Z_f(\vmu)
    }{\prod\limits_{e} Z_e(\vmu)},
    \label{sec:SNFG:eqn:5}
  \end{align}
  where
  \begin{alignat}{2}
    Z_f(\vmu)
    &\defeq
      \sum\limits_{\xf}
        f(\xf)
          \cdot
          \prod\limits_{e \in \setpf}
            \mu_{\etof}(\xe),
            \ 
            \qquad f \in \setF,
                \nonumber \\
    Z_e(\vmu)
      &\defeq
        \sum\limits_{\xe}
          \prod\limits_{f \in \setpe}
          \mu_{\ef}(\xe),
          \qquad e \in \setEfull.
              \label{sec:SNFG:eqn:6}
  \end{alignat}
  (For NFGs with full and half edges, the product in the denominator on the
  right-hand side of~\eqref{sec:SNFG:eqn:5} is only over all full
  edges. Moreover, note that the expression for $\ZBSPA(\graphN,\vmu)$ is scaling
  invariant, \ie, if any of the messages $\mu_{\ef}$ is multiplied by a
  non-zero complex number then $\ZBSPA(\graphN,\vmu)$ remains unchanged.)
  The SPA-based Bethe partition function $\ZBSPA^{*}(\graphN)$ is then defined
  to be
  \begin{align*}
    \begin{aligned}
      \ZBSPA^{*}(\graphN)
          &\defeq \max\limits_{\vmu}
          \ZBSPA(\graphN,\vmu)
          \nonumber\\
          &\qquad \mathrm{s.t.} \
          \text{$\vmu$ is an SPA fixed-point message vector 
          satisfy~\eqref{eqn:positivity of SNFG Ze}.}
      \end{aligned}
  \end{align*}
  The above definition is viewed as the pseudo-dual formulation of the Bethe partition function.
  \edefinition
\end{definition}

 \begin{remark}\label{sec:SNFG:remk:1}
   Let us make some remarks on the pseudo-dual formulation of the Bethe partition function in Definition~\ref{sec:SNFG:def:2}.
   %---------------------------------------------------------------------------
   \begin{enumerate}
        \item Consider an arbitrary cycle-free S-NFG $\graphN$. 
        Proposition~\ref{prop: property of the SPA and ZB for cycle-free SNFG} 
        states that the SPA on $\graphN$ converges to a unique fixed point
         and that $ \ZB(\graphN) = Z(\graphN) $.
        Denote the associated SPA fixed-point message vector by~$\vmu$. 
        One can verify that $\ZBSPA(\graphN,\vmu) = Z(\graphN)$. 

        \item By Assumption~\ref{asmp: assume messages are non-negative}, in this thesis, we only consider the S-NFGs where $  \ZBSPA^{*}(\graphN) $, as defined in Definition~\ref{sec:SNFG:def:2}, is well-defined, \ie, the SPA fixed-point messages $ \vmu $ such that
        %-----------------------------------------------------------------------
        \begin{align*}
            \ZBSPA(\graphN,\vmu) = \ZBSPA^{*}(\graphN), \qquad
            Z_e(\vmu) >0, \qquad e \in \setEfull.
        \end{align*}
        %-----------------------------------------------------------------------
        See Item~\ref{sec:SNFG:remk:1:item:2} for a non-trivial S-NFG where this condition is not satisfied.

        \begin{figure}[t]
            \centering
            \captionsetup{font=scriptsize}
            %----------------------------
            \begin{minipage}[t]{0.4\linewidth}
              \centering
              \begin{tikzpicture}[node distance=2cm, remember picture]
                \tikzstyle{state}=[shape=rectangle,fill=white, draw, minimum width=0.2cm, minimum height = 0.2cm,outer sep=-0.3pt]
                \tikzset{dot/.style={circle,fill=#1,inner sep=0,minimum size=3pt}}
                \node[state] (f2) at (0,0) [label=above: $f_{2}$] {};
                \node[state] (f1) at (0,1) [label=above: $f_{1}$] {};
                \node [dot=black] at (f2.west) {};
                \node [dot=black] at (f1.west) {};
                \node (x1) at (-0.8,0.5) [label=left: $x_{1}$] {};
                \node (x2) at (0.8,0.5) [label=right: $x_{2}$] {};
                %----------------------------------------------------------------------------
                \begin{pgfonlayer}{background}
                  \draw[-,draw]
                      (1,0) -- (-1,0) (1,0) -- (1,1) (-1,0) -- (-1,1) 
                      (-1,1) -- (1,1);
                \end{pgfonlayer}
                %----------------------------------------------------------------------------
              \end{tikzpicture}
            \end{minipage}%
            \caption{S-NFG $\sfN$ in Item~\ref{sec:SNFG:remk:1:item:2} in Remark~\ref{sec:SNFG:remk:1}.\label{sec:SNFG:fig:2}}
        \end{figure}

        \item\label{sec:SNFG:remk:1:item:2} Let us consider the S-NFG in Fig.~\ref{sec:SNFG:fig:2}, where the dots are used for denoting the row index of the matrices $ \matr{f}_{1} $ and $ \matr{f}_{2} $ w.r.t. the function nodes $ f_{1} $ and $ f_{2} $, respectively. In particular, these matrices are given by
        %-------------------------------------------------------------------
        \begin{align*}
            \matr{f}_{1} &\defeq 
            \bigl( f_{1}(x_{1},x_{2}) \bigr)_{x_{1} \in \setx{1},x_{2} \in \setx{2}},
            \qquad 
            \matr{f}_{2} \defeq 
            \bigl( f_{2}(x_{1},x_{2}) \bigr)_{x_{1} \in \setx{1},x_{2} \in \setx{2}}
        \end{align*}
        %-------------------------------------------------------------------
        with row indexed by $ x_{1} $ and column indexed by $ x_{2} $. Running the SPA on this S-NFG is equivalent to applying the power method to the matrix $ \matr{f}_{1} \cdot \matr{f}_{2}^{\tran} $. If we set
        %-----------------------------------------------------------------------
        \begin{align*}
            \matr{f}_{1} 
            = \begin{pmatrix}
              1 & 1 \\ 0 & 1
            \end{pmatrix},
            \qquad 
            \matr{f}_{2} 
            = \begin{pmatrix}
              1 & 0 \\ 0 & 1
            \end{pmatrix},
        \end{align*}
        %-----------------------------------------------------------------------
        then the SPA fixed-point messages are given by
        %-------------------------------------------------------------------
        \begin{align*}
            \vmu_{1,f_{1}}= \vmu_{2,f_{2}} = \begin{pmatrix}
                0, & 1
            \end{pmatrix}^{\!\!\tran}, \quad
            \vmu_{2,f_{1}}= \vmu_{1,f_{2}} = \begin{pmatrix}
                1, & 0
            \end{pmatrix}^{\!\!\tran}, \quad
        \end{align*}
        %-------------------------------------------------------------------
        which results in 
        %-----------------------------------------------------------------------
        \begin{align*}
            Z_{1}(\vmu) 
            = Z_{2}(\vmu) = 0.
        \end{align*}
        %-----------------------------------------------------------------------
        We cannot evaluate $ \ZBSPA(\graphN,\vmu) $ for this S-NFG $ \graphN $.

        \item~\cite[Section VI-C]{Yedidia2005} A sufficient condition for $ \ZBSPA^{*}(\graphN) $ being well-defined for the S-NFG $ \sfN $ is that the local functions in
        $ \sfN $ satisfy
        %-----------------------------------------------------------------------
        \begin{align*}
            f(\vxf) > 0, \qquad \vxf \in \setxf,\, f \in \setF.
        \end{align*}
        %-----------------------------------------------------------------------
        For such S-NFG $ \sfN $, both the SPA fixed-point messages and the beliefs at the locations of the local minima of the Bethe free energy function are positive-valued. 
        Thus the function $ \ZBSPA^{*}(\graphN) $ is well-defined. The statement proven by Yedidia \etal~in~\cite[Theorem 3]{Yedidia2005} further implied that 
        \begin{align*}
            \ZB(\graphN) = \ZBSPA^{*}(\graphN).
        \end{align*}

        \item 
        Consider running the SPA defined in Definition~\ref{sec:SNFG:def:4} on some S-NFG $\graphN$. 
        Suppose that we obtain an SPA fixed-point message vector $\vmu$. Following Definition~\ref{def:belief at SPA fixed point for S-NFG}, we obtain 
        the collection of the beliefs $ \vbeli \in \LMP(\sfN) $, which is induced by $\vmu$. 
        Then we relate the value of the Bethe free energy function evaluated at $ \vbeli $
        to the value of the $\vmu$-based Bethe partition function by the following equation:
        \begin{align*}
            \exp \bigl( -\FB(\vbeli) \bigr) = \ZBSPA(\graphN,\vmu).
        \end{align*}
   \end{enumerate}
   %---------------------------------------------------------------------------
  \eremark
\end{remark}

\section[The S-NFG for the Matrix Permanent]{The S-NFG \texorpdfstring{$\sfN(\mtheta)$}{} for the Permanent of a Non-negative Square Matrix \texorpdfstring{$\mtheta$}{} }
\label{sec: snfg for permanent of nonnagtive square matrix}

% Factor graphs are powerful graphical tools to depict the factorization of
% multivariate functions. Many well-known approximation methods, \eg, the
% Bethe~\cite{Yedidia2005} and Sinkhorn\cite{Linial2000} approximations, are
% formulated based on factor graphs. In this thesis, we consider standard normal
% factor graphs (S-NFGs), \ie, factor graphs where the local functions are
% non-negative real-valued, and variables are associated with
% edges~\cite{Forney2001, Loeliger2004}.

%***************************************************************************

The main purpose of this section is to present an S-NFG whose partition
function equals the permanent of a non-negative matrix~$\mtheta$ of size
$n \times n$. There are different ways to develop such
S-NFGs~\cite{Greenhill2010, Vontobel2013a}. (For a comparison between the ways
in~\cite{Greenhill2010, Vontobel2013a}, see\cite[Section
VII-E]{Vontobel2013a}.) In this thesis, we consider the S-NFG $ \sfN(\mtheta) $
defined in~\cite{Vontobel2013a}. Throughout this section, let $\mtheta$ be a
fixed non-negative matrix of size $n \times n$.

%***************************************************************************

\begin{definition}
  \label{sec:1:def:5}
  \label{def: details of S-NFG of permanent}
  % unify the above labels
  \index{Normal factor graph!S-NFG for the matrix permanent}
  We define the S-NFG
  $ \sfN(\mtheta) \defeq \sfN( \setF, \setEfull, \set{X}) $ as follows (see
  also Fig.~\ref{sec:1:fig:15} for the special case $n = 3$). Note that $ \sfN(\mtheta) $ is a special case of the S-NFG defined in Definition~\ref{def: def of snfg}.
  \begin{enumerate}

  \item The set of vertices (also called function nodes) is
    \begin{align*}
      \set{F}
        &\defeq
           \{ \fr{i} \}_{i \in [n]} 
           \ \cup \ 
           \{ \fc{j} \}_{j \in [n]}.
    \end{align*}
    Here, the letter ``$\mathrm{r}$'' in $ \fr{i} $ means that $ \fr{i} $
    corresponds to the $ i $-th row of $\mtheta$, whereas the letter
    ``$\mathrm{c}$'' in $ \fc{j} $ means that $ \fc{j} $ corresponds to the
    $ j $-th column of $\mtheta$.

  \item The set of edges is defined to be
    \begin{align*}
      \setEfull
        &\defeq
           [n] \times [n] 
          = \bigl\{ 
              (i,j) 
            \bigm|
              i , j \in [n] 
            \bigr\}.
    \end{align*}

  \item The alphabet associated with edge $ e = (i,j) \in \setEfull $ is
  \begin{align*}
    \setx{e} = \setx{i,j} \defeq \{0,1\}.
  \end{align*}

  \item The set $ \set{X} $ is defined to be
    \begin{align*}
      \set{X}
        &\defeq
           \prod\limits_{e}
             \setx{e} 
         = \prod\limits_{i,j}
             \setx{i,j}
    \end{align*}
    and is called the configuration set of $ \sfN(\mtheta) $.

  \item A matrix like
    \begin{align*}
      \vgamzo
        &\defeq
           \bigl( 
             \gamzo(i,j)
           \bigr)_{\! (i,j) \in [n] \times [n]}
             \in \setX
    \end{align*} 
    is used to denote an element of $ \set{X} $, \ie, a configuration of
    $ \sfN(\mtheta) $. Here, $ \gamma(i,j) $ is associated with edge
    $ e = (i,j) \in \setEfull$. In the following, the $i$-th row and the
    $j$-th column of $\vgamzo$ will be denoted by $ \vgamzo(i,:) $ and
    $ \vgamzo(:,j) $, respectively.

  \item For a given configuration $\vgamzo \in \setX$, the value of the local
    functions of $\sfN(\mtheta)$ are defined as follows. Namely, for every
    $ i \in [n] $, the local function $ \fr{i} $ is defined to
    be\footnote{Note that, with a slight abuse of notation, we use the same
      label for a function node and its associated local function.}
    \begin{align*}
      \fr{i}\bigl( \vgamzo(i,:) \bigr)
        &\defeq
           \begin{cases}
             \sqrt{ \theta(i,j) } & \vgamzo(i,:) = \vuj  \\
             0                    & \text{otherwise}
          \end{cases}.
    \end{align*}
    For every $ j \in [n] $, the local function $ \fc{j} $ is defined to be
    \begin{align*}
        \fc{j}\bigl( \vgamzo(:,j) \bigr) \defeq
        \begin{cases}
            \sqrt{ \theta(i,j) } & \vgamzo(:,j) = \vui \\
            0 & \text{otherwise}
        \end{cases}.
    \end{align*}
    Here we used the following notation: the vector $\vuj$ stands for a row
    vector of length $n$ whose entries are all equal to~$0$ except for the
    $j$-th entry, which is equal to~$1$. The column vector $\vui$ is defined
    similarly.
  
  \item For a given configuration $\vgamzo \in \setX$, the value of the global
    function $g$ of $\sfN(\mtheta)$ is defined to be the product of the values
    of the local functions, \ie,
    \begin{align*}
      \gtheta(\vgamzo)
        &\defeq
           \Biggl( 
             \prod\limits_{i} 
               \fr{i}\bigl( \vgamzo(i,:) \bigr)
           \Biggr)
          \cdot
          \Biggl(
            \prod\limits_{j}
              \fc{j}\bigl( \vgamzo(:,j) \bigr) 
          \Biggr) .
    \end{align*}
  
  \item The set of valid configurations of $ \sfN(\mtheta) $ is
    defined to be
    \begin{align*}
      \setA(\mtheta)
        &\defeq
           \setA\bigl( \sfN(\mtheta) \bigr)
         \defeq 
           \bigl\{
             \vgamzo \in \setX
           \bigm|
             \gtheta(\vgamzo) > 0
           \bigr\}.
    \end{align*}
    
  \item The partition function of $ \sfN(\mtheta) $ is defined to be
    \begin{align*}
      Z\bigl( \sfN(\mtheta) \bigr)
        &\defeq 
           \sum\limits_{\vgamzo \in \setX }
             \gtheta(\vgamzo)
         = \sum\limits_{\vgamzo \in \setA(\mtheta) }
             \gtheta(\vgamzo).
    \end{align*}
    \label{partition function of N theta}

  \end{enumerate}
  \vspace{-0.25cm}
  \edefinition
\end{definition}

%***************************************************************************

\begin{figure}[t]
  \centering
  \input{figures/example_snfg_permanent_3by3.tex}
  \caption{The S-NFG $ \sfN(\mtheta) $ for the special case $ n = 3 $.}
  \label{sec:1:fig:15}
  \vspace{-0.0cm}
\end{figure}

%***************************************************************************

We make the following observations:
\begin{enumerate}

\item One can verify that $Z\bigl( \sfN(\mtheta) \bigr) = \perm(\mtheta)$.

\item If $ \gtheta(\vgamzo) > 0 $, then $ \vgamzo \in \setP_{[n]} $,
  \ie, $\vgamzo$ is a permutation matrix of size $n \times n$. (In
  fact, if all entries of $\mtheta$ are strictly positive, then
  $ \gtheta(\vgamzo) > 0 $ if and only if $ \vgamzo \in \setP_{[n]} $.)

\end{enumerate}

%***************************************************************************

For the following definition, recall that $\Gamma_{n}$ is the set of all
doubly stochastic matrices of size $n \times n$ and that $\Gamma_{M,n}$ is the
set of all doubly stochastic matrices of size $n \times n$ where all entries
are multiples of $1/M$.

For the following definition, recall that $\Gamma_{n}$ is the set of all
doubly stochastic matrices of size $n \times n$ and that $\Gamma_{M,n}$ is the
set of all doubly stochastic matrices of size $n \times n$ where all entries
are multiples of $1/M$.

%***************************************************************************

\begin{definition}
  \label{sec:1:def:10}

  Consider the S-NFG $ \sfN(\mtheta) $. Let $M \in \sZpp $.  We make the
  following definitions.
  \begin{enumerate}

  \item We define $\Gamnthe$ to be the set of matrices in $\Gamma_{n}$ whose
    support is contained in the support of $\mtheta$, \ie,
    \begin{align*}
      \Gamnthe
        &\defeq 
           \bigl\{ 
             \vgam \in \Gamma_{n}
           \bigm|
             \gamma(i,j) = 0 \text{ if } \theta(i,j) = 0
           \bigr\}.
    \end{align*}

  \item We define $\GamMnthe$ to be the set of matrices in $\Gamma_{M,n}$
    whose support is contained in the support of $\mtheta$, \ie,
     \begin{align*}
       \GamMnthe
         &\defeq 
            \bigl\{ 
              \vgam \in \Gamma_{M,n}
            \bigm|
              \gamma(i,j) = 0 \text{ if } \theta(i,j) = 0
            \bigr\}.
     \end{align*}

  \item We define $\PiA{\mtheta}$ to be the set of vectors representing
    probability mass functions over $\setP_{[n]}$ whose support is contained
    in $\setA(\mtheta)$, \ie,
    \begin{align*}
    \PiA{\mtheta} \defeq
         \left\{ 
           \vp = \bigl( p(\mP_{\sigma}) \bigr)_{\! \mP_{\sigma} \in \setP_{[n]}}
        \ \middle| \!
          \begin{array}{c}
            p(\mP_{\sigma}) \geq 0, \, 
              \forall \mP_{\sigma} \in \setP_{[n]} \\
            p(\mP_{\sigma}) = 0, \, 
              \forall \mP_{\sigma} 
                        \notin \set{A}(\mtheta) \\
            \sum\limits_{\mP_{\sigma} \in \setP_{[n]}} p(\mP_{\sigma}) = 1
          \end{array} \!
        \right\}.
    \end{align*}

  \item Let $\vgam \in \Gamma_n$. We define
    \begin{align*}
      \PiA{\mtheta}( \vgam )
        &\defeq 
          \left\{ 
              \vp 
              \in \PiA{\mtheta}
          \ \middle| \ 
              \sum\limits_{\mP_{\sigma} \in \setP_{[n]}}
                p(\mP_{\sigma}) \cdot \mP_{\sigma} = \vgam
          \right\}.
    \end{align*}

   \item For any matrix $ \matr{B} \in \sR^{n \times n}_{\geq 0} $, we define
   \begin{align*}
       \setS_{[n]}(\matr{B})
         &\defeq 
            \bigl\{ 
              \sigma \in \setS_{[n]} 
            \bigm|
              B\bigl( i, \sigma(i) \bigr) > 0, \, \forall i \in [n]
            \bigr\} \\
         &= \left\{ 
              \sigma \in \setS_{[n]} 
            \ \middle| \ 
              \prod_i 
                B\bigl( i, \sigma(i) \bigr) > 0
            \right\}.
   \end{align*}
   
   \item We define 
     \begin{align*}
       \vsigma_{[M]}
         &\defeq
            (\sigma_{m})_{m \in [M]} \in (\setS_{[n]})^{M} , \\
       \mP_{\vsigma_{[M]}}
         &\defeq
            ( \mP_{\sigma_{m}} )_{ m \in [M] } \in (\setP_{[n]})^{M}.
     \end{align*}

   \item Let $ \vgam \in \Gamma_{M,n} $. We define
     \begin{align*}
       \mtheta^{ M \cdot \vgam }
         &\defeq
            \prod\limits_{i,j}
              \bigl(
                \theta(i,j)
              \bigr)^{\! M \cdot \gamma(i,j)}.
     \end{align*}

   \item Let $\pmtheta$ be the probability mass function on $\setS_{[n]}$
     induced by~$\mtheta$, \ie,
     \begin{align*}
       \pmtheta(\sigma)
         &\defeq
            \frac{ \prod\limits_{i \in [n]} \theta\bigl( i,\sigma(i) \bigr)}
                 { \perm(\mtheta) }
          = \frac{ \mtheta^{\mP_{\sigma}}}
                 { \perm(\mtheta) } , 
                   \quad \sigma \in \setS_{[n]}.
     \end{align*}
     Clearly, $\pmtheta(\sigma) \geq 0$ for all $\sigma \in \setS_{[n]}$ and
     $\sum\limits_{\sigma \in \setS_{[n]}} \pmtheta(\sigma) = 1$.

   \end{enumerate}
  \edefinition
\end{definition}

%***************************************************************************
Note that if $\matr{B}$ is a doubly stochastic matrix of size $n \times n$, then
  the set $\setS_{[n]}(\matr{B})$ is non-empty. Indeed, this follows from
  considering
   \begin{align*}
      \perm( \matr{B} ) 
      = \sum\limits_{\sigma \in \setS_{[n]}(\matr{B})} 
      \prod\limits_{i}
      B\bigl( i, \sigma(i) \bigr) 
      \overset{(a)}{\geq} n!/n^{n} > 0 ,
   \end{align*}
   where step $(a)$ follows from van der Waerden's inequality, which was
   proven in~\cite[Theorem 1]{Egorychev1981} and~\cite[Theorem
   1]{Falikman1981}.
% %----------------------------------------------------------------------------
% \begin{lemma}
%    \label{lem: nonemptyness of setS n B}
%    If $ \matr{B} \in \sR^{n \times n}_{\geq 0} $ is a doubly stochastic matrix, 
%    then the set $ \setS_{[n]}(\matr{B}) $ is non-empty.
% \end{lemma}
% %----------------------------------------------------------------------------
% %----------------------------------------------------------------------------
% \begin{proof}
   % The non-emptyness of $ \setS_{[n]}(\matr{B}) $ follows from
   % \begin{align*}
   %    \perm( \matr{B} ) 
   %    = \sum\limits_{\sigma \in \setS_{[n]}(\matr{B})} 
   %    \prod\limits_{i}
   %    B\bigl( i, \sigma(i) \bigr) 
   %    \overset{(a)}{\geq} n!/n^{n} > 0,
   % \end{align*}
   % where step $(a)$ follows from van der Waerden's inequality proven
   % in~\cite[Theorem 1]{Egorychev1981} and~\cite[Theorem 1]{Falikman1981}. 
% \end{proof}
% %----------------------------------------------------------------------------

\begin{lemma}
  \label{sec:1:lemma:8}

  The set of valid configurations $ \setA(\mtheta) $ of $\sfN(\mtheta)$
  satisfies
  \begin{align*}
    \setA(\mtheta) = 
      \left\{ 
         \left. 
         \mP_{\sigma} \in \setP_{[n]}
      \ \right| \ 
          \sigma \in \setS_{[n]}(\vtheta)
      \right\} , 
      % \label{sec:1:eqn:118}
  \end{align*}
  which implies a bijection between $ \setA(\mtheta) $ and
  $ \setS_{[n]}(\vtheta) $.
\end{lemma}

%***************************************************************************

\begin{proof}
  This result follows from the Observation~2 after
  Definition~\ref{sec:1:def:5}.
\end{proof}

\section{Free Energy Functions Associated with 
               \texorpdfstring{$\sfN(\mtheta)$}{}}
\label{sec:free:energy:functions:1}

We introduce various free energy functions associated with
$\sfN(\mtheta)$. First we introduce the Gibbs free energy function, whose
minimum yields $\perm(\mtheta)$. Afterwards, we introduce the Bethe free
energy function and the scaled Sinkhorn free energy function, which are
approximations of the Gibbs free energy functions. Their minima yield the
Bethe permanent $\permb(\mtheta)$ and the scaled Sinkhorn permanent
$\permscs(\mtheta)$, respectively, which are approximations of
$\perm(\mtheta)$. Finally, we introduce the degree-$M$ Bethe permanent
$\permbM{M}(\mtheta)$ and the degree-$M$ scaled Sinkhorn permanent
$\permscsM{M}(\mtheta)$.

%***************************************************************************

\begin{definition}
  \label{sec:1:def:18}
  \index{Gibbs free energy function!for the permanent case}
  Consider the S-NFG $ \sfN(\mtheta) $. The Gibbs free energy function
  associated with $ \sfN(\mtheta) $ (see~\cite[Section III]{Yedidia2005}) is
  defined to be
  \begin{align*}
    \FGthe: \ \PiA{\mtheta} &\to \sR \\
              \vp           &\mapsto \UGthe(\vp) - \HG(\vp)
  \end{align*}
  with
  \begin{align*}
    \UGthe(\vp)   
      &\defeq 
         - 
         \sum\limits_{\mP_{\sigma} \in \setP_{[n]}}
           p( \mP_{\sigma} ) \cdot \log \bigl( g( \mP_{\sigma} ) \bigr) , \\
    \HG(\vp)   
      &\defeq 
         - 
         \sum\limits_{\mP_{\sigma} \in \setP_{[n]}} 
         p( \mP_{\sigma} ) \cdot \log \bigl( p( \mP_{\sigma} ) \bigr) ,
  \end{align*}
  where $ \UGthe $ and $ \HG $ are called the Gibbs average
  energy function and the Gibbs entropy function, respectively.\footnote{Note
    that $\set{P}_{[n]} \subseteq \set{X}$, where the configuration set
    $\set{X}$ of $\sfN(\mtheta)$ was defined in Definition~\ref{def: details
      of S-NFG of permanent}. With this, $\mP_{\sigma} \in \setP_{[n]}$ is a
    configuration of $\sfN(\mtheta)$ and $g(\mP_{\sigma})$ is well defined.}
  \edefinition
\end{definition}

%***************************************************************************

One can verify that
\begin{align*}
  Z\bigl( \sfN(\mtheta) \bigr) 
    &= \exp
         \biggl( 
           - \min_{\vp \in \PiA{\mtheta}} \FGthe(\vp)
         \biggr).
\end{align*}

%***************************************************************************

Recall the Birkhoff--von Neumann theorem (see, \eg,~\cite[Theorem
4.3.49]{Horn2012}), which states that the convex hull of all permutation
matrices of size $n \times n$ equals the set of all doubly stochastic matrices
of size $ n \times n $. Because $\vp \in \PiA{\mtheta}$ is a probability mass
function over the set of all permutation matrices of size $n \times n$ with support contained in the support of $ \mtheta $, the
Gibbs free energy function can be reformulated as follows as a function over
$\Gamma_{n}(\mtheta)$ instead of over~$\PiA{\mtheta}$.

% For this reformulation, we will need the set
% $\PiA{\mtheta}( \vgam )$, which is defined to be the set of all
% $\vp \in \PiA{\mtheta}$ that are ``compatible'' with
% $\vgam \in \Gamma_{n}(\mtheta)$, \ie,
% \begin{align*}
%   \PiA{\mtheta}( \vgam )
%     &\defeq 
%       \left\{ 
%           \vp 
%           \in \PiA{\mtheta}
%       \ \middle| \ 
%           \sum\limits_{ \mP_{\sigma} \in \setA(\mtheta) }
%             p(\mP_{\sigma}) \cdot \mP_{\sigma} = \vgam
%       \right\}.
% \end{align*}

%***************************************************************************

\begin{definition}
  \label{sec:1:def:18:part:2}
  \index{Gibbs free energy function!for the permanent case}
  Consider the S-NFG $ \sfN(\mtheta) $. The modified Gibbs free energy
  function associated with $ \sfN(\mtheta) $ is defined to be
  \begin{align*}
    \FGthe': \ \Gamma_{n}(\mtheta) &\to \sR \\
               \vgam               &\mapsto \UGthe'(\vgam) - \HG'(\vgam)
  \end{align*}
  with
  \begin{align*}
    \UGthe'(\vgam)
      &\defeq 
         - 
         \sum\limits_{ i,j }  
           \gamma(i,j) \cdot \log\bigl( \theta(i,j) \bigr) , \\ 
    \HG'(\vgam)
      &\defeq
         \max\limits_{ \vp \in \PiA{\vgam}( \vgam ) }
           \HG(\vp) ,
  \end{align*}
  where $ \UGthe' $ and $ \HG' $ are called the modified Gibbs average energy
  function, and the modified Gibbs entropy function, respectively.%
  \footnote{\label{footnote:PiA:support:condition:1}%
    Note that the definition
    $\HG'(\vgam) \defeq \max\limits_{ \vp \in \PiA{\mtheta}( \vgam ) } \HG(\vp)$
    would have been more natural/straightforward. However, we prefer the
    definition
    $\HG'(\vgam) \defeq \max\limits_{ \vp \in \PiA{\vgam}( \vgam ) } \HG(\vp)$, as it
    does not depend on $\mtheta$ and, most importantly is equivalent to the
    former definition for $\vgam \in \Gamma_{n}(\mtheta)$. This equivalence
    stems from the following observation. Namely, let
    $\vp \in \PiA{\mtheta}( \vgam )$. Then $p(\mP_{\sigma}) > 0$ only if
    $\supp(\mP_{\sigma}) \subseteq \supp(\vgam)$.} \\
  \mbox{} \edefinition
\end{definition}

% (Note that if $\supp(\mtheta) \subsetneq \supp(\vgam)$, then
% $\PiA{\mtheta}( \vgam )$ is empty.)

%***************************************************************************

One can verify that
\begin{align*}
  Z\bigl( \sfN(\mtheta) \bigr) 
    &= \exp
     \biggl( 
       - \min_{\vgam \in \Gamnthe} \FGthe'(\vgam)
     \biggr).
\end{align*}

%***************************************************************************

However, neither formulation of the Gibbs free energy function leads to a
tractable optimization problem in general unless~$n$ is small. This motivates
the study of approximations of the permanent based on approximating the Gibbs
free energy function, especially based on approximating the modified Gibbs
free energy function in the reformulation presented in
Definition~\ref{sec:1:def:18:part:2}. In the following, we will first discuss
an approximation known as the Bethe free energy function and then an
approximation known as the scaled Sinkhorn free energy function.

%***************************************************************************

In general, the Bethe free energy function associated with an S-NFG is a
function over the so-called local marginal polytope associated with the
S-NFG. (For a general discussion of the local marginal polytope associated
with an S-NFG, see, \eg, \cite[Section 4]{JWainwright2008}.) The thesis
\cite[Section~IV]{Vontobel2013a} shows that the local marginal polytope of the
$\sfN(\mtheta)$ can be parameterized by the set of doubly stochastic
matrices. This is a consequence of the Birkhoff--von Neumann theorem.

%***************************************************************************

\begin{lemma}\!\!\cite[Corollary 15]{Vontobel2013a}
  \label{sec:1:def:16}
  % \index{Bethe free energy function!for the permanent case}
  Consider the S-NFG $ \sfN(\mtheta) $. The Bethe free energy function
  associated with $ \sfN(\mtheta) $, as defined in Definition~\ref{sec:SNFG:def:1}, 
  can be reformulated by the mapping $ \FBthe $, which is defined to be
  \begin{align*}
    \FBthe: \ \Gamma_{n}(\mtheta) &\to \sR, \\
              \vgam              &\mapsto \UBthe( \vgam ) - \HBthe( \vgam )
  \end{align*}
  with
  \begin{align*}
    \UBthe(\vgam)
      &\defeq
        -\sum\limits_{ i,j }
        \gamma(i,j) \cdot \log\bigl( \theta(i,j) \bigr) , \\
    \HBthe( \vgam )
      &\defeq 
        -
        \sum\limits_{ i,j }
          \gamma(i,j) \cdot \log\bigl(  \gamma(i,j) \bigr) \\
      &\quad\ 
        + 
        \sum\limits_{ i,j }
          \bigl( 1 - \gamma(i,j) \bigr) 
            \cdot 
            \log\bigl( 1 - \gamma(i,j) \bigr).
  \end{align*}
  Here, $ \UBthe $ and $ \HBthe $ are called the Bethe average energy function
  and the Bethe entropy function, respectively. With this, we have
  \begin{align*}
      \ZB\bigl( \sfN(\mtheta) \bigr) = \permb(\mtheta),
  \end{align*}
  where the Bethe permanent $ \permb(\mtheta) $ is defined to be
  \begin{align}
    \permb(\mtheta)
      &\defeq 
         \exp
           \biggl(
             - \min_{\vgam \in \Gamnthe} \FBthe( \vgam ) 
           \biggr).
        \label{sec:1:eqn:192}
  \end{align}
  \index{Bethe permanent}
  \elemma
\end{lemma}

%***************************************************************************

The degree-$M$ Bethe permanent of $\mtheta$ was defined
in~\cite{Vontobel2013a} and yields a combinatorial characterization of the
Bethe permanent. (Note that the definition of $\permb(\mtheta)$
in~\eqref{sec:1:eqn:192} is analytical in the sense that $\permb(\mtheta)$ is
given by the solution of an optimization problem.)

%***************************************************************************

\begin{definition}
  \label{def:matrix:degree:M:cover:1}
  \index{Degree-$M$ Bethe permanent}
  % \index{Degree-$M$ Bethe partition function!for the permanent case}
  Consider the S-NFG $ \sfN(\mtheta) $. Let $M \in \sZpp$.  The degree-$M$
  Bethe permanent is defined to be
  \begin{align*}
    \permbM{M}(\mtheta)
      &\defeq
        \sqrt[M]{
          \bigl\langle
            \perm( \mtheta^{\uparrow \mP_{M}} )
          \bigr\rangle_{\mP_{M} \in \tPsi_{M} }
       }  , 
         % \label{eq:def:degree:M:Bethe:permanent:1:part:2}
  \end{align*}
  where 
  \begin{align*}
    \tPsi_{M}
      &\defeq 
        \left\{ 
          \mP_{M} \defeq \Bigl( \mP^{(i,j)} \Bigr)_{\!\! i,j\in [n]}
        \ \middle| \ 
          \mP^{(i,j)} \in \setP_{[M]}
        \right\} ,
  \end{align*}
  where $ \setP_{[M]} $ is defined in~\eqref{sec:1:eqn:181} and the
  $ \mP_{M} $-lifting of $ \mtheta $ is defined to be a real-valued matrix of size $ M n $-by-$Mn$:
  \begin{align*}
    \!\!
    \mtheta^{\uparrow \mP_{M}}
      &\defeq 
        \begin{pmatrix}
            \theta(1,1) \cdot  \mP^{(1,1)} 
            & \cdots & \theta(1,n) \cdot  \mP^{(1,n)} \\
            \vdots & \ddots & \vdots \\
            \theta(n,1) \cdot  \mP^{(n,1)} 
            & \cdots & \theta(n,n) \cdot  \mP^{(n,n)} 
           \end{pmatrix} 
        .
  \end{align*}
  \edefinition
\end{definition}

%***************************************************************************

\begin{proposition}
  \label{sec:1:thm:8}
  \index{Graph-cover theorem!for the permanent case}
  Consider the S-NFG $ \sfN(\mtheta) $. It holds that
  \begin{align}
    \permbM{M}(\mtheta) &= \ZBM\bigl( \sfN(\mtheta) \bigr), 
    \qquad M \in \sZpp, \nonumber\\
    \limsup_{M \to \infty} \,
      \permbM{M}(\mtheta)
      &= \permb(\mtheta). 
            \label{sec:1:eqn:193}
  \end{align}
  The limit in~\eqref{sec:1:eqn:193} is visualized in
  Fig.~\ref{fig:combinatorial:characterization:1}~(left).
\end{proposition}

%***************************************************************************

\begin{proof}
  See~\cite[Section~VI]{Vontobel2013a}.
\end{proof}

As we will see in Section~\ref{sec: finite graph covers}, the degree-$M$ Bethe permanent $ \permbM{M}(\mtheta) $ is equal to the degree-$M$ Bethe partition function of the S-NFG $ \sfN(\mtheta) $, and the proof of Proposition~\ref{sec:1:thm:8} is based on using the method of types to characterize the valid configurations in the degree-$M$ covers of the S-NFG $ \sfN(\mtheta) $.

%***************************************************************************

The Sinkhorn free energy function was defined in~\cite{Vontobel2014}. Here we
consider a variant called the scaled Sinkhorn free energy function that was
introduced in~\cite{N.Anari2021}. These free energy functions are also defined
over $\Gamma_{n}(\mtheta)$.

%***************************************************************************

\begin{definition} 
  \label{sec:1:def:20}
  \index{(Scaled) Sinkhorn free energy function for the permanent case}
  Consider the S-NFG $ \sfN(\mtheta) $. The scaled Sinkhorn free energy
  function associated with $ \sfN(\mtheta) $ is defined to be
  \begin{align*}
    \FscSthe: \ \Gamma_{n}(\vtheta) &\to \sR \\
                \vgam          &\mapsto \UscSthe( \vgam ) - \HscSthe( \vgam )
  \end{align*}
  with
  \begin{align*}
    \UscSthe(\vgam) 
      &\defeq
         - 
         \sum\limits_{ i,j }
           \gamma(i,j) \cdot \log\bigl( \theta(i,j) \bigr) , \\
    \HscSthe(\vgam)
      &\defeq 
         - \,
         n
         -
         \sum\limits_{ i,j }
           \gamma(i,j) \cdot 
             \log\bigl( \gamma(i,j) \bigr).
    \end{align*}
    Here, $ \UscSthe $ and $ \HscSthe $ are called the scaled Sinkhorn average
    energy function and the scaled Sinkhorn entropy function,
    respectively. With this, the scaled Sinkhorn permanent is defined to be
    \begin{align}
      \permscs(\mtheta)
        &\defeq 
           \exp
             \biggl(
               - \min_{\vgam \in \Gamnthe} \FscSthe( \vgam ) 
             \biggr).
          \label{sec:1:eqn:192:part:2}
    \end{align}
    \index{(Scaled) Sinkhorn permanent}
    \edefinition
\end{definition}

%***************************************************************************

Note that $\HscSthe(\vgam)$ can be obtained from $\HBthe( \vgam )$ by
approximating $\bigl( 1 - \gamma(i,j) \bigr) \cdot \log( 1 - \gamma(i,j) )$ by
the first-order Taylor series expansion term $- \gamma(i,j)$ and using the
fact that the sum of all entries of the matrix
$ \vgam \in \Gamma_{n}(\vtheta) $ equals $n$.

%***************************************************************************

In the following, we introduce the degree-$M$ scaled Sinkhorn permanent of
$\mtheta$, which yields a combinatorial characterization of the scaled
Sinkhorn permanent of $\mtheta$. (Note that the definition of
$\permscs(\mtheta)$ in~\eqref{sec:1:eqn:192:part:2} is analytical in the sense
that $\permscs(\mtheta)$ is given by the solution of an optimization problem.)

%***************************************************************************

\begin{definition}
  \label{def:Sinkhorn:degree:M:cover:1}
  \index{Degree-$M$ (scaled) Sinkhorn permanent}
  Consider the S-NFG $ \sfN(\mtheta) $. Let $M \in \sZpp$.  The degree-$M$
  scaled Sinkhorn permanent is defined to be
  \begin{align*}
    \permscsM{M}(\mtheta)
      &\defeq 
         \sqrt[M]{
           \perm
             \Bigl(
               \bigl\langle
                 \mtheta^{\uparrow \mP_{M}}
               \bigr\rangle_{\mP_{M} \in \tPsi_{M} }
             \Bigr)
         } \nonumber \\
      &= \sqrt[M]{
           \perm\bigl( \mtheta \otimes \mU_{M,M}\bigr)
         }  , 
   % \label{eq:def:degree:M:scaled:Sinkhorn:permanent:1}
  \end{align*}
  where $\otimes$ denotes the Kronecker product of two matrices and where
  $ \mU_{M,M} $ is the matrix of size $M \times M$ with all entries equal to
  $1/M$.\edefinition
\end{definition}

%***************************************************************************

\begin{proposition}
  \label{sec:1:prop:14}
  \label{SEC:1:PROP:14}
    
  Consider the S-NFG $ \sfN(\mtheta) $. It holds that
  \begin{align}
    \limsup_{M \to \infty} \,
      \permscsM{M}(\mtheta)
      &= \permscs(\mtheta). 
            \label{sec:1:eqn:193:part:2}
  \end{align}
  The limit in~\eqref{sec:1:eqn:193:part:2} is visualized in
  Fig.~\ref{fig:combinatorial:characterization:1}~(right).
\end{proposition}

%***************************************************************************

\begin{proof}
  See Appendix~\ref{apx:24}. (Note that this proof needs some results from
  Lemma~\ref{lem: expression of permanents w.r.t. C} and
  Proposition~\ref{prop:coefficient:asymptotitic:characterization:1}, which
  are proven only in the upcoming sections. Of course, Lemma~\ref{lem:
    expression of permanents w.r.t. C} and
  Proposition~\ref{prop:coefficient:asymptotitic:characterization:1} are
  proven independently of the statement in Proposition~\ref{sec:1:prop:14}.)
\end{proof}

%***************************************************************************

Except for small values of $n$, computing $ \permbM{M}(\mtheta) $ and
$ \permscsM{M}(\mtheta) $ appears to be intractable in general for finite
$M$.\footnote{More formally, we leave it as an open problem to show that
  computing $ \permbM{M}(\mtheta) $ and $ \permscsM{M}(\mtheta) $ is in the
  complexity class \#P for finite $ M \in \sZpp $.} However, as $ M $ goes to
infinity, the limit superior of $ \permbM{M}(\mtheta) $ and
$ \permscsM{M}(\mtheta) $ are equal to $ \permb(\mtheta) $ and
$ \permscs(\mtheta) $, respectively, and they can be computed
efficiently~\cite{Vontobel2013a,Linial2000}.

Note that the modified Gibbs entropy function $ \HG'(\vgam) $, the Bethe entropy function  $ \HBthe( \vgam ) $, and the scaled Sinkhorn entropy function $ \HscSthe(\vgam) $, as defined in Definitions~\ref{sec:1:def:18:part:2},~\ref{sec:1:def:16}, and~\ref{sec:1:def:20}, respectively, are well defined for $ \vgam \in \Gamma_n $, not just for $ \vgam \in \Gamma_n(\vtheta) $. This observation allows us to consider $ \vgam \in \Gamma_{M,n} $ 
in Proposition~\ref{prop:coefficient:asymptotitic:characterization:1}, \ie, $\vgam$ is an element of a set that is independent of $\mtheta$.

\section{Double-Edge Normal Factor Graphs (DE-NFGs)}
\label{sec: basic of denfg}

Now we move to the quantum case. Instead of using the bra-ket notation of quantum mechanics, we will use standard linear algebra notation similar as in~\cite{Loeliger2017,Loeliger2020}.

In this section, we review DE-NFGs, which are based on the works in~\cite{Loeliger2017,Loeliger2020,Cao2017}. In general, a DE-NFG can
contain four types of edges: full double edges, full single edges, half double
edges, and half single edges. For the purposes of this thesis, it is sufficient
to consider only DE-NFGs where all the edges are full double edges. The
reasoning is as follows: similar to Section~\ref{sec:SNFG}, half
(double/single) edges can be turned into full
(double/single) edges by attaching a dummy $1$-valued function node to every
half (double/single) edge without changing any marginals or the partition
function. Moreover, DE-NFGs with full single edges can be turned into DE-NFGs
with only full double edges by changing full single edges into full double
edges and suitably reformulating function nodes.

We start by presenting an example DE-NFG.
%*****************************************************************************

\begin{example}
  \label{sec:DENFG:exp:1}

  % \input{figures/penfg_example_main}

  % The NFG on the left-hand side in Fig.~\ref{sec:DENFG:fig:1} depicts the following factorization
  % \begin{align*}
  %   g(x_{1},\ldots,x_{5},x_{1}',\ldots,x_{\upfif})
  %   &\defeq
  %     f_{1}(x_{1},x_{2},x_{3},x_{1}',x_{2}',x_{\upthre}) 
  %     \cdot f_{2}(x_{1},x_{4},x_{1}',x_{\upfor})\\
  %   &\quad \cdot f_{3}(x_{2},x_{5},x_{2}',x_{\upfif})
  %     \cdot f_{4}(x_{3},x_{4},x_{5},x_{\upthre},x_{\upfor},x_{\upfif}).
  % \end{align*}
  % In an NFG, each edge corresponds to a variable, and every pair of edges, denoted by $ (e, \upe) $ represents a pair of variables $\tx_{e} = (\xe, x_{\unpair{e}})$. Here $x_e \in \tset{X}_{e}$ and $x_{\unpair{e}} \in \tset{X}_{\upe}$ with $ \tset{X}_{e} = \tset{X}_{\upe} $ are variables corresponding to edges $ e $ and $ \unpair{e} $, respectively. The term "paired" indicates that if edge $e$ is incident on a function node $f$, then $ \unpair{e} $ is also incident on $ f $, implying that the variable pair $\tx_{e}$ is an argument of 
  % the local function $f$. 

  % In Fig.~\ref{sec:DENFG:fig:1}, both the DE-NFG on the left-hand side and the DE-NFG on the right-hand side are equivalent. Specifically, for the DE-NFG on the right-hand side in Fig.~\ref{sec:DENFG:fig:1}, we establish an association between each edge $ e \in \{ 1,\ldots,5 \} $ and another edge $ \unpair{e} \in \{ \upone, \uptwo,\ldots,\upfif \} $. 

  \input{figures/denfg_example}
  The DE-NFG in Fig.~\ref{sec:DENFG:fig:4} depicts the following factorization
  \begin{align*}
      g(x_{1},\ldots,x_{5},x_{1}',\ldots,&x_{5}')
      \defeq f_{1}(x_{1},x_{2},x_{3},x_{1}',x_{2}',x_{3}') 
          \nonumber\\
          &
          \cdot f_{2}(x_{1},x_{4},x_{1}',x_{4}')
          \cdot f_{3}(x_{2},x_{5},x_{2}',x_{5}')
          \cdot f_{4}(x_{3},x_{4},x_{5},x_{3}',x_{4}',x_{5}').
  \end{align*}
  For each double edge $e \in [5]$, we associate it with the variable
  $\tx_e = (x_e, x'_e)$, where $x_e$ and $x'_e$ take values in the same
  alphabet $\set{X}_e$. Moreover, if $e$ is incident on a function node
  $f$, then $\tx_e$ is an argument of the local function $f$, \ie, both
  $x_e$ and $x'_e$ are arguments of $f$.

  \eexample
\end{example}

% In general, a DE-NFG
% contains full edges and half edges. For the same reason as mentioned in the statements below Example~\ref{sec:SNFG:exp:1}, it is sufficient for us to consider the DE-NFGs where all the edges are full edges.

%***************************************************************************

\begin{definition}\label{sec:DENFG:def:4}
  \index{Normal factor graph!DE-NFG}
  A DE-NFG $\graphN(\setF,\setEfull,\tset{X})$ consists of the following
  objects:
  \begin{enumerate}

    \item\label{sec:DENFG:def:4:item:5} A graph $(\setF,\setEfull)$ with vertex set $\setF$ and edge set
    $\setEfull$. The set $ \setF $ represents the set of function nodes, while $ \setEfull $ represents the set of full double edges. (As
    mentioned above, every edge $e \in \setEfull$ will be assumed to be a full
    double edge connecting two function nodes.)

    \item\label{sec:DENFG:def:4:item:1} An alphabet $\tset{X}\defeq\prod\limits_{e \in \setEfull}\tset{X}_e$, where
    $\tset{X}_e\defeq \set{X}_e\times \set{X}_e$ is the alphabet of variable
    $\tx_e \defeq (x_e,x_e')$ associated with edge $e \in \setEfull$. 
    Without loss of generality,
    for each $ e \in \setEfull $, we suppose that $ (0,0) \in \tset{X}_e $.

  \end{enumerate}
    For a given DE-NFG $\graphN(\setF,\setEfull,\tset{X})$, we fix the order of the function nodes by $ \set{F} = \{ f_{1},\ldots,f_{|\set{F}|} \}. $ Then we make the following definitions.
    %------------------------------------------------------------------------
    % \begin{align*}
        
        % \setEfull = \{ \pone,\ldots,\pair{(|\setEfull|-1)} \}.
        % \qquad
        % \setEfull = \{ 1,\ldots,(|\setEfull|-1) \},\qquad
        % \unpair{\setEfull} = \{ \upone,\ldots,\unpair{(|\setEfull|-1)} \}.
    % \end{align*}
    %------------------------------------------------------------------------
    % where 
    % %-----------------------------------------------------------------------
    % \begin{align*}
    %     |\setEfull| = |\setEfull| = |\unpair{\setEfull}|.
    % \end{align*}
    % %-----------------------------------------------------------------------
    
  \begin{enumerate}
    \setcounter{enumi}{3}
%  \item A set of local functions $\{f\}_{f \in \setF}$. Note that each
%    $f$ corresponds to function node $f \in \setF$.

    \item\label{sec:DENFG:def:4:item:2} For every function node $f \in \setF$, the set $\setpf$ is defined to be
    the set of edges incident on $f$. 

    % The degree of function node $ f $ in DE-NFG is given by the cardinality of $ \setpf $, \ie, $ |\setpf| $.

    \item For every edge $e \in \setEfull$, the set $\setpe$ represents the set
    of function nodes that are connected to $e$. 
    % For consistency, we also define $ \setpe $ and $ \setupe $ to be the set of function nodes that are connected to edges $ e $ and $ \upe $, respectively. It is straightforward to see that $ \setppe = \setpe = \setupe. $
    % In the following, we use $ \setpe $ for convenience.

    \item\label{sec:DENFG:def:4:item:3} Consider $ \set{I} \subseteq \setEfull $. The alphabets $ \tset{X}_{\set{I}} $ and $ \set{X}_{\set{I}} $ are defined to be
    %------------------------------------------------------------------------
    \begin{align*}
      \tset{X}_{\set{I}} \defeq \prod\limits_{e \in \set{I}} \tset{X}_{e}
      = \prod\limits_{e \in \set{I}} ( \set{X}_e \times \set{X}_e ), \qquad 
      \set{X}_{\set{I}} \defeq \prod\limits_{e \in \set{I}} \set{X}_{e}.
    \end{align*}
    %------------------------------------------------------------------------
    For any mapping $ h: \tset{X}_{\set{I}} \to \sC $, we define
    the following complex-valued square matrix:
    \begin{align}
        \matr{C}_{h}
        &\defeq 
        \bigl( 
          h( \vx_{\set{I}}, \vx_{\set{I}}' )
        \bigr)_{ \! \tvx_{\set{I}} = 
        (\vx_{\set{I}}, \vx_{\set{I}}') \in \tset{X}_{\set{I}}}
        \in \sC^{|\set{X}_{\set{I}}| \times |\set{X}_{\set{I}}|}
        \label{eqn: def of Choi matrix representation}
    \end{align}
    with row indices $\vx_{\set{I}} \in \set{X}_{\set{I}}$ and column indices 
    $\vx_{\set{I}}' \in \set{X}_{\set{I}}$.

    % For a collection of finite sets $ \set{I} = (\set{I},\upset{I}) $, 
    % the set $\tset{X}_{\set{I}}$ is defined to be
    % %-----------------------------------------------------------------------
    % \begin{align*}
    %    \tset{X}_{\set{I}} \defeq 
    %    \prod\limits_{e \in \set{I}} \tset{X}_{e}
    %    \times 
    %    \prod\limits_{\upe \in \upset{I}} \tset{X}upe.
    %   % = \tset{X}_{\set{I}} \times \tset{X}_{\upset{I}}
    %   % = \tset{X}_{\set{I}}^{2}.
    % \end{align*}
    % %-----------------------------------------------------------------------
    % Note that $ |\set{I}| $ can be different from $ |\upset{I}| $.
    
    \item Consider $ \set{I} \subseteq \setEfull $. The sets 
    $ \setHerm{ \set{X}_{\set{I}} }$ and 
    $ \setPSD{ \set{X}_{\set{I}} } $ are defined to be 
    the set of Hermitian matrices of size 
    $ |\set{X}_{\set{I}}| \times |\set{X}_{\set{I}}|$
    and 
    % the set of Hermitian matrices of size $|\tset{X}_{\set{I}}| \times |\tset{X}_{\set{I}}|$ and trace one, 
    the set of Hermitian, positive semi-definite (PSD) matrices of size 
    $|\set{X}_{\set{I}}| \times |\set{X}_{\set{I}}|$, respectively.

    \item An assignment $\tvx \defeq ( \tvx_{e} )_{e \in \setEfull} \in \tset{X}$
    is called a configuration of the DE-NFG. For each $f \in \setF$, a
    configuration $\tvx \in \tset{X}$ induces the vector
    \begin{align*}
      \tvx_{\setpf} = (\tx_{e})_{e \in \setpf} \in \tset{X}_{\setpf}.
    \end{align*}

    % \item The alphabet for $f \in \setF$ is given by
    % $ \tset{X}_{\setpf} \defeq \tset{X}_{\setpf} \times \tset{X}_{\upsetpf} $,
    % where $ \tset{X}_{\setpf} \defeq \prod\limits_{e \in \setpf} \tset{X}_{e} $ and 
    % $ \tset{X}_{\setpf} \defeq \prod\limits_{e \in \upsetpf} \tset{X}upe = \tset{X}_{\setpf} $.

  \item\label{sec:DENFG:def:4:item:4} For each function node $f \in \setF$, there is a local function associated with it,  with some slight abuse of notation, also called $f$. Depending on the
    conditions imposed on $f$, we distinguish between a strict-sense and a
    weak-sense DE-NFG.
    \begin{enumerate}

    \item\label{sec:DENFG:def:4:item:4:item:1} In the case of a \textbf{strict-sense DE-NFG}, the local function
      $f$ can be an arbitrary mapping from $\tset{X}_{\setpf}$ to $\sC$ satisfying
      the following property: the square matrix 
      \begin{align}
          \matr{C}_f
          \in \setPSD{ \set{X}_{\setpf} }
          \label{sec:DENFG:eqn:6}
      \end{align}
      with row indices $\vx_{\setpf}$ and column indices $\vx_{\setpf}'$ is a complex-valued,
      Hermitian, positive semi-definite (PSD) matrix. Motivated by the use of
      DE-NFG in quantum information processing (see, \eg,
      \cite{Wood2015}), the matrix $\matr{C}_f$ will be
      called the Choi-matrix representation of $f$ or the Choi matrix associated
      with $f$. 
  
    \item In the case of a \textbf{weak-sense DE-NFG}, we require that
      $\matr{C}_f \in \setHerm{ \set{X}_{\setpf} } $, \ie, the matrix $ \matr{C}_f $ is a complex-valued, Hermitian matrix. However, we do
      \emph{not} require $\matr{C}_f$ to be a PSD matrix.  

    \item For each $ f \in \setF $ in a strict-sense or weak-sense DE-NFG, 
        the mappings 
        \begin{align*}
          u_{f, \ell_f}: \tset{X}_{\setpf} 
          &\to \sC, \nonumber\\
          \lambda_{f}: \set{L}_f  &\to \sR,
        \end{align*}
        where $\ell_f \in \set{L}_f$,
        are defined to be the mappings satisfying the following equalities and inequalities:
        \begin{align}
          f\bigl( \tvx_{\setpf} \bigr) 
            &= \sum\limits_{\ell_f \in \set{L}_f}
                \lambda_{f}(\ell_f) \cdot
                 u_{f,\ell_f}(\vx_{\setpf}) 
                 \cdot 
                 \overline{u_{f,\ell_f}(\vx_{\setpf}')},
                 \qquad \tvx_{\setpf} \in \tset{X}_{\setpf},
                 \label{sec:DENFG:eqn:5}
            \\
            [\ell_f \!=\! \ell_f']
            &= \sum\limits_{\tvx_{\setpf} \in \tset{X}_{\setpf}}
            u_{f,\ell_f}(\vx_{\setpf}) 
            \cdot 
            \overline{u_{f,\ell_f'}(\vx_{\setpf}')},\qquad
            \ell_f,\ell_f' \in \set{L}_f, 
            \\
            \lambda_{f}(0) &\geq \lambda_{f}(1) \geq 
            \cdots \geq \lambda_{f}(|\set{L}_f|-1),
            \label{sec:DENFG:eqn:8}
        \end{align}
        where $\set{L}_f \defeq \{0,\ldots,|\tset{X}_{\setpf}|-1\}$ 
        is a finite set with $ |\set{L}_f| = |\tset{X}_{\setpf}| $.
        If we define $ \vect{u}_{f,\ell_f} $ to be the column vector associated with $ u_{f,\ell_f} $ as follows
        %-----------------------------------------------------------------------
        \begin{align*}
          \vect{u}_{f,\ell_f} \defeq \bigl( 
            u_{f,\ell_f}(\vx_{\setpf}) 
          \bigr)_{\! \vx_{\setpf} \in \set{X}_{\setpf} } 
          \in \sC^{|\set{X}_{\setpf}|},
          % \label{exp:uf}
        \end{align*}
        %-----------------------------------------------------------------------
        then the conditions in~\eqref{sec:DENFG:eqn:5}--\eqref{sec:DENFG:eqn:8} are equivalent to
        %-------------------------------------------------------------------
        \begin{align}
          \matr{C}_f &= \sum\limits_{\ell_f \in \set{L}_f} 
          \lambda_{f}(\ell_{f}) \cdot \vect{u}_{f,\ell_f} 
          \cdot \bigl( \vect{u}_{f,\ell_f} \bigr)^{\!\Herm},
          \label{sec:DENFG:eqn:9}\\
          \bigl( \vect{u}_{f,\ell'_f} \bigr)^{\!\Herm} 
          \cdot \vect{u}_{f,\ell_f} &= 
          \bigl[ \ell'_f \!=\! \ell_f \bigr], \qquad 
          \ell_f, \ell'_f \in \set{L}_f.
          \nonumber
        \end{align}
        %-------------------------------------------------------------------
        
        In summary, the vectors $ \bigl( \vect{u}_{f,\ell_f} \bigr)_{ \ell_{f} \in \set{L}_f } $ and the coefficients $ \bigl( \lambda_{f}(\ell_{f}) \bigr)_{ \ell_{f} \in \set{L}_f } $ form an eigenvalue decomposition of the Hermitian matrix $ \matr{C}_f $, where $ \lambda_{f}(\ell_{f}) $ is the eigenvalue associated with the right-eigenvector 
        $\vect{u}_{f,\ell_f}$.
  \end{enumerate}

  \item The global function $g$ is defined to be the mapping
    \begin{align*}
      g: \tset{X} \to \sC, \qquad \tvx \mapsto \prod\limits_{f \in \setF} 
      f\bigl( \tvx_{\setpf} \bigr).
    \end{align*}

  \item A configuration $\tvx \in \tset{X}$ satisfying $g(\tvx) \neq 0$ is
    called a valid configuration.

  \item The partition function is defined to be
   \begin{align*}
      Z(\graphN) \defeq \sum\limits_{\tvx \in \tset{X}}g(\tvx).
   \end{align*}
    For a strict-sense
    DE-NFG, the partition function satisfies $Z(\graphN) \in \sRp$, whereas
    for a weak-sense DE-NFG, the partition function satisfies
    $Z(\graphN) \in \sR$ (see
    Appendix~\ref{apx:property of ZN}).
    \label{partition function of denfg}

  \end{enumerate}
  \edefinition
\end{definition}

%***************************************************************************
If there is no ambiguity, when we consider S-NFG, we will use the short-hands
$\sum\limits_{\tvx}$, 
$\sum\limits_{\txe}$, 
$\sum\limits_{e}$, 
$\sum\limits_{f}$, 
$\prod\limits_{\tvx}$,  
$\prod\limits_{\txe}$,
$\prod\limits_{e}$, 
$\prod\limits_{f}$,
$ (\cdot)_{\txe} $,
and
$ (\cdot)_{e} $ for
$\sum\limits_{\tvx \in \tset{X}}$, 
$\sum\limits_{\txe \in \tset{X}_{e}}$, 
$\sum\limits_{e \in \setEfull}$, 
$\sum\limits_{f \in \setF}$,
$\prod\limits_{\tvx \in \tset{X}}$, 
$\prod\limits_{\txe \in \tset{X}_{e}}$, 
$\prod\limits_{e \in \setEfull}$, 
$\prod\limits_{f \in \setF}$, 
$ (\cdot)_{\txe \in \tset{X}_{e}} $,
and 
$ (\cdot)_{e\in \setEfull} $,
respectively. 
For any $ \set{I} \subseteq \setEfull$, we will use the short-hands
$\sum\limits_{\tvx}$, $\sum\limits_{\tvx_{\set{I}}}$, $\prod\limits_{\tvx}$, $\prod\limits_{\tvx_{\set{I}}}$, 
$(\cdot)_{ \tvx_{\set{I}} }$,
and $\max\limits_{\tvx_{\set{I}}}$ for
$\sum\limits_{\tvx \in \tset{X}}$, 
$\sum\limits_{\tvx_{\set{I}} \in \tset{X}_{\set{I}}}$, 
$\prod\limits_{\tvx \in \tset{X}}$, 
$\prod\limits_{ \tvx_{\set{I}} \in \tset{X}_{\set{I}} }$, 
$(\cdot)_{ \tvx_{\set{I}} \in \tset{X}_{\set{I}} }$,
and $\max\limits_{\tvx_{\set{I}} \in \tset{X}_{\set{I}}}$,
respectively. 
Moreover, $\setpf \setminus e$ will be short-hand notation for
$\setpf \setminus \{ e \}$.

%***************************************************************************

\begin{assumption}\label{sec:DENFG:asum:2}
  In the following, unless stated otherwise, when referring to a DE-NFG, we mean a strict-sense DE-NFG, \ie, we consider $ \lambda_{f}(\ell_{f}) \in \sR_{\geq 0} $ for all $ \ell_{f} \in \set{L}_{f} $ and $ f \in \setF $. 

  Furthermore, we will only consider the (strict-sense) DE-NFG
  $\graphN$ for which $Z(\graphN) \in \sRpp$.\footnote{In this thesis, we only consider the NFGs whose pseudo-dual formulation of the Bethe partition function is well-defined. (See the detailed discussion for S-NFG and DE-NFGs in 
  Remark~\ref{sec:SNFG:remk:1} and Assumptions~\ref{sec:DENFG:asum:2} and~\ref{sec:DENFG:remk:1}.)}
  \eassumption
\end{assumption}

%***************************************************************************

One of the motivations for considering DE-NFGs is that many NFGs
with complex-valued local functions in quantum information processing
can be transformed into DE-NFGs. (See~\cite{Cao2017} for further examples and details on this transformation.)

%--------------------------
\input{figures/penfg_example_qip}
%------------------------------------
\begin{example}[\cite{Loeliger2012,Loeliger2017}]\label{sec:DENFG:exp:2}
  Let us consider an NFG in Fig.~\ref{sec:DENFG:fig:2}(left). Note that at the end of each edge there is a dot used for denoting the row index of the matrix w.r.t. the local function. In particular, it describes a quantum system with two consecutive unitary evolutions:
  \begin{itemize}
    \item At the first stage, a quantum mechanical system is 
    prepared in a mixed state represented by a complex-valued PSD density matrix:
    \begin{align*}
        \matr{\rho} \defeq \bigl( \rho(x_{1},x_{1}') 
        \bigr)_{\! x_{1},x_{1}' \in \set{X}_{1} } \in \setPSD{ \set{X}_{1} }
    \end{align*}
    with trace one, where the row is indexed by $ x_{1} $ and the column indexed by $ x_{1}' $.

    \item Then the system experiences two consecutive unitary 
    evolutions, which are represented by the pair of function nodes 
    $U$ and $B$, respectively. The associated matrices $ \bigl( U(x_{2},x_{1}) \bigr)_{\!x_{2} \in \set{X}_{2},x_{1} \in \set{X}_{1}} $ with row indices $ x_{2} $ and column indices $ x_{1} $ and $ \bigl( B(x_{3},x_{2}) \bigr)_{\!x_{3} \in \set{X}_{3},x_{2} \in \set{X}_{2}} $ with row indices $ x_{3} $ and column indices $ x_{2} $ are unitary.
  \end{itemize}
  In order to transform this NFG into a DE-NFG, certain modifications are made:
  %----------------------------------------------------------------------------
  \begin{itemize}
    \item The variables $ x_{e} \in \set{X}_{e} $ 
    and $ x_{e}' \in \set{X}_{e} $ are merged into the variable 
    \begin{align*}
      \tx_{e} = ( x_{e}, x_{e}' ) \in \tset{X}_{e}, \qquad
      e \in [3].
    \end{align*}

    \item New local functions are defined based on the rearranged variables.
  \end{itemize}
  %----------------------------------------------------------------------------
  The resulting DE-NFG is shown in Fig.~\ref{sec:DENFG:fig:2}(right). 
  As an example, local function $\tilde{U}$ is defined to be 
  \begin{align*}
    \tilde{U}(\tx_{2},\tx_{1}) \defeq U(x_{2},x_{1})
    \cdot \overline{U(x_{2}',x_{1}')}, \qquad
    \tx_{2} \in \tset{X}_{2},\, 
    \tx_{1} \in \tset{X}_{1}.
  \end{align*}
  One can verify that the associated Choi-matrix representation
  $\matr{C}_{\tilde{U}}$ is a PSD matrix with row indices 
  $(x_{2},x_{1})$ and 
  column indices 
  $(x_{2}',x_{1}')$.
  \eexample
\end{example}

%***************************************************************************

%***************************************************************************
\begin{definition}\label{sec:DENFG:def:2}
  \index{SPA!on DE-NFG}
   The SPA for DE-NFG is a straightforward application of the SPA for S-NFG. Consider an arbitrary DE-NFG $\graphN$. The details of the SPA for $\graphN$ are presented as follows.
    \begin{enumerate}

        \item (Initialization) For every $e \in \setEfull$, $f \in \setpe$, the
        messages $\mu_{\ef}^{(0)}: \tset{X}_{e} \to \sC$, are initialized as some
        arbitrary function. (A typical choice is 
        $\mu_{\ef}^{(0)}(\txe) \defeq |\tset{X}_{e}|^{-1} $ for
        all $\tx_e \in \tset{X}_{e}$.) 
     
        \item\label{details of SPA in DE-NFG} (Iteration) For $t = 1, 2, 3, \ldots$, the following calculations are performed until a termination criterion is met.\footnote{The termination
          criterion is typically a combination of numerical convergence and an
          upper bound on the number of iterations.}
        \begin{enumerate}

        \item\label{sec:DENFG:def:2:item:1} For $ e = (f_{i},f_{j}) \in \setEfull $,
        we define the scaling factors $\kappa_{\efj}^{(t)} \in \sC $ to be
        %------------------------------------------------------------------------
        \begin{align*}
            \kappa_{\efj}^{(t)} &\defeq 
            \sum\limits_{\tvx_{\setpfi}}
            f_{i}\bigl( \tvx_{\setpfi} \bigr) 
            \cdot \prod\limits_{e' \in \setpfi \setminus e}
            \mu_{\epfi}^{(t-1)}(\tx_{e'}).
        \end{align*}
        %------------------------------------------------------------------------
        One can obtain the definition of $ \kappa_{\efi}^{(t)} $ by switching $ f_{i}  $ with $ f_{j} $ in the above definition.

        \item If $ \prod\limits_{e \in \setEfull} \bigl(\kappa_{\efi}^{(t)} \cdot \kappa_{\efj}^{(t)} \bigr) = 0 $, we generate random values for each element in the message vector $\vmu_{\etof}^{(t)}$ by sampling from a uniform distribution in the interval $[0, 1]$. 
        
        \item If $ \prod\limits_{e \in \setEfull} \bigl(\kappa_{\efi}^{(t)} \cdot \kappa_{\efj}^{(t)} \bigr) \neq 0 $, valid message updates are performed. 
        For each $ e = (f_{i},f_{j}) \in \setEfull $, we define the message from the edge $ e $ to 
        $ f_{j} $ via the mapping
        \begin{align*}
          \mu_{\efj}^{(t)}: \tset{X}_{e} &\to \sC, \nonumber\\
          \tze& \mapsto
          \bigl( \kappa_{\efj}^{(t)} \bigr)^{\!-1}
             \cdot
             \sum\limits_{\tvx_{\setpfj}: \, \txe = \tze}
              f_{j}\bigl( \tvx_{\setpfj} \bigr)
               \cdot \prod\limits_{e'\in \setpfj \setminus e}
                \mu_{\epfj}^{(t-1)}(\tx_{e'}).
          % \label{sec:DENFG:eqn:10}
        \end{align*}
        (If desired or necessary, other generalization procedures is allowed.)
        % \begin{align}
        %     \mu_{\efi}^{(t)}(\tze)
        %       &\defeq 
                 
        %       % \mu_{\upefi}^{(t)}(\zupe)
        %       % &\defeq 
        %       %    \bigl( \kappa_{\upefi}^{(t)} \bigr)^{\!-1}
        %       %    \cdot
        %       %    \sum\limits_{\tvx_{\setpfj}: \, \xupe = \zupe}
        %       %     f_{j}\bigl( \tvx_{\setpfj} \bigr)
        %       %      \cdot
        %       %      \mu_{\efj}^{(t-1)}(\xe)
        %       %      \cdot
        %       %      \prod\limits_{e'\in \setpfj \setminus e}
        %       %       \mu_{\epfj}^{(t-1)}(x_{e'})
        %       %       \cdot \mu_{\upepfj}^{(t-1)}(x_{\upe'}),
        %       % \quad \zupe \in \upset{X}_{e}. \label{sec:DENFG:eqn:11}
        % \end{align}
        One can obtain the definition of $ \mu_{\efi}^{(t)} $ by switching $ f_{i} $ with $ f_{j} $ in the above definition.

        \item For every $f \in \setF$, we define
        %--------------------------------------------------------------------
        \begin{align*}
          \kappa_f^{(t)} \defeq
          \sum\limits_{ \tvx_{\setpf} }
          f\bigl( \tvx_{\setpf} \bigr)
          \cdot
          \prod\limits_{e \in \setpf}
          \mu_{\ef}^{(t)}(\txe).
        \end{align*}
        %--------------------------------------------------------------------
        If $ \kappa_f^{(t)} \neq 0 $, then we define the belief function
        at the function node $f$ to be the mapping
        \begin{align}
            \beli_f: \tset{X}_{\setpf} &\to \sC, \nonumber\\
            \tvx_{\setpf} & \mapsto
              \bigl( \kappa_f^{(t)} \bigr)^{\!-1}
              \cdot
              f\bigl( \tvx_{\setpf} \bigr)
              \cdot
              \prod\limits_{e \in \setpf}
                \mu_{\ef}^{(t)}(\txe).
            \label{sec:DENFG:eqn:1}
        \end{align}
      % The Choi matrix associated with $ \beli_f^{(t)} $ is given by
      % %-----------------------------------------------------------------------
      % \begin{align*}
      %   \matr{C}_{\beli_f^{(t)}} \defeq 
      %   \Bigl( 
      %       \beli_f^{(t)}\bigl( \tvx_{\setpf},\tvx_{\upsetpf} \bigr) 
      %   \Bigr)_{ \! \tvx_{\setpf} \in \tset{X}_{\setpf}}
      % \end{align*}
      % %-----------------------------------------------------------------------
      % with row indices $ \tvx_{\setpf} $ and column indices $ \tvx_{\upsetpf} $.
      
    \item For every $e \in \setEfull$, we define
    %------------------------------------------------------------------------
    \begin{align*}
        \kappa_e^{(t)} \defeq
        \sum\limits_{\txe}
        \prod\limits_{f \in \setpe}
        \mu_{\ef}^{(t)}(\txe).
    \end{align*}
    %------------------------------------------------------------------------
    If $ \kappa_e^{(t)} \neq 0 $, then we define the belief function at edge $e$ to be the mapping
      \begin{align}
        \beli_{e}: \tset{X}_{e} &\to \sC, \nonumber\\
        \txe
        &\mapsto
          \bigl( \kappa_e^{(t)} \bigr)^{\!-1}
          \cdot
          \prod\limits_{f \in \setpe}
          \mu_{\ef}^{(t)}(\txe).
        \label{sec:DENFG:eqn:2}
      \end{align}

    \end{enumerate}
     \item At iteration $t \in \sZpp$, the SPA messages can be collected as the SPA message vector $ \vmu^{(t)} $:
      %--------------------------------------------------------------------
      \begin{align*}
          \vmu^{(t)} \defeq 
          \Bigl( 
            \vmu_{\etof}^{(t)}
          \Bigr)_{\! e \in \setpf,\, f \in \setF}, 
      \end{align*}
      %--------------------------------------------------------------------
      where
      \begin{align*}
          \vmu_{\etof}^{(t)} \defeq 
          \Bigl( \mu_{\etof}^{(t)}(\txe) \Bigr)_{\!\txe}.
      \end{align*}
      With this, the SPA message update rules at Step~\ref{details of SPA in DE-NFG} can be written as
      %--------------------------------------------------------------------
      \begin{align*}
        \vmu^{(t)} = \fSPAN\bigl( \vmu^{(t-1)} \bigr)
      \end{align*}
      %--------------------------------------------------------------------
      for some suitably defined function $\fSPAN$.

      \item \index{SPA fixed point!for DE-NFG} 
      An SPA message vector 
      $ \vmu \defeq ( \vmu_{\etof} )_{e \in \setpf, f \in \setF} $ with 
      $
        \vmu_{\etof} \defeq 
        \bigl( \mu_{\etof}(\txe) \bigr)_{ \! \txe }
      $
      is called an SPA fixed-point message vector if 
      %--------------------------------------------------------------------
      \begin{align}
        \vmu = \fSPAN(\vmu).  \label{sec:DENFG:eqn:15}
      \end{align}
      %-------------------------------------------------------------------- 
      The definitions of the belief functions $ \beli_f $ and $ \beli_e $
      evaluated by the SPA fixed-point message vector $ \vmu $ 
      are defined similarly as the definitions in~\eqref{sec:DENFG:eqn:1} 
      and~\eqref{sec:DENFG:eqn:2}, where we need to replace 
      $ \vmu_{\etof}^{(t)} $ with $ \vmu_{\etof} $ for all
      $ e \in \setpf $ and $ f \in \setF $.
      % where $ \mu_{\etof} : \tset{X}_{e} \to \sC $ is an arbitrary mapping.

  \end{enumerate}
  (If desired or necessary, other normalization procedures for messages and
  beliefs can be specified. Note that other normalization procedures for beliefs do not change the results in this thesis.)
  \edefinition
\end{definition}

\begin{assumption}
  \label{asmp: assume messages are PSD}
  In the context of DE-NFG, if not mentioned otherwise, we will only consider SPA
  message vectors $ \vmu^{(t)} $ for
  which 
  \begin{align*}
    \vmu_{\etof}^{(t)} \in \setPSD{\setxe}, \qquad 
  \end{align*}
  for all $e \in \setpf$,
  $f \in \setF$.
  Furthermore, we consider the following normalization procedure for the SPA message vectors $ \vmu^{(t)} $ and the collection of the beliefs $ \vbeli^{(t)} $:
  \begin{align*}
      \sum\limits_{\txe} \mu_{\etof}^{(t)}(\txe) &= 1, \qquad e \in \setpf, \, f \in \setF,
      \nonumber\\
      \sum\limits_{ \txe } \beli^{(t)}_{e}(\txe) &= 1, \qquad e \in \setEfull, \nonumber\\
      \sum\limits_{ \tvx_{\setpf} } \beli^{(t)}_{f}(\tvx_{\setpf}) &= 1, \qquad f \in \setF.
  \end{align*}

  Furthermore, we only consider the SPA fixed-point message vector $ \vmu $ such that
  \begin{align}
      \sum\limits_{\txe}
      \prod\limits_{f \in \setpe}
      \mu_{\ef}(\txe) \neq 0, \qquad e \in \setEfull.
      \label{eqn:assume that Ze is positive}
  \end{align}
  \eassumption
\end{assumption}

\begin{lemma}
  \label{sec:DENFG:lem:1}
  Consider running the SPA defined in Definition~\ref{sec:DENFG:def:2} on a DE-NFG $\graphN$. The resulting belief functions $ \beli_f^{(t)} $ and $ \beli_e^{(t)} $ for $ f \in \setF $ and $ e \in \setEfull $ have the following properties.
  %---------------------------------------------------------------------------
  \begin{enumerate}

    \item\label{sec:DENFG:lem:1:item:1} For all $t \in \sZpp $, we have 
    $ C_{\mu_{\ef}^{(t)}} \in \setPSD{\set{X}_{e}} $ for all $ e \in \setpf $ and $ f \in \setF $, which further implies
    $\matr{C}_{\beli_f^{(t)}} \in \setPSD{\tset{X}_{\setpf}}$ for all $f \in \setF$ and 
    $\matr{C}_{\beli_{e}^{(t)}} \in \setPSD{\set{X}_e}$ for all $e \in \setE$.

    \item\label{sec:DENFG:lem:1: PSD of SPA fixed point message} For any SPA fixed-point message vector $ \vmu $, it holds that $ \matr{C}_{\mu_{\ef}} \in \setPSD{\set{X}_{e}} $ for all $ e \in \setpf $ and $ f \in \setF $

    \item\label{sec:DENFG:lem:1:item:2} The beliefs induced by an SPA fixed-point message vector $\vmu$ satisfy the following local consistency properties:
    %-----------------------------------------------------------------------
    \begin{align*}
         \sum\limits_{ \tvz_{\setpf}:\, \tze = \tx_{e}} 
         \beli_f(\tvz_{\setpf}) 
         = \beli_{e}(\tx_{e}), \qquad 
         \txe \in \tset{X}_{e}, \, f \in \setpe,\, e \in \setEfull.
    \end{align*}
    %-----------------------------------------------------------------------
  \end{enumerate}
  %---------------------------------------------------------------------------
\end{lemma}

%***************************************************************************

\begin{proof}
  The proof for each statement in the lemma is listed as follows.
  %----------------------------------------------------------------------------
  \begin{enumerate}
    \item The statement in Item~\ref{sec:DENFG:lem:1:item:1} is a consequence of Assumption~\ref{asmp: assume messages are PSD} and a variant of Schur's product theorem, which states that the entry-wise
  product of two PSD matrices is again a PSD
  matrix.\footnote{Schur's product theorem states that the entry-wise product
    of two positive definite matrices is a positive definite matrix. Note that
    the entry-wise product of two matrices is also known as the Hadamard
    product.} 

    \item The statement in Item~\ref{sec:DENFG:lem:1: PSD of SPA fixed point message} is a straightforward result of Item~\ref{sec:DENFG:lem:1:item:1}.

    \item The statement in Item~\ref{sec:DENFG:lem:1:item:2} is a consequence of the SPA message update rules.
  \end{enumerate}
  %----------------------------------------------------------------------------
  
\end{proof}

%***************************************************************************

The following definition generalizes Definition~\ref{sec:SNFG:def:2}
from S-NFGs to DE-NFGs.

%***************************************************************************

\begin{definition}(The pseudo-dual formulation of the Bethe partition function for DE-NFGs)\label{sec:DENFG:def:3}
  \index{Pseudo-dual Bethe partition function!for DE-NFG} 
  Consider a DE-NFG $\graphN$ and let
  $\vmu $ be an SPA
  message vector for $\graphN$, as defined in~\eqref{sec:DENFG:eqn:15}.\footnote{The SPA message vector $\vmu$ does
    not necessarily have to be an SPA fixed-point message vector.} Define  the following quantities:
    %------------------------------------------------------------------------
    \begin{align*}
        Z_{e}(\vmu)
      &\defeq
        \sum\limits_{\txe}
            \mu_{\efi}(\txe)
            \cdot
            \mu_{\efj}(\txe), \qquad e \in \setEfull.
    \end{align*}
    %------------------------------------------------------------------------
   If 
   \begin{align}
    Z_{e}(\vmu) \neq 0, \qquad e \in \setEfull,
    \label{eqn:positivity of Ze}
   \end{align}
   then we define the $\vmu$-based Bethe partition function to be
  \begin{align}
    \ZBSPA(\graphN,\vmu)
      &\defeq \frac{\prod\limits_f Z_f(\vmu)}
        {\prod\limits_{e} Z_{e}(\vmu)},
        \label{sec:DENFG:eqn:3}
  \end{align}
  where
  \begin{align*}
    Z_f(\vmu)
    &\defeq
      \sum\limits_{\tvx_{\setpf}}
    f\bigl( \tvx_{\setpf} \bigr)
        \cdot
        \prod\limits_{e \in \setpf}
        \mu_{\ef}(\txe), \qquad
        f \in \setF.
  \end{align*}
  (For DE-NFGs with full and half edges, the product in the denominator on the
  right-hand side of~\eqref{sec:DENFG:eqn:3} is only over all full
  edges. Note that the expression for $\ZBSPA(\graphN,\vmu)$ is scaling
  invariant, meaning that if any of the messages in $\vmu$ is multiplied by a
  non-zero complex number, then $\ZBSPA(\graphN,\vmu)$ remains unchanged.)
  (As we will see in the following lemma, the function $ \ZBSPA(\graphN,\vmu) $ is non-negative real-valued for any SPA message vector $ \vmu $.)
  The SPA-based Bethe partition function $\ZBSPA^{*}(\graphN)$ is then defined
  to be
  \begin{align*}
    \begin{aligned}
      \ZBSPA^{*}(\graphN)
        &\defeq 
        \max\limits_{\vmu}
        \ZBSPA(\graphN,\vmu)
        \nonumber\\
        &\qquad \mathrm{s.t.} \
        \text{$\vmu$ is an SPA fixed-point message vector 
        satisfy~\eqref{eqn:positivity of Ze}}.
    \end{aligned}
  \end{align*}
  \edefinition
\end{definition}

\begin{lemma}\label{lem: non-negative of ZBSPA for DE-NFG}
  Consider a DE-NFG $\graphN$ and an SPA message vector $ \vmu $. It holds that 
  %------------------------------------------------------------------------
    \begin{align*}
        Z_f(\vmu) &\in  \sR_{\geq 0}, \qquad f \in \setF, 
        \nonumber\\
        Z_{e}(\vmu)
        & \in \sRp,
        \qquad e \in \setEfull.
        % \label{sec:DENFG:eqn:14}
    \end{align*}
    %------------------------------------------------------------------------
\end{lemma}

%***************************************************************************

\begin{proof}
    By Lemma~\ref{sec:DENFG:lem:1}, we know that $ \vmu_{\ef} \in \setPSD{\set{X}_{e}} $ for all $ e \in \setpf $ and $ f \in \setF $. Then we
    can use Schur's product theorem to show that 
    the following matrix with row indices $ \vx_{\setpf} $ and column indices $ \vx_{\setpf}' $ is a PSD matrix:
    \begin{align*}
      \Biggl( 
        f\bigl( \tvx_{\setpf} \bigr)
        \cdot
        \prod\limits_{e \in \setpf}
        \mu_{\ef}(\txe)
      \Biggr)_{ \!\! \vx_{\setpf}, \vx_{\setpf}' \in \set{X}_{\setpf} }.
    \end{align*}
    Therefore, we get
    \begin{align*}
      Z_f(\vmu) 
      = 
      \sum\limits_{\tvx_{\setpf}} 
      f\bigl( \tvx_{\setpf} \bigr)
      \cdot
      \prod\limits_{e \in \setpf}
      \mu_{\ef}(\txe) \in \sR_{\geq 0}.
    \end{align*}

    The proof of $ Z_{e}(\vmu) \in \sR_{\geq 0} $ for all $ e \in \setEfull $ is similar and thus it is omitted here.
\end{proof}

\begin{assumption}\label{sec:DENFG:remk:1}
    In Assumption~\ref{asmp: assume messages are PSD}, we assume that the inequalities in~\eqref{eqn:assume that Ze is positive} hold.
    We further assume that $\ZBSPA^{*}(\graphN)$ is well-defined, \ie, there exists an SPA fixed-point message vector $ \vmu $ that satisfies 
    %-----------------------------------------------------------------------
    \begin{align*}
        Z_{e}(\vmu) &\in \sR_{> 0}, \qquad e \in \setEfull, \nonumber\\
        \ZBSPA(\graphN,\vmu) &= \ZBSPA^{*}(\graphN).
    \end{align*}
    %-----------------------------------------------------------------------
    \eremark
\end{assumption}

\begin{proposition}
  \label{prop: sufficient condition for pd messages}
  \label{PROP: SUFFICIENT CONDITION FOR PD MESSAGES}

  A sufficient condition for $ \ZBSPA(\graphN,\vmu) \in \sR_{>0} $ for all SPA fixed-point message vector $ \vmu $ is that 
  for each $ f \in \setF $ in $ \graphN $, 
  the associated Choi-matrix representation $ \matr{C}_{f} $ is a positive definite matrix.
  
\end{proposition}
%----------------------------------------------------------------------------
\begin{proof}
    See Appendix~\ref{apx: sufficient condition for pd messages}.
\end{proof}
%----------------------------------------------------------------------------

%----------------------------------------------------------------------------
\begin{proposition}\label{sec:DENFG:prop:1}
    For any cycle-free DE-NFG $\graphN$,
    the SPA has a unique fixed-point message vector $ \vmu $, 
    and the SPA-based Bethe partition function 
    gives the exact value of the partition function, 
    \ie, 
    \begin{align*}
      \ZBSPA^{*}(\graphN) 
      = \ZBSPA(\graphN,\vmu) 
      = Z(\graphN).
    \end{align*}
\end{proposition}
%----------------------------------------------------------------------------
%----------------------------------------------------------------------------
\begin{proof}
    The proof is straightforward and thus is omitted here. 
\end{proof}

%***************************************************************************

Although \emph{formal}
Definitions~\ref{sec:SNFG:def:3} and~\ref{sec:SNFG:def:1} could be potentially extended from S-NFGs to DE-NFGs, there is a challenge in dealing with the multi-valued nature of the complex logarithm function. We leave this issue for
future research.

In conclusion, this section provides interesting insights and remarks regarding the DE-NFGs and their associated SPA. In the following, we summarize these points:
%----------------------------------------------------------------------------
\begin{enumerate}
  \item For each DE-NFG, the SPA fixed point messages are utilized to define the Bethe partition function, which serves as an approximation of the partition function.

  \item For any strict-sense DE-NFG $ \sfN $, both the partition function $Z(\graphN)$ and the Bethe partition function $\ZBSPA^{*}(\graphN)$ are non-negative real-valued functions. 

  \item Each S-NFG can be viewed as a special case of a DE-NFG. This observation indicates that the concepts, techniques, and theoretical results that are developed for DE-NFGs, can be applied to S-NFGs as well.

\end{enumerate}
%----------------------------------------------------------------------------

\section{Finite Graph Covers}
\label{sec: finite graph covers}
\label{sec:GraCov}

%***************************************************************************

%***************************************************************************

This section reviews the concept of finite graph covers of a factor graph in general~\cite{Stark:Terras:96:1, Koetter2007,Vontobel2013}.

%***************************************************************************

\begin{definition}
  A graph $\hat{\mathsf{G}} = (\hat{\set{V}}, \hat{\set{E}})$ with vertex set
  $\hat{\set{V}}$ and edge set $\hat{\set{E}}$ is the cover of a graph
  $\mathsf{G} = (\set{V}, \set{E})$ if there exists a graph homomorphism
  $\pi: \hat{\set{V}} \to \set{V}$ such that for each $v \in \set{V}$ and
  $\hat{v} \in \pi^{-1}(v)$, where $\pi^{-1}$ is the pre-image of $v$ under
  the mapping $\pi$, the neighborhood of $\hat{v}$, \ie, 
  the set of all vertices in $ \hat{\set{V}} $ adjacent to $ \hat{v} $, is
  mapped bijectively onto the neighborhood of $v$, \ie, the set of all vertices in $ \set{V} $ adjacent to $ v $. Given a cover
  $\hat{\mathsf{G}} = (\hat{\set{V}}, \hat{\set{E}})$, if there is a positive
  integer $M$ such that $|\pi^{-1}(v)| = M$ for all $v\in \set{V}$, then
  $\hat{\mathsf{G}}$ is called an $M$-cover.
  \edefinition
\end{definition}

%***************************************************************************

Without loss of generality, we can denote the vertex set of an $M$-cover by
$\hat{\set{V}} \defeq \set{V} \times [M]$ and define the mapping
$\pi: \hat{\set{V}} \to \set{V}$ to be such that $\pi\bigl( (v,m) \bigr) = v$ for all
$(v,m) \in \hat{\set{V}}$. Consider an edge in the original graph connecting nodes $v$ and $v'$, and the corresponding edge in the $M$-cover connecting $(v,m)$ and $(v',m')$. It holds that
\begin{align*}
  \Bigl( \pi \bigl( (v,m) \bigr), \pi \bigl( (v',m') \bigr) \Bigr) = (v,v').
\end{align*}
For every edge $e =(v,v')$ in the original graph, there exists a permutation $\sigma_e \in \set{S}_{[M]}$ associated with the corresponding edges in the $M$-cover such that the set of the corresponding edges can be represented by the following set:
\begin{align*}
  \left\{ 
    \Bigl( (v,m),\bigl(v',\sigma_{e}(m)\bigr) \Bigr) 
    \ \middle| \ 
    m \in [M] 
  \right\}.
\end{align*}
This observation motivates the following definition.

%***************************************************************************

\begin{definition}
  \label{sec:GraCov:def:1}
  \index{Finite graph covers} 
  Given an NFG $\mathsf{N} = (\set{F}, \set{E}, \set{X})$, we define
  $\hat{\set{N}}_{M}$ to be the set of all $M$-covers. Each $M$-cover
  $\hat{\mathsf{N}} = (\hat{\set{F}}, \hat{\set{E}})$ consists of the vertex set
  $\hat{\set{F}} \defeq \set{F} \times [M]$ and the edge set $\hat{\set{E}}$,
  where 
  \begin{align*}
    \hat{\set{E}}
      &\defeq \left\{
            \Bigl( \bigl( f_{i},m_{i} \bigr), \bigl(f_{j}, \sigma_{e}(m_{i}) \bigr) \Bigr)
         \ \middle|
           \ e = (f_{i}, f_{j}) \in \set{E},\, m_{i} \in [M]
         \right\}
  \end{align*}
  for some collection of permutations 
  $( \sigma_{e} )_{e \in \setEfull} \in \set{S}_{[M]}^{|\setEfull|}$.
  \edefinition
\end{definition}

%***************************************************************************

For a given graph $\mathsf{G}$, we consider all graphs in $\set{G}_{M}$ to be
distinct graphs, despite the fact that some graphs in the set $\set{G}_{M}$
might be isomorphic to each other. (See also the comments on labeled graph
covers after~\cite[Definition~19]{Vontobel2013}.)

%***************************************************************************

Let $\graphN$ be some S-NFG or DE-NFG. It is straightforward to
define the set $ \setcovN{M} $ of $M$-covers $\hgraphN$ of $\graphN$. (For a
technical description, see~\cite{Vontobel2013}.) More importantly, for finite graph covers of DE-NFGs, the double edges are permuted together\footnote{ A graph-cover definition based on splitting the two edges comprising a double edge could be considered. However, the resulting graph covers would not be immediately relevant for analyzing the SPA in the way that it is formulated. We leave it for future research.}

\begin{figure}[t]
  \begin{center}
    \input{figures/graph_cover_example.tex}
    \medskip
    \caption{Left: NFG $\graphN$. Right: samples of possible 2-covers
    $\mathsf{\hat N}$ of $\graphN$.\label{sec:GraCov:fig:9}}

    \vspace{1cm} 

    \input{figures/graph_cover_example_denfg.tex}
    \medskip
    \caption{
      Left: DE-NFG $\graphN$. Right: samples of possible 2-covers
      $\mathsf{\hat N}$ of $\graphN$.
    }\label{fig:Exam_Graph_cover_DE_NFG}
  
  \end{center}
\end{figure}

%***************************************************************************

\begin{example}
  % \begin{figure}[t]
  %   \centering
  %   \captionsetup{font=scriptsize}
  %   \subfloat{
  %     \input{figures/graph_cover_example.tex}
  %     % \input{figures/graph_cover_example_penfg.tex}

  %     \input{figures/graph_cover_example_denfg.tex}
  %   }
  % \end{figure}

  An S-NFG $\graphN$ is depicted in
  Fig.~\ref{sec:GraCov:fig:9}(left), consisting of $4$ vertices and $5$ edges. Fig.~\ref{sec:GraCov:fig:9}(right) shows possible $2$-covers of $\graphN$. It is mentioned that  for $ M \in \sZpp$, any $M$-cover of $\graphN$ has $M\cdot 4$ function nodes and
  $M\cdot 5$ edges.

  % A DE-NFG $\graphN$ is illustrated in
  % Fig.~\ref{sec:GraCov:fig:10}(left), which is a single-cycle DE-NFG consisting of $2$ vertices and $2$ edges. Fig.~\ref{sec:GraCov:fig:10}(right) shows possible $2$-covers of $\graphN$. Similar to S-NFG, it is stated that any $M$-cover of $\graphN$ has $M\cdot 2$ function nodes and $M\cdot 2$ edges for $ M \in \sZpp$.

  For comparison, a DE-NFG $\graphN$ is shown in
  Fig.~\ref{fig:Exam_Graph_cover_DE_NFG}(left), which consists of $4$ vertices and $5$ double edges. 
  The possible $2$-covers of $\graphN$ are
  shown in Fig.~\ref{fig:Exam_Graph_cover_DE_NFG}(right). Similarly, for $ M \in \sZpp$, each $M$-cover of $\graphN$ has $M\cdot 4$ function nodes and
  $M\cdot 5$ double edges.
  \eexample
\end{example}

%***************************************************************************
%***************************************************************************

% \subsection{Characterization of the Bethe Approximation}

%***************************************************************************

\begin{definition}\label{sec:GraCov:def:2}
  \index{Degree-$M$ Bethe partition function}
  Let $\graphN$ be an S-NFG or a DE-NFG. For any positive integer $M$, we
  define the degree-$M$ Bethe partition function of $\graphN$ to be
  \begin{align*}
    \ZBM(\graphN) 
      \defeq \sqrt[M]{
           \Bigl\langle
             Z\bigl( \hgraphN \bigr)
           \Bigr\rangle_{ \!\! \hgraphN \in \hat{\set{N}}_{M}}
         } \ ,
  \end{align*}
  where
  \begin{align*}
    \bigl\langle Z\bigl( \hgraphN \bigr) \bigr\rangle_{ \hgraphN \in \hat{\set{N}}_{M}}\
    \defeq
    \frac{1}{|\hat{\set{N}}_{M}|}
    \cdot \sum\limits_{ \hgraphN \in \hat{\set{N}}_{M} }
    Z\bigl( \hgraphN \bigr)
  \end{align*}
  represents the average (more precisely, the arithmetic mean) of
  $Z\bigl( \hgraphN \bigr)$ over all $\hgraphN \in \hat{\set{N}}_{M}$, 
  where
  $\hat{\set{N}}_{M}$ is the set of all $M$-covers of $\graphN$ (see
  Definition~\ref{sec:GraCov:def:1}).
  \edefinition
\end{definition}

%***************************************************************************

% The following characterization of S-NFGs' Bethe partition functions in
% terms of finite graph covers was given in~\cite{Vontobel2013}.

%***************************************************************************

\begin{figure}[t]
  \centering
  \begin{tikzpicture}
    \pgfmathsetmacro{\Ws}{1.3};
    % Expression on the LHS and RHS
    %----------------------------------------------------------------------------
    \begin{pgfonlayer}{main}
      \node (ZB1)     at (0,0) [] 
      {$\qquad\qquad\qquad\quad\left. \ZBM(\sfN) \right|_{M = 1} = Z(\sfN)$};
      \node (ZBM)     at (0,\Ws) [] {$\hspace{0.5 cm}\ZBM(\sfN)$};
      \node (ZBinfty) at (0,2*\Ws) {$\qquad\qquad\qquad\qquad
      \left. \ZBM(\sfN) \right|_{M \to \infty} 
      = \ZB(\sfN)$};
    %----------------------------------------------------------------------------
      \draw[]
        (ZB1) -- (ZBM) -- (ZBinfty);
    \end{pgfonlayer}
    % \begin{pgfonlayer}{glass}
    %   \node (1) at (0.25, \Ws) [] {$\ZBM(\sfN)$};
    % \end{pgfonlayer}
    %\vspace{0.3cm}
  \end{tikzpicture}
  \caption[A combinatorial characterization of the Bethe partition function.]{A combinatorial characterization of the Bethe partition function for an S-NFG.
  \label{fig: combinatorial chara for snfg}}
\end{figure}

\begin{theorem}[\!\!\cite{Vontobel2013}]
  \label{sec:GraCov:thm:1}
  \index{Graph-cover theorem!for S-NFG}
  Consider some S-NFG $\graphN$. It holds that
  \begin{align*}
    \limsup_{M\to \infty}
      \ZBM(\graphN)
      &= \ZB(\graphN).
      \nonumber
  \end{align*}
  This combinatoral characterization is visualized in Fig.~\ref{fig: combinatorial chara for snfg}.
  \etheorem
\end{theorem}

\begin{remark}
  Consider the S-NFG $ \sfN(\mtheta) $ defined in Definition~\ref{def: details of S-NFG of permanent}. 
  Vontobel in~\cite[Definition 38]{Vontobel2013a} showed that 
  \begin{align*}
    \ZBM\bigl( \sfN(\mtheta) \bigr) 
    = \permbM{M}\bigl( \mtheta \bigr),
  \end{align*}
  where the degree-$M$ Bethe permanent is defined in Definition~\ref{def:matrix:degree:M:cover:1}.
  Also the limit~\eqref{sec:1:eqn:193} in Theorem~\ref{sec:GraCov:thm:1} is equivalent to the limit proven in Proposition~\ref{sec:1:thm:8} in this case.
\end{remark}

%***************************************************************************

Given the fact that DE-NFGs are in many ways a natural generalization of
S-NFGs, one wonders if there is a characterization of the Bethe partition
function of a DE-NFG in terms of the partition functions of its graph covers
analogous to Theorem~\ref{sec:GraCov:thm:1}, which motivates the following conjecture.
Note that the right-hand side of the expression in the following conjecture
features $\ZBSPA^{*}(\graphN)$, not $\ZB(\graphN)$, because only the
former is defined for DE-NFGs.

%***************************************************************************

\begin{conjecture}
  \label{sec:GraCov:conj:1}

  Consider some DE-NFG $\graphN$. It holds that
  \begin{align}
    \limsup_{M\to \infty}
      \ZBM(\graphN)
      &= \ZBSPA^{*}(\graphN).
      \label{sec:CheckCon:eqn:25}
  \end{align}
  \econjecture
\end{conjecture}

%***************************************************************************

As already mentioned in the introduction, the proof of
Theorem~\ref{sec:GraCov:thm:1} heavily relies on
the method of types and the primal formulation of the Bethe partition function. However, the method of types does not seem to be useful
for proving Conjecture~\ref{sec:GraCov:conj:1}, and it is unclear how to generalize the primal formulation of the Bethe partition function, \ie, Definition~\ref{sec:SNFG:def:1}, to the case of DE-NFGs, so novel techniques have to be developed. In Theorem~\ref{sec:CheckCon:thm:1} in Section~\ref{sec:CheckCon}, we prove
Conjecture~\ref{sec:GraCov:conj:1} for a class of DE-NFGs
satisfying an easily checkable condition. We conjecture that it holds for more general DE-NFGs.

In the following, we provide promising numerical results supporting
Conjecture~\ref{sec:GraCov:conj:1}.

%***************************************************************************

% \input{figures/Numerical_Results_cover_conjecture_2.tex}

\input{figures/graph_cov_thm/numerical_result_main}
\begin{example}
  \label{sec:GraCov:exp:1}

  Consider the DE-NFG $\graphN$ in Fig.~\ref{sec:GraCov:fig:1}, where
  $\set{X}_e = \{ 0, 1 \}$ and $ \tset{X}_e  = \{ (0,0),(0,1),(1,0), (1,1) \} $ 
  for all $e \in \setEfull$. Here, $\graphN$ is a
  randomly generated DE-NFG, where the local functions for function nodes are independently
  generated. 

  For each instance of $\graphN$, we run the SPA until it reaches
  a fixed point and compute $\ZBSPA^{*}(\graphN)$. (We assume that
  $\ZBSPA^{*}(\graphN) = \ZBSPA(\graphN,\vmu)$, where $\vmu$ is the
  obtained SPA fixed-point message vector.) For
  each instance of $\graphN$, we compute $\ZBM(\graphN)$ for $ M \in [4]$. The following simulation results are for $3000$ randomly
  generated instances of $\graphN$. 

  Specifically, Fig.~\ref{sec:GraCov:fig:3} shows the cumulative distribution
  function (CDF) of
  $\bigl( \ZBM(\graphN) \! - \! \ZBSPA^{*}(\graphN) \bigr) /
  \ZBSPA^{*}(\graphN)$ for $ M \in \{1,2,3,4\} $. (Note that 
  $ Z_{\mathrm{B},1}(\graphN) = Z(\graphN) $.) 
  Fig.~\ref{sec:GraCov:fig:4} presents the mean $\mu$ 
  and standard deviation $\sigma$ of the empirical distribution of
  $\bigl| \ZBM(\graphN) \! - \! \ZBSPA^{*}(\graphN) \bigr| /
  \ZBSPA^{*}(\graphN)$. 

  The simulation results indicate that, as $M$ increases, the degree-$M$ Bethe
  partition function $\ZBM(\graphN)$ approaches SPA-based Bethe partition function $\ZBSPA^{*}(\graphN)$. 

  \hfill \eexample
\end{example}
%***************************************************************************

\begin{example}
  \label{sec:GraCov:exp:2}

  The setup in this example is essentially the same as the setup in Example~\ref{sec:GraCov:exp:1}. However,
  instead of the DE-NFG in Fig.~\ref{sec:GraCov:fig:1}, we consider the DE-NFG in
  Fig.~\ref{sec:GraCov:fig:5}. The simulation results for $734$ randomly
  generated instances of $\graphN$ are presented in
  Figs.~\ref{sec:GraCov:fig:7} and~\ref{sec:GraCov:fig:8} 
  \eexample
\end{example} 

% %***************************************************************************
% %***************************************************************************

\chapter[Bound for the Matrix Permanent]{Finite-Graph-Covers-based Bound for the Permanent of a Non-negative Square Matrix}
\label{CHAPT: FGC BOUND PERMANENT}
\label{chapt: finite graph cover based bound for matrix permanent}

In this chapter, we study the degree-$M$ Bethe permanent and degree-$M$ Sinkhorn permanent, whose definitions are motivated by the degree-$M$ finite graph covers of the associated S-NFG, as defined in Section~\ref{sec: snfg for permanent of nonnagtive square matrix}. In particular, in Section~\ref{sec:recursions:coefficients:1}, we present the recursions of the coefficients in expressions of the degree-$M$ Bethe permanent and degree-$M$ Sinkhorn permanent. Based on these recursions, in Section~\ref{sec:5}, we prove the main results in Section~\ref{sec:main:constributions:1}, \ie, we bound the permanent via its degree-$M$ Bethe permanent and degree-$M$ Sinkhorn permanent.
Then in Section~\ref{sec:asymptotic:1}, we gives asymptotic expressions for the coefficients in the degree-$M$ Bethe permanent and degree-$M$ Sinkhorn permanent. Finally, in Section~\ref{apx:22}, we focus on the case where $ M = 2 $ and connects some of the results in this chapter with results in other papers.

\section[Recursions for the Coefficients]{Recursions for \texorpdfstring{the \\
              Coefficients
              $\CM{n}( \vgam )$,
              $\CBM{n}( \vgam )$, 
              $\CscSgen{M}{n}( \vgam )$}{}}
\label{sec:recursions:coefficients:1}

%***************************************************************************

In the first part of this section, we prove Lemma~\ref{lem: expression of
  permanents w.r.t. C}. Afterwards, we prove some lemmas that will help us in
the proofs of Theorems~\ref{thm: inequalities for the coefficients}
and~\ref{th:main:permanent:inequalities:1} in the next section. These latter
lemmas are recursions (in $M$) on the coefficients $\CM{n}( \vgam )$,
$\CBM{n}( \vgam )$, and $\CscSgen{M}{n}( \vgam )$. Note that these recursions
on the coefficients $\CM{n}( \vgam )$, $\CBM{n}( \vgam )$, and
$\CscSgen{M}{n}( \vgam )$ are somewhat analogous to recursions for multinomial
coefficients, in particular also recursions for binomial coefficients like the
well-known recursion  
\begin{align*}
  \binom{M}{k} 
  = \binom{M-1}{k-1} 
  + \binom{M-1}{k}.
\end{align*}
(See also the discussion in
Example~\ref{example:pascal:triangle:generalization:1}).

%***************************************************************************

Throughout this section, the positive integer $n$ and the matrix
$\mtheta \in \sR_{\geq 0}^{n \times n}$ are fixed. Recall the definition of
the notation $\vsigma_{[M]}$ in Definition~\ref{sec:1:def:10}.

%***************************************************************************

\begin{definition}
  \label{sec:1:lem:43}

  Consider $M \in \sZpp $ and $ \vgam \in \Gamma_{M,n} $.
  \begin{enumerate}

  \item The coefficient $ \CM{n}( \vgam ) $ is defined to be the
    number of $ \vsigma_{[M]} $ in $ \setS_{[n]}^{M} $ such that
    $ \vsigma_{[M]} $ decomposes $ \vgam $, \ie,
    \begin{align}
      \CM{n}( \vgam )
        &\defeq
           \sum\limits_{ \vsigma_{[M]} \in \setS_{[n]}^{M} }
             \Bigl[ 
               \vgam = \bigl\langle \mP_{\sigma_{m}} \bigr\rangle_{m \in [M]}
             \Bigr] ,
               \label{sec:1:eqn:56}
    \end{align}
    where
    $\bigl\langle \mP_{\sigma_{m}} \bigr\rangle_{m \in [M]} \defeq \frac{1}{M}
    \cdot \sum\limits_{m\in [M]} \mP_{\sigma_{m}}$ and where the notation
    $[S]$ represents the Iverson bracket, \ie, $[S] \defeq 1$ if the statement
    $S$ is true and $[S] \defeq 0$ if the statement is false.
    
  \item The coefficient $ \CBM{n}(\vgam) $ is defined to be
    \begin{align*}
      \CBM{n}(\vgam) 
        &\defeq
           (M!)^{ 2n -  n^2} 
           \cdot
           \prod\limits_{i,j}
             \frac{ \bigl(M - M \cdot \gamma(i,j) \bigr)! }
                  { \bigl(M \cdot \gamma(i,j) \bigr)! } .
        % \label{sec:1:eqn:213}
    \end{align*}
    
  \item The coefficient $ \CscSgen{M}{n}(\vgam) $ is defined to be
    \begin{align*}
      \CscSgen{M}{n}(\vgam)
        &\defeq
           M^{-n\cdot M} 
           \cdot
           \frac{ (M!)^{2n} }
                { \prod\limits_{i,j}
                    \bigl( M \cdot \gamma(i,j) \bigr)! 
                }.
        % \label{sec:1:eqn:214}
    \end{align*}

  \end{enumerate}
  Note that for $M = 1$ it holds that
  \begin{align*}
    \Cgen{1}{n}(\vgam) 
      &= \CBgen{1}{n}(\vgam) 
       = \CscSgen{1}{n}(\vgam) 
       = 1 ,
           \quad \vgam \in \Gamma_{1,n}.
  \end{align*}
  \edefinition
\end{definition}

%***************************************************************************

We are now in a position to prove Lemma~\ref{lem: expression of permanents
  w.r.t. C} in Section~\ref{sec:main:constributions:1}.

%***************************************************************************

\medskip

\noindent
\textbf{Proof of Lemma~\ref{lem: expression of permanents w.r.t. C}}.
\begin{enumerate}

\item The expression in~\eqref{sec:1:eqn:43} follows from
  \begin{align*}
    \begin{aligned}
      \bigl( \perm(\mtheta) \bigr)^{\! M}
        &= \Biggl(
             \sum\limits_{\sigma \in \setS_{[n]}} 
               \prod\limits_{i \in [n]}
                 \theta\bigl( i, \sigma(i) \bigr)
           \Biggr)^{\!\!\! M} \\
        &\overset{(a)}{=}
           \sum\limits_{\vgam \in \GamMnthe} 
             \mtheta^{ M \cdot \vgam }
             \cdot
             \CM{n}( \vgam ) , 
     \end{aligned}
  \end{align*}
  where step~$(a)$ follows from expanding the right-hand side of the first
  line, along with using the Birkhoff--von Neumann theorem and
  Eq.~\eqref{sec:1:eqn:56}.

  While it is clear that $\CM{n}( \vgam ) \geq 0$ for all
    $\vgam \in \GamMnthe$, note that the first inequality
    in~\eqref{sec:1:eqn:58} in Theorem~\ref{thm: inequalities for the
      coefficients} (proven later in this thesis) can be used to show that
    $\CM{n}( \vgam ) > 0$ for all $\vgam \in \GamMnthe$.

\item 

    Eq.~\eqref{sec:1:eqn:190} follows from~\cite[Lemma~29]{Vontobel2013},
    where the role of the set $\set{B}'_M$ in~\cite[Lemma~29]{Vontobel2013} is
    played by the set $\Gamma'_{M,n}(\mtheta)$ in this thesis.

    The set $\Gamma'_{M,n}(\mtheta)$ is defined as follows. Namely, as
    mentioned in Section~\ref{sec:free:energy:functions:1}, the paper
    \cite[Section~IV]{Vontobel2013a} shows that the local marginal polytope of
    $\sfN(\mtheta)$ can be parameterized by the set of doubly stochastic
    matrices, more precisely it can be parameterized by
    $\Gamma_{n}(\mtheta)$. For arbitrary $ M \in \sZ_{\geq 1} $, the set
    $\Gamma'_{M,n}(\mtheta) \subseteq \Gamma_{n}(\mtheta)$ is then defined to
    be the parameterization of the set of pseudo-marginal vectors in the local
    marginal polytope that are $M$-cover lift-realizable pseudo-marginal
    vectors. (For the definition of $M$-cover lift-realizable pseudo-marginal
    vectors, see~\cite[Section~VI]{Vontobel2013}.)

    We claim $\Gamma'_{M,n}(\mtheta) = \Gamma_{M,n}(\mtheta)$. Indeed, it
    follows from the construction of pseudo-marginal vectors based on
    $M$-covers that $\Gamma'_{M,n}(\mtheta) \subseteq \Gamma_{M,n}(\mtheta)$,
    and it follows from the construction in the proof
    of~\cite[Lemma~29]{Vontobel2013} that for every element of
    $\Gamma_{M,n}(\mtheta)$ the corresponding pseudo-marginal vector is
    $M$-cover lift-realizable, \ie,
    $\Gamma'_{M,n}(\mtheta) \supseteq \Gamma_{M,n}(\mtheta)$.

\item Eq.~\eqref{sec:1:eqn:168} follows
  from~\cite[Theorem~2.1]{Friedland1979}.

\elemmaproof

\end{enumerate}

%***************************************************************************

The matrices that are introduced in the following definition will turn out to
be useful to formulate the upcoming results.

%***************************************************************************

\begin{definition}    
  \label{sec:1:def:10:part:2}
  \label{sec:1:def:15}

  Consider a matrix $ \vgam \in \Gamma_n $.
  \begin{enumerate}
    
  \item We define
    \begin{align*}
      \setR 
        &\defeq 
           \setR(\vgam)
         \defeq 
           \bigl\{ 
             i \in [n] 
           \bigm|
             \exists j\in [n], \text{ s.t. } 0 < \gamma(i,j) < 1 
           \bigr\} , \\
      \setC
        &\defeq 
           \setC(\vgam)
         \defeq
           \bigl\{ 
             j \in [n]
           \bigm|
             \exists i \in [n], \text{ s.t. } 0 < \gamma(i,j) < 1
           \bigr\} ,
    \end{align*}
    \ie, the set $\setR$ consists of the row indices of rows where $\vgam$
    contains at least one fractional entry and, similarly, the set $\setC$
    consists of the column indices of columns where $\vgam$ contains at least
    one fractional entry. Based on $\setR$ and $\setC$, we define
    \begin{align*}
      \vgamRCp 
        &\defeq 
           \bigl( \gamma(i,j) \bigr)_{\! (i,j) \in \setR \times \setC}.
    \end{align*}
    On the other hand, the set $[n] \setminus \setR$ consists of the row
    indices of rows where $\vgam$ contains exactly one $1$ and $0$s
    otherwise. Similarly, the set $[n] \setminus \setC$ consists of the
    column indices of columns where $\vgam$ contains exactly one $1$ and $0$s
    otherwise.
  
  \item We define $ r(\vgam) \defeq |\setR|$. One can verify the
    following statements:
    \begin{itemize}

    \item $|\setR| = |\setC|$.

    \item Either $r(\vgam) = 0$ or $r(\vgam) \geq 2$, \ie,
      $r(\vgam) = 1$ is impossible.

    \item If $r(\vgam) = 0$, then
      $\bigl| \setS_{[n]}(\vgam) \bigr| = 1$.

    \item If $r(\vgam) > 0$, then $\vgamRCp \in \Gamma_{r(\vgam)}$, \ie, the
      matrix $\vgamRCp$ is a doubly stochastic matrix of size
      $r(\vgam) \times r(\vgam)$.

  \end{itemize}

  \item Based on the matrix $\vgamRCp$, we define the matrix
    \begin{align*}
      \hvgamRCp
        &\defeq
           \bigl(
             \hgamRCp(i,j)
           \bigr)_{\! (i,j) \in \setR \times \setC} ,
    \end{align*}
    where
    \begin{align*}
      \hgamRCp(i,j)
        &\defeq 
          \frac{\gamRCp(i,j) \cdot \bigl( 1- \gamRCp(i,j) \bigr)}
               {\Biggl(
                  \prod\limits_{(i',j') \in \setR \times \setC} 
                     \bigl( 1 - \gamRCp(i',j') \bigr)
                \Biggr)^{\!\!\! 1 / r(\vgam)}
               }.
    \end{align*}
    Note that the permanent of $ \hvgamRCp $ is given by
    \begin{align}
      \perm(\hvgamRCp)
        &= \sum\limits_{\sigma \in \setS_{\setR\to\setC} }
             \prod\limits_{i \in \setR}
               \frac{\gamRCp\bigl(i,\sigma(i) \bigr) 
                     \cdot
                     \Bigl( 1- \gamRCp\bigl(i,\sigma(i) \bigr) \Bigr)
                    }
                 {\Biggl(
                    \prod\limits_{(i',j') \in \setR \times \setC} 
                       \bigl( 1 - \gamRCp(i',j') \bigr)
                  \Biggr)^{\!\!\! 1 / r(\vgam)}
                 } 
        \nonumber \\
        &= \frac{ 
                  \sum\limits_{\sigma \in \setS_{\setR\to\setC} }
                    \prod\limits_{i \in \setR}
                      \gamRCp\bigl(i,\sigma(i) \bigr)
                      \cdot
                      \Bigl( 1- \gamRCp\bigl(i,\sigma(i) \bigr) \Bigr)
                }
                { 
                  \prod\limits_{i \in \setR,j \in \setC}
                    \bigl( 1- \gamRCp(i,j) \bigr) 
                }.
        \label{eqn: permament of hvgamRCp}
    \end{align}

  \end{enumerate}
  \edefinition
\end{definition}

%***************************************************************************

\begin{example}
  \label{example:vgam:1:1}

  Consider the matrix
  \begin{align*}
    \vgam 
       \defeq
         \frac{1}{3}
         \cdot
         \begin{pmatrix}
           3 & 0 & 0 & 0 \\
           0 & 0 & 3 & 0 \\
           0 & 1 & 0 & 2 \\
           0 & 2 & 0 & 1
         \end{pmatrix}
           \in \Gamma_4.
  \end{align*}
  We obtain $\setR = \{ 3, 4 \}$, $\setC = \{ 2, 4 \}$, $r(\vgam) =
  2$, 
  \begin{align*}
    \vgamRCp
      &= \frac{1}{3}
         \cdot
         \begin{pmatrix}
           1 & 2 \\
           2 & 1
         \end{pmatrix} ,
    \qquad
    \hvgamRCp
       = \begin{pmatrix}
           1 & 1 \\
           1 & 1
         \end{pmatrix} ,
  \end{align*}
  and $\perm(\hvgamRCp) = 2$.
  \eexample
\end{example}

%***************************************************************************

\begin{definition}    
  \label{sec:1:def:10:part:2:add}

  Consider $M \in \sZ_{\geq 2}$, $ \vgam \in \Gamma_{M,n} $, and
  $ \sigma_1 \in \setS_{[n]}(\vgam) $. (The discussion after
    Definition~\ref{sec:1:def:10} ensures that such a $ \sigma_1 $ exists.)
  We define
  \begin{align*}
    \vgam_{\sigma_1}
      &\defeq
         \frac{1}{M-1} 
         \cdot
         \bigl( M \cdot \vgam - \mP_{\sigma_1} \bigr).
  \end{align*}
  Note that $ \vgam_{\sigma_1} \in \Gamma_{M-1,n} $. 
  The matrix
  $\vgam_{\sigma_1}$ can be thought of, modulo normalizations, as being
  obtained by ``peeling off'' the permutation matrix $\mP_{\sigma_1}$ from
  $\vgam$. 
  \edefinition
\end{definition}

%***************************************************************************

\begin{example}
  \label{example:vgam:1:2}

  Consider again the matrix $\vgam \in \Gamma_4$ that was introduced in
  Example~\ref{example:vgam:1:1}. Let $M = 3$. Note that
  $ \vgam \in \Gamma_{3,4} $. If $\sigma_1 \in \setS_{[n]}(\vgam)$ is chosen
  to be
  \begin{align*}
    \sigma_1(1) = 1 , \ \ 
    \sigma_1(2) = 3 , \ \ 
    \sigma_1(3) = 2 , \ \ 
    \sigma_1(4) = 4 ,
  \end{align*}
  then
  \begin{align*}
    \vgam_{\sigma_1}
       \defeq
         \frac{1}{2}
         \cdot
         \begin{pmatrix}
           2 & 0 & 0 & 0 \\
           0 & 0 & 2 & 0 \\
           0 & 0 & 0 & 2 \\
           0 & 2 & 0 & 0
         \end{pmatrix}
           \in \Gamma_{2,4}.
  \end{align*}
  On the other hand, if $\sigma_1 \in \setS_{[n]}(\vgam)$ is
  chosen to be
  \begin{align*}
    \sigma_1(1) = 1 , \ \ 
    \sigma_1(2) = 3 , \ \ 
    \sigma_1(3) = 4 , \ \ 
    \sigma_1(4) = 2 ,
  \end{align*}
  then
  \begin{align*}
    \vgam_{\sigma_1}
       \defeq
         \frac{1}{2}
         \cdot
         \begin{pmatrix}
           2 & 0 & 0 & 0 \\
           0 & 0 & 2 & 0 \\
           0 & 1 & 0 & 1 \\
           0 & 1 & 0 & 1
         \end{pmatrix}
           \in \Gamma_{2,4}.
  \end{align*}
  \eexample
\end{example}

%***************************************************************************

The following three lemmas are the key technical results that will allow us to
prove Theorems~\ref{thm: inequalities for the coefficients}
and~\ref{th:main:permanent:inequalities:1} in the next section.

%***************************************************************************

\begin{lemma}
  \label{sec:1:lem:29:copy:2}
  
  Let $ M \in \sZ_{\geq 2} $ and $ \vgam \in \Gamma_{M,n} $. It holds that
  \begin{align*}
     \CM{n}( \vgam )
        &= \sum\limits_{\sigma_1 \in \setS_{[n]}( \vgam )}
             \Cgen{M-1}{n}\bigl( 
                \vgam_{\sigma_1}
        \bigr).
    \end{align*}
\end{lemma}

%***************************************************************************

\begin{proof}
  If $r(\vgam) = 0$, then $\bigl| \setS_{[n]}( \vgam ) \bigr| = 1$ and
  the result in the lemma statement is straightforward. If
  $r(\vgam) \geq 2$, then the result in the lemma statement follows
  from the fact that all permutations in $ \setS_{[n]}(\vgam) $ are distinct
  and can be used as $ \sigma_1 $ in $ \vsigma_{[M]} $ for decomposing the
  matrix $ \vgam $.
\end{proof}

%***************************************************************************

\begin{lemma}
  \label{sec:1:lem:15:copy:2}
  \label{SEC:1:LEM:15:COPY:2}

  Let $M \in \sZ_{\geq 2}$ and $\vgam \in \Gamma_{M,n}$. Based on $\vgam$, we
  define the set $\setR$, the set $\setC$, and the matrix $\hvgamRCp$
  according to Definition~\ref{sec:1:def:15}. Moreover, for any
  $\sigma_1 \in \setS_{[n]}( \vgam )$, we define $\vgam_{\sigma_1}$ according
  to Definition~\ref{sec:1:def:10:part:2:add}. It holds
  that\footnote{Consider
       the case where both $\setR$ and $\setC$ are non-empty sets. Because each entry
      in $ \hvgamRCp $ is positive-valued, we have $\perm(\hvgamRCp) > 0$, and
      so the fraction $1 / \perm(\hvgamRCp)$ is well defined.}
  \begin{align*}
    \CBM{n}(\vgam)
      &= \frac{1}{\perm(\hvgamRCp)}
         \cdot
         \sum\limits_{\sigma_1 \in \setS_{[n]}( \vgam )}
           \CBgen{M-1}{n}
             \bigl( 
               \vgam_{\sigma_1}
             \bigr) ,
  \end{align*}
  where we define $\perm(\hvgamRCp) \defeq 1$ if $\setR$ and $\setC$ are
  empty sets.
\end{lemma}

%***************************************************************************

\begin{proof}
  See Appendix~\ref{app:sec:1:lem:15:copy:2}.
\end{proof}

%***************************************************************************

\begin{lemma}
  \label{sec:1:lem:30}
  \label{SEC:1:LEM:30}
  
  Let $M \in \sZ_{\geq 2}$ and $\vgam \in \Gamma_{M,n}$. For any
  $\sigma_1 \in \setS_{[n]}( \vgam )$, we define $\vgam_{\sigma_1}$ according
  to Definition~\ref{sec:1:def:10:part:2:add}. It holds that\footnote{Note
    that $\vgam \in \Gamma_{M,n}$ and so $\perm(\vgam) \geq n! / n^n$ by van
    der Waerden's inequality. With this, the fraction $1 / \perm(\vgam)$ is
    well defined.}
  \begin{align*}
    \CscSgen{M}{n}(\vgam)
      &\!=\! \frac{1}{
           \bigr( \chi(M) \bigl)^{\! n}
           \, \cdot \,
           \perm(\vgam)
         }
        \cdot \!\!\!
        \sum\limits_{\sigma_1 \in \setS_{[n]}( \vgam )} \!\!\!
          \CscSgen{M-1}{n}( \vgam_{\sigma_1} ) , 
    % \label{sec:1:eqn:176}
  \end{align*}
  where
  \begin{align*}
    \chi(M)
      &\defeq
         \biggl( \frac{M}{M-1} \biggr)^{\!\! M-1} ,
  \end{align*}
  which starts at $\bigl. \chi(M) \bigr|_{M = 2} \! = \! 2$
      and increases to $\bigl. \chi(M) \bigr|_{M \to \infty} \! = \! \e$.
\end{lemma}

%***************************************************************************

\begin{proof}
  See Appendix~\ref{apx:23}.
\end{proof}

%***************************************************************************

All objects and quantities in this thesis assume $M \in \sZ_{\geq
  1}$. Interestingly, it is possible to define these objects and quantities
also for $M = 0$, at least formally. However, we will not do this except in
the following example, as it leads to nicer figures (see
Figs.~\ref{fig:pascal:triangle:generalization:1}%
--\ref{fig:pascal:triangle:generalization:3}).

%***************************************************************************

\begin{example}
  \label{example:pascal:triangle:generalization:1}

  The purpose of this example is to visualize the recursions in
  Lemmas~\ref{sec:1:lem:29:copy:2}--\ref{sec:1:lem:30} for the special case
  $n = 2$, \ie, for matrices $\vgam$ of size $2 \times 2$. For,
  $k_1, k_2 \in \sZ_{\geq 0}$, define\footnote{The case $k_1 = k_2 = 0$ is
    special: the matrix $\vgam^{(0,0)}$ has to be considered as some formal
    matrix and the set $\Gamma_{0,2}$ is defined to be the set containing only
    the matrix $\vgam^{(0,0)}$.}
  \begin{align*}
    \vgam^{(k_1,k_2)}
      &\defeq
         \frac{1}{k_1 + k_2} 
         \cdot
         \begin{pmatrix}
           k_1 & k_2 \\
           k_2 & k_1
         \end{pmatrix}
           \in \Gamma_{k_1+k_2,2}.
             \quad
  \end{align*}
  Define
  \begin{align*}
    \Cgen{0}{n}\bigl( \vgam^{(0,0)} \bigr)
      &\defeq
         1 , \ 
    \CBgen{0}{n}\bigl( \vgam^{(0,0)} \bigr)
       \defeq
         1 , \ 
    \CscSgen{0}{n}\bigl( \vgam^{(0,0)} \bigr)
       \defeq
         1.
  \end{align*}
  One can the verify that with these definitions, the recursions in
  Lemmas~\ref{sec:1:lem:29:copy:2}--\ref{sec:1:lem:30} hold also for $M = 1$.

%***************************************************************************

  With these definitions in place, we can draw diagrams that can be seen as
  generalizations of Pascal's triangle,\footnote{These generalizations of
    Pascal's triangle are drawn such that $M$ increases from the bottom to the
    top. This is different from the typical drawings of Pascal's triangle
    (where $M$ increases from the top to the bottom), but matches better the
    convention used in Fig.~\ref{fig:combinatorial:characterization:1} and
    similar figures in~\cite{Vontobel2013} (\ie, \cite[Figs.~7
    and~10]{Vontobel2013}), along with better matching the fact that
    $\permbM{M}(\mtheta)$ is based on \emph{liftings} of~$\mtheta$.} see
  Figs.~\ref{fig:pascal:triangle:generalization:1}%
  --\ref{fig:pascal:triangle:generalization:3}. The meaning of these figures
  is as follows. For example, consider the quantity
  $\Cgen{3}{n}\bigl(\vgam^{(1,2)}\bigr)$. According to
  Lemma~\ref{sec:1:lem:29:copy:2}, this quantity can be obtained via
  \begin{align}
    \Cgen{3}{n}\bigl(\vgam^{(1,2)}\bigr)
      &= 1 
           \cdot
           \Cgen{2}{n}\bigl(\vgam^{(1,1)}\bigr)
         +
         1
           \cdot
           \Cgen{2}{n}\bigl(\vgam^{(0,2)}\bigr) \nonumber \\
      &= 1
           \cdot 2
         +
         1
           \cdot 1 \nonumber \\
      &= 3 .
      \label{eq:CM:recursion:special:case:1}
  \end{align}
  The recursion in~\eqref{eq:CM:recursion:special:case:1} is visualized in
  Fig.~\ref{fig:pascal:triangle:generalization:1} as follows. First, the
  vertex/box above the label ``$\Cgen{3}{n}\bigl(\vgam^{(1,2)}\bigr)$''
  contains the value of $\Cgen{3}{n}\bigl(\vgam^{(1,2)}\bigr)$. Second, the
  arrows indicate that $\Cgen{3}{n}\bigl(\vgam^{(1,2)}\bigr)$ is obtained by
  adding $1$ times $\Cgen{2}{n}\bigl(\vgam^{(1,1)}\bigr)$ and $1$ times
  $\Cgen{3}{n}\bigl(\vgam^{(0,2)}\bigr)$.

  Moreover, notice that the value of $\Cgen{3}{n}\bigl(\vgam^{(1,2)}\bigr)$
  equals the number of directed paths from the vertex labeled
  ``$\Cgen{0}{n}\bigl(\vgam^{(0,0)}\bigr)$'' to the vertex labeled
  ``$\Cgen{3}{n}\bigl(\vgam^{(1,2)}\bigr)$''. In general, for
  $k_1, k_2 \in \sZ_{\geq 0}$, the value of
  $\Cgen{k_1+k_2}{n}\bigl(\vgam^{(k_1,k_2)}\bigr)$ equals the number of
  directed paths from the vertex labeled
  ``$\Cgen{0}{n}\bigl(\vgam^{(0,0)}\bigr)$'' to the vertex labeled
  ``$\Cgen{k_1+k_2}{n}\bigl(\vgam^{(k_1,k_2)}\bigr)$''. In this context,
  recall from
  Proposition~\ref{prop:coefficient:asymptotitic:characterization:1} that
  \begin{align*}
    \Cgen{k_1+k_2}{n}\bigl(\vgam^{(k_1,k_2)}\bigr)
      &\approx
         \exp
           \Bigl(
             (k_1 + k_2) \cdot \HG'\bigl( \vgam^{(k_1,k_2)} \bigr)
           \Bigr) ,
  \end{align*}
  where the approximation is up to $o(k_1+k_2)$ in the exponent.

%***************************************************************************

  Similar statements can be made for
  $\CBgen{k_1+k_2}{n}\bigl(\vgam^{(k_1,k_2)}\bigr)$ and
  $\CscSgen{k_1+k_2}{n}\bigl(\vgam^{(k_1,k_2)}\bigr)$ in
  Fig.~\ref{fig:pascal:triangle:generalization:2} and
  Fig.~\ref{fig:pascal:triangle:generalization:3}, respectively. Namely, for
  $k_1, k_2 \in \sZ_{\geq 0}$, the value of
  $\CBgen{k_1+k_2}{n}\bigl(\vgam^{(k_1,k_2)}\bigr)$ equals the sum of weighted
  directed paths from the vertex labeled
  ``$\CBgen{0}{n}\bigl(\vgam^{(0,0)}\bigr)$'' to the vertex labeled
  ``$\CBgen{k_1+k_2}{n}\bigl(\vgam^{(k_1,k_2)}\bigr)$'', and \\
  $\CscSgen{k_1+k_2}{n}\bigl(\vgam^{(k_1,k_2)}\bigr)$ equals the sum of
  weighted directed paths from the vertex labeled
  ``$\CscSgen{0}{n}\bigl(\vgam^{(0,0)}\bigr)$'' to the vertex labeled
  ``$\CscSgen{k_1+k_2}{n}\bigl(\vgam^{(k_1,k_2)}\bigr)$''. (Here, the weight
  of a directed path is given by the product of the weights of the directed
  edges in the path.) In this context, recall from
  Proposition~\ref{prop:coefficient:asymptotitic:characterization:1} that
  \begin{align*}
  \CBgen{k_1+k_2}{n}\bigl(\vgam^{(k_1,k_2)}\bigr)
    &\approx
       \exp
         \Bigl(
           (k_1+k_2) \cdot \HBthe\bigl( \vgam^{(k_1,k_2)} \bigr)
         \Bigr) , \\
  \CscSgen{k_1+k_2}{n}\bigl(\vgam^{(k_1,k_2)}\bigr)
    &\approx
       \exp
         \Bigl(
           (k_1+k_2) \cdot \HscSthe\bigl( \vgam^{(k_1,k_2)} \bigr)
         \Bigr) ,
  \end{align*}
  where the approximation is up to $o(k_1+k_2)$ in the exponent.
  \eexample
\end{example}

%***************************************************************************

\begin{figure}
  \begin{center}
    \scalebox{0.95}{
      \input{figures/pascal_triangle1_1.tex}
    }
    \vspace{-0.3cm}
    % \medskip
    \caption[Generalization of Pascal's triangle in
      Lemma~\ref{sec:1:lem:29:copy:2}.]{Generalization of Pascal's triangle visualizing the recursion in
      Lemma~\ref{sec:1:lem:29:copy:2}.}
    \label{fig:pascal:triangle:generalization:1}

    \vspace{0.7cm}

    \scalebox{0.95}{
      \input{figures/pascal_triangle2_1.tex}
    }
    % \medskip
    \vspace{-0.3cm}
    \caption[Generalization of Pascal's triangle in
      Lemma~\ref{sec:1:lem:15:copy:2}.]{Generalization of Pascal's triangle visualizing the recursion in
      Lemma~\ref{sec:1:lem:15:copy:2}.}
    % \medskip
    \label{fig:pascal:triangle:generalization:2}
  
    \vspace{0.7cm}
%   \end{center}
% \end{figure}

% \begin{figure}
%   \centering
  \scalebox{0.95}{\input{figures/pascal_triangle3_1.tex}}
  % \medskip
  \vspace{-0.3cm}
  \caption[Generalization of Pascal's triangle in
    Lemma~\ref{sec:1:lem:30}.]{Generalization of Pascal's triangle visualizing the recursion in
    Lemma~\ref{sec:1:lem:30}.}
  % \medskip
  \label{fig:pascal:triangle:generalization:3}
  \end{center}
\end{figure}

%***************************************************************************

In the context of Example~\ref{example:pascal:triangle:generalization:1}, note
that for $n = 2$ the matrix
\begin{align*}
  \matr{1}_{2 \times 2} \defeq \begin{pmatrix} 1 & 1 \\ 1 &
  1 \end{pmatrix} 
\end{align*}
achieves the upper bound $2^{n/2} = 2$
in~\eqref{eq:ratio:permanent:bethe:permanent:1}, and with that the matrix
$\matr{1}_{2 \times 2}$ is the ``worst-case matrix'' for $n = 2$ in the sense
that the ratio $\frac{\perm(\mtheta)}{\permb(\mtheta)} = 2$ for
$\mtheta = \matr{1}_{2 \times 2}$ is the furthest away from the ratio
$\frac{\perm(\mtheta)}{\permb(\mtheta)} = 1$ for diagonal
matrices~$\mtheta$. However, this behavior of the all-one matrix for $n = 2$
is atypical. Namely, letting $\matr{1}_{n \times n}$ be the all-one matrix for
$n \in \sZ_{\geq 1}$, one obtains, according to
\cite[Lemma~48]{Vontobel2013a},
\begin{align*}
  \left.
    \frac{\perm(\mtheta)}{\permb(\mtheta)}
  \right|_{\mtheta = \matr{1}_{n \times n}}
    &= \sqrt{\frac{2 \pi n}{\e}}
       \cdot
       \bigl( 1 + o_n(1) \bigr) ,
\end{align*}
\ie, the ratio
$\frac{\perm(\matr{1}_{n \times n})}{\permb(\matr{1}_{n \times n})}$ is much
closer to $1$ than to $2^{n/2}$ for large values of $n$. (In the above
expression, the term $o_n(\ldots)$ is based on the usual little-o notation for
a function in $n$.)

%***************************************************************************

In the case of the scaled Sinkhorn permanent, the ``worst-case matrix'' is
given by the all-one matrix not just for $n = 2$, but for all
$n \in \sZ_{\geq 2}$, \ie, the ratio
$\frac{\perm(\mtheta)}{\permscs(\mtheta)} = \e^n \cdot \frac{n!}{n^n}$ for
$\mtheta = \matr{1}_{n \times n}$ is the furthest away from the value of the
ratio $\frac{\perm(\mtheta)}{\permscs(\mtheta)} = \e^n$ for diagonal matrices
$\mtheta$ for all $n \in \sZ_{\geq 2}$.

%***************************************************************************
%***************************************************************************

\ifx\withclearpagecommands\x
\clearpage
\fi

\section{Bounding the Permanent of \texorpdfstring{$\mtheta$}{}}
\label{sec:5}

%***************************************************************************
%***************************************************************************

In this section, we first prove Theorem~\ref{thm: inequalities for the
  coefficients} with the help of Lemmas~\ref{sec:1:lem:15:copy:2}
and~\ref{sec:1:lem:30}, where, in order for Lemmas~\ref{sec:1:lem:15:copy:2}
and~\ref{sec:1:lem:30} to be useful, we need to find lower and upper bounds on
$\perm(\hvgamRCp)$ and on $\perm(\vgam)$, respectively. Afterwards, we use
Theorem~\ref{thm: inequalities for the coefficients} to prove
Theorem~\ref{th:main:permanent:inequalities:1}.

%***************************************************************************
%***************************************************************************

\subsection{Bounding the Ratio \texorpdfstring{$\perm(\mtheta) /  
                                                                                     \permbM{M}(\mtheta)$}{}}

%***************************************************************************

In this subsection, we first prove the inequalities in~\eqref{sec:1:eqn:58} in
Theorem~\ref{thm: inequalities for the coefficients} that give bounds on the
ratio $\CM{n}(\vgam) / \CBM{n}(\vgam)$. Then, based on these bounds, we prove
the inequalities in~\eqref{SEC:1:EQN:147} in
Theorem~\ref{th:main:permanent:inequalities:1} that give bounds on the ratio
$\perm(\mtheta) / \permbM{M}(\mtheta)$.

%***************************************************************************

\begin{lemma}
  \label{sec:1:lem:17}

  Consider the setup in Lemma~\ref{sec:1:lem:15:copy:2}. The matrix
  $\hvgamRCp$ satisfies
  \begin{align}
    1
      &\leq  
         \perm(\hvgamRCp)
       \leq
         2^{n/2}.
           \label{sec:1:eqn:145}
  \end{align}
\end{lemma}

%***************************************************************************

\begin{proof}
  It is helpful to first prove that $\permb(\hvgamRCp) = 1$. Indeed,
  \begin{align*}
    \permb(\hvgamRCp)
      &\overset{(a)}{=} 
         \exp
           \biggl(
             - \min_{\vgam \in \Gamma_{|\setR|}(\hvgamRCp)}
                 F_{\mathrm{B},\hvgamRCp}( \vgam )
           \biggr) \\
      &\overset{(b)}{=}
         \exp\bigl( - F_{\mathrm{B},\hvgamRCp}( \vgamRCp ) \bigr) \\
      &\overset{(c)}{=} \exp(0) \\
      &= 1 ,
  \end{align*}
  where step~$(a)$ follows from the definition of the Bethe permanent in
  Definition~\ref{sec:1:def:16}, where step~$(b)$ follows from the convexity
  of the minimization problem
  $\min_{\vgam \in \Gamma_{|\setR|}(\hvgamRCp)} F_{\mathrm{B},\hvgamRCp}(
  \vgam )$ (which was proven in~\cite{Vontobel2013a}) and from studying the
  Karush–Kuhn–Tucker (KKT) 
  conditions for this minimization problem at $\vgam = \vgamRCp$ (the
  details are omitted), and where step~$(c)$ follows from simplifying the
  expression for $F_{\mathrm{B},\hvgamRCp}( \vgamRCp )$.

  With this, the first inequality in~\eqref{sec:1:eqn:145} follows from
  \begin{align*}
    1
      &= \permb(\hvgamRCp)
       \overset{(a)}{\leq}
         \perm(\hvgamRCp) ,
  \end{align*}
  where step~$(a)$ follows from the first inequality
  in~\eqref{eq:ratio:permanent:bethe:permanent:1}, and the second inequality
  in~\eqref{sec:1:eqn:145} follows from
  \begin{align*}
    \perm(\hvgamRCp) 
      &\overset{(b)}{\leq}
         \underbrace
         {
           2^{|\setR|/2}
         }_{\leq \ 2^{n/2}} 
         \cdot
         \underbrace{
           \permb(\hvgamRCp)
         }_{= \ 1}
       \leq 
         2^{n/2} ,
  \end{align*}
  where step~$(b)$ follows from the second inequality
  in~\eqref{eq:ratio:permanent:bethe:permanent:1}.

  An alternative proof of the first inequality in~\eqref{sec:1:eqn:145} goes
  as follows. Namely, this inequality can be obtained by 
  combining~\eqref{eqn: permament of hvgamRCp} and
  \begin{align*}
    \hspace{4cm}&\hspace{-4cm}
    \sum\limits_{\sigma \in \setS_{\setR \to \setC}}
      \prod\limits_{i \in \setR} 
        \gamRCp\bigl(i,\sigma(i)\bigr)
        \cdot 
        \Bigl( 1 \! - \! \gamRCp\bigl(i,\sigma(i) \bigr) \Bigr) 
      \overset{(a)}{\geq}
         \!\! \prod\limits_{(i,j) \in \setR \times \setC} \!\!
           \bigl( 1 \! - \! \gamRCp(i,j) \bigr) ,
  \end{align*}
  where step~$(a)$ follows from~\cite[Corollary~1c of
  Theorem~1]{Schrijver1998}. Note that this alternative proof for the first
  inequality in~\eqref{sec:1:eqn:145} is not really that different from the
  above proof because the first inequality
  in~\eqref{eq:ratio:permanent:bethe:permanent:1} is also based
  on~\cite[Corollary~1c of Theorem~1]{Schrijver1998}
\end{proof}

%***************************************************************************

\noindent
\textbf{Proof of the inequalities in~\eqref{sec:1:eqn:58} in Theorem~\ref{thm:
    inequalities for the coefficients}}.
\begin{itemize}

\item We prove the first inequality in~\eqref{sec:1:eqn:58} by induction on
  $M$. The base case is $M = 1$. From Definition~\ref{sec:1:lem:43}, it
  follows that $\CM{n}(\vgam) = \CBM{n}(\vgam) = 1$ for
  $\vgam \in \Gamma_{M,n}$ and so
  \begin{align*}
    \frac{\CM{n}(\vgam)}
         {\CBM{n}(\vgam)}
      &\geq 1 ,
         \quad \vgam \in \Gamma_{M,n} ,
  \end{align*}
  is indeed satisfied for $M = 1$.

  Consider now an arbitrary $M \in \sZ_{\geq 2}$. By the induction hypothesis,
  we have
  \begin{align}
    \frac{ \Cgen{M-1}{n}(\vgam) }{\CBgen{M-1}{n}(\vgam) }
      &\geq 
        1 , \qquad \vgam \in \Gamma_{M-1,n}.
    \label{eqn: ineq. used in induction hypothesis w.r.t. cbm}
  \end{align}
  Then we obtain
  \begin{align*}
    \CM{n}( \vgam )
    &\overset{(a)}{=} \sum\limits_{\sigma_1 \in \setS_{[n]}( \vgam )}
         \Cgen{M-1}{n}\bigl( 
            \vgam_{\sigma_1}
    \bigr) 
    \nonumber \\
    &\overset{(b)}{\geq}
    \sum\limits_{\sigma_1 \in \setS_{[n]}( \vgam )}
         \CBgen{M-1}{n}\bigl( 
            \vgam_{\sigma_1}
    \bigr)
    \nonumber \\ 
    &\overset{(c)}{=} \perm(\hvgamRCp) \cdot \CBM{n}(\vgam)
    \nonumber \\
    &\overset{(d)}{\geq} \CBM{n}(\vgam) ,
\end{align*}
where step~$(a)$ follows from Lemma~\ref{sec:1:lem:29:copy:2}, where step
$(b)$ follows from the induction hypothesis in~\eqref{eqn: ineq. used in
  induction hypothesis w.r.t. cbm}, where step~$(c)$ follows from
Lemma~\ref{sec:1:lem:15:copy:2}, and where step~$(d)$ follows from
Lemma~\ref{sec:1:lem:17}. 

\item A similar line of reasoning can be used to prove the second inequality
  in~\eqref{sec:1:eqn:58}; the details are omitted.

\end{itemize}
\etheoremproof

%***************************************************************************

\medskip

\noindent
\textbf{Proof of the inequalities in~\eqref{SEC:1:EQN:147} in
  Theorem~\ref{th:main:permanent:inequalities:1}}.
\begin{itemize}

\item We prove the first inequality in~\eqref{SEC:1:EQN:147} by observing that
  \begin{align*}
    \bigl( \perm(\mtheta) \bigr)^{\!M}
      &\overset{(a)}{=} 
         \sum\limits_{\vgam \in \GamMnthe} 
           \mtheta^{ M \cdot \vgam }
           \cdot
           \CM{n}( \vgam )
             \nonumber \\
      &\overset{(b)}{\geq}
         \sum\limits_{\vgam \in \GamMnthe} 
           \mtheta^{ M \cdot \vgam }
           \cdot
           \CBM{n}( \vgam )
             \nonumber \\
      &\overset{(c)}{=}
         \bigl( \permbM{M}(\mtheta) \bigr)^{\!M} ,
  \end{align*}
  where step~$(a)$ follows from~\eqref{sec:1:eqn:43}, where step~$(b)$ follows
  from the first inequality in~\eqref{sec:1:eqn:58}, and where step~$(c)$
  follows from~\eqref{sec:1:eqn:190}. The desired inequality is then obtained
  by taking the $M$-th root on both sides of the expression.

\item A similar line of reasoning can be used to prove the second inequality
  in~\eqref{SEC:1:EQN:147}; the details are omitted.

\end{itemize}
\etheoremproof

%***************************************************************************

\medskip

The expression in the following lemma can be used to obtain a slightly
alternative proof of the inequalities in~\eqref{SEC:1:EQN:147} in
Theorem~\ref{th:main:permanent:inequalities:1}. However, this result is also
of interest in itself.

%***************************************************************************

\begin{lemma}
  \label{lemma:ratio:perm:permbM:1}
  \label{LEMMA:RATIO:PERM:PERMBM:1}
  Let $\pmtheta$ be the probability mass function on $\setS_{[n]}$ induced
  by~$\mtheta$ (see Definition~\ref{sec:1:def:10}). It holds that
  \begin{align}
    \frac{\perm(\mtheta)}
         {\permbM{M}(\mtheta)} &= \left( \!
           \sum\limits_{ \vsigma_{[M]} \in \setS_{[n]}^{M} }
             \left(
               \prod\limits_{m \in [M]} \!\!
                 \pmtheta(\sigma_m)
             \right)
             \! \cdot \!
             \frac{\CBM{n}
                     \Bigl(
                       \bigl\langle \mP_{\sigma_{m}} \bigr\rangle_{m \in [M]}
                     \Bigr)}
                  {\CM{n}
                     \Bigl(
                       \bigl\langle \mP_{\sigma_{m}} \bigr\rangle_{m \in [M]}
                     \Bigr)}
         \right)^{\!\!\! -1/M} \!\! .
           \label{eq:ratio:perm:permbM:1}
  \end{align}
  Note that from the first inequality in~\eqref{sec:1:eqn:58}
    in Theorem~\ref{thm: inequalities for the coefficients} it follows that
    $ \CM{n}(\vgam) \geq \CBM{n}(\vgam) > 0 $ for all $ \vgam \in \GamMn $,
    \ie, the ratio appearing in~\eqref{eq:ratio:perm:permbM:1} is well
    defined.
\end{lemma}

%***************************************************************************

\begin{proof}
  See Appendix~\ref{app:proof:lemma:ratio:perm:permbM:1}.
\end{proof}

Based on Lemma~\ref{lemma:ratio:perm:permbM:1}, we  
give an alternative proof of the inequalities in~\eqref{SEC:1:EQN:147} in
  Theorem~\ref{th:main:permanent:inequalities:1} in Appendix~\ref{apx: alternative proof of perm geq permb}.
%***************************************************************************

%***************************************************************************
Note that the first inequality in~\eqref{SEC:1:EQN:147} is tight. Indeed,
choosing the matrix $\mtheta$ to be diagonal with strictly positive diagonal
entries shows that 
\begin{align*}
  \frac{\perm(\mtheta)}{\permbM{M}(\mtheta)} = 1, \qquad M \in \sZpp.
\end{align*}
Furthermore,
\begin{align*}
  \frac{ \perm(\mtheta) }{\permb(\mtheta)} = 
  \liminf_{M \to \infty} \frac{ \perm(\mtheta) }{\permbM{M}(\mtheta)}
  = 1.
\end{align*}

%***************************************************************************

On the other hand, the second inequality in~\eqref{SEC:1:EQN:147} is
\emph{not} tight for finite $M \geq 2$. Indeed, for this inequality to be
tight, all the inequalities
in~\eqref{eq:ratio:perm:permbM:proof:upper:bound:1} in Appendix~\ref{app:proof:lemma:ratio:perm:permbM:1} need to be tight. In
particular, the inequality in step~$(a)$ needs to be tight, which means that
for all $\vsigma_{[M]} \in \setS_{[n]}^{M}$ with
$\prod\limits_{m \in [M]} \pmtheta(\sigma_m) > 0$, it needs to hold that
$\CBM{n}(\ldots) / \CM{n}(\ldots) = (2^{n/2})^{-(M-1)}$. However, this is not
the case for $\vsigma_{[M]} = (\sigma, \ldots, \sigma)$ with
$\sigma \in \setS_{[n]}(\vtheta)$, for which it holds that
$\CBM{n}(\ldots) / \CM{n}(\ldots) = 1$.

%***************************************************************************

Observe that the non-tightness of the second inequality
in~\eqref{SEC:1:EQN:147} for finite $M \geq 2$ is in contrast to the case
$M \to \infty$, where it is known that for
\begin{align*}
  \mtheta = \matr{I}_{(n/2) \times (n/2)} \otimes \begin{pmatrix} 1
  & 1 \\ 1 & 1 \end{pmatrix},
\end{align*}
it holds that
 \begin{align*}
    \frac{\perm(\mtheta)}{ \permb(\mtheta)} 
    = \liminf_{M \to \infty} 
    \frac{\perm(\mtheta)}{\permbM{M}(\mtheta)}
    = 2^{n/2}.
 \end{align*}
 (Here, $n$ is assumed to be even and
$\matr{I}_{(n/2) \times (n/2)}$ is defined to be the identity matrix of size
$(n/2) \times (n/2)$.)

%***************************************************************************

\vspace{0.2 cm}

\noindent
We also give another alternative proof of the inequalities in~\eqref{SEC:1:EQN:147} in
  Theorem~\ref{th:main:permanent:inequalities:1} in Appendix~\ref{apx: alternative prooof of perm geq permb: 2}.
In particular, Lemma~\ref{prop: a special case of Navin conjecture} proven in Appendix~\ref{apx: alternative prooof of perm geq permb: 2} proves the $M_{1}= 1$ case of the following conjecture.
\begin{conjecture}\label{Navin conjecture}
  For any integers $ M_{1} \in \sZpp $ and $ M_{2} \in \sZpp $, it holds that
  \begin{align*}
    \frac{ \bigl( \permbM{M_{1}}(\mtheta) \bigr)^{\! M_{1}} \cdot
    \bigl( \permbM{M_{2}}(\mtheta) \bigr)^{\! M_{2}}
    }{ \bigl( \permbM{M_{1} + M_{2}}(\mtheta) \bigr)^{\! M_{1} + M_{2}}}
    \geq 1.
  \end{align*}
  \econjecture
\end{conjecture}
A useful consequence of Conjecture~\ref{Navin conjecture} is that
\begin{align*}
  \lim\limits_{M \to \infty} 
  \permbM{M}(\mtheta) 
  = \limsup_{M \to \infty} 
  \permbM{M}(\mtheta) 
  =
  \permb(\mtheta) ,
\end{align*}
where the first equality follows from Conjecture~\ref{Navin conjecture} and Fekete’s lemma~\cite{Fekete1923} (see, \eg,~\cite[Lemma 11.6]{Lint2001}), and where the second equality follows from~\cite[Theorem 39]{Vontobel2013a}.

%***************************************************************************
%***************************************************************************

\subsection{Bounding the Ratio \texorpdfstring{$\perm(\mtheta) /  
                                                                                     \permscsM{M}(\mtheta)$}{}}

%***************************************************************************

In this subsection, we first prove the inequalities in~\eqref{sec:1:eqn:180}
in Theorem~\ref{thm: inequalities for the coefficients} that give bounds on
the ratio $\CM{n}(\vgam) / \CscSgen{M}{n}(\vgam)$. Then, based on these
bounds, we prove the inequalities in~\eqref{sec:1:eqn:200} in
Theorem~\ref{th:main:permanent:inequalities:1} that give bounds on the ratio
$\perm(\mtheta) / \permscsM{M}(\mtheta)$.

%***************************************************************************

\begin{lemma}
  \label{sec:1:lem:36}

  Consider the setup in Lemma~\ref{sec:1:lem:30}. The matrix $\vgam$ satisfies
  \begin{align*}
    \frac{n!}{n^{n}}
      &\leq 
         \perm(\vgam)
       \leq 1.
  \end{align*}
\end{lemma}

%***************************************************************************

\begin{proof}
  The lower bound follows from van der Waerden's inequality proven
  in~\cite[Theorem 1]{Egorychev1981} and~\cite[Theorem 1]{Falikman1981}. The
  upper bound follows from
  \begin{align*}
    \perm(\vgam) 
      &\overset{(a)}{\leq }
         \prod\limits_{i} 
           \Biggl( \sum\limits_{j} \gamma(i,j) \Biggr)
       = 
         \prod\limits_{i} 
           1
       = 1 ,
  \end{align*}
  where step~$(a)$ follows from the observation that $\perm(\vgam)$ contains a
  subset of the (non-negative) terms in
  $\prod\limits_{i} \bigl( \sum\limits_{j} \gamma(i,j) \bigr)$.
\end{proof}

%***************************************************************************

\noindent
\textbf{Proof of the inequalities in~\eqref{sec:1:eqn:180} in Theorem~\ref{thm:
    inequalities for the coefficients}}.
\begin{itemize}

\item We prove the first inequality in~\eqref{sec:1:eqn:180} by induction on
  $M$. The base case is $M = 1$. From Definition~\ref{sec:1:lem:43}, it
  follows that $\CM{n}(\vgam) = \CscSgen{M}{n}(\vgam) = 1$ for all
  $\vgam \in \Gamma_{M,n}$ and so
  \begin{align*}
    \frac{\CM{n}(\vgam)}
         {\CscSgen{M}{n}(\vgam)}
      &\geq
         1 , 
           \qquad \vgam \in \Gamma_{M,n}.
  \end{align*}
  (Note that when $M = 1$, then the first inequality in~\eqref{sec:1:eqn:180}
  simplifies to $1$.)

  Consider now an arbitrary $ M \in \sZ_{\geq 2} $. By the induction
  hypothesis, we have
  \begin{align}
    \frac{\Cgen{M-1}{n}(\vgam)}
         {\CscSgen{M-1}{n}(\vgam)} 
      &\! \geq \!
         \Biggl( \!
           \frac{(M \! - \! 1)^{M-1}}
                {(M \! - \! 1)!} \!
         \Biggr)^{\!\!\! n} \!\!
         \cdot \!
         \biggl(
           \frac{ n! }
                { n^{n} }
         \biggr)^{\!\! M-2} \!\!\!\! ,
           \qquad \vgam \in \Gamma_{M-1,n}.
    \label{eqn: ineq. used in induction hypothesis w.r.t. cscsm}
  \end{align}
  Then we obtain
  \begin{align*}
    \CM{n}( \vgam ) 
    &\overset{(a)}{=}
       \sum\limits_{\sigma_1 \in \setS_{[n]}( \vgam )}
         \Cgen{M-1}{n}
           \bigl( 
             \vgam_{\sigma_1}
           \bigr) 
             \nonumber \\
    &\overset{(b)}{\geq}
       \biggl( \frac{(M \! - \! 1)^{M-1}}{(M \! - \! 1)!} \biggr)^{\!\! n}
       \cdot
       \biggl(
         \frac{ n! }{ n^{n} }
       \biggr)^{\!\! M-2} 
       \!\!\!\!\!\!
       \cdot
       \sum\limits_{\sigma_1 \in \setS_{[n]}( \vgam )} \!\!
         \CscSgen{M-1}{n}(\vgam_{\sigma_1})
           \nonumber \\ 
    &\overset{(c)}{=} 
       \biggl( \frac{M^{M}}{M!} \biggr)^{\!\! n}
       \cdot
       \biggl(
         \frac{ n! }{ n^{n} }
       \biggr)^{\!\! M-2} 
       \cdot
       \perm(\vgam)
       \cdot
       \CscSgen{M}{n}(\vgam)
         \nonumber \\
    &\overset{(d)}{\geq}
       \biggl( \frac{M^{M}}{M!} \biggr)^{\!\! n}
       \cdot
       \biggl(
         \frac{ n! }{ n^{n} }
       \biggr)^{\!\! M-1} 
       \cdot 
       \CscSgen{M}{n}(\vgam) ,
  \end{align*}
  where step~$(a)$ follows from Lemma~\ref{sec:1:lem:29:copy:2}, where step
  $(b)$ follows from the induction hypothesis 
  in~\eqref{eqn: ineq. used in induction hypothesis w.r.t. cscsm}, 
  where step~$(c)$ follows from
  Lemma~\ref{sec:1:lem:30}, and where step~$(d)$ follows from
  Lemma~\ref{sec:1:lem:36}.

\item A similar line of reasoning can be used to prove the second inequality
  in~\eqref{sec:1:eqn:180}; the details are omitted.

\end{itemize}
\etheoremproof

%***************************************************************************

\medskip

\noindent
\textbf{Proof of the inequalities in~\eqref{sec:1:eqn:200} in
  Theorem~\ref{th:main:permanent:inequalities:1}}.
\begin{itemize}

\item We prove the first inequality in~\eqref{sec:1:eqn:200} by observing that
  \begin{align*}
    \bigl( \perm(\mtheta) \bigr)^{\!M} 
      &\overset{(a)}{=}
         \sum\limits_{\vgam \in \GamMnthe} 
           \mtheta^{ M \cdot \vgam }
         \cdot
         \CM{n}( \vgam )
           \nonumber \\
      &\overset{(b)}{\geq} 
         \biggl( \frac{M^{M}}{M!} \biggr)^{\!\! n}
         \cdot
         \biggl(
           \frac{ n! }{ n^{n} }
         \biggr)^{\!\! M-1} 
         \! \cdot \!\!\!\!\!
         \sum\limits_{\vgam \in \GamMnthe} 
           \mtheta^{ M \cdot \vgam }
         \cdot
         \CscSgen{M}{n}( \vgam )
          \nonumber \\
      &\overset{(c)}{=} 
         \biggl( \frac{M^{M}}{M!} \biggl)^{\!\! n}
         \cdot
         \biggl(
           \frac{ n! }{ n^{n} }
         \biggl)^{\!\! M-1} 
         \cdot 
         \bigl( \permscsM{M}(\mtheta) \bigr)^{\!M} ,
  \end{align*}
  where step~$(a)$ follows from~\eqref{sec:1:eqn:43}, where step~$(b)$ follows
  from the first inequality in~\eqref{sec:1:eqn:180}, and where step~$(c)$
  follows from~\eqref{sec:1:eqn:168}. The desired inequality is then obtained
  by taking the $M$-th root on both sides of the expression.

\item A similar line of reasoning can be used to prove the second inequality
  in~\eqref{sec:1:eqn:200}; the details are omitted.

\etheoremproof
\end{itemize}

%***************************************************************************

The expression in the following lemma can be used to obtain a slightly
alternative proof of the inequalities in~\eqref{sec:1:eqn:200} in
Theorem~\ref{th:main:permanent:inequalities:1}. However, this result is also
of interest in itself.

%***************************************************************************

\begin{lemma}
  \label{lemma:ratio:perm:permscM:1}

  Let $\pmtheta$ be the probability mass function on $\setS_{[n]}$ induced
  by~$\mtheta$ (see Definition~\ref{sec:1:def:10}). It holds that
  \begin{align*}
    \frac{\perm(\mtheta)}
         {\permscsM{M}(\mtheta)} 
      &= \left( \!
           \sum\limits_{ \vsigma_{[M]} \in \setS_{[n]}^{M} } \!
             \Biggl(
               \prod\limits_{m \in [M]} \!\!\!
                 \pmtheta(\sigma_m)
             \Biggr)
             \!\! \cdot \!
             \frac{\CscSgen{M}{n}
                     \Bigl(
                       \bigl\langle \mP_{\sigma_{m}} \bigr\rangle_{m \in [M]}
                     \Bigr)}
                  {\CM{n}
                     \Bigl(
                       \bigl\langle \mP_{\sigma_{m}} \bigr\rangle_{m \in [M]}
                     \Bigr)} \!
         \right)^{\!\!\! -1/M} \!\! .
         % \label{eq:ratio:perm:permscM:1}
  \end{align*}
\end{lemma}

%***************************************************************************

\begin{proof}
  The proof of this statement is analogous to the proof of
  Lemma~\ref{lemma:ratio:perm:permbM:1} in
  Appendix~\ref{app:proof:lemma:ratio:perm:permbM:1}; one simply has to
  replace $\permbM{M}(\mtheta)$ and $\CBM{n}(\ldots)$ by
  $\permscsM{M}(\mtheta)$ and $\CscSgen{M}{n}(\ldots)$, respectively.
\end{proof}

%***************************************************************************

\noindent
\textbf{Alternative proof of the inequalities in~\eqref{sec:1:eqn:200} in
  Theorem~\ref{th:main:permanent:inequalities:1}}.

\noindent
Based on Lemma~\ref{lemma:ratio:perm:permscM:1}, an alternative proof of the
inequalities in~\eqref{sec:1:eqn:200} can be given that is very similar to the
alternative proof presented
in Appendix~\ref{apx: alternative proof of perm geq permb}; 
the details are omitted.
\etheoremproof

%***************************************************************************

\medskip

Note that the second inequality in~\eqref{sec:1:eqn:200} is tight. Indeed,
choosing the matrix $\mtheta$ to be diagonal with strictly positive diagonal
entries shows that
\begin{align*}
  \frac{ \perm(\mtheta) }{ \permscsM{M}(\mtheta) }
  = \frac{M^{n}}{(M!)^{n/M}}, \qquad M \in \sZpp.
\end{align*}
Furthermore,
\begin{align*}
  \frac{ \perm(\mtheta) }{\permscs(\mtheta)} 
  = \liminf_{M \to \infty} \frac{ \perm(\mtheta) }{ \permscsM{M}(\mtheta) }
  = \e^n.
\end{align*}

%***************************************************************************

On the other hand, the first inequality in~\eqref{sec:1:eqn:200} is \emph{not}
tight for finite $M \geq 2$. (The proof is similar to the proof of the
non-tightness of the second inequality in~\eqref{SEC:1:EQN:147} for finite
$M \geq 2$.) This is in contrast to the case $M \to \infty$, where it is known
that for $\mtheta$ being equal to the all-one matrix of size $n \times n$ it
holds that
\begin{align*}
  \frac{ \perm(\mtheta) }{ \permscs(\mtheta) } 
  = \liminf_{M \to \infty} 
  \frac{ \perm(\mtheta) }{ \permscsM{M}(\mtheta) }
  = \e^n \cdot \frac{n!}{n^n}.
\end{align*}

%***************************************************************************
%***************************************************************************

\ifx\withclearpagecommands\x
\clearpage
\fi

\section[Asymptotic Expressions for the 
              Coefficients]{Asymptotic Expressions for the 
              Coefficients \texorpdfstring{$\CM{n}( \vgam )$, 
                                   $\CBM{n}( \vgam )$, $\CscSgen{M}{n}( \vgam )$}{}}
\label{sec:asymptotic:1}

%***************************************************************************

In this section, we prove
Proposition~\ref{prop:coefficient:asymptotitic:characterization:1}. Whereas
Definitions~\ref{sec:1:def:18:part:2}, \ref{sec:1:def:16},
and~\ref{sec:1:def:20} give an analytical definition of the entropy functions
$\HG'$, $ \HBthe $, and $ \HscSthe $, respectively, the expressions in
Eqs.~\eqref{sec:1:eqn:125}--\eqref{sec:1:eqn:207} in
Proposition~\ref{prop:coefficient:asymptotitic:characterization:1} give a
combinatorial characterization of these entropy functions. The proof of
Proposition~\ref{prop:coefficient:asymptotitic:characterization:1} is based on
standard results of the method of types and considerations that are similar to
the considerations in~\cite[Sections~IV and~VII]{Vontobel2013}.

%***************************************************************************

% Throughout this section, the positive integer $n$ and the matrix
% $\mtheta \in \sR_{\geq 0}^{n \times n}$ are fixed. In fact, the matrix
% $\mtheta$ is not really needed in this section, and so, for simplicity, we
% assume that the set $\setA(\mtheta)$ (see Definition~\ref{sec:1:def:5}) is the
% largest possible set for any $\mtheta \in \sR_{\geq 0}^{n \times n}$, \ie,
% $\setA(\mtheta) = \setP_{[n]}$, where $\setP_{[n]}$ is the set of
% permutation matrices of size $n \times n$.

%***************************************************************************
\begin{definition}
  \label{sec:1:def:17}

  Fix some $ M \in \sZpp $ and $ \vgam \in \Gamma_{n} $. 
  Recall the definition of the notations
  $\vsigma_{[M]}$ and $\mP_{\vsigma_{[M]}}$ in
  Definition~\ref{sec:1:def:10}. We define the following objects.
  \begin{enumerate}

  \item Let $\mP_{\vsigma_{[M]}} \in \setA(\vgam)^{M}$. The type of
    $ \mP_{\vsigma_{[M]}} $ is defined to be
    \begin{align*}
      \bm{t}_{\mP_{\vsigma_{[M]}}}
        &\defeq 
           \bigl(
             t_{ \mP_{\vsigma_{[M]}} }( \mP_{\sigma} ) 
          \bigr)_{ \! \mP_{\sigma} \in \setA(\vgam) }
           \in \PiA{\vgam} ,
    \end{align*}
    where
    \begin{align*}
      t_{\mP_{\vsigma_{[M]}}}( \mP_{\sigma} )
        &\defeq
           \frac{1}{M} 
           \cdot 
           \Bigl| 
             \bigl\{ 
               m \in [M] 
             \bigm|
               \mP_{\sigma_{m}} = \mP_{\sigma} 
             \bigr\}
           \Bigr|.
    \end{align*}

  \item Let $ \set{B}_{\setA(\vgam)^{M}} $ be the set of all the possible types that are based on the elements in $ \setA(\vgam)^{M} $:
    \begin{align*}
      \set{B}_{\setA(\vgam)^{M}}
        &\defeq 
           \left\{ 
             \bm{t} 
               \defeq 
                 \Bigl( t( \mP_{\sigma} ) \Bigr)_{ \!\! \mP_{\sigma} \in \setA(\vgam) }
           \ \middle|
             \begin{array}{c}
               \exists \mP_{\vsigma_{[M]}} \in \setA(\vgam)^{M} \\
               \text{s.t.} \ \bm{t}_{\mP_{\vsigma_{[M]}}} = \bm{t}
             \end{array}
           \right\}.
    \end{align*}

  \item Let $ \bm{t} \in \set{B}_{\setA(\vgam)^{M}} $. The type class of
    $ \bm{t} $ is defined to be
    \begin{align*}
      \set{T}_{\bm{t}}
        &\defeq 
           \bigl\{ 
             \mP_{\vsigma_{[M]}} \in \setA(\vgam)^{M}
           \bigm|
             \bm{t}_{\mP_{\vsigma_{[M]}}} = \bm{t}
           \bigr\}.
    \end{align*}

  \end{enumerate}
  \edefinition
\end{definition}

%***************************************************************************

\begin{lemma}
  \label{sec:1:lem:21}

  Consider $ M \in \sZpp $ and $ \vgam \in \Gamma_{n} $. It holds that
  \begin{align*}
    \bigl| \set{B}_{\setA(\vgam)^{M}} \bigr|
      &= \binom{M + |\setA(\vgam)| - 1}{|\setA(\vgam)| - 1} 
           \in O\bigl( M^{|\setA(\vgam)|-1} \bigr) , \nonumber \\
    |\set{T}_{\bm{t}}| 
      &= \frac{ M! }
              { \prod\limits_{ \mP_{\sigma} \in \setA(\vgam) } 
                  \bigl( M \cdot t(\mP_{\sigma}) \bigr)! } ,
           \qquad \bm{t} \in \set{B}_{\setA(\vgam)^{M}}.
  \end{align*}
\end{lemma}
%***************************************************************************

\begin{proof}
  These are standard combinatorial results and so the details of the
  derivation are omitted.
\end{proof}

%***************************************************************************

\noindent
\textbf{Proof sketch of
  Proposition~\ref{prop:coefficient:asymptotitic:characterization:1}}.
\begin{itemize}

\item The expression in~\eqref{sec:1:eqn:125} can be proven as
  follows. Namely, let $\vgam \in \GamMn$. For a large integer $M$, we obtain
  \begin{align*}
    \CM{n}( \vgam )
      &\overset{(a)}{=}
         \sum\limits_{ \vsigma_{[M]} \in \setS_{[n]}^{M} }
           \Bigl[ 
             \vgam = \bigl\langle \mP_{\sigma_{m}} \bigr\rangle_{m \in [M]}
           \Bigr] \\
      &\overset{(b)}{=}
         \sum\limits_{ \bm{t} \in \set{B}_{\setA(\vgam)^{M}} }
           \underbrace
           {
             \sum\limits_{ \mP_{\vsigma_{[M]}} \in \set{T}_{\bm{t}}}
               \Bigl[ 
                 \vgam = \bigl\langle \mP_{\sigma_{m}} \bigr\rangle_{m \in [M]}
               \Bigr]
           }_{\defeq \ (*)} \\
      &\overset{(c)}{=}
         \sum\limits_{ \bm{t} \in \set{B}_{\setA(\vgam)^{M}} \cap \PiA{\vgam}( \vgam )}
           |\set{T}_{\bm{t}}| \\
      &\overset{(d)}{=}
         \sum\limits_{ \bm{t} \in \set{B}_{\setA(\vgam)^{M}} \cap \PiA{\vgam}( \vgam )}
           \exp\bigl( M \cdot \HG( \bm{t} ) + o( M ) \bigr) \\
      &\overset{(e)}{=}
         \max\limits_{ \bm{t} \in \set{B}_{\setA(\vgam)^{M}} \cap \PiA{\vgam}( \vgam )}
           \exp\bigl( M \cdot \HG( \bm{t} ) + o( M ) \bigr)  ,
  \end{align*}
  where step~$(a)$ follows from the definition of $\CM{n}( \vgam )$
  in~\eqref{sec:1:eqn:56}, where step~$(b)$ follows from the definition of
  $\set{B}_{\setA(\vgam)^{M}}$ and $\set{T}_{\bm{t}}$, where
  step~$(c)$ follows from the observation that $(*)$ equals
  $|\set{T}_{\bm{t}}|$ if 
  $\bm{t} \in \set{B}_{\setA(\vgam)^{M}} \cap \PiA{\vgam}(\vgam)$ and equals zero
  if $\bm{t} \in \set{B}_{\setA(\vgam)^{M}} \setminus \PiA{\vgam}(\vgam)$,
  where step~$(d)$ follows from 
  Stirling's approximation of the factorial function (see, \eg,~\cite{Robbins1955}), 
  Lemma~\ref{sec:1:lem:21}, and the definition of
  $\HG$ in Definition~\ref{sec:1:def:18}, and where step~$(e)$ follows from
  the fact that the size of the set $\set{B}_{\setA(\vgam)^{M}}$ is polynomial
  in $M$ (see Lemma~\ref{sec:1:lem:21}) and therefore can be absorbed in the
  $o(M)$ term in the exponent.
  
  In the limit $M \to \infty$, the set $\set{B}_{\setA(\vgam)^{M}}$ is dense
  in $\PiA{\vgam}( \vgam )$, and so
  \begin{align*}
    \CM{n}( \vgam )
      &= \max\limits_{ \bm{t} \in \PiA{\vgam}( \vgam )}
           \exp\bigl( M \cdot \HG( \bm{t} ) + o( M ) \bigr) \\
      &\overset{(a)}{=}
         \exp\bigl( M \cdot \HG'( \vgam ) + o( M ) \bigr) ,
  \end{align*}
  where step~$(a)$ used the definition of $\HG'$ in
  Definition~\ref{sec:1:def:18:part:2}.

\item For the proof of the expression in~\eqref{sec:1:eqn:206},
  see~\cite[Section~VII]{Vontobel2013}.

\item A similar line of reasoning as above can be used to prove the expression
  in~\eqref{sec:1:eqn:207}; the details are omitted.

  \epropositionproof

\end{itemize}

%**************************************************************************
%**************************************************************************

\ifx\withclearpagecommands\x
\clearpage
\fi

\section[Results for the \texorpdfstring{Degree-$2$}{} Case]{Results for the Permanent and the \texorpdfstring{Degree-$2$ Bethe Permanent}{}}
\label{apx:22}

%**************************************************************************

The upcoming proposition, Proposition~\ref{prop:ratio:perm:permbM:2:1}, is the
same as~\cite[Proposition~2]{KitShing2022}. The purpose of this section is
to rederive this proposition based on some of the results that have been
established so far in this thesis.

%**************************************************************************

\begin{proposition}
  \label{prop:ratio:perm:permbM:2:1}
  \label{PROP:RATIO:PERM:PERMBM:2:1}
  Let $\pmtheta$ be the probability mass function on $\setS_{[n]}$ induced
  by~$\mtheta$ (see Definition~\ref{sec:1:def:10}). It holds that
  \begin{align}
    \frac{ \perm(\mtheta) }{ \permbM{2}(\mtheta) }
      &= \Biggl(
            \sum\limits_{ \sigma_1,\sigma_{2} \in \setS_{[n]} } 
              \pmtheta( \sigma_1 )
              \cdot
              \pmtheta( \sigma_{2} ) 
              \cdot
              2^{-c(\sigma_1,\sigma_2)}
         \Biggr)^{\!\!\! -1/2} ,
           \label{eq:ratio:perm:permbM:2:result:1}
  \end{align}
  where $c(\sigma_1,\sigma_2)$ is the number of cylces of length larger than
  one in the cycle notation
  expression\footnote{\label{footnote:cycle:notation:example:1}%
    Let us recall the cycle notation for permutations with the help of an
    example. Namely, consider the permutation $\sigma: [7] \to [7]$ defined by
    $\sigma(1) = 1$, $\sigma(2) = 2$, $\sigma(3) = 4$, $\sigma(4) = 3$,
    $\sigma(5) = 6$, $\sigma(6) = 7$, $\sigma(7) = 5$. In terms of the cycle
    notation, $\sigma = (1)(2)(34)(567)$. As can be seen from this expression,
    the permutation $\sigma$ contains two cycles of length $1$ (one cycle
    involving only~$1$ and one cycle involving only~$2$), one cycle of length
    $2$ (involving $3$ and $4$), and one cycle of length $3$ (involving $5$,
    $6$, and $7$).} of the permutation ${\sigma_1 \circ \sigma_2^{-1}}$. Here,
  $\sigma_1 \circ \sigma_2^{-1}$ represents the permutation that is obtained
  by first applying the inverse of the permutation $\sigma_2$ and then the
  permutation $\sigma_1$.
  \eproposition
 \end{proposition}

%**************************************************************************

Note that $c(\sigma_1,\sigma_2)$ satisfies\footnote{While the lower bound on
  $c(\sigma_1,\sigma_2)$ is immediately clear, the upper bound on
  $c(\sigma_1,\sigma_2)$ follows from permutations of $[n]$ that contain
  $\lfloor n/2 \rfloor $ cycles of length~$2$.}
\begin{align*}
  0 
    &\leq 
       c(\sigma_1,\sigma_2) 
     \leq 
       \Bigl\lfloor \frac{n}{2} \Bigr\rfloor  , 
       \qquad \sigma_1,\sigma_{2} \in \setS_{[n]}(\mtheta).
\end{align*}
These bounds on $c(\sigma_1,\sigma_2)$ yield
\begin{align}
  1
    &\overset{(a)}{\leq}
       \frac{ \perm(\mtheta) }{ \permbM{2}(\mtheta) }
     \overset{(b)}{\leq}
       \sqrt{2^{\lfloor n/2 \rfloor}}
     \leq
        2^{n/4} ,
          \label{eq:bounds:on"ratio:perm:permbM:2:1}
\end{align}
where the inequality~$(a)$ is obtained by lower bounding the expression
in~\eqref{eq:ratio:perm:permbM:2:result:1} with the help of
$c(\sigma_1,\sigma_2) \geq 0$ and where the inequality~$(b)$ is obtained by
upper bounding the expression in~\eqref{eq:ratio:perm:permbM:2:result:1} with
the help of $c(\sigma_1,\sigma_2) \leq \lfloor n/2 \rfloor$. Note that the
lower and upper bounds in~\eqref{eq:bounds:on"ratio:perm:permbM:2:1} on the
ratio $\perm(\mtheta) / \permbM{2}(\mtheta)$ match exactly the lower and upper
bounds in~\eqref{SEC:1:EQN:147} for $M = 2$.

%**************************************************************************

\begin{remark}
  A variant of Proposition~\ref{prop:ratio:perm:permbM:2:1} (and with that
  of~\cite[Proposition 2]{KitShing2022}), appears in~\cite[Lemma
  4.2]{Csikvari2017}. The setup in~\cite{Csikvari2017} is more special than
  that in~\cite{KitShing2022} in the sense that Csikv{\'a}ri
  in~\cite{Csikvari2017} considered $\{0,1\}$-valued matrices only, not
  arbitrary non-negative square matrices. However, the result in~\cite{Csikvari2017}
  is more general than the result in~\cite{KitShing2022} in the sense that
  Csikv{\'a}ri studied matchings of size $k \in [n]$, and not just matchings
  of size $n$ (\ie, perfect matchings).  \eremark
\end{remark}

%***************************************************************************

\noindent
\textbf{Alternative proof Proposition~\ref{prop:ratio:perm:permbM:2:1}.} 
Let us now show how the results in this thesis can be used to rederive
Proposition~\ref{prop:ratio:perm:permbM:2:1}. Namely, we obtain
\begin{align*}
  \frac{\perm(\mtheta)}
       {\permbM{2}(\mtheta)}
    &\overset{(a)}{=}
       \left( \!
         \sum\limits_{ \sigma_1, \sigma_2 \in \setS_{[n]} }
           \pmtheta(\sigma_1)
           \cdot
           \pmtheta(\sigma_2)
           \cdot
           \frac{\CBtwo{n}\Bigl( 
                            \frac{1}{2} 
                            \bigl( 
                              \mP_{\sigma_1} \! + \! \mP_{\sigma_2}
                            \bigr)
                          \Bigr)}
                {\Ctwo{n}\Bigl( 
                           \frac{1}{2} 
                           \bigl( 
                             \mP_{\sigma_1} \! + \! \mP_{\sigma_2}
                           \bigr)
                         \Bigr)}
       \right)^{\!\!\! -1/2} \\
    &\overset{(b)}{=}
       \Biggl( 
         \sum\limits_{ \sigma_1, \sigma_2 \in \setS_{[n]} }
           \pmtheta(\sigma_1)
           \cdot
           \pmtheta(\sigma_2)
           \cdot
           \frac{1}
                {2^{c(\sigma_1,\sigma_2)}}
       \Biggr)^{\!\!\! -1/2} ,
\end{align*}
which is the same expression as
in~\eqref{eq:ratio:perm:permbM:2:result:1}. Here, step~$(a)$ follows from
Lemma~\ref{lemma:ratio:perm:permbM:1} for $M = 2$ and step~$(b)$ follows from
\begin{alignat}{2}
  \CBtwo{n}\biggl( 
             \frac{1}{2} 
             \bigl( 
               \mP_{\sigma_1} \! + \! \mP_{\sigma_2} 
             \bigr)
           \biggr)
    &= 1  , 
         &&\quad \sigma_1, \sigma_2 \in \setS_{[n]} ,
             \label{eq:CBtwo:result:1} \\
  \Ctwo{n}\biggl( 
            \frac{1}{2} 
            \bigl( 
              \mP_{\sigma_1} \! + \! \mP_{\sigma_2} 
            \bigr)
          \biggr)
    &= 2^{c(\sigma_1,\sigma_2)},
         &&\quad \sigma_1, \sigma_2 \in \setS_{[n]}.
             \label{eq:Ctwo:result:1}
\end{alignat}
Eqs.~\eqref{eq:CBtwo:result:1} and~\eqref{eq:Ctwo:result:1} are established in
Appendix~\ref{app:alt:proof:prop:ratio:perm:permbM:2:1}.
\epropositionproof
%***************************************************************************
%***************************************************************************

\ifx\withclearpagecommands\x
\clearpage
\fi

% \section{Conclusion and Open Problems}
% \label{sec:7}

% % %***************************************************************************

\chapter{Loop-calculus Transform (LCT)}
\label{chapt: LCT}
\label{CHAPT: LCT}

\input{figures/lct_examples/snfg/main_pic.tex}

%***************************************************************************
%***************************************************************************
Loop calculus, as defined by Chertkov and Chernyak~\cite{Chertkov2006,
  Chernyak2007} and further developed by Mori~\cite{Mori2015}, can be
formulated as a certain holographic transform~\cite{AlBashabsheh2011} applied
to an S-NFG, and will henceforth be called the loop-calculus transform
(LCT). After applying the LCT to an S-NFG, the partition function equals the sum of the Bethe partition function and some correction terms, where the sum of the correction terms are related to a weighted sum over all generalized loops in the S-NFG. 

In this chapter, we introduce the LCT for both S-NFGs and DE-NFGs. 
The LCT for DE-NFGs has been discussed by Cao in~\cite[Theorem 2.21]{Cao2021}. Let us make a comparison between the LCT proposed in~\cite{Mori2015,Cao2021} and our proposed LCT.
%----------------------------------------------------------------------------
\begin{itemize}
  \item Mori in~\cite{Mori2015} considered S-NFGs and defined
  the LCT based on the SPA fixed-point message vector that consists of positive-valued entries only. Note that he did not define the LCT for the SPA fixed-point message vector that consists of at least one zero-valued entry.

  \item Cao in~\cite{Cao2021} considered DE-NFGs and defined the LCT for DE-NFGs based on the SPA fixed-point message vector 
  where the entries are allowed to take values in the field of complex numbers.
  However, he only
  stated the constraints that the LCT has to satisfy, 
  and he did not provide explicit expressions that enable us to applying the LCT 
  for DE-NFGs.

  \item In this section, we explicitly define the LCT for both S-NFGs and DE-NFGs. Our proposed LCT is defined based on the SPA fixed-point message vector whose entries can take values in the field of complex numbers, \ie, the SPA fixed-point message vector is allowed to have zero-valued entries. Our proposed LCT generalizes the results in~\cite{Mori2015,Cao2021}.

  \item In Chapter~\ref{CHAPT: GCT DENFG}, we use the LCT for DE-NFGs to prove the graph-cover theorem, \ie, Conjecture~\ref{sec:GraCov:conj:1} for a class of DE-NFGs that satisfy an easily checkable condition.
\end{itemize}
%----------------------------------------------------------------------------

%***************************************************************************
%***************************************************************************

\section{LCT for S-NFGs}
\label{sec:LCT:SNFG:1}

% In the following, we first give a high-level introduction of our proposed LCT for S-NFGs with the help of Fig.~\ref{sec:LCT:fig:1}.

%***************************************************************************

Before giving the general definition of the LCT for S-NFGs, we consider a
specific S-NFG $\graphN$ and assume that a part of it looks as shown in
Fig.~\ref{sec:LCT:fig:3}. Let
$\vmu $ be some SPA
fixed-point message vector for $\graphN$. The LCT is a sequence of
modifications to the original S-NFG that leave the partition function
unchanged. This is done by applying opening-the-box and closing-the box
operations~\cite{Loeliger2004} as follows.
\begin{enumerate}

\item The S-NFG in Fig.~\ref{sec:LCT:fig:4} is obtained from the S-NFG
  in Fig.~\ref{sec:LCT:fig:3} by applying suitable opening-the-box
  operations~\cite{Loeliger2004}. Importantly, the new function nodes $h_i$
  are based on $\vmu$ and are defined in such a way that the partition function remains unchanged.

  For example, for the edge $1$ that connects the function nodes $f_{i}$ and $f_{j}$, as shown in Fig.~\ref{sec:LCT:fig:3}, we
  introduce the function nodes $h_3$ and $h_8$ with the local functions
  $h_3: \set{X}_{1} \times \LCTset{X}_{1} \to \sC$ and
  $h_8: \set{X}_{1} \times \LCTset{X}_{1} \to \sC$, where
  $\LCTset{X}_{1} \defeq \set{X}_{1} $. The partition function of the S-NFG remains unchanged because $h_3$ and $h_8$ are
  defined such that the following equalities hold
  \begin{align*}
    \sum\limits_{\LCT{x}_{1} \in \LCTset{X}_{1}} 
    h_3(x_{1,f_{i}}, \LCT{x}_{1}) \cdot h_8(x_{1,f_{j}}, \LCT{x}_{1}) =
    [x_{1,f_{i}} \! = \! x_{1,f_{j}}], \qquad 
    x_{1,f_{i}}, x_{1,f_{j}} \in \set{X}_{1}.
  \end{align*}

\item The S-NFG in Fig.~\ref{sec:LCT:fig:5} is essentially the same as
  the S-NFG in Fig.~\ref{sec:LCT:fig:4}.
  
\item The S-NFG in Fig.~\ref{sec:LCT:fig:6} is obtained from the S-NFG
  in Fig.~\ref{sec:LCT:fig:5} by applying suitable closing-the-box
  operations~\cite{Loeliger2004} to the S-NFG in
  Fig.~\ref{sec:LCT:fig:5}. Consider the dashed boxes on the left-hand side in Fig.~\ref{sec:LCT:fig:5}. Its exterior function is defined to be the sum, over the internal variables, of the product of the internal local functions. Replacing this dashed box by a single function node that represents this exterior function is known as closing-the-box operation~\cite{Loeliger2004}.
  For example, the local function
  $\LCT{f}_{i}$ with the set of adjacent edges $\setpfi = \{1,2,3\}$ is
  defined to be
  \begin{align*}
    \LCT{f}_{i}(\LCT{x}_{1}, \LCT{x}_{2}, \LCT{x}_{3})
      &\defeq \!\!\!\!\! 
      \sum\limits_{
          x_{1,f_{i}} \in \set{X}_{1}, x_{2,f_{i}} \in  \set{X}_{2}, 
          x_{3,f_{i}} \in \set{X}_{2}
        } \!\!\!\!\!\!\!
            f_{i}(x_{1,f_{i}}, x_{2,f_{i}}, x_{3,f_{i}})
            \cdot
            h_3(x_{1,f_{i}}, \LCT{x}_{1})
            \nonumber\\
            &\hspace{3.5 cm}
            \cdot
            h_2(x_{2,f_{i}}, \LCT{x}_{2})
            \cdot
            h_4(x_{3,f_{i}}, \LCT{x}_{3}).
  \end{align*}
  The partition function of the S-NFG remains unchanged due to the way that $\LCT{f}$ is defined.
  
\end{enumerate}

%***************************************************************************

Let us now give the general definition of the LCT for an S-NFG $\graphN$ that
we use in this thesis, along with stating its main properties.

%***************************************************************************

\begin{definition}
  \label{sec:LCT:def:1}
  \index{LCT!for S-NFG}
  Consider some S-NFG $\graphN$ and let
  $\vmu $ be some SPA
  fixed-point message vector. 
  In the following, for each $ e = (f_{i}, f_{j}) \in \setEfull $, we use the
  short-hand notation $Z_e$ for $Z_e(\vmu)$, where
  (see~\eqref{sec:SNFG:eqn:6})
  \begin{align*}
    Z_e
      = \sum\limits_{x_e} \mu_{\efi}(\xe) \cdot \mu_{\efj}(\xe).
  \end{align*}
  (Recall that in Assumption~\ref{asmp: assume messages are non-negative}, we assume that $ Z_e(\vmu) >0 $ for all $ e \in \setEfull $.) 
  We apply the following steps.
  \begin{enumerate}
  
  \item (generalization of the first step above) For every edge $e \in
    \setEfull$, we do the following.
  
    \begin{itemize}

    \item For every edge $e \in \setE$, we pick a special element of
      $\set{X}_e$. In order to simplify the notation, and without loss of
      generality, this special element will be called $0$ for all edges, \ie, we suppose that $ 0 \in \setxe $.
  
    \item We define the alphabet
      $\LCTset{X}_{e} \defeq \set{X}_e$.\footnote{The set $ \LCTset{X}_{e} $ can be different from $ \set{X}_e $. Essentially, the only requirement is that 
      $ |\LCTset{X}_{e}| = |\set{X}_e|$. In this thesis, we use the definition $ \LCTset{X}_{e} \defeq \set{X}_e $ for simplicity.}
  
    \item The edge $e$ (with associated variable $x_e$) connecting $f_{i}$ to $f_{j}$
      is replaced by the following alternating sequence of edges and function
      nodes: an edge with associated variable $x_{e,f_{i}} \in \set{X}_e$, a
      function node $M_{\efi}$, an edge with associated variable
      $\LCT{x}_e \in \LCTset{X}_{e}$, a function node $M_{\efj}$, and an edge
      with associated variable $x_{e,f_{j}} \in \set{X}_e$. Here:
      \begin{itemize}
  
      \item The function node $M_{\efi}$ is associated with the local function
        \begin{align*}
          M_{\efi}: \set{X}_e \times \LCTset{X}_{e} \to \sR.
        \end{align*}
      
      \item The function node $M_{\efj}$ is associated with the local function
        \begin{align*}
          M_{\efj}: \set{X}_e \times \LCTset{X}_{e} \to \sR.
        \end{align*}
  
      \end{itemize}
      In order to ensure that the partition function of the resulting factor graph equals the
      partition function of the original S-NFG, the local functions $M_{\efi}$
      and $M_{\efj}$ have to satisfy
      \begin{alignat}{2}
        \sum\limits_{\LCT{x}_e}
          M_{\efi}(x_{\efi}, \LCT{x}_e)
          \cdot
          M_{\efj}(x_{\efj}, \LCT{x}_e) 
          &= [x_{\efi} \! = \!  x_{\efj}],
                \qquad x_{e,f_i}, x_{e,f_{j}} \in \set{X}_e,
                  \label{sec:LCT:exp:1}
      \end{alignat}
      where $ \sum\limits_{\LCT{x}_e} $ represents $ \sum\limits_{\LCT{x}_e \in \LCTset{X}_{e}} $ for simplicity.
      Note that~\eqref{sec:LCT:exp:1} implies\footnote{Let
        $\matr{M}_1$ and $\matr{M}_2$ be square matrices satisfying
        $\matr{M}_1 \cdot \matr{M}_2^\tran = \matr{I}$, where $\matr{I}$ is
        the identity matrix of the corresponding size. From linear algebra it
        is well known that this implies
        $\matr{M}_2^\tran \cdot \matr{M}_1 = \matr{I}$. Transposing both
        sides, one obtains $\matr{M}_1^\tran \cdot \matr{M}_2 = \matr{I}$.}
      \begin{alignat}{2}
        \sum\limits_{x_e}
          M_{\efi}(\xe, \LCT{x}_{e,f_{i}})
          \cdot
          M_{\efj}(\xe, \LCT{x}_{e,f_{j}})
          &= \bigl[ \LCT{x}_{e,f_{i}} \! = \!  \LCT{x}_{e,f_{j}} \bigr]
          \qquad \LCT{x}_{e,f_{i}}, \LCT{x}_{e,f_{j}} \in \LCTset{X}_{e}.
          \label{sec:LCT:exp:2}
      \end{alignat}
  
    \item Let $\zeta_{\efi}$ and $\zeta_{\efj}$ be some arbitrary non-zero
      real-valued constants. At the heart of the LCT are the following definitions for 
      $M_{\efi}$ and $M_{\efj}$:
      \begin{alignat}{2}
        M_{\efi}(\xe,\LCT{x}_e)
          &= \zeta_{\efi}
               \cdot
               \mu_{\efi}(\xe),
                 \qquad x_e \in \set{X}_e, \ \LCT{x}_e = 0,
                     \label{sec:LCT:exp:3} \\
        M_{\efj}(\xe,\LCT{x}_e)
          &= \zeta_{\efj}
               \cdot
               \mu_{\efj}(\xe),
                 \qquad x_e \in \set{X}_e, \ \LCT{x}_e = 0.
                     \label{sec:LCT:exp:4}
      \end{alignat}
      Evaluating~\eqref{sec:LCT:exp:2} for $\LCT{x}_{\efi} = 0$ and
      $\LCT{x}_{\efj} = 0$, yields the following constraints on $\zeta_{\efi}$
      and $\zeta_{\efj}$:
      \begin{align}
        \zeta_{\efi}
        \cdot
        \zeta_{\efj}
          &= Z_e^{-1}.
               \label{sec:LCT:exp:5}
      \end{align}
      % where we use $Z_e$ for $Z_{e}(\vmu)$ as defined in~\eqref{sec:SNFG:eqn:6} for simplicity.

    \item A possible choice for the remaining function values of $M_{\efi}$ and
      $M_{\efj}$ is given as follows:
      \begin{align}
        \nonumber\\
        \nonumber\\
        &\hspace*{-1.6 cm}
        M_{\efi}(\xe,\LCT{x}_e)
          \defeq \zeta_{\efi}
                \cdot
                \chi_{\efi}
          \nonumber\\
          &\hspace*{-1.5 cm} \cdot
              \begin{cases}
                  - \mu_{\efj}(\LCT{x}_e)
                    & x_e = 0, \ 
                       \LCT{x}_e \in \LCTset{X}_{e}\setminus\{0\} \\
                  \delta_{\efi} 
                    \cdot
                    [x_e \!=\! \LCT{x}_e] 
                  +
                  \epsilon_{\efi}
                    \cdot
                    \mu_{\efi}(\xe) 
                    \cdot 
                    \mu_{\efj}(\LCT{x}_e)
                    & x_e \in \set{X}_e \setminus \{ 0 \}, \ 
                       \LCT{x}_e \in \LCTset{X}_{e}\setminus\{0\}
                \end{cases} 
        \label{sec:LCT:eqn:14},\\
        &\hspace*{-1.6 cm}
        M_{\efj}(\xe,\LCT{x}_e)
          \defeq 
          \zeta_{\efj}
                \cdot
                \chi_{\efj}
          \nonumber\\
          &\hspace*{-1.5 cm} 
            \cdot
                \begin{cases}
                  - \mu_{\efi}(\LCT{x}_e)
                    & x_e = 0, \ 
                       \LCT{x}_e \in \LCTset{X}_{e}\setminus\{0\} \\
                  \delta_{\efj} 
                    \cdot
                    [x_e \!=\! \LCT{x}_e] 
                  +
                  \epsilon_{\efj}
                    \cdot
                    \mu_{\efj}(\xe) 
                    \cdot 
                    \mu_{\efi}(\LCT{x}_e)
                    & x_e \in \set{X}_e \setminus \{ 0 \}, \ 
                       \LCT{x}_e \in \LCTset{X}_{e}\setminus\{0\}
                \end{cases},
            \label{sec:LCT:eqn:15}
      \end{align}
      where $\chi_{\efi}$ and $\chi_{\efj}$ are real-valued constants satisfying
      \begin{align}
        \chi_{\efi}
        \cdot
        \chi_{\efj}
          &= 1,
               \label{sec:LCT:exp:6}
      \end{align}
      where $\delta_{\efi}$, $\delta_{\efj}$ are real-valued constants satisfying
      \begin{align}
        \delta_{\efi} 
          \cdot
          \delta_{\efj}
          &= Z_e,
               \label{sec:LCT:exp:7}
      \end{align}
      where for $ \beli_e(0) \neq 1 $, the constants $\delta_{\efi}$, $\delta_{\efj}$, $\epsilon_{\efi}$, and
      $\epsilon_{\efj}$ satisfy
      \begin{align}
        \delta_{\efi} 
        + 
        Z_e
          \cdot 
          \bigl( 1 - \beli_e(0) \bigr)
          \cdot
          \epsilon_{\efi}
          &= \mu_{\efj}(0),
               \label{sec:LCT:exp:8} \\
        \delta_{\efj} 
        + 
        Z_e
          \cdot 
          \bigl( 1 - \beli_e(0) \bigr)
          \cdot
          \epsilon_{\efj}
          &= \mu_{\efi}(0),
               \label{sec:LCT:exp:9}
      \end{align}
      where for $ \beli_e(0) = 1 $, the constants $\delta_{\efi}$, $\delta_{\efj}$ satisfy~\eqref{sec:LCT:exp:8} and~\eqref{sec:LCT:exp:9}, respectively, and the constants $\epsilon_{\efi}$ and
      $\epsilon_{\efj}$ satisfy
      \begin{align}
        1
        +
        \delta_{\efi}
          \cdot
          \epsilon_{\efj}
        +
        \delta_{\efj} 
          \cdot
          \epsilon_{\efi} &= 0, 
          \label{extra constraint on epsilon for beli = [x = 0]}
      \end{align}
      and where $\beli_e(0) = \mu_{\efi}(0) \cdot \mu_{\efj}(0) / Z_e$ (see
      Definition~\ref{def:belief at SPA fixed point for S-NFG}). 
      The proof showing that the above gives a valid
      choice for $M_{\efi}$ and $M_{\efj}$ is given in
      Appendix~\ref{apx:LCT SNFGs}.

      Let $ \matr{M}_{\efi} $ and $ \matr{M}_{\efj} $ to be the matrices associated with $ M_{\efi} $ and $ M_{\efi} $, respectively, 
      with rows indices $ \xe $ and column indices $ \LCT{x}_e $.
      If $ \mu_{\efi}(0) = \mu_{\efj}(0) = 1 $ and $\mu_{\efi}(\xe) = \mu_{\efj}(\xe) = 0$ for $ \xe \in \setxe \setminus \{0\} $, then from~\eqref{sec:LCT:eqn:14} and~\eqref{sec:LCT:eqn:15}, the matrices $ \matr{M}_{\efi} $ and $ \matr{M}_{\efj} $ become
      %--------------------------------------------------------------------
      \begin{align*}
          \matr{M}_{\efi} = \matr{M}_{\efj} = \matr{I},
      \end{align*}
      %--------------------------------------------------------------------
      where $ \matr{I} $ is the identity matrix of size $ |\setxe| \times |\LCTset{X}_{e}| $.
      To sum up, in this case, the LCT is to set “0” in the associated off-diagonal entries in the matrices $ \matr{M}_{\efi} $ and $ \matr{M}_{\efj} $.

            % suppose that the row $ \matr{M}_{\efi}(\ze,:) $ is at the $ (i_{\ze}-1) $-th row in the matrix $ \matr{M}_{\efi} $ with $ i_{\ze} > 1 $. 
            % the expressions for
            % $ M_{\efi}(\ze,:) $, $ M_{\efi}(:,i_{\ze}-1) $, 
            % $ M_{\efj}(\ze,:) $, and $ M_{\efj}(:,i_{\ze}-1) $ becomes
            % %--------------------------------------------------------------------
            % \begin{align*}
            %     M_{\efi}(\ze,\LCT{x}_e) &= 
            %     \zeta_{\efj}
            %     \cdot
            %     \chi_{\efj}
            %     \cdot [ \LCT{x}_e = i_{\ze}-1 ], 
            %     \qquad ( \LCT{x}_e \in \LCTset{X}_{e} ), \nonumber\\
            %     M_{\efi}(\xe,i_{\ze}-1) &= 
            %     \zeta_{\efj}
            %     \cdot
            %     \chi_{\efj}
            %     \cdot [ \xe = \ze ], 
            %     \qquad ( \xe \in \setxe ),\nonumber\\
            %     M_{\efj}(\ze,\LCT{x}_e) &= M_{\efi}(\ze,\LCT{x}_e), \qquad 
            %     ( \LCT{x}_e \in \LCTset{X}_{e} ), \nonumber\\
            %     M_{\efj}(\xe,i_{\ze}-1) &= M_{\efi}(\xe,i_{\ze}-1), \qquad 
            %     ( \xe \in \setxe ).
            % \end{align*}
            % %--------------------------------------------------------------------
  
      Note that when $|\set{X}_e| = 2$, say $\set{X}_e = \{ 0, 1 \}$, and,
      correspondingly, $\LCTset{X}_{e} = \{ 0, 1 \}$, the function values of
      $M_{\efi}(1,1)$ and $M_{\efj}(1,1)$ are uniquely fixed:
        \begin{align*}
            M_{\efi}(1,1)
            &= \zeta_{\efi}
            \cdot
            \chi_{\efi}
            \cdot
            \mu_{\efj}(0), \\
            M_{\efj}(1,1)
            &= \zeta_{\efj}
            \cdot
            \chi_{\efj}
            \cdot
            \mu_{\efi}(0),
        \end{align*}
      \ie, although different choices are possible for $\delta_{\efi}$,
      $\delta_{\efj}$, $\epsilon_{\efi}$, and $\epsilon_{\efj}$, they all
      lead to the same function values for $M_{\efi}(1,1)$ and
      $M_{\efj}(1,1)$.
  
    \end{itemize}
  
  \item (generalization of the second step above) Nothing is done in this step.
  
  \item (generalization of the third step above) For every function node
    $f \in \setF$, we do the following.
    \begin{itemize}
  
    \item Let 
      $\LCTvxf \defeq ( \LCT{x}_e )_{e \in \setpf} \in \LCTsetx_{\setpf}$ and 
      $ \LCTsetx_{\setpf} \defeq \prod\limits_{e \in \setpf} \LCTset{X}_{e} $.
  
    % \item Assume that $\setpf = \{ e_1, \ldots, e_{|\setpf|} \}$.
  
    \item Define
      \begin{align}
        \LCT{f}(\LCTvxf)
          &\defeq \sum\limits_{\vx_{\setpf}}
            f(\vx_{\setpf})
            \cdot
            \prod\limits_{e \in \setpf}
            M_{\ef}(x_{e}, \LCT{x}_e), \qquad 
            \LCTvxf \in \LCTsetx_{\setpf}.
        \label{sec:LCT:exp:10}
      \end{align}
  
    \item Finally, the LCT $\LCT{\graphN}$ of $\graphN$ based on $\vmu$ is
      defined to be the NFG with function node set $\setF$, edge set $\setE$,
      edge-function node incidence as for $\graphN$, where for every function
      node $f \in \setF$, the associated local function is $\LCT{f}$,
      and where for every edge $e \in \setE$, the associated variable 
      is $\LCT{x}_e$ and the associated alphabet is $ \LCTset{X}_{e} $.
      With this, the global function of $\LCT{\graphN}$ is defined to be the mapping
      \begin{align*}
          \LCT{g}: \prod\limits_{e} \LCTset{X}_{e} &\to \sR, \nonumber\\
          \LCTv{x} &\mapsto \prod_f \LCT{f}(\LCTvxf).
      \end{align*}

      % \footnote{Note that we use the same labels to denote
      %   vertices and edges in $\LCT{\graphN}$ as in $\graphN$. However, now the
      %   local function associated with function node $f \in \setF$ is $\LCT{f}$
      %   and the variable associated with edge $e \in \setE$ is $\LCT{x}_e \in \LCTset{X}_{e}$.}
  
    \end{itemize}

  \end{enumerate}
  \edefinition
\end{definition}

\begin{remark}
  \label{sec:LCT:rem:1}
  The following remarks provide insights into the limitations and characteristics of the LCT in different scenarios.
  %----------------------------------------------------------------------------
  \begin{enumerate}
    \item There are S-NFGs where it is not possible to obtain $ M_{\efi} $ and $ M_{\efj} $ for $ e =(f_{i}, f_{j}) \in \setEfull $ based on the SPA fixed-point messages. 
    Consider a single-cycle S-NFG in Item~\ref{sec:SNFG:remk:1:item:2} in Remark~\ref{sec:SNFG:remk:1}. 
    % with local functions 
    % $\left(\begin{smallmatrix}
    %     1 & 1 \\ 0 & 1
    % \end{smallmatrix}\right)$ 
    % and 
    % $\left(\begin{smallmatrix}
    %     1 & 0 \\ 0 & 1
    % \end{smallmatrix}\right)$. 
    The SPA fixed-point message vector satisfies
    $ Z_e = 0 $ for all $ e \in \setEfull $ and thus the LCT does not work due to the equations~\eqref{sec:LCT:exp:3}--\eqref{sec:LCT:exp:5}. The existence of the LCT for this single-cycle S-NFG would contradict the non-diagonalizability of the matrix 
    \begin{align*}
        \begin{pmatrix}
          1 & 1 \\ 0 & 1
      \end{pmatrix}.
    \end{align*}
    % For more details, see also Item~\ref{sec:SNFG:remk:1:item:2} in Remark~\ref{sec:SNFG:remk:1}.

    \item\label{remk: LCT on snfg results an NFG with negative fun nodes}  The LCT $\LCT{\graphN}$ of $\graphN$ is in general \textbf{not} an S-NFG. This
    is because some of the local functions of $\LCT{\graphN}$ might take on
    negative real values. With that, also the global function of $\LCT{\graphN}$
    might take on negative real values.

    \item In the case of a binary alphabet $\set{X}_e$, the LCT we obtain is the same 
    as the one presented in~\cite{Chertkov2006}. However, for non-binary alphabets $\set{X}_e$,
    our LCT, although inspired by previous works such as~\cite{Chertkov2006, Chernyak2007, Mori2015}, is in general different. Notably, our approach does not require any assumption of non-zeroness on the SPA fixed-point messages.
  \end{enumerate}
  
 \eremark
\end{remark}

%***************************************************************************
%***************************************************************************

\begin{proposition}
  \label{sec:LCT:prop:1}
  \label{SEC:LCT:PROP:1}
  Consider an S-NFG $\graphN$ and let
  $\vmu $ be an SPA
  fixed-point message vector for $\graphN$.  
  The LCT of $ \graphN $, as specified in
  Definition~\ref{sec:LCT:def:1}, is denoted by $\LCT{\graphN}$. 
  (Recall that in Assumption~\ref{asmp: assume messages are non-negative}, 
  we assume $ Z_e(\vmu) >0 $ for all $ e \in \setEfull $, which ensures that the LCT of $ \graphN $ is well-defined.) 
  The main properties of $\LCT{\graphN}$ are listed as follows.
  \begin{enumerate}

  \item\label{sec:LCT:prop:1:item:1} The partition function remains the same:
    \begin{align*}
      Z(\graphN)
        &= Z\bigl( \LCT{\graphN} \bigr).
    \end{align*}

  \item\label{sec:LCT:prop:1:item:2} The SPA-based Bethe partition function of the original S-NFG is equal to the global function value of the all-zero configuration of
    $\LCT{\graphN}$:
    \begin{align*}
      \ZBSPA(\graphN,\vmu)
        &= \LCT{g}(\vect{0}).
    \end{align*}

  \item\label{sec:LCT:prop:1:item:3} For every $f \in \setF$ and every $\LCTvxf \in \LCTsetxf$ with
    exactly one non-zero component, \ie,
    %------------------------------------------------------------------------
    \begin{align*}
         \wh( \LCTvxf ) 
         = 1,
    \end{align*}
    %------------------------------------------------------------------------
    % where $  \wh( \LCTvxf )  \defeq \sum\limits_{e \in \setpf} [\LCT{x}_{e} \!\neq\! 0] $,
    it holds that
    \begin{align*}
      \LCT{f}(\LCTvxf)
        &= 0,
    \end{align*}
    where $ \wh( \cdot ) $ is the Hamming weight of the argument.

  \item\label{sec:LCT:prop:1:item:4} 
    For any configuration $\LCTv{x} \in \prod\limits_{e} \LCTset{X}_{e}$ in $ \LCT{\graphN} $, 
    we define 
    the associated subgraph $\bigl( \setF,\setE'(\LCTv{x}) \bigr)$, which consists of
    the set of vertices $ \setF $
    and the set of edges
    \begin{align*}
        \setE'(\LCTv{x})= \{ e \in \setE \ | \ \LCT{x}_{e} \neq 0 \}.
    \end{align*}
    The subgraph $\bigl( \setF,\setE'(\LCTv{x}) \bigr)$ is called a generalized loop 
    if this subgraph
    does not contain any leaves (vertices of degree one). (Note that for the all-zero
    configuration $\LCTv{x} = \vect{0}$, the subgraph $(\setF,\setE'(\LCTv{x}))$ 
    is also considered as a generalized loop.)

    It holds that if $ \LCTv{x} $ is a valid configuration for $ \LCT{\sfN} $,
    then the subgraph $ \bigl( \setF,\setE'(\LCTv{x}) \bigr) $ is a generalized
    loop, or equivalently, 
    if $ \bigl( \setF,\setE'(\LCTv{x}) \bigr) $ is not a generalized loop
    for the configuration $ \LCTv{x} $, then $\LCT{g}(\LCTv{x}) = 0$. 
    If $\graphN$ is cycle-free and with that
    also $\LCT{\graphN}$ is cycle-free, then 
    $\LCT{\graphN}$ has only one valid configuration, which is 
    the all-zero configuration and corresponds to a generalized loop.

  \item\label{sec:LCT:prop:1:item:5} 
  If $ \ZBSPA(\graphN,\vmu) \in \sR_{>0} $,
  it holds that
    \begin{align*}
      Z(\graphN)
        &= \ZBSPA(\graphN,\vmu)
           \cdot
           \left(
             1
             +
             \sum\limits_{\LCTv{x}}
               \frac{\LCT{g}(\LCTv{x})}
                    {\LCT{g}(\vect{0})}
           \right), 
    \end{align*}
    where the summation is over all $ \LCTv{x} \in \prod\limits_{e} \LCTset{X}_{e} $ such that $ \setE'(\LCTv{x}) $ corresponds to a non-zero generalized loop of
    $\LCT{\graphN}$. The term $\LCT{g}(\LCTv{x})/\LCT{g}(\vect{0})$ can
    be considered to be ``correction terms'' to the Bethe approximation. Note
    that if $\graphN$ is cycle-free then there are no correction terms,
    implying the well-known result that
    $Z(\graphN) = \ZBSPA(\graphN,\vmu)$.

  \item\label{sec:LCT:prop:1:item:6} The SPA fixed-point message vector
    $\vmu $ for
    $\graphN$ induces an SPA fixed-point message vector
    $\LCTv{\mu} \defeq ( \LCTv{\mu}_{\ef} )_{f \in \setpe, e \in \setEfull}$ for
    $\LCT{\graphN}$ such that
    \begin{align*}
      \LCTv{\mu}_{\ef}
        &= \begin{pmatrix}
          1, & 0, & \cdots & 0
        \end{pmatrix}^{\!\!\!\tran} \in \sC^{|\LCTset{X}_{e}|},
        \qquad e \in \setE,\, f \in \setpe.
    \end{align*}

  \item\label{sec:LCT:prop:1:item:7} While the above properties are shared by all variants of the LCT, the
    LCT in Definition~\ref{sec:LCT:def:1} has the following additional
    properties (that are not necessarily shared by the other LCT variants):
    \begin{enumerate}
 
    \item For every $e = (f_{i},f_{j}) \in \setE$, the functions $M_{\efi}$ and
      $M_{\efj}$ are formally symmetric in the sense that swapping $f_{i}$ and
      $f_{j}$ in the definition of $M_{\efi}$ leads to the definition of
      $M_{\efj}$, and vice-versa.

    \item\label{sec:LCT:prop:1:item:7:b} Assume that $M_{\efi}$ and $M_{\efj}$ are defined based on the
      choices $\zeta_{\efi} = \zeta_{\efj} \defeq Z_e^{-1/2}$,
      $\chi_{\efi} = \chi_{\efj} \defeq 1$,
      $\delta_{\efi} = \delta_{\efj} \defeq Z_e^{1/2}$. Let
      $e = (f_{i},f_{j}) \in \setE$ be an edge such that
      $\mu_{\efi}(\xe) = \mu_{\efj}(\xe) \in \sR$ for all $x_e \in \set{X}_e$. Then
      the matrices $\matr{M}_{\efi}$
      and $\matr{M}_{\efj}$ are equal
      and orthogonal.

    \end{enumerate}

  \end{enumerate}
\end{proposition}

%***************************************************************************

\begin{proof}
  See Appendix~\ref{apx:property of SNFGs}.
\end{proof}

We conclude this section with some remarks.
%---------------------------------------------------------------------------
\begin{itemize}
  
    \item Property~\ref{sec:LCT:prop:1:item:2} in Proposition~\ref{sec:LCT:prop:1} states that the LCT is a reparamerization of the global function based on the SPA fixed point message vector $ \vmu $ such that 
    the global function evaluated at $ \LCTv{x} = \vect{0} $ equals the $\vmu$-based Bethe partition function. The LCT can be viewed as shifting the bulk of the sum over the variables of the global function to some known configuration as well as keeping the partition function unchanged.

    \item Property~\ref{sec:LCT:prop:1:item:6} in Proposition~\ref{sec:LCT:prop:1} highlights that after applying the LCT on $\graphN$ based on the SPA fixed-point message vector $\vmu$, one of the SPA fixed-point message vectors for the LCT-transformed factor graph $\LCT{\graphN}$ has a simple structure. 
    Specifically, this simple SPA fixed-point message vector, denoted by $ \LCT{\vmu} $, satisfies 
    \begin{align*}
      \ZBSPA( \graphN, \vmu ) &= \ZBSPA( \LCT{\graphN}, \LCT{\vmu} ), \nonumber\\
      \LCT{\mu}_{\efi}(\LCT{x}_e) 
      &= \LCT{\mu}_{\efj}(\LCT{x}_e) = [\LCT{x}_e \! = \! 0], \qquad
      \LCT{x}_{e} \in \LCTset{X}_{e},\, e = (f_{i}, f_{j}) \in \setEfull.
    \end{align*}
    This simple structure partially motivates the proof of the graph-cover theorem in Section~\ref{sec:CheckCon}.

\end{itemize}
%---------------------------------------------------------------------------

%***************************************************************************
%***************************************************************************

\section{LCT for DE-NFGs}
\label{sec:LCT:DENFG:1}

The LCT for DE-NFGs has many similarities to the LCT for S-NFGs. Therefore, we will mostly highlight the differences in this section. 

\begin{definition}
    \label{def:DENFG:LCT:1}
    \index{LCT!for DE-NFG}
    The LCT for DE-NFGs is a natural extension of the LCT for S-NFG in Definition~\ref{sec:LCT:def:1}. 
    Further definitions are made for DE-NFGs.
    For each edge $ e \in \setEfull $, we apply the following changes.
    \begin{itemize}

        \item The alphabet $\set{X}_e$ is replaced by the alphabet $\tset{X}_e = \set{X}_e \times \set{X}_e$.

        \item The special element $0 \in \set{X}_e$ is replaced by the special element
        $\tzero \defeq (0,0) \in \tset{X}_e$.

        \item The variable $x_e \in \set{X}_e$ is replaced by the variable
        $\tx_e = (x_e,x'_e) \in \tset{X}_e$.

        \item The alphabet
        $\LCTset{X}_e \defeq \set{X}_e$ is
        replaced by the alphabet
        $\LCTtset{X}_e \defeq \tset{X}_e$.

        \item The element $0 \in \LCTset{X}_e$ is replaced by the element
        $\tzero\in \LCTtset{X}_e$.

        \item The variable $\LCT{x}_e \in \LCTset{X}_e$ is replaced by the variable
        $\LCTt{x}_e = (\LCT{x}_e, \LCT{x}'_e) \in \LCTtset{X}_{e}$.

    \end{itemize}
    Note that in the definition for DE-NFG, we consider that the constants $\zeta_{\efj}$, $\chi_{\efj}$, $\delta_{\ef}$, and $\epsilon_{\ef}$ are real-valued.
    For any finite set $ \set{I} \subseteq \setEfull $, we define
    %--------------------------------------------------------------------
    \begin{align*}
        \LCTtset{X}_{ \set{I} }
        \defeq \prod\limits_{e \in \set{I}} \LCTtset{X}_{ e }
        = \LCTset{X}_{ \set{I} } \times \LCTset{X}_{ \set{I} },
    \end{align*}
    %--------------------------------------------------------------------
    where
    \begin{align*}
      \LCTset{X}_{ \set{I} } \defeq 
        \prod\limits_{e \in \set{I}} \LCTset{X}_{e}.
    \end{align*}
    The associated collection of variables is defined to be
    \begin{align*}
      \LCTtv{x}_{\set{I}} \defeq ( \LCTv{x}_{\set{I}}, \LCTv{x}_{\set{I}}' )
      \in \LCTtset{X}_{ \set{I} }.
    \end{align*}
     For each $ f \in \setF $, we define
      \begin{align}
        \LCT{f}(\LCTtv{x}_{\setpf})
            &\defeq \sum\limits_{\tvx_{\setpf}}
            f(\tvx_{\setpf})
            \cdot
            \prod\limits_{e \in \setpf}
              M_{\ef}(\tx_{e}, \LCTt{x}_e),
        \qquad \LCTtv{x}_{\setpf} \in \LCTtset{X}_{\setpf}.
        \label{sec:LCT:eqn:1}
      \end{align}

    \edefinition
\end{definition}

If there is no ambiguity, for any $ \set{I} \subseteq \setEfull $,
we will use $ \sum\limits_{ \LCTtv{x}_{\set{I}} } $ instead of 
$ \sum\limits_{ \LCTtv{x}_{ \set{I} } \in \LCTtset{X}_{\set{I}} } $ for simplicity.

%***************************************************************************

\begin{proposition}
  \label{prop:DENFG:LCT:1}
  \label{PROP:DENFG:LCT:1}

  Consider a DE-NFG $\graphN$ and an SPA
  fixed-point message vector $ \vmu $ for $\graphN$. 
  Let $\LCT{\graphN}$ denote the LCT of $\graphN$ obtained by $\vmu$, as specified in
  Definition~\ref{def:DENFG:LCT:1}. The main properties of $\LCT{\graphN}$ are listed as follows.
  \begin{enumerate}

  \item Natural extension of Property~\ref{sec:LCT:prop:1:item:1} in Proposition~\ref{sec:LCT:prop:1}:
  $
      Z(\graphN)
        = Z\bigl( \LCT{\graphN} \bigr).
  $
  % (modulo the notational changes mentioned in Definition~\ref{def:DENFG:LCT:1})
  \item\label{prop:DENFG:LCT:1:item:4} Natural extension of Property~\ref{sec:LCT:prop:1:item:2} in Proposition~\ref{sec:LCT:prop:1}:
    \begin{align*}
      \ZBSPA(\graphN,\vmu)
        = \LCT{g}(\tv{0}) = \prod\limits_{f} \LCT{f}(\tv{0}),
    \end{align*}
    where $ \tv{0} = ( \tzero,\ldots,\tzero ) $.

  \item\label{sec:LCT:prop:1:item:8} Natural extension of Property~\ref{sec:LCT:prop:1:item:3} in Proposition~\ref{sec:LCT:prop:1}: for every $f \in \setF$ and every $\LCTv{x}_{\setpf} \in \LCTsetx_{\setpf}$ with exactly one non-zero component, 
  \ie, 
    %------------------------------------------------------------------------
    \begin{align*}
        \wh( \LCTv{x}_{\setpf} )
        = 1,
    \end{align*}
    %------------------------------------------------------------------------
    with $ \wh( \LCTv{x}_{\setpf} ) \defeq 
    \sum\limits_{ e \in \setpf } \bigl[ \LCTt{x}_{e} \!\neq\! \tzero \bigr] $,
    it holds that
    \begin{align*}
      \LCT{f}(\LCTtv{x}_{\setpf})
        &= 0.
    \end{align*}
  % (modulo the notational changes mentioned in
    % Definition~\ref{def:DENFG:LCT:1} and the replacement of ``non-zero
    % component'' by ``component not equal to $(0,0)$'')

  \item Natural extension of Property~\ref{sec:LCT:prop:1:item:4} in Proposition~\ref{sec:LCT:prop:1}.
    % (modulo the notational changes mentioned in
    % Definition~\ref{def:DENFG:LCT:1}).

  \item Natural extension of Property~\ref{sec:LCT:prop:1:item:5} in Proposition~\ref{sec:LCT:prop:1}.
    % (modulo the notational changes mentioned in
    % Definition~\ref{def:DENFG:LCT:1}).

  \item Natural extension of Property~\ref{sec:LCT:prop:1:item:6} in Proposition~\ref{sec:LCT:prop:1}
    % (modulo the notational changes mentioned in
    % Definition~\ref{def:DENFG:LCT:1}). 
    Nevertheless, because of its importance
    for later parts of this thesis, the statement is included here. Namely, the
    SPA fixed-point message vector
    $\vmu$ for
    $\graphN$ induces an SPA fixed-point message vector $ \LCTv{\mu} 
    = ( \LCTv{\mu}_{e,f} )_{e \in \setpf, f \in \setF} $
    for
    $\LCT{\graphN}$ with
    \begin{align*}
        \LCT{\mu}_{e,f}(\LCTt{x}_{e})
        &= \left[ \LCTt{x}_{e} \! = \! \tzero \right],
        \qquad
        \LCTt{x}_{e} \in \LCTtset{X}_{e}.
        % \\
        % \LCT{\mu}_{\upefi}(\LCT{x}_{\upe}) & 
        % = \bigl[ \LCT{x}_{\upe} \! = \! 0 \bigr],
        %    \qquad e =(e,\upe) \in \setEfull, \, 
        %    f \in \setpe,\, 
        %    \LCT{x}_{\upe} \in \LCTset{X}_{\upe}. \nonumber
    \end{align*}
    %-----------------------------------------------------------------------
    % Recall $  \setppe = \setpe =  \setupe $ as defined in Definition~\ref{sec:DENFG:def:4}. 
    % where is defined in item~\ref{def:DENFG:LCT:1:item:3} in Definition~\ref{def:DENFG:LCT:1}.

  \item Natural extension of Property~\ref{sec:LCT:prop:1:item:7} in Proposition~\ref{sec:LCT:prop:1}.
    % (modulo the notational changes mentioned in
    % Definition~\ref{def:DENFG:LCT:1}). 
    Note that Property~\ref{sec:LCT:prop:1:item:7:b} is slightly different:
    Assume that $M_{\efi}$ and $M_{\efj}$ are defined based on the
    choices $\zeta_{\efi} = \zeta_{\efj} \defeq Z_e^{-1/2}$,
    $\chi_{\efi} = \chi_{\efj} \defeq 1$,
    $\delta_{\efi} = \delta_{\efj} \defeq Z_e^{1/2}$. Let
    $e = (f_{i},f_{j}) \in \setE$ be an edge such that
    $\mu_{\efi}(\txe) = \overline{\mu_{\efj}(\txe)}$ for all $x_e \in \set{X}_e$. In this case, both $ \matr{M}_{\efi} $ and $ \matr{M}_{\efj} $ are unitary matrices with row indices $ \xe $ and column indices $ \LCT{x}_{e} $, and they satisfy
    $ \matr{M}_{\efi} = \matr{M}_{\efj}^{\Herm} $.

    \item\label{prop:DENFG:LCT:1:item:1} For each $ e = (f_{i}, f_{j}) $, the Choi-matrix representations 
    $ \matr{ C }_{ M_{\efi} } $ and $ \matr{ C }_{ M_{\efi} } $
    are complex-valued Hermitian matrices.

    \item\label{prop:DENFG:LCT:1:item:2} For each $f \in \setF$, 
      the Choi-matrix representation $\matr{C}_{\LCT{f}}$ is a complex-valued Hermitian matrix,
      where $ \matr{C}_{\LCT{f}} $ is defined to be a matrix with row indices $ \LCTv{x}_{\setpf} $ and column indices $ \LCTv{x}_{\setpf}' $ such that
      %-------------------------------------------------------------------
      \begin{align*}
          \matr{C}_{\LCT{f}} \defeq 
          \Bigl( \LCT{f}(\LCTv{x}_{\setpf}, \LCTv{x}_{\setpf}') 
          \Bigr)_{ (\LCTv{x}_{\setpf}, \LCTv{x}_{\setpf}') 
          \in \LCTtset{X}_{\setpf} }.
          % \label{sec:LCT:eqn:6}
      \end{align*}

    \item The resulting DE-NFG $\LCT{\graphN}$ is again a weak-sense DE-NFG. 

  \end{enumerate}

\end{proposition}

\begin{proof}
  See Appendix~\ref{apx:LCT DENFGs}.
\end{proof}

Most of the LCT properties for S-NFG can be straightforwardly extended to the case of DE-NFG. We note that applying the LCT to a strict-sense DE-NFG or a weak-sense DE-NFG yields in another weak-sense DE-NFG, which is similar to the property of the LCT for S-NFG as discussed in Item~\ref{remk: LCT on snfg results an NFG with negative fun nodes} in Remark~\ref{sec:LCT:rem:1}.

\chapter{Symmetric Subspace Transform (SST)}
\label{chapt:SST}
\label{CHAPT:SST}

\input{figures/sst_examples/snfg/main_pic}

In this chapter, we introduce the symmetric-subspace transform (SST) for both S-NFG and DE-NFG. The SST provides a different perspective to understand the $M$-covers of S-NFGs or DE-NFGs. While this transformation has been applied in quantum physics~\cite{Wood2015,Harrow2013}, to the best of our knowledge, it is the first time to introduce the SST in the factor-graph literature.

The developments in this chapter were motivated by the work of Wood \etal~\cite{Wood2015}, where, in terms of the language of the present paper, they transformed a certain integral into the average partition function of some NFGs, with the average taken over double covers. (Although~\cite{Wood2015} focus on double covers, it is clear that their results can be extended to general $M$-covers.) However, in contrast to~\cite{Wood2015}, our approach takes the opposite direction. We express the average partition function of a given NFG in terms of some integral, where the average is over all $M$-covers of the considered NFG.

%***************************************************************************

The developments in this chapter were partially also motivated by the results presented in~\cite{Vontobel2016}.

Let $\graphN$ be some S-NFG or DE-NFG. In
Definition~\ref{sec:GraCov:def:2}, the degree-$M$ Bethe
partition function of $\graphN$ is 
\begin{align*}
  \ZBM(\graphN) 
    = \sqrt[M]{
         \Bigl\langle
           Z\bigl( \hgraphN \bigr)
         \Bigr\rangle_{ \hgraphN \in \hat{\set{N}}_{M}}
       }.
\end{align*}
(See also
Definition~\ref{sec:GraCov:def:1}.) Our
particular interest lies in the limit superior of this quantity as
$M \to \infty$, as stated in Conjecture~\ref{sec:GraCov:conj:1}. Toward proving
Conjecture~\ref{sec:GraCov:conj:1}, we aim to reformulate
$\bigl( \ZBM(\graphN) \bigr)^M$ (or equivalently
$\bigl\langle Z\bigl( \hgraphN \bigr) \bigr\rangle_{ \hgraphN \in \hat{\set{N}}_{M}}$) as
an integral that can be analyzed in 
the limit $M \to \infty$. This approach allows us to study the behavior of the Bethe partition function for large values of $M$.

%***************************************************************************

%----------------------------------------------------------------------------

In order to introduce the approach, we consider a specific example of an S-NFG. In this example, the mathematical formalism that we will present is general, but we will illustrate
it under a particularly simple setup. Specifically, 
we will focus a part of the S-NFG looks like
Fig.~\ref{sec:SST:fig:3} and set $\set{X}_e = \{ 0, 1 \}$ and $M = 2$.

%***************************************************************************
%***************************************************************************

\section{Specifying an \texorpdfstring{$M$}{}-Cover of an S-NFG}

%***************************************************************************
The following construction process defines an $M$-cover of the original S-NFG $\graphN$.

\begin{definition}
  \label{def:graph:cover:construction:1}

  The construction process for an $M$-cover of $\graphN$ is outlined as follows (see also
  Definition~\ref{sec:GraCov:def:1}):
  \begin{enumerate}
  
  \item For each edge $e \in \setEfull$, we specify a permutation
    $\sigma_e \in \set{S}_{[M]}$. We collect all these permutations into a vector $\vsigma \defeq (\sigma_e)_{e \in \setEfull} \in \set{S}_{[M]}^{|\setEfull|}$.
  
  \item\label{def:graph:cover:construction:1:item:1} For each function node $f \in \set{F}$, we draw $M$ copies of $f$. Each copy
    of function node $f$ is associated with a collection of sockets and variables. Specifically,
    if $e \in \setpf$, then the $m$-th copy of $f$ has a socket with
    associated variable denoted by $x_{e,f,m}$.
  
  \item For every edge $e = (f_i, f_{j}) \in \setEfull$, we specify how the
    sockets corresponding to $(x_{\efi,m})_{m \in [M]}$, are connected to the
    sockets corresponding to $(x_{\efj,m})_{m \in [M]}$. 
    In this thesis, the connection is determined based on the chosen permutation $\sigma_e \in \set{S}_{[M]}$. Specifically, the socket corresponding to $x_{\efi,m}$ is connected to the socket corresponding to $x_{\efj,\sigma_{e}(m)}$ for $m \in [M]$. The possible cases for this connection are illustrated in Figs.~\ref{sec:SST:fig:4} and~\ref{sec:SST:fig:5}, considering that $|\set{S}_{[M]}| = 2! = 2$ for $M = 2$.

    % we do
    % it as follows: if the permutation $\sigma_e \in \set{S}_{[M]}$ is chosen, then
    % the socket corresponding to $x_{\efi,m}$ is connected to the socket
    % corresponding to $x_{\efj,\sigma_{e}(m)}$ for $ m \in [M] $.

    % Figs.~\ref{sec:SST:fig:4} and~\ref{sec:SST:fig:6} show the two possible cases
    % that can happen in the illustrated S-NFG. (Note that $|\set{S}_{[M]}| = 2! = 2$
    % for $M = 2$.)
  
  \item The resulting $M$-cover is denoted by $\hgraphN_{M,\vsigma}$.

  \end{enumerate}
  \edefinition
\end{definition}

%***************************************************************************

In order to simplify the notation, let $\vSigma_{M} \defeq \vSigma \defeq \bigl( \Sigmae \bigr)_{\! e}$ be the random vector consisting of
independent and uniformly distributed (i.u.d.) random permutations from
$\set{S}_{[M]}^{|\setEfull|}$.\footnote{We use $\Sigmae$ instead of $\Sigma_e$ for these random
  variables in order to distinguish them from the symbol used to denote a sum
  over all edges.} The probability distributions of $ \Sigmae $ and $ \vSigma $ are given by
%----------------------------------------------------------------------------
\begin{align}
  p_{\Sigmae}(\sigma_e) 
  &= \frac{1}{|\set{S}_{[M]}|} = \frac{1}{M!}, \qquad e \in \setEfull,\,
  \sigma_e \in \set{S}_{[M]},
  \label{sec:SST:eqn:pdf of SST for each edge}\\
  p_{\vSigma}(\vsigma) &= \prod\limits_{e}p_{\Sigmae}(\sigma_e)
  = \frac{1}{(M!)^{|\setEfull|}}, \qquad 
  \vsigma \in \set{S}_{[M]}^{|\setEfull|} 
  .\label{sec:SST:eqn:31}
\end{align}
%----------------------------------------------------------------------------
With this, $\hgraphN_{M,\vsigma}$ represents a uniformly
sampled $M$-cover of $\graphN$. We can rewrite the quantity of interest, $\bigl( \ZBM(\graphN) \bigr)^M$, as follows:
\begin{align}
  \bigl( \ZBM(\graphN) \bigr)^M
    &= \Bigl\langle
         Z\bigl( \hgraphN \bigr)
       \Bigr\rangle_{ \hgraphN \in \hat{\set{N}}_{M}}
     = \sum\limits_{\vsigma}
         p_{\vSigma}(\vsigma)
         \cdot
         Z\bigl( \hgraphN_{M,\vsigma} \bigr), \label{sec:SST:eqn:15}
\end{align}
where $ \sum\limits_{\vsigma} $ denotes $ \sum\limits_{\vsigma \in \set{S}_{[M]}^{|\setEfull|}} $.
%***************************************************************************
%***************************************************************************

\section{Replacement of \texorpdfstring{$M$}{}-Cover S-NFG 
                       by another S-NFG}

%***************************************************************************

In order to proceed, it is convenient to replace each $M$-cover
$\hgraphN_{M,\vsigma}$ with an S-NFG denoted by $\hgraphPsig$,
which satisfies the property 
%----------------------------------------------------------------------------
\begin{align}
  Z\bigl( \hgraphN_{M,\vsigma} \bigr)
  = Z\bigl( \hgraphPsig \bigr). \label{sec:SST:eqn:33}
\end{align}
%----------------------------------------------------------------------------
Before providing the general definition of $\hgraphPsig$, we briefly
discuss this replacement in terms of the illustrated example. Instead
of working with the S-NFG in Fig.~\ref{sec:SST:fig:4}, we will
work with the S-NFG in Fig.~\ref{sec:SST:fig:6}. Similarly,
instead of working with the S-NFG in
Fig.~\ref{sec:SST:fig:5}, we will work with the S-NFG in
Fig.~\ref{sec:SST:fig:7}.

%***************************************************************************

\begin{definition}
  \label{sec:SST:def:1}

  For a given $\vsigma \in \set{S}_{[M]}^{\setEfull}$, the S-NFG
  $\hgraphPsig$ is defined as follows:
  \begin{enumerate}
  
  \item \label{step: for f covers} For every $f \in \set{F}$, we draw $M$ copies of $f$. 
   Each copy of $f$ is associated with a collection of sockets and variables. Specifically, if $e \in \setpf$, the $m$-th copy of $f$ has a socket with an associated variable $x_{\ef,m}$. (This step is the same as Step~\ref{def:graph:cover:construction:1:item:1} in Definition~\ref{def:graph:cover:construction:1}.)
  
  \item For every edge $e = (f_{i}, f_{j}) \in \setEfull$, we draw a function node
    $P_{e,\sigma_e}$. This node is connected to the sockets associated with the collection of
    variables $(x_{\efi,m})_{m \in [M]}$ and to the sockets associated with the collection of variables
    $(x_{\efj,m})_{m \in [M]}$. The local function $P_{e,\sigma_e}: \set{X}_e^{M} \times \set{X}_e^{M}  \to  \{0,1\} $ is
    defined to be
    \begin{align*}
      P_{e,\sigma_e}\bigl( \vx_{\efi,[M]}, \vx_{\efj,[M]} \bigr)
        &\defeq
           \prod\limits_{m \in [M]}
     \bigl[
       x_{\efi,m} \! = \! x_{\efj,\sigma_e(m)}
     \bigr],
       \quad \vx_{\efi,[M]}, \vx_{\efj,[M]} \in \set{X}_e^M,
    \end{align*}
    where 
    \begin{align*}
      \vx_{\ef,[M]}
      &\defeq
       (x_{\ef,1}, \ldots, x_{\ef,M}) \in  \set{X}_e^M, \qquad 
       f \in \{f_{i}, f_{j}\}.
    \end{align*}
    (Note that $P_{e,\sigma_e}$ is an indicator
    function, as it takes only the values zero and one.)
  
  \end{enumerate}
  \edefinition
\end{definition}

%***************************************************************************

In the case of Fig.~\ref{sec:SST:fig:4}, the permutation
$\sigma_e$ satisfies $\sigma_e(1) = 1$ and $\sigma_e(2) = 2$. Therefore, the corresponding local function $P_{e,\sigma_e}$ with $M = 2$ is given by:
\begin{align*}
  P_{e,\sigma_e}\bigl( \vx_{\efi,[M]}, \vx_{\efj,[M]} \bigr)
  &=
     [x_{\efi,1} \! = \! x_{\efj,1}]
     \cdot
     [x_{\efi,2} \! = \! x_{\efj,2}].
\end{align*}
Representing this function in matrix form, with rows indexed by
$(x_{\efi,1}, x_{\efi,2}) $ and columns
indexed by
$(x_{\efj,1}, x_{\efj,2})$, which take the values following the order $ (0,0), \, (0,1), \, (1,0), \, (1,1) $, we have
\begin{align*}
  \matr{P}_{e,\sigma_e}
    &= \begin{pmatrix}
      1 & 0 & 0 & 0 \\
      0 & 1 & 0 & 0 \\
      0 & 0 & 1 & 0 \\
      0 & 0 & 0 & 1  
     \end{pmatrix}.
\end{align*}
The resulting S-NFG is shown in Fig.~\ref{sec:SST:fig:6}.

%***************************************************************************

On the other hand, in the case of Fig.~\ref{sec:SST:fig:5}, the
permutation $\sigma_e$ satisfies $\sigma_e(1) = 2$ and $\sigma_e(2) = 1$.  Thus, the corresponding local function $P_{e,\sigma_e}$ is given by:
\begin{align*}
  P_{e,\sigma_e}\bigl( \vx_{\efi,[M]}, \vx_{\efj,[M]} \bigr)
  &=
     [x_{\efi,1} \! = \! x_{\efj,2}]
     \cdot
     [x_{\efi,2} \! = \! x_{\efj,1}].
\end{align*}
Representing this function in matrix form (with the same row and column indexing as
above), we get
\begin{align*}
  \matr{P}_{e,\sigma_e}
  &= \begin{pmatrix}
    1 & 0 & 0 & 0 \\
    0 & 0 & 1 & 0 \\
    0 & 1 & 0 & 0 \\
    0 & 0 & 0 & 1  
   \end{pmatrix}.
\end{align*}
The resulting S-NFG is shown in Fig.~\ref{sec:SST:fig:7}.

%***************************************************************************
%***************************************************************************

\section[An NFG Representing the Average \texorpdfstring{$M$}{}-Cover]{An NFG Representing the Average Degree \texorpdfstring{$M$}{}-Cover 
                       of an S-NFG}
\label{sec:SST:NFG:average:cover:1}

%***************************************************************************

The next step is to construct an NFG $\hgraphNavg$, whose partition sum equals
$\sum\limits_{\vsigma} p_{\vSigma}(\vsigma) \cdot Z\bigl( \hgraphN_{M,\vsigma} \bigr)$. We do this
as follows. (See Fig.~\ref{sec:SST:fig:8} for an illustration.)

%***************************************************************************

\begin{definition}
  \label{sec:SST:def:2}

  % The construction of $\hgraphNavg$ is essentially the same as the
  % construction of $\hgraphPsig$ for any $\vsigma$, but with the
  % following modifications and additions.

  The construction of $\hgraphNavg$ is essentially the same as the
  construction of $\hgraphPsig$ for any $\vsigma$ in Definition~\ref{sec:SST:def:2}, but with the
  following modifications and additions. 
  Namely, for every $e \in \setEfull$, the
  permutation $\sigma_e$ is not fixed, but a variable. With this, the function
  $P_{e,\sigma_e}\bigl( \vx_{\efi,[M]}, \vx_{\efj,[M]} \bigr)$ is now considered to be a function
  not only of $\vx_{\efi,[M]}$ and $\vx_{\efj,[M]}$, but also of
  $\sigma_e$. Graphically, this is done as follows:
  \begin{itemize}

  \item For every $e \in \setEfull$, we draw a function node $p_{\Sigmae}$.

  \item For every $e \in \setEfull$, we draw an edge that connects
    $P_{e,\sigma_e}$ and $p_{\Sigmae}$ and that represents the variable
    $\sigma_e$. Recall $ p_{\Sigmae}( \sigma_{e} ) = (M!)^{-1} $ for all $ e \in \setEfull $ and $ \sigma_{e} \in \set{S}_{[M]} $, as defined in~\eqref{sec:SST:eqn:pdf of SST for each edge}.

  \end{itemize}
  % We define $ p_{\vSigma}(\vsigma) 
  % \defeq \prod\limits_{e \in \setEfull} p_{\Sigmae}( \sigma_{e} ) $.
  \edefinition
\end{definition}

%***************************************************************************

\begin{lemma}
  \label{sec:SST:lem:1}

  The partition function of $\hgraphNavg$ satisfies
  \begin{align*}
    Z\bigl( \hgraphNavg \bigr)
      &= \sum\limits_{\vsigma}
           p_{\vSigma}(\vsigma) 
           \cdot Z\bigl( \hgraphN_{M,\vsigma} \bigr)
      = \Bigl\langle
           Z\bigl( \hgraphN \bigr)
         \Bigr\rangle_{ \hgraphN \in \hat{\set{N}}_{M}}
      = \bigl( \ZBM(\graphN) \bigr)^M.
  \end{align*}
\end{lemma}

%***************************************************************************

\begin{proof}
  The first equality  follows directly from the definition of $\hgraphNavg$ in Definition~\ref{sec:SST:def:2}, and the last two equalities follow from the equalities in~\eqref{sec:SST:eqn:15}.
\end{proof}

%***************************************************************************

In the subsequent discussion, we fix an arbitrary edge $e = (f_{i}, f_{j}) \in \setEfull$.

%***************************************************************************

%---------------------------------------------------------------------------
\begin{definition}\label{sec:SST:def:4}
  We define
  %---------------------------------------------------------------------------
  \begin{align*}
    P_e\bigl( \vx_{\efi,[M]}, \vx_{\efj,[M]} \bigr)
      &\defeq
         \sum\limits_{\sigma_e \in \set{S}_{[M]}}
           p_{\Sigmae}(\sigma_e)
           \cdot
             P_{e,\sigma_e}\bigl( \vx_{\efi,[M]}, \vx_{\efj,[M]} \bigr) \\
      &= \frac{1}{M!}
           \cdot
           \sum\limits_{\sigma_e \in \set{S}_{[M]}}
             P_{e,\sigma_e}\bigl( \vx_{\efi,[M]}, \vx_{\efj,[M]} \bigr),
  \end{align*}
  where the last equality follows from $p_{\Sigmae}(\sigma_e) = 1/M!$ for all
  $\sigma_e \in \set{S}_{[M]}$. Also, we define the following vectors:
  %---------------------------------------------------------------------------
  \begin{align*}
    \vx_{[M]} &\defeq ( x_{e,f,m} )_{e \in \setpf, f \in \setF, m \in [M]}
    \in \prod\limits_{e \in \setpf, f \in \setF} \setx{e}^{M},
    \nonumber\\
    \vx_{\setpff,m} &\defeq ( x_{e,f,m} )_{e \in \setpf}
    \in \setx{\setpf},
    \quad f \in \setF, \ m \in [M]. 
  \end{align*}
  \edefinition
\end{definition}
%---------------------------------------------------------------------------

%***************************************************************************
In the case of Fig.~\ref{sec:SST:fig:9}, the function $P_e$ corresponds to the exterior function of the dashed box in Fig.~\ref{sec:SST:fig:8}. This function is obtained through the closing-the-box operation~\cite{Loeliger2004}. For the specific example in Fig.~\ref{sec:SST:fig:8}, the
matrix representation of $ P_e $ is given by (the same row and column indexing as
$\matr{P}_{e,\sigma_e}$)
\begin{align}
  \matr{P}_e
    &\defeq
       \sum\limits_{\sigma_{e} \in \set{S}_2}
         \frac{1}{2!}
         \cdot
         \matr{P}_{e,\sigma_e}
     = \begin{pmatrix}
         1 & 0   & 0   & 0 \\
         0 & 1/2 & 1/2 & 0 \\
         0 & 1/2 & 1/2 & 0 \\
         0 & 0   & 0   & 1
       \end{pmatrix}
         \label{sec:SST:eqn:1}.
\end{align}

%***************************************************************************

%----------------------------------------------------------------------------
\begin{lemma}\label{sec:SST:lem:6}
  \label{SEC:SST:LEM:6}
  The partition function $ Z\bigl( \hgraphNavg \bigr) $ can be expressed as
  %---------------------------------------------------------------------------
  \begin{align*}
    Z\bigl( \hgraphNavg \bigr) = \sum\limits_{\vx_{[M]}}
    \left( \prod\limits_{m \in [M]} \prod_f f(\vx_{\setpff,m}) \right)
    \cdot \prod\limits_{e} P_e\bigl( \vx_{\efi,[M]}, \vx_{\efj,[M]} \bigr).
  \end{align*}
  %---------------------------------------------------------------------------
\end{lemma}
%----------------------------------------------------------------------------
%----------------------------------------------------------------------------
\begin{proof}
  See Appendix~\ref{apx:alternative expression of ZBM by Pe}.
\end{proof}
%----------------------------------------------------------------------------

In the subsequent analysis, we use the language of the method of types\footnote{For a detailed exposition of the method of types, we refer to the book by Cover and Thomas~\cite{T.M.Cover2006}, although we use different notation in our discussion.} to characterize the function $P_e$. 

% In the following, we use the language of the method of types to characterize the
% function $P_e$. 

%***************************************************************************

\begin{definition}\label{sec:SST:def:6}
  We introduce the following objects:\footnote{Note that, for simplicity, we use $x$ instead of the more precise notation $\xe$.}
  \begin{itemize}

  \item The type
    $\vt_e(\vv_e) \defeq \bigl( t_{e,x}(\vv_e) \bigr)_{x \in \set{X}_e} \in \Pi_{\setxe}$ of a
    vector $\vv_e \defeq (v_{e,m})_{m \in [M]} \in \set{X}_e^M$ is defined to be
    \begin{align*}
      t_{e,x}(\vv_e)
        &\defeq
           \frac{1}{M}
           \cdot
           \bigl| \hskip0.5mm 
             \left\{
               m \in [M]
             \ \middle| \ 
               v_{e,m} = x
             \right\} \hskip0.5mm 
           \bigr|,
             \qquad x \in \set{X}_e.
    \end{align*}

  \item Let $\set{B}_{\set{X}_e^M}$ be the set of possible types of vectors of
    length $M$ over $\set{X}_e$, \ie,
    \begin{align*}
      \set{B}_{\set{X}_e^M}
        &\defeq
           \bigl\{
             \vt_e \in \Pi_{\setxe}
           \bigm|
             \text{there exists $\vv_e \in \set{X}_e^M$ 
                   such that $\vt(\vv_e) = \vt_e$}
           \bigr\}.
    \end{align*}
  Let $ \set{B}_{\set{X}^M} $ be the Cartesian product $ \prod\limits_{e} \set{B}_{\set{X}_e^M} $.
    
  \item Let $\vt_e \in \set{B}_{\set{X}_e^M}$. Then the type class of $\vt_e$ is defined to be the set
    \begin{align*}
      \set{T}_{e,\vt_e}
      &\defeq
         \bigl\{
           \vv_e \in \set{X}_e^M
         \bigm|
           \vt_e(\vv_e) = \vt_e
         \bigr\}.
    \end{align*}
  
  \end{itemize}
  \edefinition
\end{definition}

%***************************************************************************

%---------------------------------------------------------------------------
\begin{lemma}
  \label{sec:SST:lem:2}

  It holds that
  \begin{align*}
    \bigl| \set{B}_{\set{X}_e^M} \bigr|
      &= \binom{|\set{X}_e| + M - 1}{M}, \\
    |\set{T}_{e,\vt_e}|
      &= \frac{M!}{
        \prod\limits_{x \in \set{X}_e} 
        \bigl( (M \cdot t_{e,x})! \bigr)
      },
      \quad \vt_e \in \set{B}_{\set{X}_e^M}.
  \end{align*}
\end{lemma}

\begin{proof}
  These expressions can be derived from standard combinatorial results.
\end{proof}

%***************************************************************************

For the illustrated example, the number of possible types is given by the equation
$\bigl| \set{B}_{\set{X}_e^M} \bigr| = \binom{2 + 2 - 1}{2} = 3$ and the set 
$\set{B}_{\set{X}_e^M} $ is given by $ \{ (1,0), \, (1/2,1/2), \, (0,1) \} $. The corresponding type classes have the following
sizes: 
\begin{align*}
  |\set{T}_{e,(1,0)}| = 1,\qquad |\set{T}_{e,(1/2,1/2)}| = 2,\qquad
  |\set{T}_{e,(0,1)}| = 1.
\end{align*}

%***************************************************************************

\begin{lemma}\label{sec:SST:lem:4}
  The function $P_e$ for $ e = (f_{i}, f_{j}) $ satisfies
  \begin{align*}
    P_e\bigl( \vx_{\efi,[M]}, \vx_{\efj,[M]} \bigr)
      &= \begin{cases}
           |\set{T}_{e,\vt_e}|^{-1}
             &  
             \vt_e = \vt_e\bigl(\vx_{\efi,[M]}\bigr) 
             = \vt_e\bigl( \vx_{\efj,[M]} \bigr) \\
           0
             & \text{otherwise}
         \end{cases}.
  \end{align*}
\end{lemma}

%***************************************************************************

\begin{proof}
  This result follows from the definition of the function $P_{e,\sigma_e}$ in
  Definition~\ref{sec:SST:def:1}, along with the properties of the uniform distribution $p_{\Sigmae}$ and the symmetry of $ P_e $, as shown in Definition~\ref{sec:SST:def:4}.
\end{proof}

%***************************************************************************

For the illustrated example, the expression in the above lemma is
corroborated by the matrix in~\eqref{sec:SST:eqn:1}. Here, the matrix is reproduced with horizontal and vertical lines to highlight the three type classes of the row and column indices, respectively:
\begin{align*}
  \matr{P}_e
    &= \left(
         \begin{array}{c|cc|c}
           1 & 0   & 0   & 0 \\
         \hline
           0 & 1/2 & 1/2 & 0 \\
           0 & 1/2 & 1/2 & 0 \\
         \hline
           0 & 0   & 0   & 1
         \end{array}
       \right).
\end{align*}

%***************************************************************************

The matrix $\matr{P}_e$ associated with $P_e$ is known as the
symmetric-subspace projection operator (see, \eg,~\cite{Harrow2013}).

Now we can express the $M$-th power of the degree-$M$ Bethe partition function as the partition function of the average degree-$M$ cover.
%----------------------------------------------------------------------------
\begin{theorem} 
  \label{thm: expression of Z for degree M Bethe partition function}
  The degree-$M$ Bethe partition function satisfies
  \begin{align*}
    \bigl( \ZBM(\graphN) \bigr)^M
    = 
    \sum\limits_{\vx_{[M]}}
    \left( \prod\limits_{m \in [M]} \prod_f f(\vx_{\setpff,m}) \right)
    \cdot \prod\limits_{e} P_e\bigl( \vx_{\efi,[M]}, \vx_{\efj,[M]} \bigr).
  \end{align*}
  % where
  % \begin{align*}
  %   \prod\limits_{e} P_e\bigl( \vx_{\efi,[M]}, \vx_{\efj,[M]} \bigr)
  %   = \begin{cases}
  %     \prod\limits_{e} \frac{1}{|\set{T}_{e,\vt_e}|}
  %      &  
  %      \vt_e = \vt_e\bigl(\vx_{\efi,[M]}\bigr) 
  %      = \vt_e\bigl( \vx_{\efj,[M]} \bigr),\, 
  %      \forall e \in \setEfull \\
  %     0
  %      & \text{otherwise}
  %   \end{cases}.
  % \end{align*}
\end{theorem}
%----------------------------------------------------------------------------
%----------------------------------------------------------------------------
\begin{proof}
  This follows from Lemmas~\ref{sec:SST:lem:1},~\ref{sec:SST:lem:6},
  and~\ref{sec:SST:lem:4}.
\end{proof}
%----------------------------------------------------------------------------

\section[Reformulation of the Function \texorpdfstring{$P_e$}{}]{Reformulation of the Function \texorpdfstring{$P_e$}{} as an Integral}
\label{sec:SST:Pe:reformulation:1}

% \begin{figure}[t]
%   % \subfloat[]{
%     % \begin{minipage}[t]{0.3\textwidth}
%       \centering
%       \begin{tikzpicture}
%         \input{figures/head_files_figs.tex}
%         \input{figures/sst_examples/snfg/length_large.tex}
%         \input{figures/sst_examples/snfg/background_nodes_lines.tex}
%         \input{figures/sst_examples/snfg/lines_not_permuted.tex}
%         \input{figures/sst_examples/snfg/otb_pe.tex}
%         \node[] (var1) at (-0.18*\ldis,-1*\ldis) [label=right: $\cvpsi_{e}$] {};
%         \node[state] (Psige) at (0,-1.5*\ldis) [label=below: $ \bigl| \set{B}_{\set{X}_e^M}\! \bigr| \cdot \muFSsimple$] {};
%         \input{figures/sst_examples/snfg/otb_back_nodes_lines.tex}
%         \begin{pgfonlayer}{background}
%           \node[state_dash6] (db1) at (0,0.1*\ldis) [label=below:$P_{e}$] {};
%         \end{pgfonlayer}
%       \end{tikzpicture}
%     % \end{minipage}
%     \caption{The NFG for illustrating the SST for S-NFG.\label{sec:SST:fig:10}}
%   % }
% \end{figure}

%***************************************************************************

In this section, we explore an approach to express the function $P_e$ as an  integral. This integral representation can be utilized to modify the NFG in
Fig.~\ref{sec:SST:fig:9}, resulting in the NFG shown in
Fig.~\ref{sec:SST:fig:10}.

%***************************************************************************

%***************************************************************************

\begin{definition}\label{sec:SST:def:3}
  For the following considerations, we consider an arbitrary edge
  $e = (f_{i}, f_{j}) \in \setEfull$ and a positive integer $M$. We introduce
  the following notation:
  \begin{itemize}

    \item The $2$-norm of a function $\psi_e: \set{X}_e \to \sC$ is defined to be
      $\lVert \psi_e \rVert_{2} \defeq \sqrt{\sum\limits_{\xe} |\psi_e(\xe)|^2}$.

    \item The 2-norm of a vector
      $\vpsi_e = \bigl( \psi_e(\xe) \bigr)_{\! \xe \in \set{X}_e} \in
      \sC^{|\set{X}_e|}$ is defined to be\\
      $\lVert \vpsi_e \rVert_{2} \defeq \sqrt{\sum\limits_{\xe} |\psi_e(\xe)|^2} =
      \sqrt{\vpsi_e^\Herm \cdot \vpsi_e}$.

    % \item The vector containing the absolute values of the entries in the vector $ \cvpsi_e $ is defined to be $ |\cvpsi_e| = \bigl( |\cpsi_e(\xe)| \bigr)_{\! \xe \in \set{X}_e} \in
    %   \sR_{\geq 0}^{|\set{X}_e|} $.

  \end{itemize}

  Now, let $\cpsi_e: \set{X}_e \to \sC$ be a function with $2$-norm equals one, \ie,
  $\lVert \cpsi_e \rVert_{2} = 1$. Based on $\cpsi_e$, we define the function
  $\funcFS_{e,\cpsi_e}: \set{X}_e^{M} \times \set{X}_e^{M} \to \sC$ to be
  \begin{align*}
    \funcFS_{e,\cpsi_e}\bigl( \vx_{\efi,[M]}, \vx_{\efj,[M]} \bigr)
    &\defeq
      \bigl| \set{B}_{\set{X}_e^M} \bigr|
      \cdot
      \Biggl(
         \prod\limits_{m \in [M]}
           \cpsi_e(x_{\efi,m})
      \Biggr)
      \cdot
      \Biggl(
         \prod\limits_{m \in [M]}
           \overline{ \cpsi_e(x_{\efj,m}) }
      \Biggr),
  \end{align*}
  where $ \vx_{\efi,[M]}, \vx_{\efj,[M]} \in \set{X}_e^M $.
  \edefinition
\end{definition}

%***************************************************************************
If we associate the function $\cpsi_e$ with the column vector
%----------------------------------------------------------------------------
\begin{align}
  \cvpsi_e \defeq \bigl( \cpsi_e(\xe) \bigr)_{\! \xe \in \set{X}_e}
  \in \sC^{|\setxe|},
  \label{sec:SST:eqn:39}
\end{align}
%----------------------------------------------------------------------------
and the
function $\funcFS_{e,\cpsi_e}$ with the matrix $\matrFS_{e,\cpsi_e}$ (with rows
indexed by $\vx_{\efi,[M]}$ and columns indexed by $\vx_{\efj,[M]}$), then we have
\begin{align*}
  \matrFS_{e,\cpsi_e}
    &= \bigl| \set{B}_{\set{X}_e^M} \bigr|
    \cdot
    \bigl( \cvpsi_e^{\otimes M} \bigr)
    \cdot
    \bigl( \cvpsi_e^{\otimes M} \bigr)^{\! \Herm}
    \in \sC^{|\setxe|^{M} \times |\setxe|^{M}},
\end{align*}
where $\cvpsi_e^{\otimes M}$ represents the $M$-fold Kronecker product of
$\cvpsi_e$ with itself.

%***************************************************************************

% Alternatively, with the introduction of some suitable locally compact
% topological group, it is a Haar measure over that group.

In the following, for every $ e \in \setEfull $, we introduce the so-called Fubini-Study measure in dimension $|\set{X}_e|$ which can be viewed as a, 
in a suitable sense, uniform measure $\muFSsimple$ on functions
$\cpsi_e: \set{X}_e \to \sC_{\geq 0}$ with a $ 2 $-norm of one. This measure is also a Haar measure.
For the purposes of this thesis, it is sufficient
to know that a sample $\cvpsi_{e} \in \sC^{|\set{X}_e|}$ from $\muFSsimple$ can be
generated as follows:
\begin{itemize}

\item Let
  \begin{align*}
    \vW_e 
      &\defeq
         \bigl(
           (W_{e,\xe,0}, W_{e,\xe,1})
         \bigr)_{\! \xe \in \set{X}_e}
           \in \sR^{|\set{X}_e|} \times \sR^{|\set{X}_e|}
  \end{align*}
  be a length-$2 |\set{X}_e|$ random vector with i.i.d. entries distributed
  according to a normal distribution with mean zero and variance one.

  \item Let $\vw_e$ be a realization of $\vW_e$. Let\footnote{Recall that
    ``$\imagunit$'' denotes the imaginary unit.}
  \begin{align*}
    \cvw_e
      &\defeq
         \vw_e / \lVert \vw_e \rVert_{2}, \nonumber \\
    \cpsi_e(\xe)
      &=
         \cw_{e,\xe,0} + \imagunit \cw_{e,\xe,1}. 
   \nonumber
  \end{align*}
   Because of the bijection between $\cvw$ and $\cvpsi$, we will use, with some slight abuse of notation, the notation $\muFSsimple$ not only for the distribution of $\cvpsi$, but also for the distribution of $\cvw$. In this case, the 2-norm of the function $ \cpsi_e $ is still one, \ie,
   %------------------------------------------------------------------------
   \begin{align*}
    \sum\limits_{\xe} \bigl| \cpsi_e(\xe) \bigr|^{2} = 1. 
   \end{align*}
   %------------------------------------------------------------------------
\end{itemize}

%***************************************************************************

%---------------------------------------------------------------------------
\begin{lemma}\label{sec:SST:prop:2}
  \label{SEC:SST:PROP:2}
  For all $ \vx_{\efi,[M]}, \vx_{\efj,[M]} \in \set{X}_e^M $, it holds that 
  \begin{align*}
    \int
      \funcFS_{e,\cpsi_e}\bigl( \vx_{\efi,[M]}, \vx_{\efj,[M]} \bigr)
    \dd{\muFSsimple\bigl( \cvpsi_e \bigr)} 
    &= P_e\bigl( \vx_{\efi,[M]}, \vx_{\efj,[M]} \bigr).
  \end{align*}
\end{lemma}
%---------------------------------------------------------------------------
\begin{proof}
  See Appendix~\ref{apx:SST}.
\end{proof}

\section[Another NFG Representing the Average \texorpdfstring{$M$}{}-Cover]{An Alternative NFG Representing 
the Average \texorpdfstring{$M$}{}-Cover of an S-NFG\label{sec:SST:sunbsec:2}}

%***************************************************************************

In Section~\ref{sec:SST:NFG:average:cover:1}, we introduced the NFG
$\hgraphNavg$, whose partition function equals
$\bigl( \ZBM(\graphN) \bigr)^M$. Now, based on the
results from Section~\ref{sec:SST:Pe:reformulation:1}, we formulate an
alternative NFG, denoted by $\hgraphNavgalt$. The partition function of $\hgraphNavgalt$ also equals
$\bigl( \ZBM(\graphN) \bigr)^M$.

%***************************************************************************

Recall that $P_e\bigl( \vx_{\efi,[M]}, \vx_{\efj,[M]} \bigr)$ represents the exterior
function of the dashed box in Fig.~\ref{sec:SST:fig:8}; in other
words, if we close the dashed box in Fig.~\ref{sec:SST:fig:8}, then we
obtain a single function node representing the function
$P_e\bigl( \vx_{\efi,[M]}, \vx_{\efj,[M]} \bigr)$, as shown in Fig.~\ref{sec:SST:fig:9}.

%***************************************************************************

The NFG $\hgraphNavgalt$ in Fig.~\ref{sec:SST:fig:10} is defined
such that it is the same as the NFG $\hgraphNavg$ in
Fig.~\ref{sec:SST:fig:9}, except for the content of the function $ P_{e} $
for every $e \in \setEfull$. However, the content of the dashed box in
Fig.~\ref{sec:SST:fig:10} is set up in such a way that its exterior function
equals the function $ P_{e} $ in
Fig.~\ref{sec:SST:fig:9} as well as the exterior function of the dashed box in Fig.~\ref{sec:SST:fig:8}. 

By closing the dashed box in Fig.~\ref{sec:SST:fig:10}, a single function node representing the function $P_e$ is formed, as shown in Fig.~\ref{sec:SST:fig:9}. This ensures that the partition function of the NFG in Fig.~\ref{sec:SST:fig:10} equals the partition functions of the NFGs in Figs.~\ref{sec:SST:fig:8} and~\ref{sec:SST:fig:9}. In other words, it can be stated that\footnote{In terms of NFG language, the
  NFG in Fig.~\ref{sec:SST:fig:10} can be obtained
  by, first applying a closing-the-box operation and
  then an opening-the-box operation~\cite{Loeliger2004} for every $e \in \setEfull$. Note that
  closing-the-box and opening-the-box operations maintain the partition
  function of the NFG.}
\begin{align*}
  Z\bigl( \hgraphNavgalt \bigr) = Z\bigl( \hgraphNavg \bigr) = \bigl( \ZBM(\graphN) \bigr)^M.
\end{align*}

% With that, closing the dashed box in
% Fig.~\ref{sec:SST:fig:10} results in a single function node
% representing the function $P_e$, as shown in Fig.~\ref{sec:SST:fig:9}. 

% This guarantees
% that the partition function of the NFG in Fig.~\ref{sec:SST:fig:10}
% equals the partition functions of the NFGs in
% Figs.~\ref{sec:SST:fig:8} and~\ref{sec:SST:fig:9}, \ie,
% $Z\bigl( \hgraphNavgalt \bigr) = Z\bigl( \hgraphNavg \bigr) = \bigl( \ZBM(\graphN) \bigr)^M$.\footnote{In terms of NFG language, the
%   NFG in Fig.~\ref{sec:SST:fig:10} can be obtained
%   by, first applying a closing-the-box operation and
%   then an opening-the-box operation~\cite{Loeliger2004} for every $e \in \setEfull$. Note that
%   closing-the-box and opening-the-box operations maintain the partition
%   function of the NFG.}

%***************************************************************************

In the NFG shown in Fig.~\ref{sec:SST:fig:10}, compared to the NFG in Fig.~\ref{sec:SST:fig:8}, the following new elements are introduced for every $ e \in \setEfull $:
\begin{itemize}

\item An edge is added with an associated variable (vector)
  $\cvpsi_e \in \sC^{|\set{X}_e|}$.

\item A function node is added, representing the function
  \begin{align}
    p_e\bigl( \cvpsi_e \bigr)
      &\defeq
         \bigl| \set{B}_{\set{X}_e^M} \bigr|
         \cdot
         \muFSsimple\bigl( \cvpsi_e \bigr).
     \label{eqn: def of pe}
  \end{align}

\item $M$ function nodes representing the function $\cpsi_e$ are added, and $M$ function nodes representing the function $\overline{\cpsi_e}$ are added.  The function $\cpsi_e$ is specified by $\cvpsi_e$ in a straightforward manner.

\end{itemize}

%***************************************************************************

%---------------------------------------------------------------------------
\begin{definition}
  \label{sec:SST:def:5}
  \index{SST!for S-NFG}
  We make the following definitions for $ \avgalt{\graphN} $.
  %-----------------------------------------------------------------------
  \begin{enumerate}

    \item We define a set of variables
    \begin{align*}
      \vxavgalt \defeq ( x_{\ef,m} )_{e \in \setpf, \, f \in \setF, \, m \in [M]} 
      \in \setx{}^{2M}.
    \end{align*}

    \item We define a set of vectors 
    \begin{align*}
      \cvpsiavgalt \defeq ( \cvpsi_e )_{e \in \setEfull},
    \end{align*}
    where $ \cvpsi_{e} $ is defined in~\eqref{sec:SST:eqn:39} for all $ e \in \setEfull $.

    \item We define the measure $ \muFSsimple\bigl( \cvpsiavgalt \bigr) $ to be the product of Fubini-Study measures: 
    \begin{align*}
      \muFSsimple\bigl( \cvpsiavgalt \bigr) 
      \defeq \prod\limits_{e} \muFSsimple\bigl( \cvpsi_e \bigr).
    \end{align*}

    \item For each $ e = (f_{i}, f_{j}) \in \setEfull $ such that $ i<j $, we define 
    \begin{align*}
      \cpsi_{\efi}(\xe)
      \defeq \cpsi_{e}(\xe), \qquad
      \cpsi_{\efj}(\xe)
      \defeq \overline{\cpsi_{e}(\xe)}, \qquad
      \xe \in \setxe.
    \end{align*}
  \end{enumerate}

  The construction of $\avgalt{\graphN}$ is given as follows.
  %----------------------------------------------------------------------------
  \begin{enumerate}
    \item For every $f \in \set{F}$, we follow the same step as in Step~\ref{step: for f covers} in Definition~\ref{sec:SST:def:1}.

    \item For every $ e = (f_{i}, f_{j}) \in \setEfull $, we add the following elements:
    \begin{itemize}

    \item an edge associated with variable (vector) 
    $\cvpsi_e \in \sC^{|\set{X}_e|}$;

    \item a function node representing the function $ p_e\bigl( \cvpsi_e \bigr) $ as defined in~\eqref{eqn: def of pe};

    \item $M$ function nodes representing the function $\cpsi_e$;

    \item $M$ function nodes representing the function $\overline{\cpsi_e}$. 

    \end{itemize}
    Then we connect each of the $M$ sockets $ (x_{e,f_{i},m})_{m \in [M]} $ to distinct function nodes representing the function $ \cpsi_e $. Similarly, we connect each of the $M$ sockets $ (x_{e,f_{j},m})_{m \in [M]} $ to distinct function nodes representing the function $ \overline{\cpsi_e} $. 

  \end{enumerate}
  %----------------------------------------------------------------------------

  \edefinition
\end{definition}
%---------------------------------------------------------------------------

With this, the global function of $\hgraphNavgalt$ is
\begin{align}
  g\bigl( \vxavgalt, \cvpsiavgalt \bigr)
    &= \left(
       \prod_f
         \prod\limits_{m \in [M]} f(\vx_{\setpff,m})
     \right)
     \nonumber\\
     &\quad \cdot
     \left(
       \prod\limits_{e=(f_{i},f_{j})}
          \Biggl( 
           \bigl| \set{B}_{\set{X}_e^M} \bigr|
           \cdot
           \prod\limits_{m \in [M]}
            \bigl( 
              \cpsi_{\efi}(x_{\efi,m})
              \cdot
              \cpsi_{\efj}(x_{\efj,m})
            \bigr)  
          \Biggr)
     \right) \label{sec:SST:eqn:26} \\
     &= \Biggl( \prod_f \prod\limits_{m \in [M]} f(\vx_{\setpff,m}) \Biggr)
     \cdot \prod\limits_{e} 
     \funcFS_{e,\cpsi_e}\bigl( \vx_{\efi,[M]}, \vx_{\efj,[M]} \bigr), 
     \label{sec:SST:eqn:10}
\end{align}
%***************************************************************************
for all $ \vx_{\efi,[M]}, \vx_{\efj,[M]} \in \set{X}_e^M, $
where the last equality follows from the definition of $ \funcFS $ in Definition~\ref{sec:SST:def:3}. We define $ \ZSSTM $ to be the partition function of $\hgraphNavgalt$ for fixed $ \cvpsiavgalt $:
%----------------------------------------------------------------------------
\begin{align}
  \ZSSTM\bigl( \avgalt{\graphN}, \cvpsiavgalt \bigr) \defeq 
  \sum\limits_{\vxavgalt} g\bigl( \vxavgalt, \cvpsiavgalt \bigr), 
  \label{sec:SST:eqn:43}
\end{align}
%----------------------------------------------------------------------------
where $ \sum\limits_{\vxavgalt} $ represents $ \sum\limits_{ \vxavgalt \in \setx{}^{2M} } $.

%---------------------------------------------------------------------------
\begin{proposition}
  \label{sec:SST:prop:1}
  \label{SEC:SST:PROP:1}

  The partition function of $\hgraphNavgalt$ satisfies
  \begin{align*}
    Z\bigl( \hgraphNavgalt \bigr)
    &= \int \ZSSTM\bigl( \avgalt{\graphN}, \cvpsiavgalt \bigr)
    \dd{\muFSsimple\bigl( \cvpsiavgalt \bigr)}
    \nonumber\\
    &= 
    \left( \prod\limits_{e}\bigl| \set{B}_{\set{X}_e^M} \bigr| \right)
      \cdot  
      \int \prod_f \Bigl( \ZSSTf\bigl( \cvpsi_{\setpff} \bigr) \Bigr)^{ \! M}
      \dd{\muFSsimple\bigl( \cvpsiavgalt \bigr)} 
      \\
    &= Z\bigl( \hgraphNavg \bigr),
  \end{align*}
  where for each function node $ f \in \setF $, we define
  %----------------------------------------------------------------------------
  \begin{align}
    \ZSSTf\bigl( \cvpsi_{\setpff} \bigr) 
    \defeq \sum\limits_{\vx_{\setpf}}
      f(\vx_{\setpf}) 
      \cdot
      \prod\limits_{e \in \setpf} 
      \cpsi_{e,f}(x_{e}), \label{sec:SST:eqn:42}
  \end{align}
  %----------------------------------------------------------------------------
  and $ \cvpsi_{\setpff} \defeq ( \cvpsi_{\ef} )_{e \in \setpf} $.
\end{proposition}
%***************************************************************************

\begin{proof}
  See Appendix~\ref{apx:property of SST}.
\end{proof}

%***************************************************************************
Consider conditioning on $\cvpsiavgalt$. The simplification $ \prod_f \Bigl( \ZSSTf\bigl( \cvpsi_{\setpff} \bigr) \Bigr)^{ \! M} $ obtained in Proposition~\ref{sec:SST:prop:1} is made possible by decomposing the NFG into $|\setF| \cdot M$ tree-structured components. This decomposition allows us to factorize and separate the partition function into individual components, each corresponding to a tree-structured factor graph consisting of a function node. 

% The simplification that was achieved in the proof of
% Proposition~\ref{sec:SST:prop:1} was thanks to the fact that,
% conditioned on $\cvpsiavgalt$, the NFG decomposes into an NFG with
% $|\setF| \cdot M$ components, where each component is tree-structured.

%***************************************************************************
%***************************************************************************

\section{Combining the Results of this Section}
\label{sec:SST:combining:results:1}
%***************************************************************************

Combining the main results of this section, \ie,
\begin{alignat*}{2}
  \text{(Introduction of Chapter~\ref{chapt:SST})} \quad &&
  \ZBM(\graphN)
    &=
       \sqrt[M]{
         \Bigl\langle 
           Z\bigl( \hgraphN \bigr)
         \Bigr\rangle_{\hgraphN \in \hat{\set{N}}_{M}} }, 
  % \label{sec:SST:eqn:expresssion of ZBM} 
  \nonumber\\
  \text{(Lemma~\ref{sec:SST:lem:1})} \quad &&
  Z\bigl( \hgraphNavg \bigr)
    &= \bigl( \ZBM(\graphN) \bigr)^{\! M}, 
  \nonumber\\
  %\label{sec:SST:eqn:23} \\
  \text{(Proposition~\ref{sec:SST:prop:1})} \quad &&
  Z\bigl( \hgraphNavgalt \bigr)
    &= Z\bigl( \hgraphNavg \bigr), 
  \nonumber\\
  %\label{sec:SST:eqn:24} \\
  \text{(Proposition~\ref{sec:SST:prop:1})} \quad &&
  Z\bigl( \hgraphNavgalt \bigr)
    &=\left( \prod\limits_{e}\bigl| \set{B}_{\set{X}_e^M} \bigr| \right)
    \nonumber \\
    &&&\quad \cdot  
    \int \prod_f \Bigl( \ZSSTf\bigl( \cvpsi_{\setpff} \bigr) \Bigr)^{ \! M}
    \dd{\muFSsimple\bigl( \cvpsiavgalt \bigr)},
  \nonumber
\end{alignat*}
we can obtain the expression for $\ZBM(\graphN)$ as follows:
\begin{align*}
  \ZBM(\graphN)
  &=  
  \Biggl( \prod\limits_{e} \bigl| \set{B}_{\set{X}_e^M} \bigr| \Biggr)^{\!\!\! 1/M}
  \cdot \Biggl(
  \int \prod_f \Bigl( \ZSSTf\bigl( \cvpsi_{\setpff} \bigr) \Bigr)^{ \! M}
  \dd{\muFSsimple\bigl( \cvpsiavgalt \bigr)}
  \Biggr)^{\!\!\! 1/M}. 
  % \label{alternative expression of ZBM in snfg}
\end{align*}

\section{The SST for DE-NFGs}
\label{sec:SST:DENFG:1}
\index{SST!for DE-NFG}
So far, this section was about the SST for S-NFG. The extension to the SST
for DE-NFG is, by and large, straightforward. Comparing with
Fig.~\ref{sec:SST:fig:1} for S-NFGs, we note that the only major change is that for
every edge $e \in \setEfull$, the role of the alphabet $\set{X}_e$ is taken over by
the alphabet $\tset{X}_e = \set{X}_e \times \set{X}_e$.

With this, 
the value of
$\ZBM(\graphN)$ for a DE-NFG $\graphN$ can be written as (compare
with Theorem~\ref{thm: expression of Z for degree M Bethe partition function})
\begin{align}
    \bigl( \ZBM(\graphN) \bigr)^M
    = 
    \sum\limits_{\vx_{[M]}}
    \Biggl( \prod\limits_{m \in [M]} \prod_f f(\tvx_{\setpff,m}) \Biggr)
    \cdot \prod\limits_{e} P_e\bigl( \tvx_{\efi,[M]}, \tvx_{\efj,[M]} \bigr),
  \label{eqn: Z for degree M average cover of DENFG}
\end{align}
where
\begin{align*}
  P_e\bigl( \tvx_{\efi,[M]}, \tvx_{\efj,[M]} \bigr)
  = \begin{cases}
    |\set{T}_{e,\vt_e}|^{-1}
     &  
     \vt_e = \vt_e\bigl(\tvx_{\efi,[M]}\bigr) 
     = \vt_e\bigl( \tvx_{\efj,[M]} \bigr) \\
   0
     & \text{otherwise}
  \end{cases}.
\end{align*}
Following a similar idea as in Section~\ref{sec:SST:combining:results:1}, we also rewrite $\ZBM(\graphN)$ as 
\begin{align*}
    \bigl( \ZBM(\graphN) \bigr)^{\!M}
    &= \biggl( \prod\limits_{e} |\set{B}_{\tset{X}_e^M}| \biggr)^{\!\!\! 1/M}
    \!\!\!\!\cdot\! \left( \int \prod_f 
      \Biggl( \sum\limits_{\tvx_f}
      f(\tvx_f) \prod\limits_{e \in \setpf} 
      \cpsi_{e,f}(\tx_{e,f})
      \Biggr)^{\!\!\!M}
      \!\!\!
      \dd{\muFSsimple(\cvpsiavgalt)}
    \right)^{\!\!\!\! 1/M}
    \!\!\!\!,
\end{align*}
where $\cvpsiavgalt \defeq ( \cvpsi_e )_{e \in \setEfull}$, 
$ \cvpsi_e \defeq \vpsi_e /\| \vpsi_e \|_{2} $ and 
$ \vpsi_e \defeq \bigl( \psi_e(\txe) \bigr)_{\! \txe} $
for all $  e \in \setEfull $,
and where for every
$e \in \setEfull$, the norm-one vector $\cvpsi_e$ is distributed according to the
Fubini-Study measure in dimension $|\tset{X}_e|$.

\chapter{The Graph-Cover Theorem for Special DE-NFG}
\label{chapt: graph cover thm for DENFG}
\label{CHAPT: GCT DENFG}
\label{sec:CheckCon}

In this section, we present a proof for Conjecture~\ref{sec:GraCov:conj:1}, \ie, the graph-cover theorem, for a class of DE-NFGs satisfying an easily checkable condition. In Conjecture~\ref{sec:GraCov:conj:1}, we conjecture that for some DE-NFG $ \graphN $, the limit in~\eqref{sec:CheckCon:eqn:25} holds
% %---------------------------------------------------------------------------
% \begin{align}
%     \limsup_{M \to \infty} \ZBM(\graphN) = \ZBSPA^{*}(\graphN). 
    
% \end{align}
%---------------------------------------------------------------------------
% where $ \ZBM(\graphN) $ is defined in Definition~\ref{sec:GraCov:def:2} and $ \ZBSPA^{*}(\graphN) $ is defined in Definition~\ref{sec:DENFG:def:3}.
% Although the expression of $\ZBM(\graphN)$ 
% in Proposition~\ref{prop:altenative expression of ZBM by eigenvalue decomposition and SST}, thanks to the SST, is
% interesting, it is generally challenging to analyze this expression in the limit 
% $M \to \infty$. However, if we apply the SST after applying the LCT to every $M$-cover of $\graphN$, we obtain expressions that are easier to analyze.
%***************************************************************************
This section is organized as follows.

\begin{enumerate}

  \item \textbf{Application of the LCT}: We begin by applying the LCT to a DE-NFG $ \graphN $, based on an SPA fixed point that achieves $ \ZBSPA^{*}(\graphN) $. The resulting factor graph is denoted by $ \LCT{\graphN} $. Recall in Proposition~\ref{prop:DENFG:LCT:1}, we establish that $ \ZBM(\graphN) = \ZBM\bigl( \LCT{\graphN} \bigr) $ for any $ M \in \sZpp $.

  % \item \textbf{Application of the SST}: Next, we apply the SST to $ \LCT{\graphN} $, leading to a factor graph denoted by $ \LCTavgalt{\graphN}$. Following the similar idea as in Section~\ref{sec:SST}, we obtain $ \ZBM\bigl( \LCT{\graphN} \bigr) = \sqrt[M]{Z\bigl( \LCTavgalt{\graphN}\bigr)} $. 

  \item \textbf{Limit Evaluation}: In order to establish the limit stated in~\eqref{sec:CheckCon:eqn:25}, we need to evaluate $ \ZBM(\graphN) $ as $ M \to \infty $. We propose a checkable sufficient condition based on the local functions in the LCT-transformed DE-NFG  $ \LCT{\graphN} $, which ensures that $ \lim\limits_{M \to \infty}\ZBM\bigl( \LCT{\graphN} ) = \ZBSPA^{*}( \graphN ) $.
\end{enumerate}
%***************************************************************************
%***************************************************************************

The considered setup is given as follows.
%---------------------------------------------------------------------------
\begin{itemize}
    \item Consider an arbitrary DE-NFG $\graphN$.

    \item Assume that among all fixed-point SPA message vectors for $\graphN$, we obtain a vector $\vmu$ that achieves the maximum value of $\ZBSPA(\graphN,\vmu)$, \ie,
    %-----------------------------------------------------------------------
    \begin{align*}
         \ZBSPA(\graphN,\vmu) = \ZBSPA^{*}(\graphN) \in \sR_{\geq 0}, \qquad 
         Z_{e}(\vmu) >0, \qquad e \in \setEfull.
    \end{align*}
    %-----------------------------------------------------------------------
    (See also Definition~\ref{sec:DENFG:def:3} and Assumption~\ref{sec:DENFG:remk:1}.) 

    \item We apply the LCT w.r.t. $\vmu$ on $\graphN$, resulting in a new DE-NFG denoted by $ \LCT{\graphN} $. We recall the properties of $ \LCT{\graphN} $, which are summarized in Proposition~\ref{prop:DENFG:LCT:1}, 
    as follows.
    %----------------------------------------------------------------------------
    \begin{itemize}

       \item It holds that $ Z\bigl( \LCT{\graphN} \bigr) = Z(\graphN) $.

       \item The fixed-point SPA message vector $\vmu$ induces an SPA fixed-point message vector
        $\LCTv{\mu} = \bigl( \LCTv{\mu}_{\ef} \bigr)_{ e \in \setpf,\, f \in \setF}$ for
        $\LCT{\graphN}$, which has a simple structure:
        %-------------------------------------------------------------------
        \begin{align*}
            \LCTv{\mu}_{\ef}
            = \begin{pmatrix}
                1, & 0, & \cdots & 0
            \end{pmatrix}^{\!\! \tran}
            \in \sR^{|\LCTtset{X}_{e}|},
            \qquad
            e \in \setpf,\, f \in \setF.  \nonumber
        \end{align*}
        %-------------------------------------------------------------------
       % \item There is a bijection between the set of fixed-point SPA message vectors for $ \graphN $ and the set of fixed-point SPA message vectors for $\LCT{\graphN}$. Therefore, following the definition of $ \ZBSPA^{*} $ in Definition~\ref{sec:DENFG:def:3}, we obtain
       This simple structure implies
        %-----------------------------------------------------------------------
        \begin{align*}
            \ZBSPA^{*}(\graphN)
            % = \ZBSPA^{*}\bigl( \LCT{\graphN} \bigr)
            = \ZBSPA(\graphN,\vmu)
            = \ZBSPA(\LCT{\graphN}, \LCTv{\mu})
            =
            \prod\limits_{f} \LCT{f}(\bm{0}) \in \sR_{\geq 0},
            % \label{sec:CheckCon:eqn:31}
        \end{align*}
        where the property $ \ZBSPA^{*}(\graphN) = \ZBSPA(\graphN,\vmu) \in \sR_{\geq 0} $ follows from Assumption~\ref{sec:DENFG:remk:1}. 
   \end{itemize}
   %----------------------------------------------------------------------------
    
    \item For each $ M \in \sZpp $, we consider a degree-$M$ finite graph cover, 
    denoted by $ \hLCTgraphN $. This graph cover is specified by a collection of permutations 
    $ \vsigma = (\sigma_e)_{e \in \setEfull} 
    \in \set{S}^{|\setEfull|}_{M} $.

    The collection of vectors given by
    %-----------------------------------------------------------------------
    \begin{align*}
        \hat{\LCTv{\mu}}_{(e,m),(f_{i},m)}=
        \hat{\LCTv{\mu}}_{(e,m),(f_{j},\sigma_{e}(m))}
        = \begin{pmatrix}
            1, & 0, & \cdots & 0
        \end{pmatrix}^{\!\!\! \tran},
        \,\, m \in [M],\, e = (f_{i},f_{j}) \in \setEfull
    \end{align*}
    %-----------------------------------------------------------------------
    represents a collection of fixed-point SPA message vectors for $ \hLCTgraphN $.
    Here, $ \hat{\LCTv{\mu}}_{(e,m),(f_{i},m')} $ denotes the SPA message vector from the edge $ (e,m) $ to local function $ (f_{i},m') $ for $ m, m' \in [M] $.
    (Because an $M$-cover of $\hLCTgraphN$ looks locally the same as the base graph $\LCT{\graphN}$, the SPA fixed-point vector $ \LCTv{\mu}$ for $\LCT{\graphN}$ yields an SPA fixed point vector $\hat{\LCTv{\mu}}$ for $\hLCTgraphN$ simply by repeating every message $M$ times.\footnote{This observation was the original motivation for the graph cover analysis of message-passing iterative decoders~\cite{Koetter2003,Koetter2007}.})

    \item The degree-$M$ Bethe partition functions of $ \graphN $ and 
        $ \LCT{\graphN} $ are equal, \ie,
    %-----------------------------------------------------------------------
    \begin{align*}
        \ZBM(\graphN) 
        &\overset{(a)}{=} \sqrt[M]{
          \bigl\langle
             Z\bigl( \hgraphN \bigr)
           \bigr\rangle_{ \hgraphN \in \hat{\set{N}}_{M}}
        }
        \overset{(b)}{=} \sqrt[M]{
            \Bigl\langle
                Z\bigl( \hLCTgraphN \bigr)
            \Bigr\rangle_{ \hLCTgraphN \in \hLCT{\set{N}}_{M} }
        }
        \overset{(c)}{=} \ZBM\bigl( \LCT{\graphN} \bigr) ,\quad M \in \sZpp,  \nonumber
    \end{align*}
    %----------------------------------------------------------------------
    where step $(a)$ follows from the definition of the function $ \ZBM $ in Definition~\ref{sec:GraCov:def:2}, where step $(b)$ follows from Proposition~\ref{prop:DENFG:LCT:1},
    and where step $(c)$ again follows from~Definition~\ref{sec:GraCov:def:2}.

    \item In the following, we evaluate $ \ZBM\bigl( \LCT{\graphN} \bigr) $ as $ M $ grows to infinity.

    % %----------------------------------------------------------------------------
    % \begin{itemize}
    %     \item 
    %     \item 

    %     % \item Following the similar idea of Lemmas~\ref{sec:SST:lem:1} and~\ref{sec:SST:lem:6}, the $M$-th power of the degree-$M$ Bethe partition can be written as
    %     % \begin{align*}
    %     %     \ZBM(\graphN) 
    %     %     = Z\bigl( \hgraphNavg \bigr) 
    %     %     =
    %     %      \sum\limits_{ \LCTvxsM, \LCTvxlM }
    %     %       \Biggl( 
    %     %         \prod\limits_{m \in [M]} \prod\limits_{f} 
    %     %         \LCT{f}( \LCTtv{x}_{\setpf,f,m}) 
    %     %       \Biggr) 
    %     %       \cdot \prod\limits_{e = (f_{i},f_{j})} 
    %     %       P_{e}(\LCTtv{x}_{\efi,[M]}, \LCTtv{x}_{\efj,[M]})
    %     % \end{align*}
    
    % \end{itemize}
    % %---------------------------------------------------------------------------- 
    (Note that the order of operations, namely constructing a graph cover and applying the LCT, does not affect the resulting degree-$M$ Bethe partition function. In other words, changing the order of these two operations does not alter the outcome.)
\end{itemize}

%---------------------------------------------------------------------------
\begin{theorem}\label{sec:CheckCon:thm:1}
  \label{SEC:CHECKCON:THM:1}
  \index{Graph-cover theorem!for special DE-NFG}
  Consider a DE-NFG $ \graphN $, along with its associated LCT-transformed DE-NFG $ \LCT{\graphN} $. If 
  the following checkable inequality is satisfied:
  %-----------------------------------------------------------------------
  \begin{align}
    \frac{3}{2}
    \cdot \ZBSPA^{*}(\graphN)
    &>
    \prod\limits_{f}
     \left( 
        \sum\limits_{\LCTtv{x}_{\setpf}}
        \left| \LCT{f}( \LCTtv{x}_{\setpf} ) \right|
    \right),
    % \Biggr),
    % \sum\limits_{ \LCTvxs, \LCTvxl }
    % \bigl|G_{ \LCTzero{\graphN}, 1 }( \LCTvxs, \LCTvxl )\bigr|,
    \label{sec:CheckCon:eqn:70}
  \end{align}
  %-----------------------------------------------------------------------
  then we have
  %-----------------------------------------------------------------------
  \begin{align*}
    \lim\limits_{M \to \infty}
    \ZBM(\graphN)
    = \lim\limits_{M \to \infty}
    \ZBM(\LCT{\graphN})
    = \ZBSPA^{*}(\graphN). 
  \end{align*}
  %-----------------------------------------------------------------------
  % {\color{red}(whether $\limsup$ or $ \liminf $, we will come back later)}
\end{theorem}
%---------------------------------------------------------------------------
%---------------------------------------------------------------------------
\begin{proof}
  See Appendix~\ref{apx:check_graph_thm}.
\end{proof}

\chapter{Summary and Open Problems}

In this chapter, we summarize this thesis and present open problems that are left for future research.

\label{chapt: summary and outlook}

\section[Degree-\texorpdfstring{$M$}{}-Bethe-Permanent-based Bounds]{Degree-\texorpdfstring{$M$}{}-Bethe-permanent-based bounds for 
\texorpdfstring{\\}{} the permanent of a non-negative square matrix}

In Chapter~\ref{chapt: preliminaries}, we have introduced a class of S-NFGs such that the partition function of each S-NFG equals the permanent of a non-negative square matrix. We have also introduced the degree-$M$ Bethe and scaled Sinkhorn permanents, which are based on the degree-$M$ covers of the underlying S-NFG.

In Chapter~\ref{chapt: finite graph cover based bound for matrix permanent}, we have shown that it is possible to bound the permanent of a non-negative
matrix by its degree-$M$ Bethe and scaled Sinkhorn permanents, thereby, among
other statements, proving a conjecture in~\cite{Vontobel2013a}. 
A key result are the recursive
expressions for the coefficients appearing in the expressions for the $M$-th
power of the degree-$M$ Bethe permanent and the $M$-th power of the
degree-$M$ scaled Sinkhorn permanent.

While the results in Chapter~\ref{chapt: finite graph cover based bound for matrix permanent} do not yield any new numerical techniques for
bounding the permanent of a non-negative square matrix, they can, besides being of
inherent combinatorial interest, potentially be used to obtain new analytical
results about the Bethe and the Sinkhorn approximation of the permanent,
similar to the results in the paper~\cite{KitShing2022}. (See also the
discussion at the end of Section~\ref{sec:main:constributions:1}.)

Some open problems about the topic studied in Chapter~\ref{chapt: finite graph cover based bound for matrix permanent} include the following:
\begin{itemize}

\item Our proofs in Chapter~\ref{chapt: finite graph cover based bound for matrix permanent} used some rather strong results from~\cite{Schrijver1998,
    Gurvits2011, Anari2019, Egorychev1981, Falikman1981}. We leave it as an
  open problem to find ``more basic'' proofs for some of the inequalities that
  were established in this thesis.

\item One can verify that proving the first inequality
  in~\eqref{eq:ratio:permanent:bethe:permanent:1} for all~$\vtheta$ is
  equivalent to proving $ \HG'(\vgam) \geq \HBthe( \vgam ) $ for all
  $ \vgam \in \Gamma_{n} $. Note that proving the first inequality
  in~\eqref{eq:ratio:permanent:bethe:permanent:1} for all $ \vtheta $ is also
  equivalent to proving Schrijver's inequality~\cite{Schrijver1998}. We leave it as an open problem
  to prove $ \HG'(\vgam) \geq \HBthe( \vgam ) $ for all
  $ \vgam \in \Gamma_{n} $ via an approach different from Schrijver's
  inequality and its variations.

\item It was remarked that both the second inequality in~\eqref{SEC:1:EQN:147}
  and the first inequality in~\eqref{sec:1:eqn:200} are not tight. We leave it
  as an open problem to find tighter bounds, or even tight bounds.

\item Recently, Huang and Vontobel~\cite{Huang2024} extended Vontobel's
  results in~\cite{Vontobel2013a} by considering a class of bipartite S-NFGs
  where each local function is defined based on a (possibly different)
  multi-affine homogeneous real stable polynomial. The S-NFGs considered in~\cite{Huang2024} can be seen as an extension of the class of S-NFGs considered in Chapter~\ref{chapt: finite graph cover based bound for matrix permanent}. Generalizing the results in Section~\ref{sec:main:constributions:1} to the S-NFGs discussed
  in~\cite{Huang2024} remains an open problem for future research.

\end{itemize}

\section[The Graph-Cover Theorem for Special DE-FGs]{Characterizing the Bethe Partition Function of \texorpdfstring{\\}{} Double-Edges Factor Graphs via Graph Covers}

In Chapter~\ref{chapt: preliminaries}, we have also introduced the general definitions of both S-NFGs and DE-NFGs. The latter are highly relevant in the field of quantum information processing. We have defined the SPA and finite graph covers for both S-NFGs and DE-NFGs. A key difference between S-NFGs and DE-NFGs is that S-NFGs consist of non-negative local functions only, while DE-NFGs have complex-valued local functions that satisfy some positive semi-definiteness constraints. We have considered a special SPA for DE-NFGs, where the messages are initialized to satisfy certain positive semi-definiteness constraints.

In Chapters~\ref{chapt: LCT} and~\ref{chapt:SST}, we have defined the LCT and the SST for both S-NFGs and DE-NFGs. Our proposed LCT does not have the limitation on non-zero components in SPA fixed-point messages, which extends the LCT proposed by Mori~\cite{Mori2015}.  
The SST is a novel contribution to the graphical models literature and has broader applications beyond proving the main results of this thesis.

In Chapter~\ref{chapt: graph cover thm for DENFG}, we have proven the graph-cover theorem, \ie, a combinatorial characterization of the Bethe partition function in terms of finite graph covers, for a class of DE-NFGs that satisfy an easily checkable condition. The proof is based on analyzing the finite graphs cover of loop-calculus transformed DE-NFGs.  

The graph-cover theorem for DE-NFGs, beyond those satisfying the easily checkable condition stated in Chapter~\ref{chapt: graph cover thm for DENFG}, remains an open problem.
In Chapter~\ref{chapt: preliminaries}, we have supported the graph-cover conjecture for DE-NFGs with numerical results on small DE-NFGs. We also note that the SST provides a different perspective to understand finite graph covers, and we hope that the SST can be used for proving the graph-cover theorem for more general DE-NFGs.  

%=======================================================

% appendix

        %\chapter*{Appendix}
       % \addcontentsline{toc}{chapter}{Appendix}
	
	%\vskip 1cm \noindent
	\input{appendix.tex}
	%\newpage

%=======================================================

%bibliography

\ifx\sectionheaderonnewpage\x
\clearpage
\fi

\addcontentsline{toc}{chapter}{Bibliography}

% \bibliographystyle{ieee-alphabetic}
% \bibliography{biblio.bib}

\printbibliography

% \ifx\sectionheaderonnewpage\x
% \clearpage
% \fi

\printindex

% \indexpage

% % \printbibliography
% \bibliographystyle{ieee-alphabetic}
% \bibliography{thesisbiblio.bib}

\end{document}

%% file: mathcommand.tex
\newtheorem{theorem}{Theorem}

\newtheorem{assumption}[theorem]{Assumption}
\newtheorem{lemma}[theorem]{Lemma}
\newtheorem{proposition}[theorem]{Proposition}
\newtheorem{example}[theorem]{Example}
\newtheorem{remark}[theorem]{Remark}

\newtheorem{claim}[theorem]{Claim}

\newtheorem{conjecture}[theorem]{Conjecture}

\newtheorem{definition}[theorem]{Definition}

\newtheorem{observation}[theorem]{Observation}
\newcommand{\eproposition}{\hfill$\square$}
\newcommand{\epropositionproof}{\hfill\qedsymbol}
\newcommand{\elemma}{\hfill$\square$}
\newcommand{\elemmaproof}{\hfill\qedsymbol}
\newcommand{\etheorem}{\hfill$\square$}
\newcommand{\eassumption}{\hfill$\square$}
\newcommand{\etheoremproof}{\hfill\qedsymbol}
\newcommand{\edefinition}{\hfill$\triangle$}

\newcommand{\eexample}{\hfill$\square$}

\newcommand{\econjecture}{\hfill$\square$}
\newcommand{\eremark}{\hfill$\triangle$}

\newcommand{\vect}[1]{ \bm{#1} }

\newcommand{\set}[1]{\mathcal{#1}}
\newcommand{\imagunit}{\iota}
\newcommand{\beli}{\beta}
\newcommand{\vbeli}{ \vect{\beli} }

\newcommand{\vgam}{ \vect{\gamma} }
\newcommand{\setEfull}{ \mathcal{E} }
\newcommand{\setF}{ \mathcal{F} }
\newcommand{\sfN}{ \mathsf{N} }
\newcommand{\setpf}{ \partial f }
\newcommand{\setpe}{ \partial e }
\newcommand{\sR}{\mathbb{R}}
\newcommand{\sC}{\mathbb{C}}
\newcommand{\sZ}{\mathbb{Z}}
\newcommand{\sZpp}{\mathbb{Z}_{\geq 1}}

\newcommand{\FB}{F_{ \mathrm{B} }}

\newcommand{\ZB}{Z_{ \mathrm{B} }}

\newcommand{\Herm}{\mathsf{H}}
\newcommand{\tran}{\mathsf{T}}
\newcommand{\ze}{ z_{e} }
\newcommand{\xe}{ x_{e} }

\newcommand{\matr}[1]{ \bm{#1} }
\newcommand{\defeq}{ \triangleq }

\newcommand{\fr}[1]{ f_{ \mathrm{r}, #1 } }
\newcommand{\fc}[1]{ f_{ \mathrm{c}, #1 } }

\newcommand{\setx}[1]{ \set{X}_{#1} }
\newcommand{\setxe}{ \set{X}_{e} }
\newcommand{\setX}{ \set{X} }

\newcommand{\vx}{ \bm{x} }
\newcommand{\graphN}{ \mathsf{N} }

\newcommand{\vvu}{\bm{u}}

\newcommand{\vui}{\vvu_{i}}
\newcommand{\vuj}{\vvu_{j}}

\newcommand{\supp}{\mathrm{supp}}

\newcommand{\sRp}{\mathbb{R}_{\geq 0}}

\newcommand{\xf}{\bm{x}_{\partial f}}

\newcommand{\ZBn}[1]{\perm_{ \mathrm{B}, #1 }}

\newcommand{\perm}{ \mathrm{perm} }

\newcommand{\permb}{ \mathrm{perm}_{\mathrm{B}} }
\newcommand{\permbM}[1]{\perm_{ \mathrm{B}, #1 }}

\newcommand{\perms}{ \mathrm{perm}_{\mathrm{S}} }
\newcommand{\permscs}{ \mathrm{perm}_{\mathrm{scS}} }
\newcommand{\permscsM}[1]{ \mathrm{perm}_{\mathrm{scS},#1} }

\newcommand{\pmtheta}{p_{\mtheta}}

\newcommand{\hgraphN}{\hat{\graphN}}

\newcommand{\ef}{e,f}

\newcommand{\efi}{e,f_i}
\newcommand{\efj}{e,f_j}

\newcommand{\setpff}{ \setpf,f }

\newcommand{\setxf}{ \setX_{\setpf} }

\newcommand{\vxf}{\bm{x}_{\partial f}}

\newcommand{\setpfi}{\setpf_{i}}

\newcommand{\etof}{e, f}
\newcommand{\vmu}{\bm{\mu}}
\newcommand{\fSPAN}{f_{\mathrm{SPA},\sfN}}

\newcommand{\setpfj}{\setpf_{j}}

\newcommand{\wh}{ w_{\mathrm{H}} }

\newcommand{\vsigma}{ \bm{\sigma} }

\newcommand{\setC}{\set{C}}
\newcommand{\setP}{\set{P}}
\newcommand{\setR}{\set{R}}
\newcommand{\setS}{\set{S}}

\newcommand{\vgamRCp}{\vgam_{\setR,\setC}}
\newcommand{\gamRCp}{\gamma_{\setR,\setC}}

\newcommand{\CM}[1]{\Cgen{M}{#1}}
\newcommand{\Ctwo}[1]{\Cgen{2}{#1}}
\newcommand{\Cgen}[2]{C_{#1,#2}}

\newcommand{\CBM}[1]{\CBgen{M}{#1}}
\newcommand{\CBtwo}[1]{\CBgen{2}{#1}}
\newcommand{\CBgen}[2]{C_{\mathrm{B},#1,#2}}

\newcommand{\CscSgen}[2]{C_{\mathrm{scS},#1,#2}}

\newcommand{\mtheta}{\matr{\theta}}

\newcommand{\isigmaoi}{\bigl( i,\sigma_1(i) \bigr)}

\newcommand{\tvgam}{\tilde{\vgam}}
\newcommand{\tgam}{\tilde{\gamma}}

\newcommand{\sZp}{\sZ_{\geq 0}}

\newcommand{\vtheta}{\matr{\theta}}
\newcommand{\mP}{\matr{P}}

\newcommand{\mU}{\matr{U}}

\newcommand{\sfG}{\mathsf{G}}

\newcommand{\FBthe}{F_{ \mathrm{B}, \mtheta }}
\newcommand{\UBthe}{U_{ \mathrm{B}, \mtheta }}
\newcommand{\HBthe}{H_{ \mathrm{B} }}

\newcommand{\FscSthe}{F_{ \mathrm{scS}, \mtheta }}
\newcommand{\UscSthe}{U_{ \mathrm{scS}, \mtheta }}
\newcommand{\HscSthe}{H_{ \mathrm{scS} }}

\newcommand{\setA}{ \set{A} }
\newcommand{\suppmtheta}{\operatorname{supp}(\mtheta)}

\newcommand{\hvgamRCp}{ \hat{\vgam}_{\setR,\setC} }
\newcommand{\hgamRCp}{ \hat{\gamma}_{\setR,\setC} }

\newcommand{\vgamzo}{\vgam}
\newcommand{\gamzo}{\gamma}

\newcommand{\HG}{H_{ \mathrm{G} }}

\newcommand{\FGthe}{F_{ \mathrm{G}, \mtheta }}
\newcommand{\UGthe}{U_{ \mathrm{G}, \mtheta }}

\newcommand{\vp}{\bm{p} }

\newcommand{\PiA}[1]{ \Pi_{\setA(#1)} }
\newcommand{\Gamnthe}{ \Gamma_{n}(\vtheta) }
\newcommand{\GamMnthe}{ \Gamma_{M,n}(\vtheta) }
\newcommand{\Gampnthe}{ \Gamma'_{n}(\vtheta) }

\newcommand{\GamMn}{ \Gamma_{M,n} }
\newcommand{\Gamtwon}{ \Gamma_{2,n} }

\newcommand{\tPsi}{\tilde{\Psi}}
\newcommand{\eg}{\textit{e.g.}}
\newcommand{\ie}{\textit{i.e.}}
\newcommand{\etal}{\textit{et al.}}

\newcommand{\gtheta}{g_{\sfN(\mtheta)}}
\newcommand{\e}{\mathrm{e}}
\newcommand{\bigformulabottom}[3]{%
  \begin{figure*}[!b]
    \normalsize
    \vspace*{4pt}
    \hrulefill

    \setcounter{equation}{#1}
    #3

    \setcounter{equation}{#2}
  \end{figure}
}

\newcommand{\cw}{ \check{w} }

\newcommand{\vw}{ \vect{w} }
\newcommand{\cvw}{ \check{\vw} }
\newcommand{\vpsi}{ \vect{ \psi } }
\newcommand{\cpsi}{ \check{ \psi } }
\newcommand{\cvpsi}{ \check{ \vect{ \psi } } }
\newcommand{\tvz}{ \tilde{ \vect{z} } }

\newcommand{\ZSSTf}{ Z_{\mathrm{SST},f} }

\newcommand{\setE}{ \mathcal{E} }
\newcommand{\setHerm}[1]{ \mathrm{Herm}_{#1} }
\newcommand{\setPSD}[1]{ \mathrm{PSD}_{#1} }

\newcommand{\sRpp}{ \mathbb{R}_{> 0} }

\newcommand{\LMP}{ \set{B} }
\newcommand{\ZBSPA}{ Z_{\mathrm{B,SPA}} }
\newcommand{\ZBM}{ Z_{\mathrm{B},M} }
\newcommand{\setcovN}[1]{ \hat{\set{N}}_{#1} }
\newcommand{\tset}[1]{\tilde{\mathcal{#1}}}
\newcommand{\LCTtset}[1]{\mathring{\tilde{\mathcal{#1}}}}
\newcommand{\LCTset}[1]{\mathring{\mathcal{#1}}}
\newcommand{\LCTtv}[1]{ \mathring{ \tilde{ \vect{#1} } } }
\newcommand{\LCTv}[1]{ \mathring{ \vect{#1} } }
\newcommand{\LCT}[1]{ \mathring{#1} }
\newcommand{\LCTt}[1]{ \mathring{ \tilde{#1} } }

\newcommand{\Sigmae}{\Sigma_{(e)}}
\newcommand{\vSigma}{\vect{\Sigma}}
\newcommand{\hgraphNavg}{ \hat{\graphN}_{M,\mathrm{avg}} }
\newcommand{\hgraphNavgalt}{ \hat{\graphN}_{M,\mathrm{avg}}^{\mathrm{alt}} }
\newcommand{\funcFS}{ \Phi }
\newcommand{\matrFS}{ \vect{\Phi} }

\newcommand{\muFSsimple}{ \mu_{\mathrm{FS}} }
\newcommand{\vW}{ \vect{W} }
\newcommand{\vxavgalt}{ \vx_{\mathrm{avg}}^{\mathrm{alt}}  }
\newcommand{\cvpsiavgalt}{ \cvpsi_{\mathrm{avg}}^{\mathrm{alt}} }

\newcommand{\avgalt}[1]{ \hat{{#1}}_{M,\mathrm{avg}}^{\mathrm{alt}} }

\newcommand{\epfi}{ e',f_i }
\newcommand{\epfj}{ e',f_j }

\newcommand{\ZSSTM}{ Z_{\mathrm{SST},M} }

\newcommand{\unpair}[1]{ \dot{#1} }
\newcommand{\pair}[1]{ \ddot{#1} }

\newcommand{\upone}{ \unpair{1} }
\newcommand{\uptwo}{ \unpair{2} }
\newcommand{\upthre}{ \unpair{3} }
\newcommand{\upfor}{ \unpair{4} }
\newcommand{\upfif}{ \unpair{5} }

\newcommand{\pone}{ \pair{1} }
\newcommand{\ptwo}{ \pair{2} }
\newcommand{\pthre}{ \pair{3} }
\newcommand{\pfor}{ \pair{4} }
\newcommand{\pfif}{ \pair{5} }

\newcommand{\xfi}{ \bm{x}_{\setpf_i} }

\newcommand{\hLCT}[1]{ \hat{ \LCT{#1} } }
\newcommand{\hLCTgraphN}{ \hLCT{\graphN} }

\newcommand{\hgraphPsig}{ \hgraphN_{M,P_{\vsigma}} }

\newcommand{\LCTvxf}{ \LCTv{x}_{\setpf} }
\newcommand{\LCTvxfi}{ \LCTv{x}_{\setpfi} }

\newcommand{\LCTsetx}{ \LCTset{X} }

\newcommand{\LCTsetxf}{ \LCTset{X}_{\setpf} }

\newcommand{\tx}{ \tilde{x} }
\newcommand{\txe}{ \tilde{x}_{e} }
\newcommand{\tvx}{ \tilde{\vect{x}}  }
\newcommand{\LCTvxl}{ \LCTtv{x}_{\mathrm{l}} }
\newcommand{\LCTvxs}{ \LCTtv{x}_{\mathrm{s}} }
\newcommand{\LCTvxlm}{ \LCTtv{x}_{\mathrm{l},m} }
\newcommand{\LCTvxlsigm}{ \LCTtv{x}_{\mathrm{l},\sigma(m)} }
\newcommand{\LCTvxsm}{ \LCTtv{x}_{\mathrm{s},m} }
\newcommand{\LCTvxssigm}{ \LCTtv{x}_{\mathrm{s},\sigma(m)} }
\newcommand{\LCTvxlM}{ \LCTtv{x}_{\mathrm{l},[M]} }
\newcommand{\LCTvxsM}{ \LCTtv{x}_{\mathrm{s},[M]} }

\newcommand{\tzero}{\tilde{0}}

\newcommand{\gzero}{g_{0}}
\newcommand{\gone}{g_{1}}

\newcommand{\tze}{ \tilde{z}_{e} }

\newcommand{\etofone}{ e, f_1 }
\newcommand{\onetoftwo}{ 1, f_2 }

\newcommand{\tv}[1]{ \tilde{\vect{#1}} }

\newcommand{\vt}{ \bm{t} }
\newcommand{\vv}{ \bm{v} }
\newcommand{\tz}{ \tilde{z} }

%% file: abstract.tex
%!TEX root = main.tex
\noindent 

\vspace*{-0.8cm}
% Finite graph covers of a factor graph is a useful tool to analyze the SPA on this factor graph. In broad terms, a factor graph is a graph cover of another factor graph if there is a bijective mapping that map the neighborhood of each vertex in the former factor graph onto the neighborhood of the associated vertex in the latter factor graph.

Various intractable computational problems can be reduced to computing the marginals and the partition function of a suitably defined multivariate function represented by a factor graph. 
For computational problems relating to classical statistical models, we consider standard factor graphs (S-FGs), \ie, factor graphs with local functions taking on non-negative real values. For computational problems relating to quantum information processing systems, we consider so-called double-edge factor graphs (DE-FGs), a class of factor graphs where local functions take on complex values and satisfy some positive semi-definiteness constraints. The sum-product algorithm (SPA) provides an efficient iterative method for approximating the marginals and the partition function induced by each of these factor graphs. 

For an arbitrary S-FG, Yedidia \etal~showed that fixed points of the SPA on this S-FG correspond to stationary points of the Bethe free energy function. Based on this, they defined the so-called Bethe approximation of the partition function, or simply the Bethe partition function, in terms of the minimum of the Bethe free energy function. Later, Vontobel provided a combinatorial characterization of the Bethe partition function in terms of the finite graph covers of the S-FG. The proof of this characterization heavily relied on the method of types. 

In this thesis, we leverage finite graph covers to analyze the SPA and the Bethe partition function for both S-FGs and DE-FGs. There are two main contributions in this thesis.

The first main contribution concerns a special class of S-FGs where the partition function of each S-FG equals the permanent of a nonnegative square matrix. The Bethe partition function for such an S-FG is called the Bethe permanent. A combinatorial characterization of the Bethe permanent is given by the degree-$M$ Bethe permanent, which is defined based on the degree-$M$ graph covers of the underlying S-FG.
In this thesis, we prove a degree-$M$-Bethe-permanent-based lower bound on the
permanent of a non-negative square matrix, resolving a conjecture proposed by
Vontobel in [IEEE Trans. Inf. Theory, Mar.\ 2013]. We also prove a
degree-$M$-Bethe-permanent-based upper bound on the permanent of a
non-negative matrix. In the limit $M \to \infty$, these lower and upper
bounds yield known Bethe-permanent-based lower and upper bounds on the
permanent of a non-negative square matrix. Similar results are obtained for
another approximation of the permanent known as the (scaled) Sinkhorn permanent.

The second main contribution is giving a combinatorial characterization of the Bethe partition function for DE-FGs in terms of finite graph covers. In general, approximating the partition function of a DE-FG is more challenging than for an S-FG because the partition function of the DE-FG is a sum of complex values and not just a sum of non-negative real values. Moreover, one cannot apply the method of types for proving the combinatorial characterization as in the case of S-FGs. We overcome this challenge by applying a suitable loop-calculus transform (LCT). Originally, the LCT was introduced by Chertkov and Chernyak as a special linear transform for S-FGs based on SPA fixed-point messages. Mori further developed the LCT for S-FGs with non-binary alphabet, but his LCT was limited to the case where SPA fixed-point messages consist of non-zero components only. In contrast, the SPA on DE-FGs often involves messages containing zeros. We introduce an alternative LCT that is applicable for both S-FGs and DE-FGs and is without the limitation of non-zero components in SPA fixed-point messages. Currently, we provide a combinatorial characterization of the Bethe partition function in terms of finite graph covers for a class of DE-FGs satisfying an (easily checkable) condition. However, based on numerical results, we suspect that this combinatorial characterization holds more broadly. We further introduce the symmetric-subspace transform (SST), which reformulates the previously mentioned combinatorial characterizations for S-FGs and DE-FGs in terms of some integral. Both the LCT and the SST should be of interest beyond proving the main results of this thesis. 

\newpage
\begin{CJK*}{UTF8}{bsmi}
\noindent\textbf{摘要：}
%\vspace*{1.5cm}
\noindent 

各種棘手的計算問題可以歸結為計算由因子圖表示的適當定義的多變量函數的邊際和配分函數。對於與經典統計模型相關的計算問題，我們考慮標準因子圖（S-FG），即局部函數取非負實值的因子圖。對於與量子信息處理系統相關的計算問題，我們考慮雙邊因子圖（DE-FG），這是一類局部函數取複數值並滿足某些半正定約束的因子圖。乘積和算法（SPA）是一種高效的疊代方法。它可以用來逼近由這些因子圖表示的邊際和配分函數。

對於任意的S-FG，Yedidia等人顯示，SPA在此S-FG上的不動點對應於Bethe自由能函數的駐點。基於此，他們定義了Bethe配分函數近似方法，由該近似方法得到的函數我們簡稱為Bethe配分函數。Bethe配分函數是基於Bethe自由能函數的最小值來定義的。基於S-FG的有限圖覆蓋，Vontobel提供了Bethe配分函數的組合刻畫。這一刻畫的證明在很大程度上依賴於類型方法。

在本論文中，我們利用有限圖覆蓋來分析S-FG和DE-FG的SPA和Bethe配分函數。本論文有兩個主要貢獻。

第一個主要貢獻涉及一類特殊的S-FG，其中每個S-FG的配分函數等於非負方塊矩陣的積和式。這類S-FG的Bethe配分函數稱為Bethe積和式。Bethe積和式的組合刻畫由度數為M的Bethe積和式給出。其中度數為M的Bethe積和式是基於底層的S-FG的度數為M的圖覆蓋來定義的。在本論文中，我們證明了一個基於度數為M的Bethe積和式的，對於非負方塊矩陣積和式的下界，解決了Vontobel在[IEEE Trans. Inf. Theory, Mar. 2013]中提出的一個推測。我們還證明了一個基於度數為M的Bethe積和式的，對於非負矩陣積和式的上界。在M趨於無窮大的極限下，這些上下界會等於已知的，基於Bethe積和式的，對於非負方塊矩陣積和式的上下界。類似的結果也適用於另一種對於積和式的近似。這種近似被稱為（縮放的）Sinkhorn積和式。

第二個主要貢獻是給出了基於DE-FG的有限圖覆蓋的，對於DE-FG的Bethe配分函數的組合刻畫。一般來說，逼近DE-FG的配分函數比S-FG的配分函數更具挑戰性，因為DE-FG的配分函數是複數的和，而不僅僅是非負實數的和。此外，我們不能直接應用類型方法來證明該組合刻畫，這跟S-FG的情況不一樣。我們通過對DE-FG應用適當的環計算變換（LCT）克服了這一困難。最初，LCT由Chertkov和Chernyak引入，作為基於SPA不動點信息的應用在S-FG上的特殊線性變換。Mori進一步發展了對於非二元字母表的S-FG的LCT，但他提出的LCT僅適用於SPA不動點的信息僅包含非零分量的情況。相比之下，DE-FG上的SPA通常涉及包含零值分量的信息。基于此，我們拓展了Mori提出的LCT。我們引入了一種同時適用於S-FG和DE-FG的LCT，並且我們提出的LCT也適用於SPA不動點信息中含有非零分量的情況。目前，我們提供了對於一類滿足（易於檢查的）條件的DE-FG的Bethe配分函數的組合刻畫。然而，基於數值仿真結果，我們推測這一組合刻畫對於廣泛的DE-FG都是成立的。我們還進一步引入了對稱子空間變換（SST），這變換可以將前述S-FG和DE-FG的組合刻畫重新表述為某些積分。除了用在證明本論文的主要結果，LCT和SST應該會有更廣泛的應用。

\clearpage\end{CJK*}

%% file: acknowledgement.tex
%!TEX root = main.tex
\noindent 
% I would like to thank 
This thesis represents the culmination of a journey that would not have been possible without the invaluable support and guidance of numerous individuals. I am deeply indebted to all those who have inspired, encouraged, and supported me throughout my graduate studies.

First and foremost, I express my sincere gratitude to my advisor, Professor Pascal O. Vontobel, for his unwavering guidance and support. His distinguished contributions to information theory, especially his work on finite graph covers, have not only inspired the research results presented in this thesis but have also motivated me to pursue high-quality research throughout my graduate studies. His dedication to my development as a scientist has been immeasurable. I am particularly grateful for his insightful feedback and guidance on the finite-graph-cover-based analysis presented in this paper. He has generously shared his expertise, engaged in countless insightful discussions, and made great efforts to improve my writing and presentation skills.

I am also profoundly grateful to Professor Navin Kashyap, Professor Jonathan Leake, Yannan Wang, and Haiwen Cao for their invaluable contributions to my research. Their insightful discussions, particularly in the areas of quantum information processing and theoretical computer science, have sparked new ideas and deepened my understanding of these areas. Special thanks to Professor Navin Kashyap for sharing his unpublished notes, which have improved some of the results presented in this thesis. The discussions with Yannan Wang and Haiwen Cao were particularly helpful in shaping my understanding of Bell inequalities and finite graph covers, while Professor Jonathan Leake's insights were crucial to my work on factor graphs based on homogeneous real stable polynomials.

During the final year of my doctoral studies, I have been fortunate to collaborate with Professor John C.S. Lui, Professor Xiaojun Lin, and Bin Luo. Their introduction to the exciting fields of quantum networks and quantum computing has greatly enriched my research experience. I am especially grateful to Professor John C.S. Lui for his mentorship and support since the beginning of 2024. His guidance on quantum networks was particularly valuable.

I would also like to extend my sincere appreciation to my thesis committee members--Professor Chandra Nair, Professor Henry Pfister, Professor Raymond W. Yeung, and Professor Ying-Jun Angela Zhang--for their time, insightful feedback, and encouragement. Their expertise and feedback have been invaluable in refining this thesis. Professor Henry Pfister offered particularly insightful feedback on the symmetric subspace transform, which is one of the main results in this thesis.

Lastly, I dedicate this thesis to my parents, Xuemei Xu and Hongzhong Li, whose unwavering love and support have been my constant source of strength and motivation. My deepest thanks also go to Xiangyue Tong for her encouragement.

The work presented in this thesis was partially supported by the Research
Grants Council of the Hong Kong Special Administrative Region, China (Project Nos. CUHK 14209317, CUHK 14207518, CUHK 14208319).

%CUHK 2151140, CUHK 3133146

%% file: listofpubs.tex
The work presented in this thesis is based on the following papers.
%----------------------------------------------------------------------------
\begin{enumerate}
	\item Y. Huang and P. O. Vontobel, ``Characterizing the Bethe partition function of double-edge factor graphs via graph covers,'' in Proc. IEEE Int. Symp. Inf. Theory, Los Angeles, CA, USA, Jun. 2020, pp. 1331–1336.

	\item Y. Huang and P. O. Vontobel, ``Characterizing the Bethe partition function of double-edge factor graphs via graph covers (extended version),'' \textbf{in preparation} for IEEE Trans. on Inf. Theory.

	\item Y. Huang and P. O. Vontobel, ``Bounding the permanent of a non-negative matrix via its degree-M Bethe and Sinkhorn permanents,'' in Proc. IEEE Int. Symp. Inf. Theory, Taipei, Taiwan, Jun. 2023, pp. 2774–2779.

	\item Y. Huang, N. Kashyap, and P. O. Vontobel, ``Degree-$M$ Bethe and Sinkhorn permanent based bounds on the permanent of a non-negative matrix,'' to appear in IEEE Trans. Inf. Theory, pp. 1-19, 2024.
\end{enumerate}
%----------------------------------------------------------------------------

%% file: notation.tex
%!TEX root = main.tex
\noindent\begin{xltabular}{\textwidth}{L|R}
  \multicolumn{2}{l}{\textbf{\Large Mathematics}}\\
  \hline
  iff &if and only if\\
  \hline
  $\sZ$ & the ring of integers \\
  \hline
  $\sZp$ & the set of non-negative integers \\
  \hline
  $\sZpp$ & the set of positive integers \\
  \hline
  $\sR$ & the field of real numbers \\
  \hline
  $\sRp$ & the set of non-negative real numbers \\
  \hline
  $\sRpp$ & the set of positive real numbers \\
  \hline
  $\sC$ & the field of complex numbers \\
  \hline
  $\imagunit$ & the imaginary unit \\
  \hline
  $[L]$ & the finite set $\{1, \ldots, L\}$ for any $L \in \sZpp$ \\
  \hline
  $[S]$ & equals one if the statement $S$ is true and equals zero otherwise, for any statement $S$\\
  \hline
  $\Pi_{\set{X}}$ & the set of probability mass functions over a finite set $\set{X}$, 
  \textit{i.e.},
  \begin{align*}
    \Pi_{\set{X}} \defeq \Biggl\{ p: \set{X} \to \sR \ \Biggl| \ p(x) \geq 0 \ 
    \text{for all} \,
    x \in \set{X},\, \sum\limits_{x \in \set{X}} p(x) = 1 \Biggr\}.
  \end{align*}
  \\ \hline 
  $\set{S}_{\set{A}}$ & the set of all
  permutations over an arbitrary finite set $\set{A}$, \textit{i.e.}, 
  the set of all bijective mappings from $\set{A}$ to
  $\set{A}$
  \\ \hline  
  $\log(0)$ & $ -\infty $ 
  \\ \hline 
  $ 0 \cdot \log(0) $ & $0$
  \\ \hline
  $\matr{I}$ & the identity matrix 
  \\ \hline
  $\matr{A}$ & the matrix $ \matr{A} = \bigl( A(x, y) \bigr)_{ x \in \set{X}, y \in \set{Y} } $ associated with a function $ A(x, y) \in \sC $, where $x \in \set{X}$ is the row index of $\matr{A}$ and $y \in \set{Y}$ is the column index of $\matr{A}$ with finite sets 
  $ \set{X} $ and $ \set{Y} $
  \\ \hline
  $\matr{A}^{\tran}$ & the transpose of a complex matrix $ \matr{A} $
  \\ \hline
  $\overline{\matr{A}}$ & the componentwise complex conjugate of a complex matrix $ \matr{A} $
  \\ \hline
  $\matr{A}^{\Herm}$ & the Hermitian transpose of a complex matrix $ \matr{A} $, \ie, 
  $\matr{A}^{\Herm} = \overline{\matr{A}}^{\tran}$
  \\ \hline
  $\matr{A}(:, y)$ & the column $y$ of $\matr{A}$ 
  \\ \hline 
  $\matr{A}(x, :)$ & the row $x$ of $\matr{A}$
  \\ \hline 
  
\end{xltabular}

%% file: abbreviations.tex
%!TEX root = main.tex
\noindent
\begin{xltabular}{\textwidth}{L|R}
  % \multicolumn{2}{l}{\textbf{\Large Mathematics}}
  % \\
  \hline
  DE-FG & double-edge factor graph
  \\ \hline 
  DE-NFG & double-edge normal factor graph
  \\ \hline 
  LBP & loopy belief propagation 
  \\ \hline
  LCT & loop-calculus transform
  \\ \hline 
  NFG & normal factor graph
  \\ \hline 
  PSD & positive semi-definite
  \\ \hline
  S-FG & standard factor graph
  \\ \hline
  S-NFG & standard normal factor graphs
  \\ \hline
  SPA & sum-product algorithm 
  \\ \hline
  SST & symmetric-subspace transform
  \\ \hline 
\end{xltabular}

%% file: figures/relate_z_zbm_zb.tex
\begin{tikzpicture}
\pgfmathsetmacro{\Ws}{1.3};
% Expression on the LHS and RHS
%----------------------------------------------------------------------------
\begin{pgfonlayer}{main}
  \node (ZB1)     at (0,0) [] {$\qquad\qquad\qquad\quad\left. \ZBn{M}(\vtheta) \right|_{M = 1} = \perm(\vtheta)$};
  \node (ZBM)     at (0,\Ws) [] {$\ZBn{M}(\vtheta)$};
  \node (ZBinfty) at (0,2*\Ws) {$\qquad\qquad\qquad\qquad
  \left. \ZBn{M}(\vtheta) \right|_{M \to \infty} 
  = \permb(\vtheta)$};
%----------------------------------------------------------------------------
  \draw[]
    (ZB1) -- (ZBM) -- (ZBinfty);
\end{pgfonlayer}
%\vspace{0.3cm}
\end{tikzpicture}

%% file: figures/relate_z_permscsm_permscs.tex
\begin{tikzpicture}
\pgfmathsetmacro{\Ws}{1.3};
% Expression on the LHS and RHS
%----------------------------------------------------------------------------
\begin{pgfonlayer}{main}
  \node (ZscS1)     at (0,0) [] {$\qquad\qquad\qquad\quad\left. 
                               \permscsM{M}(\vtheta) \right|_{M = 1} 
                                 = \perm(\vtheta)$};
  \node (ZscSM)     at (0,\Ws) [] {$\permscsM{M}(\vtheta)$};
  \node (ZscSinfty) at (0,2*\Ws) {$\qquad\qquad\qquad\qquad
  \left. \permscsM{M}(\vtheta) \right|_{M \to \infty} 
  = \permscs(\vtheta)$};
%----------------------------------------------------------------------------
  \draw[]
    (ZscS1) -- (ZscSM) -- (ZscSinfty);
\end{pgfonlayer}
\end{tikzpicture}

%% file: figures/snfg_example.tex
\begin{figure}[t]
  \centering
  \captionsetup{font=scriptsize}
  \subfloat{
  %----------------------------
    \begin{minipage}[t]{0.4\linewidth}
      \centering
      \begin{tikzpicture}[node distance=2cm, remember picture]
        \input{figures/head_files_figs.tex}
        \node[state] (f1) at (0,0) [label=above: $f_{1}$] {};
        \node[state,right of=f1] (f2) [label=above: $f_{2}$] {};
        \node[state,below of=f1] (f3) [label=below: $f_{3}$] {};
        \node[state,right of=f3] (f4) [label=below: $f_{4}$] {};
        %----------------------------------------------------------------------------
        \begin{pgfonlayer}{background}
          \draw[-,draw]
              (f1) edge node[above]  {$x_{1}$} (f2)
              (f1) edge node[left]  {$x_{2}$} (f3)
              (f1) edge node[left]  {$x_{3}$} (f4) % (q0)  edge [bend left= 10]  node[el,above,pos=0.8] {$a=2$}  (q2)
              (f2) edge node[right] {$x_{4}$} (f4)
              (f3) edge node[below] {$x_{5}$} (f4);
        \end{pgfonlayer}
        %----------------------------------------------------------------------------
      \end{tikzpicture}
    \end{minipage}%
  }
  \caption{NFG $\mathsf{N}$ in Example \ref{sec:SNFG:exp:1}.}
  \label{sec:SNFG:fig:1}
\end{figure}
  %----------------------------
  % }%
  % \subfloat[Configuration $\bm{x}$ on the NFG $\mathsf{N}$.]{
  % %---------------------------
  % \begin{minipage}[t]{0.5\linewidth}\label{fig:exem2}
  % \centering
  % \begin{tikzpicture}[node distance=1.2cm, on grid,auto]
  %   \tikzstyle{state}=[shape=rectangle, thick, draw, minimum width=0.4cm, minimum height = 0.4cm,outer sep=0pt]
  %   \node[state] (f1) at (0,0) [label=above: $f_{1}$] {\tiny$+$};
  %   \node[state,right of=f1] (f2) [label=above: $f_{2}$] {\tiny$+$};
  %   \node[state,below of=f1] (f3) [label=below: $f_{3}$] {\tiny$+$};
  %   \node[state,right of=f3] (f4) [label=below: $f_{4}$] {\tiny$+$};

  %   \path[-,draw,line width=1.8pt]
  %       (f1) edge[red] node[left, black]  {$x_{2}$} (f3)
  %       (f1) edge[red] node[left, black]  {$x_{3}$} (f4)
  %       (f3) edge[red] node[below, black]  {$x_{5}$} (f4); % (q0)  edge [bend left= 10]  node[el,above,pos=0.8] {$a=2$}  (q2)P
  %   \path[-,draw,thick]       
  %       (f1) edge node[above] {$x_{1}$} (f2)
  %       (f2) edge node[right] {$x_{4}$} (f4);
  %   \node[state] (f5) at (0,0) [label=above: $f_{1}$] {\tiny$+$};
  %   \node[state,right of=f1] (f6) [label=above: $f_{2}$] {\tiny$+$};
  %   \node[state,below of=f1] (f7) [label=below: $f_{3}$] {\tiny$+$};
  %   \node[state,right of=f3] (f8) [label=below: $f_{4}$] {\tiny$+$};
  % \end{tikzpicture}
  % \end{minipage}%
  % %---------------------------
  % }%
  %\caption{NFG in Example \ref{example1}.}

%% file: figures/head_files_figs.tex
\tikzstyle{state}=[shape=rectangle,fill=white, draw, minimum width=0.2cm, minimum height = 0.2cm,outer sep=-0.3pt]
% the smallest dashed boxes for OTB
\tikzstyle{state_dash1}=[shape=rectangle, draw, dashed, minimum width=1.8*\sdis cm, minimum height = \ldis cm,outer sep=-0.3pt,fill=black!10]
% the large dashed boxes for OTB
\tikzstyle{state_dash2}=[shape=rectangle, draw, dashed, minimum width=2*\sdis cm, minimum height = 2*\sdis cm,outer sep=-0.3pt,fill=black!10]
% the largest dashed boxes for CTB
% the largest dashed boxes for CTB in S-NFG
\tikzstyle{state_dash3}=[shape=rectangle, draw, dashed, minimum width= 3*\sdis cm, minimum height = 3.3*\sdis cm,outer sep=-0.3pt,fill=black!10]

% the largest dashed boxes for CTB in PEs-NFG
\tikzstyle{state_dash4}=[shape=rectangle, draw, dashed, minimum width= 2.5*\sdis cm, minimum height = 3.4*\sdis cm,outer sep=-0.3pt,fill=black!10]

% the dashed boxes for OTB of local functions in PEs-NFG
\tikzstyle{state_dash5}=[shape=rectangle, draw, dashed, minimum width= 2.4*\sdis cm, minimum height = 4*\ldis cm,outer sep=-0.3pt,fill=black!10]

% the dashed boxes for SST in OTB of local functions in S-NFG

\tikzstyle{state_dash6}=[shape=rectangle, draw, dashed, minimum width= 2.9*\sdis cm, minimum height = 9*\sdis cm,outer sep=-0.3pt,fill=black!10]

\tikzstyle{state_dash7}=[shape=rectangle, draw, dashed, minimum width= 1.7*\sdis cm, minimum height = 1.7*\ldis cm,outer sep=-0.3pt,fill=black!10]

% for illustrating the quantum information processing system
\tikzstyle{state_long}=[shape=rectangle,draw ,  fill =white, minimum width=0.7cm, minimum height = 0.8cm,outer sep=0pt]

\tikzset{dot/.style={circle,fill=#1,inner sep=0,minimum size=3pt}}

\tikzstyle{state_large}=[shape=rectangle,draw, fill =white, minimum width= 1*\ldis cm, minimum height = 2.4*\ldis cm,outer sep=0pt]

% for graph cover
\tikzstyle{state_large2}=[shape=rectangle,draw, fill =white, minimum width= 0.6*\ldis cm, minimum height = 0.6*\ldis cm,outer sep=0pt]
% for illustrat

%% file: figures/example_snfg_permanent_3by3.tex
\begin{tikzpicture}[node distance=2.2cm, on grid,auto,scale=1.50]
    \tikzstyle{state}=[shape=rectangle,fill=white,draw,minimum size=0.3cm]
    \pgfmathsetmacro{\ldis}{2} 
    \pgfmathsetmacro{\sdis}{0.5} 
    \pgfmathsetmacro{\scale}{0.8} 
    \begin{pgfonlayer}{glass}
        \foreach \i in {1,2,3}{
            \node[state] (f\i) at (0,-\i*\scale) [label=above: \scriptsize$\fr{\i}$] {};
            \node[state] (g\i) at (\ldis,-\i*\scale) [label=above: \scriptsize$\fc{\i}$] {};
        }
        \foreach \i in {1,2,3}{
            \foreach \j in {1,2,3}{
                \begin{pgfonlayer}{background}
                     \draw[]
                        (f\i) -- (g\j) ;
                \end{pgfonlayer}
            }
        }
    \end{pgfonlayer}
\end{tikzpicture}
\vspace{0.3cm}

%% file: figures/denfg_example.tex
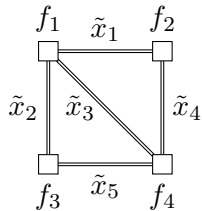
\begin{figure}[t]
  \centering
  \captionsetup{font=scriptsize}
  %----------------------------
\begin{minipage}[t]{0.45\linewidth}
  \centering
  \captionsetup{font=scriptsize}
  \begin{tikzpicture}[node distance=2cm, on grid,auto]
    \input{figures/length.tex}
    \input{figures/head_files_figs.tex}
    \node[state] (f1) at (0,0) [label=above: $f_1$] {};
    \node[state] (f2) at (1.5,0) [label=above: $f_2$] {};
    \node[state] (f3) at (0,-1.5) [label=below: $f_3$] {};
    \node[state] (f4) at (1.5,-1.5) [label=below: $f_4$] {};
    %----------------------------------------------------------------------------
    \begin{pgfonlayer}{background}
      \draw[double]
          (f1) -- node[above] {$ \tilde{x}_{1} $} (f2)
          (f1) -- node[left]  {$ \tx_{2} $} (f3) % (q0)  edge [bend left= 10]  node[el,above,pos=0.8] {$a=2$}  (q2)
          (f1) -- node[left, xshift=-0.0cm]  {$ \tx_{3} $} (f4)
          (f2) -- node[right]  {$ \tx_{4} $} (f4)
          (f3) -- node[below] {$ \tx_{5} $} (f4);
    \end{pgfonlayer}
    %----------------------------------------------------------------------------
  \end{tikzpicture}
 \end{minipage}%
  \caption{DE-NFG in Example~\ref{sec:DENFG:exp:1}.\label{sec:DENFG:fig:4}}
\end{figure}

%% file: figures/sst_examples/snfg/length.tex
% the distance between h and x
\pgfmathsetmacro{\ldis}{1} 
\pgfmathsetmacro{\sdis}{0.5} 

%% file: figures/penfg_example_qip.tex
%------------------------------------
\begin{figure}[t]
  \centering
  \captionsetup{font=scriptsize}
  \subfloat{
    \begin{minipage}[t]{0.5\textwidth}
    \centering
      \begin{tikzpicture}[node distance=0.6cm, on grid,auto]
        \input{figures/head_files_figs.tex}
        \node[state] (rho) at (0,0) [] {$\rho$};
        \node[state_long] (U) [above right =1cm and 2cm of rho] {$U$};
        \node[state_long] (UH) [below right =1cm and 2cm of rho] {$U^{H}$};
        \node[state_long] (B) [right = 2.4cm of U] {$B$};
        \node[state_long] (BH) [right = 2.4cm of UH] {$B^{H}$};
        \node[state] (eqn) [below right=1cm and 2cm of B] {$I$};
        \draw
            (rho) |-node[above right, pos=.6] {$x_{1}$} (U)--node[above] {$x_{2}$} (B)
            -| node[above left, pos=.4] {$x_{3}$} (eqn)|-node[above left, pos=.6] 
            {$x_{3}'$}(BH)
            --node[above] {$x_{2}'$} (UH)-| node[above right, pos=.4] {$x_{1}'$} (rho);
         \node [dot=black] at (rho.north) {};
         \node [dot=black] at (U.east) {};
         \node [dot=black] at (UH.west) {};
         \node [dot=black] at (B.east) {};
         \node [dot=black] at (BH.west) {};
         \node [dot=black] at (eqn.south) {};
      \end{tikzpicture}
    \end{minipage}
  }
  \subfloat{
    \begin{minipage}[t]{0.5\textwidth}
      \centering
      \begin{tikzpicture}[node distance=0.8cm, on grid,auto]
        \input{figures/head_files_figs.tex}
        \input{figures/length.tex}
        %----------------------------------------------------------------------------
        \begin{pgfonlayer}{main}
          \node[state] (rho) at (0,0) [] {$\rho$};
          \node[state_long] (U) [right =2cm of rho] {$\tilde{U}$};
          \node[state_long] (B) [right = 2.4cm of U] {$\tilde{B}$};
          \node[state] (eqn) [right= 2cm of B] {$I$};
        \end{pgfonlayer}
        %----------------------------------------------------------------------------
        \begin{pgfonlayer}{background}
          % \draw[]
          %   (0,\lvaredge) -- node[above] {$\tx_{1}$} (2,\lvaredge)
          %   -- node[above] {$\tx_{2}$} (4.4,\lvaredge)
          %   -- node[above] {$\tx_{3}$} (6.4,\lvaredge)
          %   (0,-\lvaredge) --  (2,-\lvaredge)
          %   -- (4.4,-\lvaredge)
          %   -- (6.4,-\lvaredge);
          \draw[double]
            (rho) -- node[above] {$\tx_{1}$} (U)
            -- node[above] {$\tx_{2}$} (B)
            -- node[above] {$\tx_{3}$} (eqn);
        \end{pgfonlayer}
        %----------------------------------------------------------------------------
        \node [dot=black] at (2.35,0) {};
        \node [dot=black] at (4.75,0) {};
        \node[] (f3) at (0,0) [] {};
        \node (U) [above right =1.3cm and 2cm of rho] {};
        \node (UH) [below right =1.3cm and 2cm of rho] {};
      \end{tikzpicture}
    \end{minipage}
  }
  \vspace{0.2cm}
  \caption{Left: NFG in Example \ref{sec:DENFG:exp:2}. Right: DE-NFG in Example \ref{sec:DENFG:exp:2}.}\label{sec:DENFG:fig:2}
\end{figure}
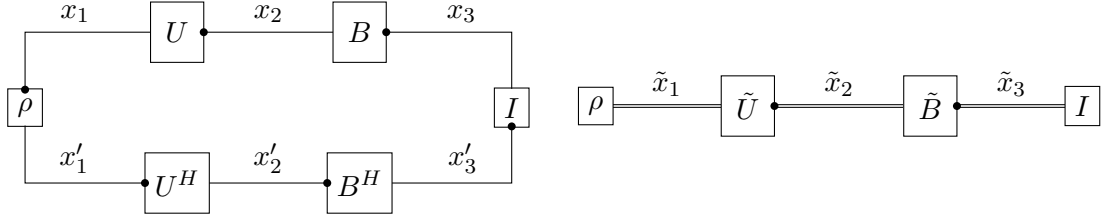
%------------------------------------

%% file: figures/graph_cover_example.tex
  % \centering
  % \captionsetup{font=scriptsize}
% \subfloat{
  \begin{minipage}[t]{0.35\linewidth}
  \centering
    \begin{tikzpicture}[node distance=2cm, on grid,auto]
      \input{figures/head_files_figs.tex}
      \node[state] (f11) [] {};
      \node[state,right of=f11] (f21) [] {};
      \node[state,below of=f11] (f31) [] {};
      \node[state,right of=f31] (f41) [] {};
      \path[-,draw]
          (f11) edge node[above]  {} (f21)
          (f11) edge node[left]  {} (f31) % (q0)  edge [bend left= 10]  node[el,above,pos=0.8] {$a=2$}  (q2)
          (f11) edge node[left] {} (f41)
          (f21) edge node[right] {} (f41)
          (f31) edge node[below] {} (f41);
    \end{tikzpicture}
  \end{minipage}
  %-----------------------
  \begin{minipage}[t]{0.2\linewidth}
    \begin{tikzpicture}[node distance=2cm, on grid,auto]
      \input{figures/head_files_figs.tex}

      \begin{pgfonlayer}{main}
        \node[state] (f11) [] {};
        \node[state,right of=f11] (f21) [] {};
        \node[state,below of=f11] (f31) [] {};
        \node[state,right of=f31] (f41) [] {};
        \path[-,draw]
          (f11) edge node[above]  {} (f21)
          (f11) edge node[left]  {} (f31) % (q0)  edge [bend left= 10]  node[el,above,pos=0.8] {$a=2$}  (q2)
          (f11) edge node[left] {} (f41)
          (f21) edge node[right] {} (f41)
          (f31) edge node[below] {} (f41);
      \end{pgfonlayer}

      \begin{pgfonlayer}{behind}
        \node[state] (f12) [above right=0.2cm and 0.2cm of f11] {};
        \node[state,right of=f12] (f22) [] {};
        \node[state,below of=f12] (f32) [] {};
        \node[state,right of=f32] (f42) [] {};
        \path[-,draw]
          (f12) edge node[above]  {} (f22)
          (f12) edge node[left]  {} (f32) % (q0)  edge [bend left= 10]  node[el,above,pos=0.8] {$a=2$}  (q2)
          (f12) edge node[left] {} (f42)
          (f22) edge node[right] {} (f42)
          (f32) edge node[below] {} (f42);
      \end{pgfonlayer}

    \end{tikzpicture}
  \end{minipage}
  %-----------------------------
  \begin{minipage}[t]{0.2\linewidth}
    \begin{tikzpicture}[node distance=2cm, on grid,auto]
      \input{figures/head_files_figs.tex}
      \begin{pgfonlayer}{main}
        \node[state] (f11) [] {};
        \node[state,right of=f11] (f21) [] {};
        \node[state,below of=f11] (f31) [] {};
        \node[state,right of=f31] (f41) [] {};
        \path[-,draw]
          (f11) edge node[above]  {} (f21)
          (f11) edge node[left]  {} (f31) % (q0)  edge [bend left= 10]  node[el,above,pos=0.8] {$a=2$}  (q2)
          (f11) edge node[left] {} (f42)
          (f21) edge node[right] {} (f41)
          (f31) edge node[below] {} (f41);
      \end{pgfonlayer}
      \begin{pgfonlayer}{behind}
        \node[state] (f12) [above right=0.2cm and 0.2cm of f11] {};
        \node[state,right of=f12] (f22) [] {};
        \node[state,below of=f12] (f32) [] {};
        \node[state,right of=f32] (f42) [] {};
          \path[-,draw]
            (f12) edge node[above]  {} (f22)
            (f12) edge node[left]  {} (f32) % (q0)  edge [bend left= 10]  node[el,above,pos=0.8] {$a=2$}  (q2)
            (f12) edge node[left] {} (f41)
            (f22) edge node[right] {} (f42)
            (f32) edge node[below] {} (f42);
      \end{pgfonlayer}
    \end{tikzpicture}
  \end{minipage}
  %--------------------------------
  \begin{minipage}[t]{0.2\linewidth}
    \begin{tikzpicture}[node distance=2cm, on grid,auto]
      \input{figures/head_files_figs.tex}
      \begin{pgfonlayer}{main}
        \node[state] (f11) [] {};
        \node[state,right of=f11] (f21) [] {};
        \node[state,below of=f11] (f31) [] {};
        \node[state,right of=f31] (f41) [] {};
        \path[-,draw]
            (f11) edge node[above]  {} (f22)
            (f11) edge node[left]  {} (f31) % (q0)  edge [bend left= 10]  node[el,above,pos=0.8] {$a=2$}  (q2)
            (f11) edge node[left] {} (f41)
            (f21) edge node[right] {} (f41);
      \end{pgfonlayer}
      \begin{pgfonlayer}{behind}
        \node[state] (f12) [above right=0.2cm and 0.2cm of f11] {};
        \node[state,right of=f12] (f22) [] {};
        \node[state,below of=f12] (f32) [] {};
        \node[state,right of=f32] (f42) [] {};
        \path[-,draw]
          (f31) edge node[below] {} (f42)
          (f12) edge node[above]  {} (f21)
          (f12) edge node[left]  {} (f32) % (q0)  edge [bend left= 10]  node[el,above,pos=0.8] {$a=2$}  (q2)
          (f12) edge node[left] {} (f42)
          (f22) edge node[right] {} (f42)
          (f32) edge node[below] {} (f41);
      \end{pgfonlayer}
    \end{tikzpicture}
  \end{minipage}
% }
% \caption{Left: NFG $\mathsf{N}$. Right: samples of possible 2-covers
%   $\mathsf{\hat N}$ of $\mathsf{N}$.\label{sec:GraCov:fig:9}}
% \end{figure}

%% file: figures/graph_cover_example_denfg.tex
  \begin{minipage}[t]{0.35\linewidth}
  \centering
  \begin{tikzpicture}[node distance=2cm, on grid,auto]
    \input{figures/head_files_figs.tex}
    \begin{pgfonlayer}{glass}
    \node[state] (f11) [] {};
    \node[state,right of=f11] (f21) [] {};
    \node[state,below of=f11] (f31) [] {};
    \node[state,right of=f31] (f41) [] {};
    \end{pgfonlayer}
    %----------------------------------------------------------------------------
    \begin{pgfonlayer}{main}
      \draw[double]
        (f11) -- (f21)
        (f11) -- (f31) % (q0)  -- [bend left= 10]  node[el,above,pos=0.8] {$a=2$}  (q2)
        (f11) -- (f41)
        (f21) -- (f41)
        (f31) -- (f41);
    \end{pgfonlayer}
    %----------------------------------------------------------------------------
  \end{tikzpicture}
  \end{minipage}
  %-----------------------
  \begin{minipage}[t]{0.2\linewidth}
    \begin{tikzpicture}[node distance=2cm, on grid,auto]
      \input{figures/head_files_figs.tex}
      %----------------------------------------------------------------------------
      \begin{pgfonlayer}{glass}
        \node[state] (f11) [] {};
        \node[state,right of=f11] (f21) [] {};
        \node[state,below of=f11] (f31) [] {};
        \node[state,right of=f31] (f41) [] {};
      \end{pgfonlayer}

      \begin{pgfonlayer}{above}
        \node[state] (f12) [above right=0.2cm and 0.2cm of f11] {};
        \node[state,right of=f12] (f22) [] {};
        \node[state,below of=f12] (f32) [] {};
        \node[state,right of=f32] (f42) [] {};
      \end{pgfonlayer}
      %----------------------------------------------------------------------------
      \begin{pgfonlayer}{background}
      \draw[double]
        (f12) -- (f22)
        (f12) -- (f32) % (q0)  -- [bend left= 10]  node[el,above,pos=0.8] {$a=2$}  (q2)
        (f12) -- (f42)
        (f22) -- (f42)
        (f32) -- (f42);
      \end{pgfonlayer}
      %----------------------------------------------------------------------------
      \begin{pgfonlayer}{behind}
      \draw[double]
        (f11) -- (f21)
        (f11) -- (f31) % (q0)  -- [bend left= 10]  node[el,above,pos=0.8] {$a=2$}  (q2)
        (f11) -- (f41)
        (f21) -- (f41)
        (f31) -- (f41);
      \end{pgfonlayer}  
    \end{tikzpicture}
  \end{minipage}
  %-----------------------------
  \begin{minipage}[t]{0.2\linewidth}
    \begin{tikzpicture}[node distance=2cm, on grid,auto]
      \input{figures/head_files_figs.tex}
       %----------------------------------------------------------------------------
      \begin{pgfonlayer}{glass}
        \node[state] (f11) [] {};
        \node[state,right of=f11] (f21) [] {};
        \node[state,below of=f11] (f31) [] {};
        \node[state,right of=f31] (f41) [] {};
      \end{pgfonlayer}

      \begin{pgfonlayer}{above}
        \node[state] (f12) [above right=0.2cm and 0.2cm of f11] {};
        \node[state,right of=f12] (f22) [] {};
        \node[state,below of=f12] (f32) [] {};
        \node[state,right of=f32] (f42) [] {};
      \end{pgfonlayer}
      %----------------------------------------------------------------------------
      \begin{pgfonlayer}{background}
        \draw[double]
            (f12) -- (f22)
            (f12) -- (f32) % (q0)  -- [bend left= 10]  node[el,above,pos=0.8] {$a=2$}  (q2)
            (f22) -- (f42)
            (f32) -- (f42);
      \end{pgfonlayer}
      %----------------------------------------------------------------------------
      \begin{pgfonlayer}{behind}
        \draw[double]
          (f12) -- (f41);
      \end{pgfonlayer}
      %----------------------------------------------------------------------------
      \begin{pgfonlayer}{main}
        \draw[double]
          (f11) -- (f21)
          (f11) -- (f31) % (q0)  -- [bend left= 10]  node[el,above,pos=0.8] {$a=2$}  (q2)
          (f11) -- (f42)
          (f31) -- (f41);
      \end{pgfonlayer}
      %----------------------------------------------------------------------------
      \begin{pgfonlayer}{above}
        \draw[double]
          (f21) -- (f41);
      \end{pgfonlayer}
    \end{tikzpicture}
  \end{minipage}
  %--------------------------------
  \begin{minipage}[t]{0.2\linewidth}
    \begin{tikzpicture}[node distance=2cm, on grid,auto]
      \input{figures/head_files_figs.tex}
      %----------------------------------------------------------------------------
       %----------------------------------------------------------------------------
      \begin{pgfonlayer}{glass}
        \node[state] (f11) [] {};
        \node[state,right of=f11] (f21) [] {};
        \node[state,below of=f11] (f31) [] {};
        \node[state,right of=f31] (f41) [] {};
      \end{pgfonlayer}

      \begin{pgfonlayer}{above}
        \node[state] (f12) [above right=0.2cm and 0.2cm of f11] {};
        \node[state,right of=f12] (f22) [] {};
        \node[state,below of=f12] (f32) [] {};
        \node[state,right of=f32] (f42) [] {};
      \end{pgfonlayer}
      %----------------------------------------------------------------------------

      \begin{pgfonlayer}{behind}
        \draw[double]
          (f11) -- (f22)
          (f11) -- (f31) % (q0)  -- [bend left= 10]  node[el,above,pos=0.8] {$a=2$}  (q2)
          
          (f31) -- (f42);
      \end{pgfonlayer}

      \begin{pgfonlayer}{background}
        \node[state] (f12) [above right=0.2cm and 0.2cm of f11] {};
        \node[state,right of=f12] (f22) [] {};
        \node[state,below of=f12] (f32) [] {};
        \node[state,right of=f32] (f42) [] {};
        \draw[double]
          (f12) -- (f21)
          (f12) -- (f42)
          (f22) -- (f42)
          (f32) -- (f41)
          (f12) -- (f32);
      \end{pgfonlayer}
      
      \begin{pgfonlayer}{main}
        \draw[double]
          (f11) -- (f41)
          (f21) -- (f41);
      \end{pgfonlayer}

    \end{tikzpicture}
  \end{minipage}
% }
% \caption{
%     Left: DE-NFG $\mathsf{\tilde N}$. Right: samples of possible 2-covers
%     $\mathsf{\hat N}$ of $\mathsf{N}$.
%   }\label{fig:Exam_Graph_cover_DE_NFG}
% \end{figure}

%% file: figures/graph_cov_thm/numerical_result_main.tex
\begin{figure}
  \centering
  \captionsetup{font=scriptsize}
  \subfloat[\label{sec:GraCov:fig:1}]{
    \begin{minipage}[t]{0.45\textwidth}
      \centering
      %----------------------
      \begin{tikzpicture}[node distance=2cm, on grid,auto]
        \input{figures/head_files_figs.tex}
        %----------------------------------------------------------------------------
        \begin{pgfonlayer}{main}
          \node[state] (f1) at (0,0) [label=above: $f_1$] {};
          \node[state,right of=f1] (f2) [label=above: $f_2$] {};
          \node[state,below of=f1] (f3) [label=below: $f_3$] {};
          \node[state,right of=f3] (f4) [label=below: $f_4$] {};
        \end{pgfonlayer}
        %----------------------------------------------------------------------------
        \begin{pgfonlayer}{background}
          \draw[double]
            (f1) -- node[above] {$\tilde{x}_{1}$} (f2)
            (f1) -- node[left]  {$\tilde{x}_{2}$} (f3) % (q0)  edge [bend left= 10]  node[el,above,pos=0.8] {$a=2$}  (q2)
            (f1) -- node[above right=.03cm of f1,xshift=.4cm] { $\tilde{x}_{3}$} (f4)
            (f2) -- node[right] {$\tilde{x}_{4}$} (f4)
            (f3) -- node[below] {$\tilde{x}_{5}$} (f4);
        \end{pgfonlayer}
        %----------------------------------------------------------------------------
      \end{tikzpicture}
    \end{minipage}
  }
  \subfloat[\label{sec:GraCov:fig:5}]{
    \begin{minipage}[t]{0.45\textwidth}
      \centering
      \begin{tikzpicture}[node distance=2cm, on grid,auto]
        \input{figures/head_files_figs.tex}
        %----------------------------------------------------------------------------
        \begin{pgfonlayer}{main}
          \node[state] (f1) at (1,1) [label=above: $f_1$] {};
          \node[state,right of=f1] (f2) [label=above:  $f_2$] {};
          \node[state,below of=f1] (f3) [label=below:  $f_3$] {};
          \node[state,right of=f3] (f4) [label=below:  $f_4$] {};
        \end{pgfonlayer}
        %----------------------------------------------------------------------------
        \begin{pgfonlayer}{background}
          \draw[double]
            (f1) -- node[above] { $\tilde{x}_{1}$} (f2)
            (f1) -- node[left]  { $\tilde{x}_{2}$} (f3) % (q0)  edge [bend left= 10]  node[el,above,pos=0.8] {$a=2$}  (q2)
            (f1) -- node[above left= 0.5cm, xshift=.4cm] {$\tilde{x}_{3}$} (f4)
            (f2) -- node[right] {$\tilde{x}_{4}$} (f4)
            (f3) -- node[below right= 0.4cm, xshift=-.8cm] {$\tilde{x}_{5}$} (f2)
            (f3) -- node[below] {$\tilde{x}_{6}$} (f4);
        \end{pgfonlayer}
        %----------------------------------------------------------------------------
        \begin{pgfonlayer}{behind}
          \draw[double]
            (f3) -- (f2);
        \end{pgfonlayer}
        %----------------------------------------------------------------------------
      \end{tikzpicture}
    \end{minipage}\label{sec:GraCov:fig:2}
  }\\
  % \subfloat[\label{sec:GraCov:fig:2}]{
  %   \begin{minipage}[t]{0.45\textwidth}
  %       \centering
  %       \includegraphics[scale=0.55]{figures/graph_cov_thm/scatter_twoedges_m4.eps}
  %       \vspace{0.1 cm}
  %   \end{minipage}%
  % }%
  % \subfloat[\label{sec:GraCov:fig:6}]{
  %   \begin{minipage}[t]{0.45\textwidth}
  %     \centering
  %     \includegraphics[scale=0.55]{figures/graph_cov_thm/scatter_sincyc_m3.eps}
  %     \vspace{0.1 cm}
  %   \end{minipage}%
  % }%
  
  \subfloat[\label{sec:GraCov:fig:3}\vspace*{-0.5cm}]{
    \begin{minipage}[t]{0.5\textwidth}
      \centering
      \includegraphics[scale=0.37]{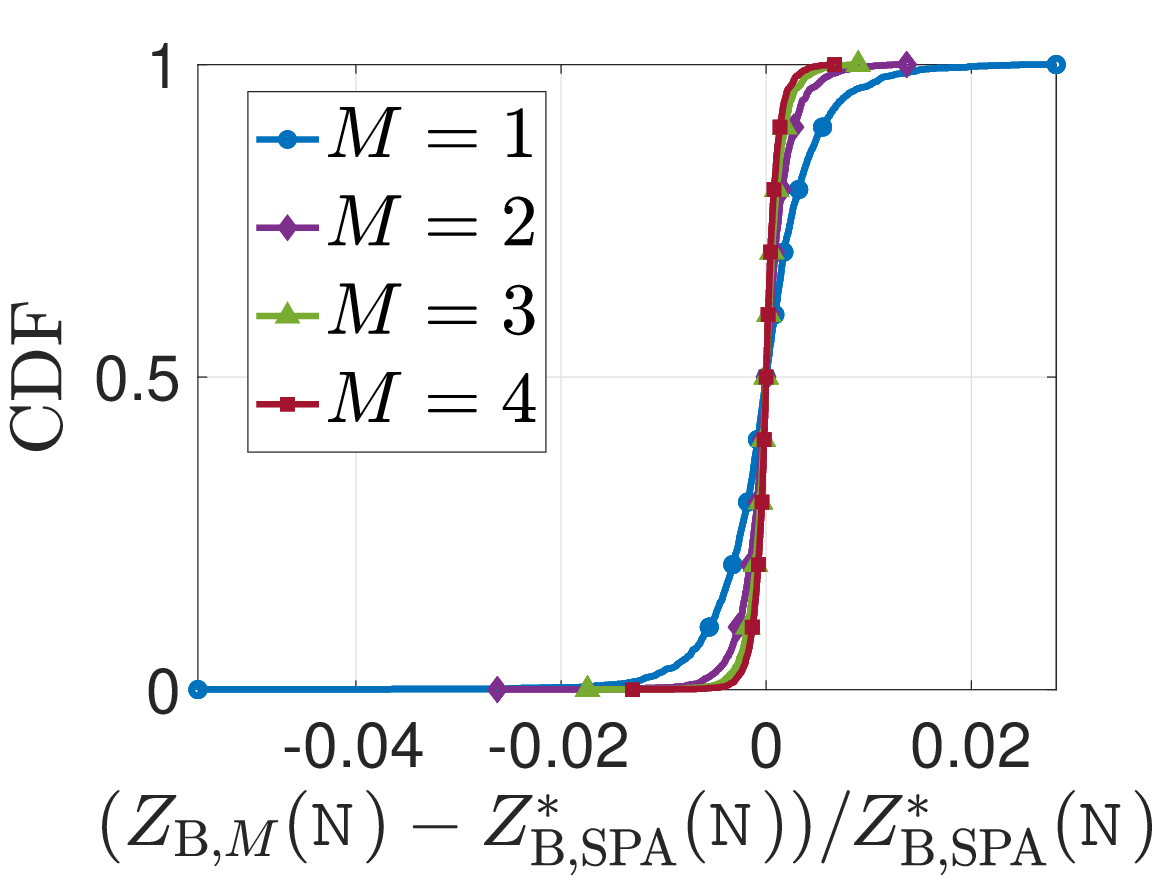}
      % \vspace{0.1 cm}
    \end{minipage}%
  }%
   \subfloat[\label{sec:GraCov:fig:7}\vspace*{-0.5cm}]{
    \begin{minipage}[t]{0.5\textwidth}
      \centering
      \includegraphics[scale=0.37]{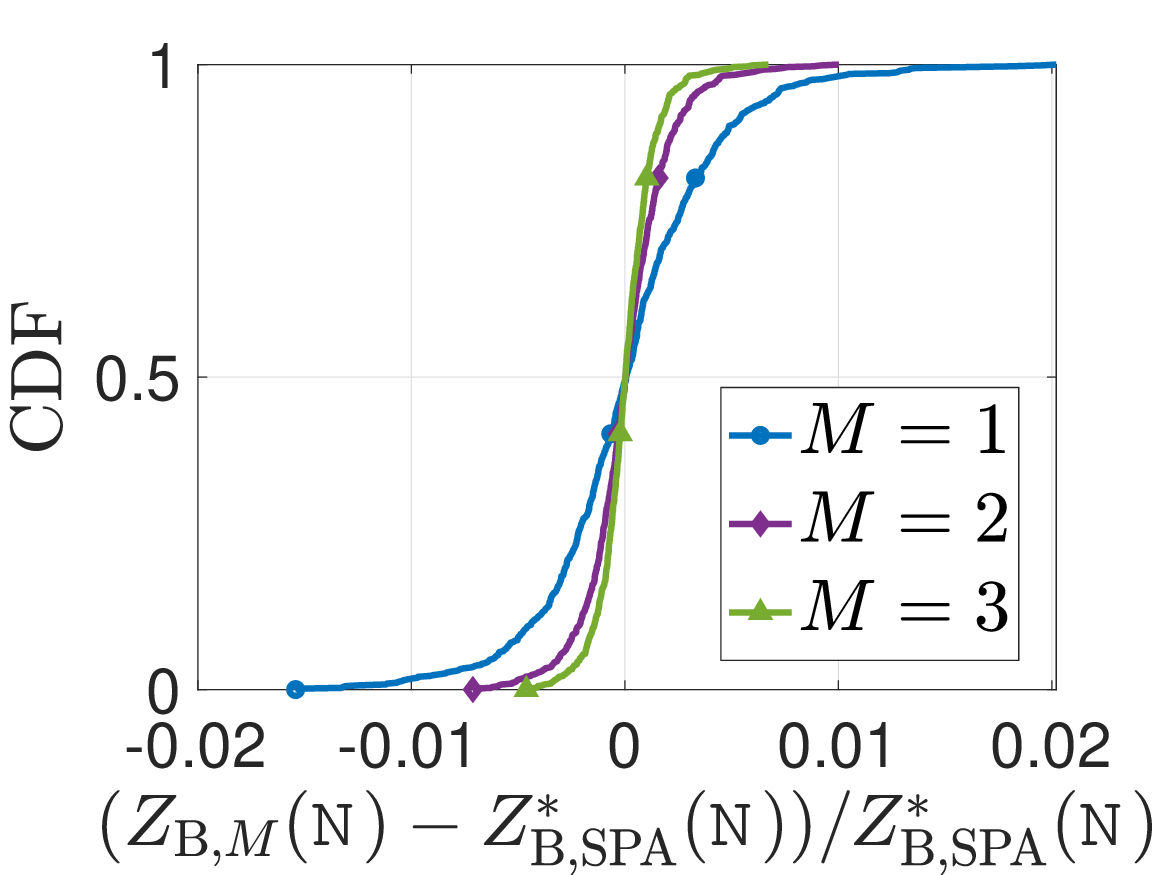}
      % \vspace{0.1 cm}
    \end{minipage}%
  }%

  \subfloat[\label{sec:GraCov:fig:4}]{
    \begin{minipage}[t]{0.5\textwidth}
      \centering
      \includegraphics[scale=0.37]{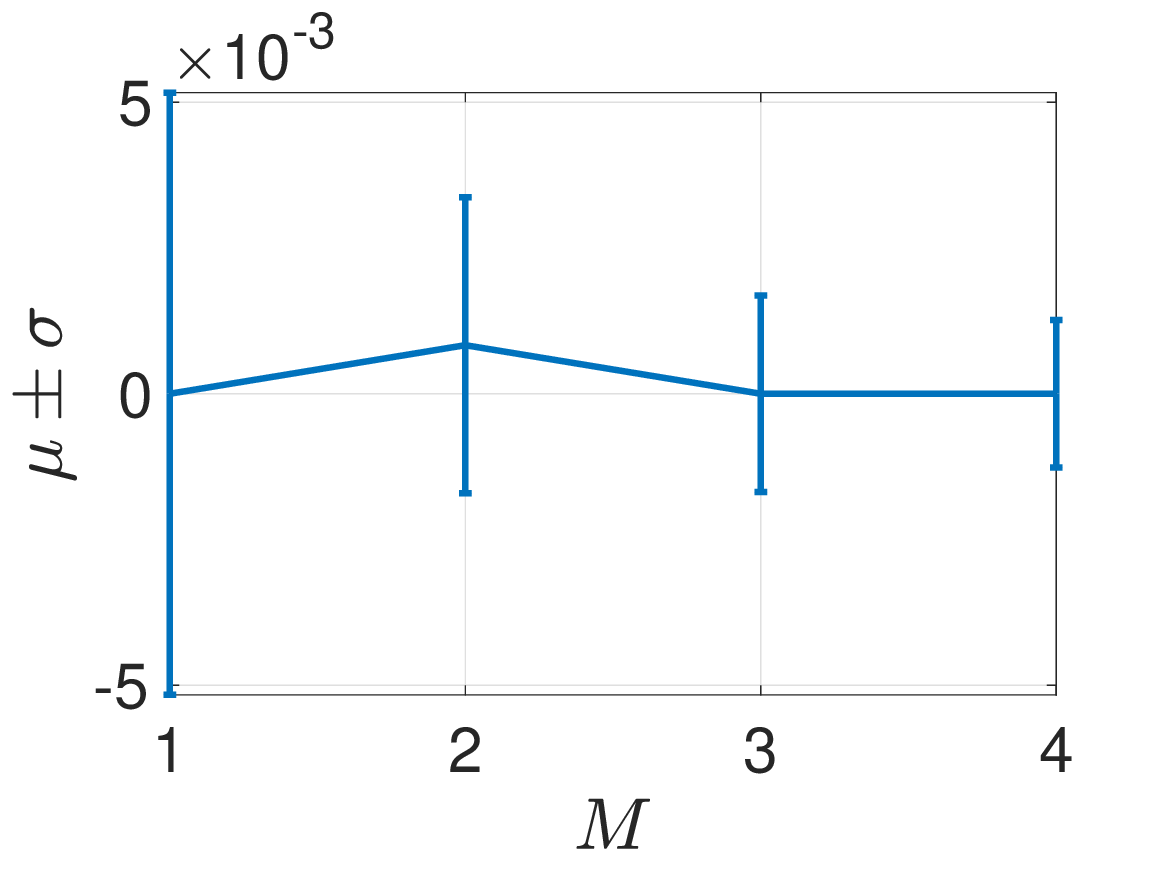}
      % \vspace{0.1 cm}
    \end{minipage}%
  }
  \subfloat[\label{sec:GraCov:fig:8}]{
    \begin{minipage}[t]{0.45\textwidth}
      \centering
      \includegraphics[scale=0.37]{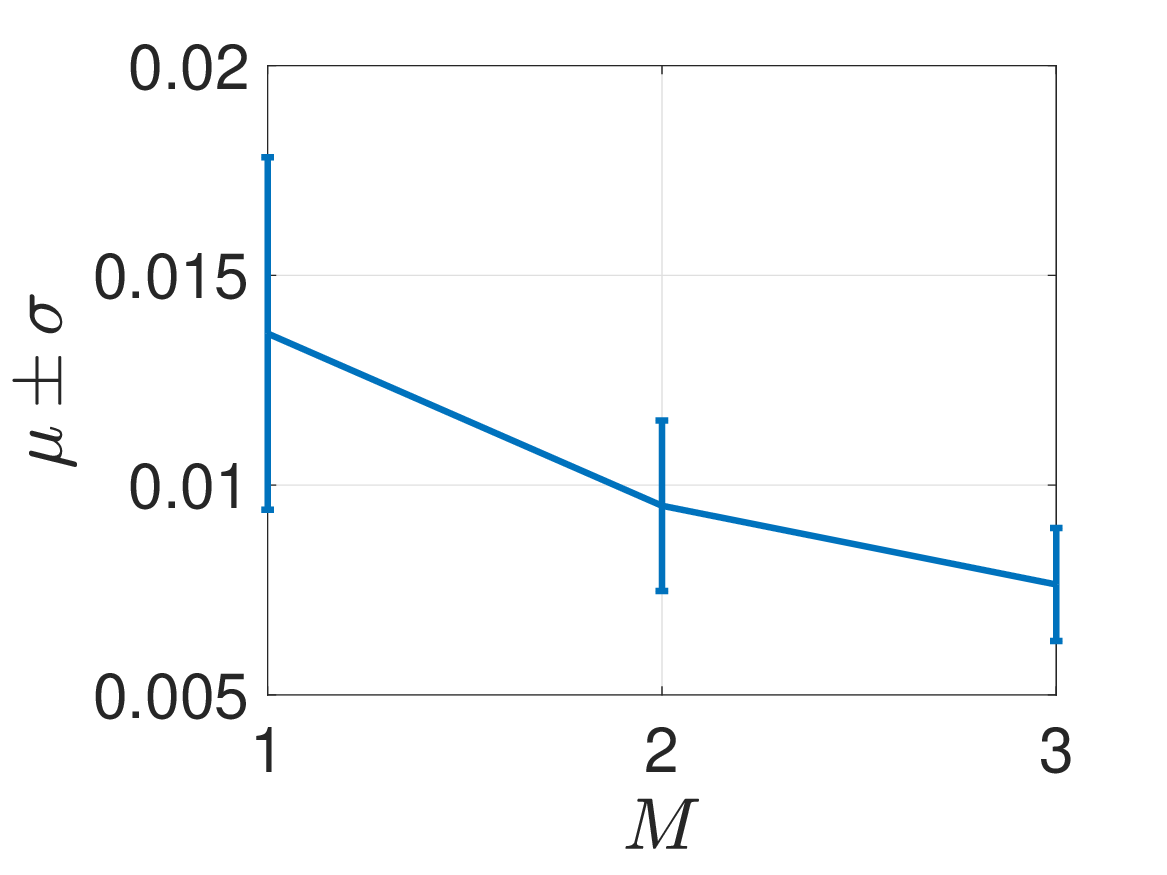}
      % \vspace{0.1 cm}
    \end{minipage}%
  }

  \caption{DE-NFGs and the numerical results in Examples~\ref{sec:GraCov:exp:1} and~\ref{sec:GraCov:exp:2}. }
\end{figure}
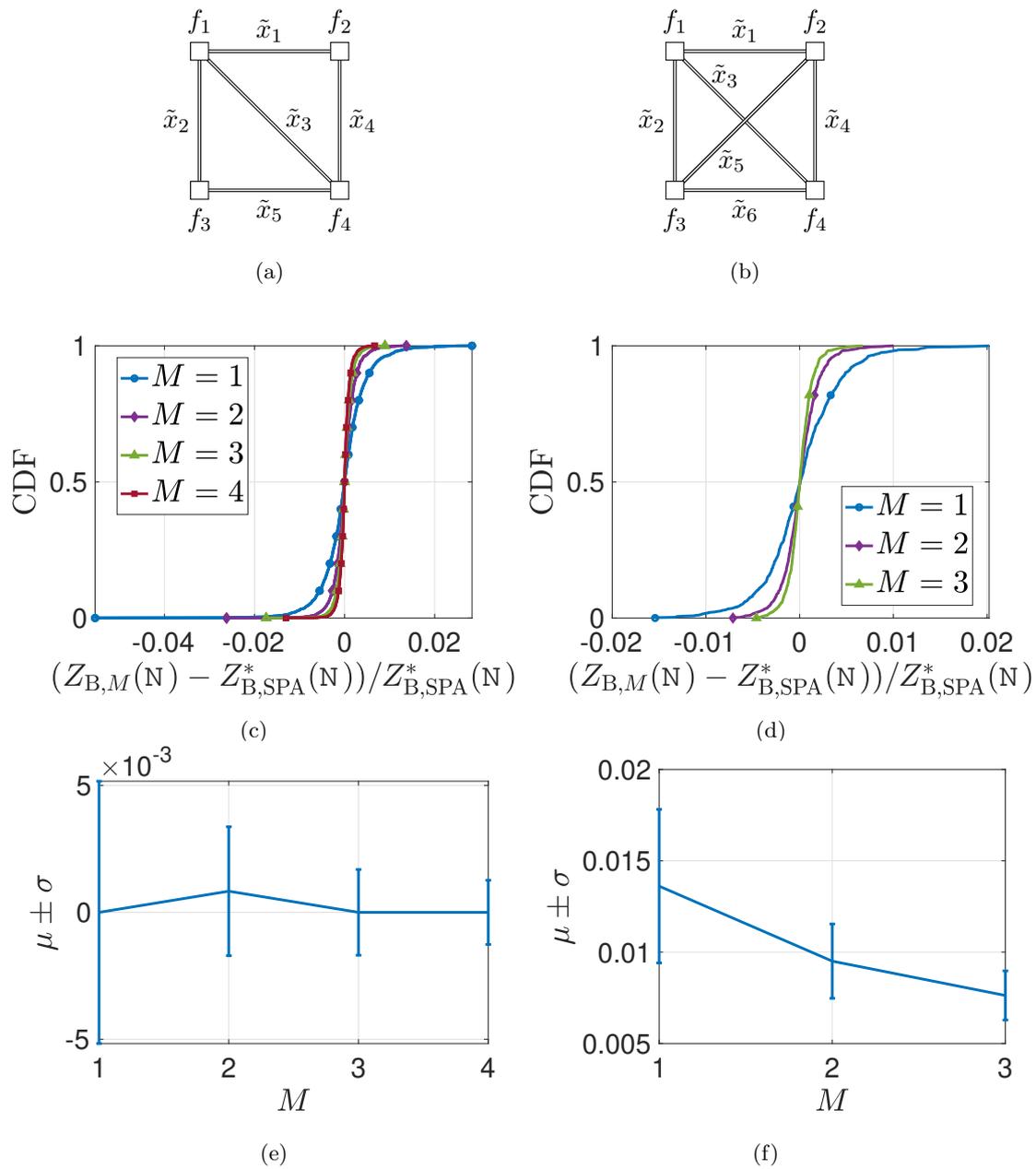

%% file: figures/pascal_triangle1_1.tex
\begin{tikzpicture}
  {\scriptsize

  \node at (0,0) 
    {\parbox{2cm}{\centering
                    \fbox{$1$} \\[0.05cm]
                    $\Cgen{0}{n}\bigl(\vgam^{(0,0)}\bigr)$
                    
                 }
    };

  \draw[->] (-0.25,+0.50) -- node[below left]  {$1$} (-0.75,+1.50);
  \draw[->] (+0.25,+0.50) -- node[below right] {$1$} (+0.75,+1.50);

  \node at (-1,+2) 
    {\parbox{2cm}{\centering
                    \fbox{$1$} \\[0.05cm]
                    $\Cgen{1}{n}\bigl(\vgam^{(1,0)}\bigr)$
                 }
    };

  \node at (+1,+2) 
    {\parbox{2cm}{\centering
                    \fbox{$1$} \\[0.05cm]
                    $\Cgen{1}{n}\bigl(\vgam^{(0,1)}\bigr)$
                 }
    };

  \draw[->] (-1.25,+2.50) -- node[below left]  {$1$} (-1.75,+3.50);
  \draw[->] (-0.75,+2.50) -- node[below right] {$1$} (-0.25,+3.50);

  \draw[->] (+0.75,+2.50) -- node[below left]  {$1$} (+0.25,+3.50);
  \draw[->] (+1.25,+2.50) -- node[below right] {$1$} (+1.75,+3.50);

  \node at (-2,+4) 
    {\parbox{2cm}{\centering
                    \fbox{$1$} \\[0.05cm]
                    $\Cgen{2}{n}\bigl(\vgam^{(2,0)}\bigr)$
                 }
    };

  \node at ( 0,+4) 
    {\parbox{2cm}{\centering
                    \fbox{$2$} \\[0.05cm]
                    $\Cgen{2}{n}\bigl(\vgam^{(1,1)}\bigr)$
                 }
    };

  \node at (+2,+4) 
    {\parbox{2cm}{\centering
                    \fbox{$1$} \\[0.05cm]
                    $\Cgen{2}{n}\bigl(\vgam^{(0,2)}\bigr)$
                 }
    };

  \draw[->] (-2.25,+4.50) -- node[below left]  {$1$} (-2.75,+5.50);
  \draw[->] (-1.75,+4.50) -- node[below right] {$1$} (-1.25,+5.50);

  \draw[->] (-0.25,+4.50) -- node[below left]  {$1$} (-0.75,+5.50);
  \draw[->] (+0.25,+4.50) -- node[below right] {$1$} (+0.75,+5.50);

  \draw[->] (+1.75,+4.50) -- node[below left]  {$1$} (+1.25,+5.50);
  \draw[->] (+2.25,+4.50) -- node[below right] {$1$} (+2.75,+5.50);

  \node at (-3,+6) 
    {\parbox{2cm}{\centering
                    \fbox{$1$} \\[0.05cm]
                    $\Cgen{3}{n}\bigl(\vgam^{(3,0)}\bigr)$
                 }
    };

  \node at (-1,+6) 
    {\parbox{2cm}{\centering
                    \fbox{$3$} \\[0.05cm]
                    $\Cgen{3}{n}\bigl(\vgam^{(2,1)}\bigr)$
                 }
    };

  \node at (+1,+6) 
    {\parbox{2cm}{\centering
                    \fbox{$3$} \\[0.05cm]
                    $\Cgen{3}{n}\bigl(\vgam^{(1,2)}\bigr)$
                 }
    };

  \node at (+3,+6) 
    {\parbox{2cm}{\centering
                    \fbox{$1$} \\[0.05cm]
                    $\Cgen{3}{n}\bigl(\vgam^{(0,3)}\bigr)$
                 }
    };

  }
\end{tikzpicture}

%% file: figures/pascal_triangle2_1.tex
\begin{tikzpicture}
  {\scriptsize

  \node at (0,0) 
    {\parbox{2cm}{\centering
                    \fbox{$1$} \\[0.05cm]
                    $\CBgen{0}{n}\bigl(\vgam^{(0,0)}\bigr)$
                 }
    };

  \draw[->] (-0.25,+0.50) -- node[below left]  {$1$} (-0.75,+1.50);
  \draw[->] (+0.25,+0.50) -- node[below right] {$1$} (+0.75,+1.50);

  \node at (-1,+2) 
    {\parbox{2cm}{\centering
                    \fbox{$1$} \\[0.05cm]
                    $\CBgen{1}{n}\bigl(\vgam^{(1,0)}\bigr)$
                 }
    };

  \node at (+1,+2) 
    {\parbox{2cm}{\centering
                    \fbox{$1$} \\[0.05cm]
                    $\CBgen{1}{n}\bigl(\vgam^{(0,1)}\bigr)$
                 }
    };

  \draw[->] (-1.25,+2.50) -- node[below left]  {$1$} (-1.75,+3.50);
  \draw[->] (-0.75,+2.50) -- node[below right] {$\frac{1}{2}$} (-0.25,+3.50);

  \draw[->] (+0.75,+2.50) -- node[below left]  {$\frac{1}{2}$} (+0.25,+3.50);
  \draw[->] (+1.25,+2.50) -- node[below right] {$1$} (+1.75,+3.50);

  \node at (-2,+4) 
    {\parbox{2cm}{\centering
                    \fbox{$1$} \\[0.05cm]
                    $\CBgen{2}{n}\bigl(\vgam^{(2,0)}\bigr)$
                 }
    };

  \node at ( 0,+4) 
    {\parbox{2cm}{\centering
                    \fbox{$1$} \\[0.05cm]
                    $\CBgen{2}{n}\bigl(\vgam^{(1,1)}\bigr)$
                 }
    };

  \node at (+2,+4) 
    {\parbox{2cm}{\centering
                    \fbox{$1$} \\[0.05cm]
                    $\CBgen{2}{n}\bigl(\vgam^{(0,2)}\bigr)$
                 }
    };

  \draw[->] (-2.25,+4.50) -- node[below left]  {$1$} (-2.75,+5.50);
  \draw[->] (-1.75,+4.50) -- node[below right] {$\frac{1}{2}$} (-1.25,+5.50);

  \draw[->] (-0.25,+4.50) -- node[below left]  {$\frac{1}{2}$} (-0.75,+5.50);
  \draw[->] (+0.25,+4.50) -- node[below right] {$\frac{1}{2}$} (+0.75,+5.50);

  \draw[->] (+1.75,+4.50) -- node[below left]  {$\frac{1}{2}$} (+1.25,+5.50);
  \draw[->] (+2.25,+4.50) -- node[below right] {$1$} (+2.75,+5.50);

  \node at (-3,+6) 
    {\parbox{2cm}{\centering
                    \fbox{$1$} \\[0.05cm]
                    $\CBgen{3}{n}\bigl(\vgam^{(3,0)}\bigr)$
                 }
    };

  \node at (-1,+6) 
    {\parbox{2cm}{\centering
                    \fbox{$1$} \\[0.05cm]
                    $\CBgen{3}{n}\bigl(\vgam^{(2,1)}\bigr)$
                 }
    };

  \node at (+1,+6) 
    {\parbox{2cm}{\centering
                    \fbox{$1$} \\[0.05cm]
                    $\CBgen{3}{n}\bigl(\vgam^{(1,2)}\bigr)$
                 }
    };

  \node at (+3,+6) 
    {\parbox{2cm}{\centering
                    \fbox{$1$} \\[0.05cm]
                    $\CBgen{3}{n}\bigl(\vgam^{(0,3)}\bigr)$
                 }
    };

  }
\end{tikzpicture}

%% file: figures/pascal_triangle3_1.tex
\begin{tikzpicture}
  {\scriptsize
  \pgfmathsetmacro{\Ws}{1.1};
  \node at (0,0) 
    {\parbox{2cm}{\centering
                    \fbox{$1$} \\[0.05cm]
                    $\CscSgen{0}{n}\bigl(\vgam^{(0,0)}\bigr)$
                 }
    };

  \draw[->] (-0.25*\Ws,+0.50) -- node[below left]  {$1$} (-0.75*\Ws,+1.50);
  \draw[->] (+0.25*\Ws,+0.50) -- node[below right] {$1$} (+0.75*\Ws,+1.50);

  \node at (-1*\Ws,+2) 
    {\parbox{2cm}{\centering
                    \fbox{$1$} \\[0.05cm]
                    $\CscSgen{1}{n}\bigl(\vgam^{(1,0)}\bigr)$
                 }
    };

  \node at (+1*\Ws,+2) 
    {\parbox{2cm}{\centering
                    \fbox{$1$} \\[0.05cm]
                    $\CscSgen{1}{n}\bigl(\vgam^{(0,1)}\bigr)$
                 }
    };

  \draw[->] (-1.25*\Ws,+2.50) -- node[below left]  {$\frac{1}{4}$} (-1.75*\Ws,+3.50);
  \draw[->] (-0.75*\Ws,+2.50) -- node[below right] {$\frac{1}{2}$} (-0.25*\Ws,+3.50);

  \draw[->] (+0.75*\Ws,+2.50) -- node[below left]  {$\frac{1}{2}$} (+0.25*\Ws,+3.50);
  \draw[->] (+1.25*\Ws,+2.50) -- node[below right] {$\frac{1}{4}$} (+1.75*\Ws,+3.50);

  \node at (-2*\Ws,+4) 
    {\parbox{2cm}{\centering
                    \fbox{$\frac{1}{4}$} \\[0.05cm]
                    $\CscSgen{2}{n}\bigl(\vgam^{(2,0)}\bigr)$
                 }
    };

  \node at ( 0,+4) 
    {\parbox{2cm}{\centering
                    \fbox{$1$} \\[0.05cm]
                    $\CscSgen{2}{n}\bigl(\vgam^{(1,1)}\bigr)$
                 }
    };

  \node at (+2*\Ws,+4) 
    {\parbox{2cm}{\centering
                    \fbox{$\frac{1}{4}$} \\[0.05cm]
                    $\CscSgen{2}{n}\bigl(\vgam^{(0,2)}\bigr)$
                 }
    };

  \draw[->] (-2.25*\Ws,+4.50) -- node[below left]  {$\frac{16}{81}$} (-2.75*\Ws,+5.50);
  \draw[->] (-1.75*\Ws,+4.50) -- node[below right] {$\frac{16}{45}$} (-1.25*\Ws,+5.50);

  \draw[->] (-0.25*\Ws,+4.50) -- node[below left]  {$\frac{16}{45}$} (-0.75*\Ws,+5.50);
  \draw[->] (+0.25*\Ws,+4.50) -- node[below right] {$\frac{16}{45}$} (+0.75*\Ws,+5.50);

  \draw[->] (+1.75*\Ws,+4.50) -- node[below left]  {$\frac{16}{45}$} (+1.25*\Ws,+5.50);
  \draw[->] (+2.25*\Ws,+4.50) -- node[below right] {$\frac{16}{81}$} (+2.75*\Ws,+5.50);

  \node at (-3*\Ws,+6) 
    {\parbox{2cm}{\centering
                    \fbox{$\frac{4}{81}$} \\[0.05cm]
                    $\CscSgen{3}{n}\bigl(\vgam^{(3,0)}\bigr)$
                 }
    };

  \node at (-1*\Ws,+6) 
    {\parbox{2cm}{\centering
                    \fbox{$\frac{4}{9}$} \\[0.05cm]
                    $\CscSgen{3}{n}\bigl(\vgam^{(2,1)}\bigr)$
                 }
    };

  \node at (+1*\Ws,+6) 
    {\parbox{2cm}{\centering
                    \fbox{$\frac{4}{9}$} \\[0.05cm]
                    $\CscSgen{3}{n}\bigl(\vgam^{(1,2)}\bigr)$
                 }
    };

  \node at (+3*\Ws,+6) 
    {\parbox{2cm}{\centering
                    \fbox{$\frac{4}{81}$} \\[0.05cm]
                    $\CscSgen{3}{n}\bigl(\vgam^{(0,3)}\bigr)$
                 }
    };

  }
\end{tikzpicture}

%% file: figures/lct_examples/snfg/main_pic.tex
\begin{figure}
  \centering
  \captionsetup{font=scriptsize}
  \subfloat[\label{sec:LCT:fig:3}]{
    \begin{minipage}[t]{0.9\textwidth}
      \centering
      \begin{tikzpicture}[on grid,auto]
        \input{figures/lct_examples/scale_small.tex}
        \input{figures/head_files_figs.tex}
        \input{figures/lct_examples/snfg/background_nodes.tex}
        % draw the lines
        \input{figures/lct_examples/snfg/background_lines.tex}
        %----------------------------------------------------------------------
        \begin{pgfonlayer}{main} 
            \node (x1) at (0,-0.2*\sdis) 
            [label=above: $x_{1}$] {}; 
            \node (x3) at (-0.5*\ldis-2.1*\sdis,1.3*\sdis) 
            [label=above: $x_{2}$] {}; 
            \node (x4) at (-0.5*\ldis-2.1*\sdis,-1.3*\sdis) 
            [label=below: $x_{3}$] {};

            \node (x3) at (0.5*\ldis+2.1*\sdis,1.3*\sdis) 
            [label=above: $x_{4}$] {}; 
            \node (x4) at (0.5*\ldis+2.1*\sdis,-1.3*\sdis) 
            [label=below: $x_{5}$] {}; 
        \end{pgfonlayer}
        %----------------------------------------------------------------------
      \end{tikzpicture}
    \end{minipage}
  }

  \subfloat[\label{sec:LCT:fig:4}]{
    \begin{minipage}[t]{0.9\textwidth}
      \centering
      \begin{tikzpicture}[on grid,auto]
        \input{figures/lct_examples/scale_large_snfg.tex}
        \input{figures/head_files_figs.tex}
        \input{figures/lct_examples/snfg/background_nodes.tex}
        \input{figures/lct_examples/snfg/background_lines.tex}
        \input{figures/lct_examples/otb_funs.tex}
        \input{figures/lct_examples/otb_dashed_boxes.tex}
        \input{figures/lct_examples/snfg/variables_after_lct.tex}
      \end{tikzpicture}
    \end{minipage}
  }\\
  \subfloat[\label{sec:LCT:fig:5}]{
    \begin{minipage}[t]{0.9\textwidth}
      \centering
      \begin{tikzpicture}[on grid,auto]
        \input{figures/lct_examples/scale_large_snfg.tex}
        \input{figures/head_files_figs.tex}
        \input{figures/lct_examples/snfg/background_nodes.tex}
        \input{figures/lct_examples/snfg/background_lines.tex}
        \input{figures/lct_examples/otb_funs.tex}
        \input{figures/lct_examples/ctb_dashed_boxes.tex}
        \input{figures/lct_examples/snfg/variables_before_lct.tex}
        \input{figures/lct_examples/snfg/variables_after_lct.tex}
      \end{tikzpicture}
    \end{minipage}
  }

   %------------------------------------------------------------------------
  \subfloat[\label{sec:LCT:fig:6}]{
    \begin{minipage}[t]{0.9\textwidth}
      \centering
      \begin{tikzpicture}[on grid,auto]
        \input{figures/lct_examples/scale_small.tex}
        \input{figures/head_files_figs.tex}
        \input{figures/lct_examples/snfg/background_nodes_after_ctb.tex}
        \input{figures/lct_examples/snfg/background_lines.tex}
        \input{figures/lct_examples/snfg/variables_after_lct.tex}
      \end{tikzpicture}
    \end{minipage}
  }
  %------------------------------------------------------------------------
  \caption{Exemplifying the LCT for a part of an example
  S-NFG.}\label{sec:LCT:fig:1}
\end{figure}

%% file: figures/lct_examples/scale_small.tex
\pgfmathsetmacro{\ldis}{1} 
% the distance for the h and f, which is shorter
\pgfmathsetmacro{\sdis}{0.42} 
% the boxes for function nodes

%% file: figures/lct_examples/snfg/background_nodes.tex
%----------------------------------------------------------------------------
\begin{pgfonlayer}{main}
  \node[state] (f1) at (-\ldis,0) [label=above: $f_{i}$] {};
  \node[state] (f2) at (\ldis,0) [label=above: $f_{j}$] {};

  \node[] (f3) at (-\ldis-2.9*\sdis, 2.9*\sdis) [] {}; 

  \node[] (f4) at (-\ldis-2.8*\sdis,-2.8*\sdis) [] {}; 

  \node[] (f5) at (\ldis+2.8*\sdis,2.8*\sdis) [] {}; 

  \node[] (f6) at (\ldis+2.8*\sdis,-2.8*\sdis) [] {};
\end{pgfonlayer}
%----------------------------------------------------------------------------

%% file: figures/lct_examples/snfg/background_lines.tex
\begin{pgfonlayer}{behind} 
	\draw[] 
	(f1) -- (f2)
	(f1) -- (f3)
	(f1) -- (f4)
	(f2) -- (f5)
	(f2) -- (f6);
\end{pgfonlayer}

%% file: figures/lct_examples/scale_large_snfg.tex
% the distance between h and x
\pgfmathsetmacro{\ldis}{2} 
\pgfmathsetmacro{\sdis}{1} 

%% file: figures/lct_examples/OTB_funs.tex
%----------------------------------------------------------------------------
\begin{pgfonlayer}{main} 
  % the functions on LHS
  \node[state] (M1) at (-\ldis-2*\sdis,2*\sdis) 
  [label=above: $h_{1}$] {};
  \node[state] (M2) at (-\ldis-\sdis,\sdis) 
  [label=above: $h_{2}$] {};
  \node[state] (M3) at (-\sdis/2,0) 
  [label=above: $h_{3}$] {};
  \node[state] (M4) at (-\ldis-\sdis,-\sdis) 
  [label=below: $h_{4}$] {};
  \node[state] (M5) at (-\ldis-2*\sdis,-2*\sdis) 
  [label=below: $h_{5}$] {};
  % the functions on RHS
  \node[state] (M6) at (\ldis+2*\sdis,2*\sdis) 
  [label=above: $h_{6}$] {};
  \node[state] (M7) at (\ldis+\sdis,\sdis) 
  [label=above: $h_{7}$] {};
  \node[state] (M8) at (\sdis/2,0) 
  [label=above: $h_{8}$] {};
  \node[state] (M9) at (\ldis+\sdis,-\sdis) 
  [label=below: $h_{9}$] {};
  \node[state] (M10) at (\ldis+2*\sdis,-2*\sdis) 
  [label=below: $h_{10}$] {};
  % the variables for OTB
\end{pgfonlayer} 
%----------------------------------------------------------------------------

%% file: figures/lct_examples/OTB_dashed_boxes.tex
%----------------------------------------------------------------------------
\begin{pgfonlayer}{background} 
  \node[state_dash1] (db1) at (0,\sdis*0.2) [] {};
  \node[state_dash2] (db2) at (-\ldis-1.5*\sdis,1.7*\sdis)  [] {};
  \node[state_dash2] (db3) at (-\ldis-1.5*\sdis,-1.7*\sdis) [] {};
  \node[state_dash2] (db4) at (\ldis+1.5*\sdis,1.7*\sdis) [] {};
  \node[state_dash2] (db5) at (\ldis+1.5*\sdis,-1.7*\sdis) [] {};
\end{pgfonlayer}
%----------------------------------------------------------------------------

%% file: figures/lct_examples/snfg/variables_after_lct.tex
%----------------------------------------------------------------------------
\begin{pgfonlayer}{main} 
    \node (x1) at (0,-0.22*\sdis) 
    [label=above: $\LCT{x}_{1}$] {}; 
    \node (x3) at (-1*\ldis-1.5*\sdis,1.5*\sdis) 
    [label=above: $\LCT{x}_{2}$] {}; 
    \node (x4) at (-1*\ldis-1.51*\sdis,-1.5*\sdis) 
    [label=below: $\LCT{x}_{3}$] {};

    \node (x3) at (1*\ldis+1.5*\sdis,1.5*\sdis) 
    [label=above: $\LCT{x}_{4}$] {}; 
    \node (x4) at (1*\ldis+1.5*\sdis,-1.5*\sdis) 
    [label=below: $\LCT{x}_{5}$] {}; 
\end{pgfonlayer}
%----------------------------------------------------------------------------

%% file: figures/lct_examples/penfg/CTB_dashed_boxes.tex
\foreach{\x} in {0,1}{
    %------------------------------------------------------------------------
    \begin{pgfonlayer}{background} 
      \node[state_dash4] (db6\x) at (-\ldis,-5.8*\sdis*\x) [] {};
      \node[state_dash4] (db7\x) at (\ldis,-5.8*\sdis*\x) [] {};
    \end{pgfonlayer}
    %------------------------------------------------------------------------
}

%% file: figures/lct_examples/snfg/variables_before_lct.tex
%----------------------------------------------------------------------------
\begin{pgfonlayer}{main} 
    \node (x1) at (-0.5*\ldis-0.22*\sdis,-0.22*\sdis) 
    [label=above: $x_{1,f_{i}}$] {};
    \node (x2) at (-1*\ldis-0.4*\sdis,0.5*\sdis) 
    [label=above: $x_{2,f_{i}}$] {}; 
    \node (x3) at (-1*\ldis-0.4*\sdis,-0.5*\sdis) 
    [label=below: $x_{3, f_{i}}$] {}; 

    \node (x4) at (0.5*\ldis+0.22*\sdis,-0.22*\sdis) 
    [label=above: $x_{1,f_{j}}$] {};
    \node (x5) at (1*\ldis+0.4*\sdis,0.5*\sdis) 
    [label=above: $x_{4,f_{j}}$] {}; 
    \node (x6) at (1*\ldis+0.4*\sdis,-0.5*\sdis) 
    [label=below: $x_{5,f_{j}}$] {};  
\end{pgfonlayer}
%----------------------------------------------------------------------------

%% file: figures/lct_examples/snfg/background_nodes_after_ctb.tex
%----------------------------------------------------------------------------
\begin{pgfonlayer}{main}
  \node[state] (f1) at (-\ldis,0) [label=above: $\LCT{f}_{i}$] {};
  \node[state] (f2) at (\ldis,0) [label=above: $\LCT{f}_{j}$] {};

  \node[] (f3) at (-\ldis-2.9*\sdis, 2.9*\sdis) [] {}; 

  \node[] (f4) at (-\ldis-2.8*\sdis,-2.8*\sdis) [] {}; 

  \node[] (f5) at (\ldis+2.8*\sdis,2.8*\sdis) [] {}; 

  \node[] (f6) at (\ldis+2.8*\sdis,-2.8*\sdis) [] {};
\end{pgfonlayer}
%----------------------------------------------------------------------------

%% file: figures/sst_examples/snfg/main_pic.tex
\begin{figure}[t]
  \captionsetup{skip=0.1cm}
  \captionsetup{font=scriptsize}
  \captionsetup{aboveskip=0pt}
  \subfloat[\label{sec:SST:fig:3}]{
    \begin{minipage}{0.9\textwidth}
      \centering
      \begin{tikzpicture}
        \input{figures/head_files_figs.tex}
        \input{figures/sst_examples/snfg/length.tex}
        \input{figures/sst_examples/snfg/pic_1.tex}
      \end{tikzpicture}
    \end{minipage}
  }\\
  \subfloat[\label{sec:SST:fig:4}]{
    \begin{minipage}[t]{0.45\textwidth}
      \centering
      \begin{tikzpicture}
        \input{figures/head_files_figs.tex}
        \input{figures/sst_examples/snfg/length.tex}
        \input{figures/sst_examples/snfg/background_nodes_lines.tex}
        \input{figures/sst_examples/snfg/lines_not_permuted.tex}
      \end{tikzpicture}
    \end{minipage}
  }
  \subfloat[\label{sec:SST:fig:5}]{
    \begin{minipage}[t]{0.45\textwidth}
      \centering
      \begin{tikzpicture}
        \input{figures/head_files_figs.tex}
        \input{figures/sst_examples/snfg/length.tex}
        \input{figures/sst_examples/snfg/background_nodes_lines.tex}
        %----------------------------------------------------------------------------
        \begin{pgfonlayer}{background}
           \foreach \x in {0,1}{
            \coordinate (f7\x) at (-1*0.5*\ldis,\x*2*\ldis) ; 
            \coordinate (f8\x) at (0.5*\ldis,\x*2*\ldis) ; 
          }
          \draw[]
            (f10) -- (f70) -- (f81) -- (f21) 
            (f11) -- (f71) -- (f80) -- (f20);
        \end{pgfonlayer}
        %----------------------------------------------------------------------------
      \end{tikzpicture}
    \end{minipage}
  }\\
  \subfloat[\label{sec:SST:fig:6}]{
    \begin{minipage}[t]{0.45\textwidth}
      \centering
      \begin{tikzpicture}
        \input{figures/head_files_figs.tex}
        \input{figures/sst_examples/snfg/length.tex}
        \input{figures/sst_examples/snfg/background_nodes_lines.tex}
        \input{figures/sst_examples/snfg/lines_not_permuted.tex}
        \input{figures/sst_examples/snfg/node_pe_sige.tex}
      \end{tikzpicture}
    \end{minipage}
  }
  \subfloat[\label{sec:SST:fig:7}]{
    \begin{minipage}[t]{0.45\textwidth}
      \centering
      \begin{tikzpicture}
        \input{figures/head_files_figs.tex}
        \input{figures/sst_examples/snfg/length.tex}
        \input{figures/sst_examples/snfg/background_nodes_lines.tex}
        \input{figures/sst_examples/snfg/lines_not_permuted.tex}
        \input{figures/sst_examples/snfg/node_pe_sige.tex}
      \end{tikzpicture}
    \end{minipage}
  }\\
  \subfloat[\label{sec:SST:fig:8}]{
    \begin{minipage}[t]{0.3\textwidth}
      \centering
      \begin{tikzpicture}
        % \tikzstyle{state_dash4}=[shape=rectangle, draw, dashed, minimum width= 1.3*\ldis cm, minimum height = 4.5*\ldis  cm,outer sep=-0.3pt]
        \input{figures/head_files_figs.tex}
        \input{figures/sst_examples/snfg/length.tex}
        \input{figures/sst_examples/snfg/background_nodes_lines.tex}
        \input{figures/sst_examples/snfg/lines_not_permuted.tex}
        \input{figures/sst_examples/snfg/node_pe_sige.tex}
        \begin{pgfonlayer}{background}
          \node[state_dash6] (db1) at (0,0.1*\ldis) [label=below:$P_{e}$] {};
          \node[] (var1) at (-0.18*\ldis,-1*\ldis) [label=right: $\sigma_{e}$] {};
          \node[state] (Psige) at (0,-1.5*\ldis) [label=below: $p_{\Sigma(e)}$] {};
        \end{pgfonlayer}
        \input{figures/sst_examples/snfg/otb_back_nodes_lines.tex}
      \end{tikzpicture}
    \end{minipage}
  }
  \subfloat[\label{sec:SST:fig:9}]{
    \begin{minipage}[t]{0.3\textwidth}
      \centering
      \begin{tikzpicture}
        \input{figures/head_files_figs.tex}
        \input{figures/sst_examples/snfg/length.tex}
        \input{figures/sst_examples/snfg/background_nodes_lines.tex}
        \input{figures/sst_examples/snfg/lines_not_permuted.tex}
        \begin{pgfonlayer}{above}
          \node[state_large] (db1) at (0,1*\ldis) [label=below:$P_{e}$] {};
        \end{pgfonlayer}
        \node (Psige) at (0,-2.59*\ldis) [] {};
      \end{tikzpicture}
    \end{minipage}
  }
  \subfloat[\label{sec:SST:fig:10}]{
    \begin{minipage}[t]{0.3\textwidth}
      \centering
      \begin{tikzpicture}
        \input{figures/head_files_figs.tex}
        \input{figures/sst_examples/snfg/length_large.tex}
        \input{figures/sst_examples/snfg/background_nodes_lines.tex}
        \input{figures/sst_examples/snfg/lines_not_permuted.tex}
        \input{figures/sst_examples/snfg/otb_pe.tex}
        \node[] (var1) at (-0.18*\ldis,-1*\ldis) [label=right: $\cvpsi_{e}$] {};
        \node[state] (Psige) at (0,-1.5*\ldis) [label=below: $ \bigl| \set{B}_{\set{X}_e^M}\! \bigr| \cdot \muFSsimple$] {};
        \input{figures/sst_examples/snfg/otb_back_nodes_lines.tex}
        \begin{pgfonlayer}{background}
          \node[state_dash6] (db1) at (0,0.1*\ldis) [label=below:$P_{e}$] {};
        \end{pgfonlayer}
      \end{tikzpicture}
    \end{minipage}
  }
  \caption{Exemplifying the SST for a part of an example S-NFG.\label{sec:SST:fig:1}}
\end{figure}

%% file: figures/sst_examples/snfg/pic_1.tex
%----------------------------------------------------------------------------
\begin{pgfonlayer}{background}
  \node[state] (f1) at (-\ldis,0) [label=above: $f_{i}$] {};
  \node[state] (f2) at (\ldis,0) [label=above: $f_{j}$] {};
  \node[] (f3) at (-2*\ldis,\ldis) [] {}; 
  \node[] (f4) at (-2*\ldis,-\ldis) [] {}; 
  \node[] (f5) at (2*\ldis,\ldis) [] {}; 
  \node[] (f6) at (2*\ldis,-\ldis) [] {}; 
\end{pgfonlayer}
%----------------------------------------------------------------------------
\begin{pgfonlayer}{behind}
  \draw[]
    (f1) --  (f2)
    (f1) -- (f3)
    (f1) -- (f4)
    (f2) -- (f5)
    (f2) -- (f6);
\end{pgfonlayer}
%----------------------------------------------------------------------------
\vspace*{-10cm}

%% file: figures/sst_examples/snfg/background_nodes_lines.tex
\foreach \x/\xin in {0/1,1/2}{
  %----------------------------------------------------------------------------
  \begin{pgfonlayer}{behind}
    \node[state] (f1\x) at (-\ldis,\x*2*\ldis) [label=above: $f_{i,\xin}$] {};
    \node[state] (f2\x) at (\ldis,\x*2*\ldis) [label=above: $f_{j,\xin}$] {};
    \node[] (f3\x) at (-2*\ldis,\ldis+\x*2*\ldis) [] {}; 
    \node[] (f4\x) at (-2*\ldis,-\ldis+\x*2*\ldis) [] {}; 
    \node[] (f5\x) at (2*\ldis,\ldis+\x*2*\ldis) [] {}; 
    \node[] (f6\x) at (2*\ldis,-\ldis+\x*2*\ldis) [] {}; 
  \end{pgfonlayer}
  %----------------------------------------------------------------------------
  \begin{pgfonlayer}{behind}
    \draw[]
      (f1\x) -- (f3\x)
      (f1\x) -- (f4\x)
      (f2\x) -- (f5\x)
      (f2\x) -- (f6\x);
  \end{pgfonlayer}
  %----------------------------------------------------------------------------
}

%% file: figures/sst_examples/snfg/lines_not_permuted.tex
\foreach \x in {0,1}{
  %----------------------------------------------------------------------------
  \begin{pgfonlayer}{behind}
    \draw[]
      (f1\x) -- (f2\x);
  \end{pgfonlayer}
  %----------------------------------------------------------------------------
}

%% file: figures/sst_examples/snfg/node_pe_sige.tex
% For illustrate SST

\begin{pgfonlayer}{glass}
  \node[state_large] (Pe) at (0,\ldis) [] {$P_{e,\sigma_e}$};
\end{pgfonlayer}

%% file: figures/sst_examples/snfg/otb_back_nodes_lines.tex
% For illustrate SST
\begin{pgfonlayer}{main}
  \foreach \x in {0,1}{
    \node[state] (eqn\x) at (0,\x*2*\ldis) [] {$=$};
  }
\end{pgfonlayer}

\begin{pgfonlayer}{behind}
  \draw[] (eqn1) -- (Psige);
\end{pgfonlayer}

%% file: figures/sst_examples/snfg/length_large.tex
% the distance between h and x
\pgfmathsetmacro{\ldis}{1.3} 
\pgfmathsetmacro{\sdis}{0.65} 

%% file: figures/sst_examples/snfg/OTB_pe.tex
\begin{pgfonlayer}{behind}
  \foreach \x in {0,1}{
    \node[state] (psi\x) at (-\sdis,\x*2*\ldis) [label=below: $\cpsi_{e}$] {};
    \node[state] (cpsi\x) at (\sdis,\x*2*\ldis) [label=below: $\overline{\cpsi_{e}}$] {};
  }  
\end{pgfonlayer}

%% file: appendix.tex
\begin{appendices}

%***************************************************************************
%***************************************************************************

%\newpage
% \renewcommand{\chaptermark}[1]{\markboth{#1}{}}
% \renewcommand{\sectionmark}[1]{\markright{#1}}

\chapter{Proofs in Chapter~\ref{CHAPT: PRELIMINARIES}}

\section{Properties of \texorpdfstring{$Z(\graphN)$}{} for DE-NFGs}
\label{apx:property of ZN}
\input{appendices/props_zn_denfg.tex}

\ifx\sectionheaderonnewpage\x
\clearpage
\fi

%------------------------------------------------------------------------
% \chapter{Proof of Lemma~\ref{sec:DENFG:lem:2}}
% \label{apx:conjugate of the messages}
% \input{appendices/the_message_conjugate_penfg}

\section{Proof of Proposition~\ref{PROP: SUFFICIENT CONDITION FOR PD MESSAGES}}
\label{apx: sufficient condition for pd messages}
\input{appendices/sufficient_condition_pdmessages}
\ifx\sectionheaderonnewpage\x
\clearpage
\fi

\chapter{Proofs in Chapter~\ref{CHAPT: FGC BOUND PERMANENT}} % (fold)
\label{cha:proofs_in_chapter_chapt: fgc bound permanent}

% chapter proofs_in_chapter_chapt: fgc bound permanent (end)

\section{Proof of Proposition~\ref{SEC:1:PROP:14}}
% \chaptermark{Proof of Proposition~\ref{sec:1:prop:14}}
\label{apx:24}

%***************************************************************************

For a large integer $M$, we obtain
\begin{align*}
  \bigl( \permscsM{M}(\mtheta) \bigr)^{\! M} 
  &\overset{(a)}{=}
     \sum_{ \vgam \in \GamMnthe }
       \mtheta^{ M \cdot \vgam } \cdot \CscSgen{M}{n}( \vgam ) \\
  &\overset{(b)}{=}  \sum_{ \vgam \in \GamMnthe }
      \exp\bigl( - M \cdot \UscSthe(\vgam) \bigr)
      \cdot
      \CscSgen{M}{n}( \vgam ) \\
  &\overset{(c)}{=}  
    \exp( o(M) ) \cdot \!
    \max_{ \vgam \in \GamMnthe }
      \exp\bigl( - M \cdot \UscSthe(\vgam) \bigr)
      \cdot
      \CscSgen{M}{n}(\vgam)
    \nonumber \\
  &\overset{(d)}{=}  
    \exp( o(M) ) \cdot 
    \max_{ \vgam \in \GamMnthe } \exp\bigl( - M \cdot \FscSthe(\vgam) \bigr),
\end{align*}
where step~$(a)$ follows from~\eqref{sec:1:eqn:168}, where step~$(b)$ follows
from the expression of $ \UscSthe $ in Definition~\ref{sec:1:def:20}, where
step~$(c)$ follows from the fact that $ \bigl| \GamMnthe \bigr| $ grows
polynomially with respect to $ M $, and where step~$(d)$ follows
from~\eqref{sec:1:eqn:207} and the definition of $\FscSthe(\vgam)$ in
Definition~\ref{sec:1:def:20}. Then we have
\begin{align*}
  \limsup_{M \to \infty}
     \permscsM{M}(\mtheta)
  &= \limsup_{M \to \infty} 
       \exp\biggl( \frac{o(M)}{M} \biggr)
       \cdot
       \sqrt[M]{\max_{ \vgam \in \GamMnthe } 
                    \exp\bigl( - M \cdot \FscSthe(\vgam) \bigr)} \\
  &= \limsup_{M \to \infty} 
       \max_{ \vgam \in \GamMnthe } 
         \exp\bigl( - \FscSthe(\vgam) \bigr) \\
  &\overset{(a)}{=}
     \sup_{ \vgam \in \Gampnthe } 
       \exp\bigl( - \FscSthe(\vgam) \bigr) \\
  &= \exp
       \biggl(
         -
         \inf_{ \vgam \in \Gampnthe }
           \FscSthe(\vgam)
       \biggr) \\
  &\overset{(b)}{=} 
     \exp
       \biggl(
         -
         \min_{ \vgam \in \Gamnthe }
           \FscSthe(\vgam)
       \biggr) \\
  &\overset{(c)}{=} 
     \permscs(\vtheta),
\end{align*}
where step~$(a)$ is based on the definition
\begin{align*}
  \Gampnthe
    &\defeq
       \bigcup_{M \in \sZpp}
         \GamMnthe,
\end{align*}
where step~$(b)$ follows from the set $\Gampnthe$ being dense in the compact
set $\Gamnthe$, and where step~$(c)$ follows from the definition of
$\permscs(\vtheta)$ in~\eqref{sec:1:eqn:192:part:2}.

%***************************************************************************
%***************************************************************************

\ifx\withclearpagecommands\x
\clearpage
\fi

\section{Proof of Lemma~\ref{SEC:1:LEM:15:COPY:2}}
\label{app:sec:1:lem:15:copy:2}

%***************************************************************************
%***************************************************************************

In this appendix, we use the following, additional notation and conventions:
\begin{itemize}

\item We use the abbreviation $r$ for $r(\vgam)$.

\item We define the matrices
  \begin{align*}
    \tvgam
      &\defeq 
         M \cdot \vgam, \\ 
    \tvgam_{\sigma_1}
      &\defeq 
         (M-1) \cdot \vgam_{\sigma_1}.
  \end{align*}
  Note that $\tvgam_{\sigma_1} = \tvgam - \mP_{\sigma_1}$.

\item A product over an empty index set evaluates to $1$. For example, if
  $\setR$ is the empty set, then $\prod_{i \in \setR} (...) = 1$.

\end{itemize}

% \bigformulatop{30}{31}{%
\begin{align}
  \hspace{1cm}&\hspace{-1cm}
  \frac{1}{\CBM{n}(\vgam)}
       \cdot
       \sum_{\sigma_1 \in \setS_{[n]}( \vgam )}
         \CBgen{M-1}{n}
           \bigl( 
             \vgam_{\sigma_1}
           \bigr)
    \nonumber \\
    &\overset{(a)}{=}
      \frac{1}{\bigl( M! \bigr)^{ \! 2n - n^2}}
      \cdot
      \Biggl( 
        \prod_{i,j}
          \frac{ \bigl(\tgam(i,j) \bigr)! }
               { \bigl( M - \tgam(i,j) \bigr)! }
      \Biggr)
      \cdot
      \sum_{\sigma_1 \in \setS_{[n]}( \vgam )}
        \bigl( (M-1)! \bigr)^{ \! 2n -  n^2} 
        \nonumber\\
        &\hspace{8 cm} \cdot
        \Biggl( 
          \prod_{i,j}
            \frac{ \bigl( M-1- \tgam_{\sigma_1}(i,j) \bigr)!}
                 { \bigl( \tgam_{\sigma_1}(i,j) \bigr)! } 
        \Biggr) \nonumber \\
    &\overset{(b)}{=}
       \frac{1}{M^{ 2n -  n^2}}
       \cdot
       \sum_{\sigma_1 \in \setS_{[n]}( \vgam )}
       \Biggl( 
         \prod_{i,j}
           \underbrace
           {
             \frac{  \bigl(\tgam(i,j) \bigr)! }
                  { \bigl( \tgam_{\sigma_1}(i,j) \bigr)! } 
             \cdot
             \frac{ \bigl( M-1- \tgam_{\sigma_1}(i,j) \bigr)! }
                  { \bigl( M - \tgam(i,j) \bigr)! }
           }_{\defeq \ (*)}
       \Biggr) \nonumber \\
    &\overset{(c)}{=}
       \frac{1}{M^{ 2n -  n^2}}
       \cdot
       \sum_{\sigma_1 \in \setS_{[n]}( \vgam )}
       \Biggl( 
         \prod_{\substack{i,j \\ j = \sigma_1(i)}}
           \tgam(i,j)
       \Biggr)
       \cdot
       \Biggl( 
         \prod_{\substack{i,j \\ j \neq \sigma_1(i)}}
           \frac{ 1 }
                { M - \tgam(i,j) }
       \Biggr) \nonumber \\
    &\overset{(d)}{=}
       \frac{1}{M^{ 2n -  n^2}}
       \cdot
       \sum_{\sigma_1 \in \setS_{[n]}( \vgam )}
       \Biggl( 
         \prod_{\substack{(i,j) \in \setR \times \setC \\ j = \sigma_1(i)}}
           \tgam(i,j)
       \Biggr)
       \cdot
       \underbrace
       {
         \Biggl( 
           \prod_{\substack{(i,j) \notin \setR \times \setC \\ j = \sigma_1(i)}}
             \tgam(i,j)
         \Biggr)
       }_{\overset{(e)}{=} \ M^{n-r}}
    \nonumber\\
      &\hspace*{3.6 cm} 
      \cdot
       \Biggl( 
         \prod_{\substack{(i,j) \in \setR \times \setC \\ j \neq \sigma_1(i)}}
           \frac{ 1 }
                { M - \tgam(i,j) }
       \Biggr)
      \cdot
       \underbrace{
         \Biggl( 
           \prod_{\substack{(i,j) \notin \setR \times \setC \\ j \neq \sigma_1(i)}}
             \frac{ 1 }
                  { M - \tgam(i,j) }
         \Biggr)
       }_{\overset{(f)}{=} \ \frac{1}{M^{n^2 - r^2 - (n - r)}}} \nonumber \\
    &= \frac{1}{M^{ 2r -  r^2}}
       \cdot
       \sum_{\sigma_1 \in \setS_{[n]}( \vgam )}
       \Biggl( 
         \prod_{\substack{(i,j) \in \setR \times \setC \\ j = \sigma_1(i)}}
           \tgam(i,j)
       \Biggr)
       \cdot
       \Biggl( 
         \prod_{\substack{(i,j) \in \setR \times \setC \\ j \neq \sigma_1(i)}}
           \frac{ 1 }
                { M - \tgam(i,j) }
       \Biggr) \nonumber \\
    &= \sum_{\sigma_1 \in \setS_{[n]}( \vgam )}
       \Biggl( 
         \prod_{\substack{(i,j) \in \setR \times \setC \\ j = \sigma_1(i)}}
           \gamma(i,j)
       \Biggr)
       \cdot
       \Biggl( 
         \prod_{\substack{(i,j) \in \setR \times \setC \\ j \neq \sigma_1(i)}}
           \frac{ 1 }
                { 1 - \gamma(i,j) }
       \Biggr) \nonumber \\
    &= \sum_{\sigma_1 \in \setS_{[n]}( \vgam )} 
         \frac{\prod_{i \in \setR}
                 \gamma(i,\sigma_1(i))
                 \, \cdot \,
                 \bigr( 1 - \gamma\isigmaoi \bigl)
              }
              {\prod\limits_{(i',j') \in \setR \times \setC}
                 \bigr( 1 - \gamma(i',j') \bigl)
              } \nonumber \\
    &= \sum_{\sigma_1 \in \setS_{[n]}( \vgam )} 
         \prod_{i \in \setR}
           \underbrace
           {
             \frac{
                   \gamma(i,\sigma_1(i))
                   \, \cdot \,
                   \bigr( 1 - \gamma\isigmaoi \bigl)
                  }
                  {
                   \Biggl(
                     \prod\limits_{(i',j') \in \setR \times \setC}
                     \bigr( 1 - \gamma(i',j') \bigl)
                   \Biggr)^{\!\!\! 1 / r}
                  }
           }_{\overset{(g)}{=} \ \hgamRCp\isigmaoi}
    \nonumber \\
    &= \sum_{\sigma_1 \in \setS_{[n]}( \vgam )} \ 
         \prod_{i \in \setR}
           \hgamRCp\isigmaoi
             \nonumber \\
    &\overset{(h)}{=}
       \perm(\hvgamRCp).
   \label{eqn: derivations of recursion of CBM}
\end{align}
% }

%***************************************************************************
%***************************************************************************

%***************************************************************************

The expression in~\eqref{eqn: derivations of recursion of CBM} at the top of
the next page is then obtained as follows:
\begin{itemize}

\item Step~$(a)$ follows from Definition~\ref{sec:1:lem:43}.

\item Step~$(b)$ follows from merging terms and using
  $\frac{(M-1)!}{M!} \! = \! \frac{1}{M}$.

\item Step~$(c)$ follows from the fact that the value of $(*)$ equals
  $\tgam(i,j)$ if $j = \sigma_1(i)$ and equals $\frac{1}{M - \tgam(i,j)}$ if
  $j \neq \sigma_1(i)$.\footnote{Note that, if $j \neq \sigma_1(i)$, then
    $\tgam(i,j) < M$, and with that $M - \tgam(i,j) > 0$.}
  The value of $(*)$ can be computed using the following relationships. \\
  If $j = \sigma_1(i)$ then
  \begin{align*}
    \tgam_{\sigma_1}(i,j)
      &= \tgam(i,j) - 1, \\
    (M-1) - \tgam_{\sigma_1}(i,j)
      &= M - \tgam(i,j).
  \end{align*}
  If $j \neq \sigma_1(i)$ then
  \begin{align*}
    \tgam_{\sigma_1}(i,j)
      &= \tgam(i,j), \\
    (M-1) - \tgam_{\sigma_1}(i,j)
      &= M - \tgam(i,j) - 1.
  \end{align*}

\item Step~$(d)$ follows from splitting the product over
  $(i,j) \in [n] \times [n]$ into two products: one over
  $(i,j) \in \setR \times \setC$ and one over
  $(i,j) \notin \setR \times \setC$.

\item Step~$(e)$ follows from the fact that this product has $n - r$ terms
  that each evaluate to $M$.

\item Step~$(f)$ follows from the fact that this product has $n^2 - r^2 -
  (n \! - \! r)$ terms that each evaluate to $\frac{1}{M}$.

\item Step~$(g)$ follows from the definition of the matrix~$\hvgamRCp$ in
  Definition~\ref{sec:1:def:10:part:2}.

\item Step~$(h)$ follows from the following considerations.

  If the sets $\setR$ and $\setC$ are empty, which happens exactly when
  $\vgam$ is a permutation matrix, we obtain
  \begin{align*}
    \sum_{\sigma_1 \in \setS_{[n]}( \vgam )} \ 
      \prod_{i \in \setR}
        \hgamRCp\isigmaoi 
      &= \!\!\!\!\!\! \sum_{\sigma_1 \in \setS_{[n]}( \vgam )} \!\!\!\!\!\!
           1
       = 1
       = \perm(\hvgamRCp),
  \end{align*}
  where the first equality follows from the convention that a product over an
  empty index set evaluates to $1$, where the second equality follows from
  $\setS_{[n]}( \vgam )$ containing only a single element, and where the third
  equality follows from the definition $\perm(\hvgamRCp) \defeq 1$ when the
  sets $\setR$ and $\setC$ are empty.

  If the sets $\setR$ and $\setC$ are non-empty, the result follows from
  \begin{align*}
    \hspace{0.5cm}&\hspace{-0.50cm}
    \sum_{\sigma_1 \in \setS_{[n]}( \vgam )} \ 
      \prod_{i \in \setR} 
        \hgamRCp\isigmaoi \\
      &= \sum_{\sigma_1 \in \setS_{\setR \to \setC}( \hvgamRCp )} \ 
           \prod_{i \in \setR} 
             \hgamRCp\isigmaoi
       = \perm(\hvgamRCp),
  \end{align*}
  where the first equality follows from the observation that there is a
  bijection between the set of permutations $\setS_{[n]}( \vgam )$ and the set
  of permutations $\setS_{\setR \to \setC}(\hvgamRCp)$, which is given by
  restricting the permutations in $\setS_{[n]}( \vgam )$ to take arguments
  only in $\setR$.

% \item Step~$(i)$ follows from the definition of the permanent of a square
%   matrix.\footnote{Note that the case where $\setR$ and $\setC$ are empty,
%     which happens exactly when $\vgam$ is a permutation matrix, is properly
%     handled by our conventions:
%     $\sum_{\sigma_1 \in \setS_{\setR \to \setC}( \hvgamRCp )} \ \prod_{i \in
%       \setR} \hgamRCp\isigmaoi = \sum_{\sigma_1 \in \setS_{\setR \to \setC}(
%       \hvgamRCp )( \vgam )} 1 = 1$, where the first equality follows from the
%   convention that a product over an empty index set evaluates to $1$, and
%   where the second equality follows from $\setS_{[n]}( \vgam )$ containing
%   only a single element.}

\end{itemize}
Finally, multiplying both sides of~\eqref{eqn: derivations of recursion of
  CBM} by $\frac{\CBM{n}(\vgam)}{\perm(\hvgamRCp)}$ yields the result promised
in the lemma statement.

%***************************************************************************
%***************************************************************************

\ifx\withclearpagecommands\x
\clearpage
\newpage
\fi

\section{Proof of Lemma~\ref{SEC:1:LEM:30}}
\label{apx:23}

%***************************************************************************
%***************************************************************************

\begin{align}
  \hspace{1cm}&\hspace{-1cm}\frac{1}{\CscSgen{M}{n}(\vgam)}
    \cdot
    \sum_{\sigma_1 \in \setS_{[n]}( \vgam )}
      \CscSgen{M-1}{n}(\vgam_{\sigma_1})
    \nonumber\\
    &\overset{(a)}{=} 
       M^{n \cdot M}
       \cdot 
       \frac{ 
             \prod_{i,j} 
               \tgam(i,j)!
            } 
            { \bigr( M! \bigr)^{\! 2n} }
       \cdot
       \sum_{\sigma_1 \in \setS_{[n]}( \vgam )} 
         (M-1)^{- n \cdot (M-1)}
         \cdot 
         \frac{ \bigr( (M-1)! \bigr)^{\! 2n} }
              { \prod_{i,j} 
                  \tgam_{\sigma_1}(i,j)!  
              } \nonumber \\
    &\overset{(b)}{=} 
       \underbrace
       {
         \biggl( \frac{M}{M-1} \biggr)^{\!\! n \cdot (M-1)}
       }_{= \ ( \chi(M) )^{n}}
       \cdot \ 
       \frac{1}{M^n}
       \cdot
       \sum_{\sigma_1 \in \setS_{[n]}( \vgam )}
         \prod_{i,j}
           \underbrace
           {
             \frac{ \tgam(i,j)! }
                  { \tgam_{\sigma_1}(i,j)! }
           }_{\defeq \ (*)}
             \nonumber \\
    &\overset{(c)}{=} 
       \bigl( \chi(M) \bigr)^{\! n}
       \cdot
       \frac{1}{M^n}
       \cdot
       \sum_{\sigma_1 \in \setS_{[n]}( \vgam )} 
         \prod_{\substack{i,j \\ j = \sigma_1(i)}}
           \tgam(i,j)
             \nonumber \\
    &= \bigl( \chi(M) \bigr)^{\! n}
       \cdot
       \frac{1}{M^n}
       \cdot
       \sum_{\sigma_1 \in \setS_{[n]}( \vgam )} 
         \prod_i
           \tgam\isigmaoi
             \nonumber \\
    &= \bigl( \chi(M) \bigr)^{\! n}
       \cdot
       \sum_{\sigma_1 \in \setS_{[n]}( \vgam )} 
         \prod_i
           \gamma\isigmaoi
             \nonumber \\
    &= \bigl( \chi(M) \bigr)^{\! n}
       \cdot 
       \perm( \vgam ).
         \label{eqn: derivations of recursion of CscSM}
\end{align}

%***************************************************************************
%***************************************************************************

We start by defining the matrices
\begin{align*}
  \tvgam
    &\defeq 
       M \cdot \vgam, \\ 
  \tvgam_{\sigma_1}
    &\defeq 
       (M-1) \cdot \vgam_{\sigma_1}.
\end{align*}
Note that $\tvgam_{\sigma_1} = \tvgam - \mP_{\sigma_1}$.

%***************************************************************************

The expression in~\eqref{eqn: derivations of recursion of CscSM}
at the top of the next page is then obtained as
follows:
\begin{itemize}

\item Step~$(a)$ follows from Definition~\ref{sec:1:lem:43}.

\item Step~$(b)$ follows from merging terms and using
  $\frac{(M-1)!}{M!} \! = \! \frac{1}{M}$.

\item Step~$(c)$ follows from the fact that the value of $(*)$ equals
  $\tgam(i,j)$ if $j = \sigma_1(i)$ and equals $1$ if $j \neq \sigma_1(i)$.
  The value of $(*)$ can be computed using the following relationships. \\
  If $j = \sigma_1(i)$ then
    \begin{align*}
      \tgam_{\sigma_1}(i,j)
        &= \tgam(i,j) - 1.
    \end{align*}
    If $j \neq \sigma_1(i)$ then
    \begin{align*}
      \tgam_{\sigma_1}(i,j)
        &= \tgam(i,j).
    \end{align*}

\end{itemize}
Finally, multiplying both sides of~\eqref{eqn: derivations of recursion of
  CscSM} by
$\frac{\CscSgen{M}{n}(\vgam_{\sigma_1})}{\bigl( \chi(M) \bigr)^{\! n} \cdot
  \perm( \vgam )}$ yields the result promised in the lemma statement.

%**************************************************************************
%**************************************************************************

\section{Proof of Lemma~\ref{LEMMA:RATIO:PERM:PERMBM:1}}
\label{app:proof:lemma:ratio:perm:permbM:1}

%**************************************************************************

We obtain
\begin{align*}
  \bigl( \permbM{M}(\mtheta) \bigr)^{\!M}
    &\overset{(a)}{=}
       \sum\limits_{ \vgam \in \GamMn } 
         \mtheta^{ M \cdot \vgam }
         \cdot
         \CBM{n}( \vgam ) \\
    &\overset{(b)}{=}
       \sum\limits_{ \vgam \in \GamMn } 
         \mtheta^{ M \cdot \vgam }
         \cdot
         \frac{\CBM{n}(\vgam)}
              {\CM{n}( \vgam )}
         \cdot
         \CM{n}( \vgam ) \\
    &\overset{(c)}{=}
       \sum\limits_{ \vgam \in \GamMn } 
         \mtheta^{ M \cdot \vgam }
         \cdot
         \frac{\CBM{n}(\vgam)}
              {\CM{n}( \vgam )}
         \cdot
         \sum_{ \vsigma_{[M]} \in \setS_{[n]}^{M} }
           \Bigl[ 
             \vgam = \bigl\langle \mP_{\sigma_{m}} \bigr\rangle_{m \in [M]}
           \Bigr] \\
    &= \sum_{ \vsigma_{[M]} \in \setS_{[n]}^{M} }
         \sum\limits_{ \vgam \in \GamMn } 
           \mtheta^{ M \cdot \vgam }
           \cdot
           \frac{\CBM{n}(\vgam)}
                {\CM{n}( \vgam )}
           \cdot
           \Bigl[ 
             \vgam = \bigl\langle \mP_{\sigma_{m}} \bigr\rangle_{m \in [M]}
           \Bigr] \\
    &= \sum_{ \vsigma_{[M]} \in \setS_{[n]}^{M} }
         \sum\limits_{ \vgam \in \GamMn } 
           \mtheta^{ M \cdot \langle \mP_{\sigma_{m}} \rangle_{m \in [M]}}
           \cdot
           \frac{\CBM{n}\Bigl(
                          \bigl\langle \mP_{\sigma_{m}} \bigr\rangle_{m \in [M]}
                        \Bigr)}
                {\CM{n}\Bigl(
                         \bigl\langle \mP_{\sigma_{m}} \bigr\rangle_{m \in [M]}
                       \Bigr)} \\
    &\hspace{5.5cm}
           \cdot
           \Bigl[ 
             \vgam = \bigl\langle \mP_{\sigma_{m}} \bigr\rangle_{m \in [M]}
           \Bigr] \\
    &= \sum_{ \vsigma_{[M]} \in \setS_{[n]}^{M} }
         \mtheta^{ M \cdot \langle \mP_{\sigma_{m}} \rangle_{m \in [M]}}
         \cdot
         \frac{\CBM{n}\Bigl(
                        \bigl\langle \mP_{\sigma_{m}} \bigr\rangle_{m \in [M]}
                      \Bigr)}
              {\CM{n}\Bigl(
                        \bigl\langle \mP_{\sigma_{m}} \bigr\rangle_{m \in [M]}
                      \Bigr)} \\
    &\hspace{4.5cm}
         \cdot
         \underbrace
         {
           \sum\limits_{ \vgam \in \GamMn } 
           \Bigl[ 
             \vgam = \bigl\langle \mP_{\sigma_{m}} \bigr\rangle_{m \in [M]}
           \Bigr]
         }_{= \ 1} \\
    &= \sum_{ \vsigma_{[M]} \in \setS_{[n]}^{M} }
         \mtheta^{ M \cdot \langle \mP_{\sigma_{m}} \rangle_{m \in [M]}}
         \cdot
         \frac{\CBM{n}\Bigl(
                        \bigl\langle \mP_{\sigma_{m}} \bigr\rangle_{m \in [M]}
                      \Bigr)}
              {\CM{n}\Bigl(
                        \bigl\langle \mP_{\sigma_{m}} \bigr\rangle_{m \in [M]}
                      \Bigr)} \\
    &= \sum_{ \vsigma_{[M]} \in \setS_{[n]}^{M} } \!
         \Biggl(
           \prod_{m \in [M]} \!
             \mtheta^{ \mP_{\sigma_{m}} } \!
         \Biggr)
         \cdot
         \frac{\CBM{n}\Bigl(
                        \bigl\langle \mP_{\sigma_{m}} \bigr\rangle_{m \in [M]}
                      \Bigr)}
              {\CM{n}\Bigl(
                       \bigl\langle \mP_{\sigma_{m}} \bigr\rangle_{m \in [M]}
                     \Bigr)},
\end{align*}
where step~$(a)$ follows from~\eqref{sec:1:eqn:190}, where step~$(b)$ is valid
because $\CM{n}(\vgam) \geq \CBM{n}(\vgam) > 0$ for
  $\vgam \in \GamMn$, as stated in the lemma statement, and where step~$(c)$
follows from~\eqref{sec:1:eqn:56}. Dividing the above expression for
$\bigl( \permbM{M}(\mtheta) \bigr)^{\!M}$ by $\bigl( \perm(\mtheta) \bigr)^{\!M}$, we
get
\begin{align*}
  \Biggl( 
    \frac{\permbM{M}(\mtheta)}
         {\perm(\mtheta)}
  \Biggr)^{\!\!\! M} 
    &= \sum_{ \vsigma_{[M]} \in \setS_{[n]}^{M} }
         \Biggl(
           \prod_{m \in [M]}
             \frac{\mtheta^{ \mP_{\sigma_{m}} }}
                  {\perm(\mtheta)}
         \Biggr)
         \cdot
         \frac{\CBM{n}\Bigl(
                        \bigl\langle \mP_{\sigma_{m}} \bigr\rangle_{m \in [M]}
                      \Bigr)}
              {\CM{n}\Bigl(
                        \bigl\langle \mP_{\sigma_{m}} \bigr\rangle_{m \in [M]}
                      \Bigr)} \\
    &= \sum_{ \vsigma_{[M]} \in \setS_{[n]}^{M} }
         \Biggl(
           \prod_{m \in [M]}
             \pmtheta(\sigma_m)
         \Biggr)
         \cdot
         \frac{\CBM{n}\Bigl(
                        \bigl\langle \mP_{\sigma_{m}} \bigr\rangle_{m \in [M]}
                      \Bigr)}
              {\CM{n}\Bigl(
                        \bigl\langle \mP_{\sigma_{m}} \bigr\rangle_{m \in [M]}
                      \Bigr)}.
\end{align*}

% XXXXX strange statement in step~$(b)$ XXXXX

%**************************************************************************
%**************************************************************************

\ifx\withclearpagecommands\x
\clearpage
\fi
\section[An Alternative proof]{An Alternative proof of the inequalities in~\eqref{SEC:1:EQN:147} in Theorem~\ref{TH:MAIN:PERMANENT:INEQUALITIES:1}}
\label{apx: alternative proof of perm geq permb}
\begin{itemize}

\item We prove the first inequality in~\eqref{SEC:1:EQN:147} by observing that
  \begin{align*}
    \frac{\perm(\mtheta)}
         {\permbM{M}(\mtheta)}
      &\overset{(a)}{\geq}
         \left( \!
           \sum_{ \vsigma_{[M]} \in \setS_{[n]}^{M} }
             \Biggl(
               \prod_{m \in [M]}
                 \pmtheta(\sigma_m)
             \Biggr)
         \right)^{\!\!\! -1/M} \\
      &= \Biggl( 
           \prod_{m \in [M]}
             \sum_{\sigma_m \in \setS_{[n]}}
               \pmtheta(\sigma_m)
         \Biggr)^{\!\!\! -1/M} \\
      &\overset{(b)}{=}
         \Biggl( 
           \prod_{m \in [M]}
             1
         \Biggr)^{\!\!\! -1/M} \\
      &= 1,
  \end{align*}
  where step~$(a)$ is obtained from~\eqref{eq:ratio:perm:permbM:1} by upper
  bounding the ratio $\CBM{n}(\ldots) / \CM{n}(\ldots)$ therein with the help
  of the first inequality in~\eqref{sec:1:eqn:58}, and where step~$(b)$ used
  the fact that $\pmtheta$ is a probability mass function.

\item We prove the second inequality in~\eqref{SEC:1:EQN:147} by observing that
  \begin{align}
    \frac{\perm(\mtheta)}
         {\permbM{M}(\mtheta)}
      &\overset{(a)}{\leq}
         \left( \!
           \sum_{ \vsigma_{[M]} \in \setS_{[n]}^{M} }
             \Biggl(
               \prod_{m \in [M]}
                 \pmtheta(\sigma_m)
             \Biggr)
             \cdot
             \bigl( 2^{n/2} \bigr)^{\! -(M-1)}
         \right)^{\!\!\! -1/M} \nonumber \\
      &= \bigl( 2^{n/2} \bigr)^{\!\! \frac{M-1}{M}}
         \cdot
         \Biggl( 
           \prod_{m \in [M]}
             \sum_{\sigma_m \in \setS_{[n]}}
               \pmtheta(\sigma_m)
         \Biggr)^{\!\!\! -1/M} \nonumber \\
      &\overset{(b)}{=}
         \bigl( 2^{n/2} \bigr)^{\!\! \frac{M-1}{M}}
         \cdot
         \Biggl(
           \prod_{m \in [M]}
             1
         \Biggr)^{\!\!\! -1/M} \nonumber \\
      &= \bigl( 2^{n/2} \bigr)^{\!\! \frac{M-1}{M}},
           \label{eq:ratio:perm:permbM:proof:upper:bound:1}
  \end{align}
  where step~$(a)$ is obtained from~\eqref{eq:ratio:perm:permbM:1} by lower
  bounding the ratio \\ $\CBM{n}(\ldots) / \CM{n}(\ldots)$ therein with the help
  of the second inequality in~\eqref{sec:1:eqn:58}, and where step~$(b)$ used
  the fact that $\pmtheta$ is a probability mass function.

\end{itemize}

\ifx\withclearpagecommands\x
\clearpage
\fi

\section[Another Alternative proof]{Another Alternative proof of the first inequalities in~\eqref{SEC:1:EQN:147} in Theorem~\ref{TH:MAIN:PERMANENT:INEQUALITIES:1}}
\label{apx: alternative prooof of perm geq permb: 2}

We start this alternative proof by proving the following lemma.
%----------------------------------------------------------------------------
\begin{lemma}\label{prop: a special case of Navin conjecture}
  For any integers $ M_{1} \in \sZpp $ and $ M_{2} \in \sZpp $, we have
  \begin{align}
      1 \leq 
      \frac{ 
        \bigl( \perm(\mtheta) \bigr)^{\! M_{1}} \cdot 
        \bigl( \permbM{M_{2}}(\mtheta) \bigr)^{\! M_{2}} 
      }{ \bigl( \permbM{M_{1}+M_{2}}(\mtheta) \bigr)^{\! M_{1}+M_{2}} }
      &\leq \bigl( 2^{n/2} \bigr)^{\! M_{1}}. 
      \label{eqn: a special case of Navin conjecture:2}
  \end{align}
\end{lemma}
%----------------------------------------------------------------------------
%----------------------------------------------------------------------------
\begin{proof}
  Consider an arbitrary integer $ M_{3} \in \sZ_{\geq 1} $. It holds that
  \begin{align}
    \perm(\mtheta) \cdot \bigl( \permbM{M_{3}}(\mtheta) \bigr)^{\! M_{3}}
    &\overset{(a)}{=}
    \Biggl(
          \sum_{\sigma_1 \in \setS_{[n]}( \mtheta )} 
          \mtheta^{ \mP_{\sigma_1} }
        \Biggr)
       \cdot 
       \sum_{\vgam \in \Gamma_{M_{3},n}(\mtheta)} 
       \mtheta^{ M_{3} \cdot \vgam }
       \cdot
      \CBgen{M_{3}}{n}(\vgam) 
    \nonumber\\
    &\overset{(b)}{=}
    \sum_{\vgam \in \Gamma_{M_{3}+1,n}(\mtheta)} 
           \mtheta^{ (M_{3}+1) \cdot \vgam }
           \cdot
           \sum_{\sigma_1 \in \setS_{[n]}( \vgam )}
           \CBgen{M_{3}}{n}
               \bigl( 
                 \vgam_{\sigma_1}
               \bigr)
    \nonumber\\
    &\overset{(c)}{=}
    \sum_{\vgam \in \Gamma_{M_{3}+1,n}(\mtheta)} 
           \mtheta^{ (M_{3}+1) \cdot \vgam }
           \cdot
           \perm(\hvgamRCp)
           \cdot
           \CBgen{M_{3}+1}{n}(\vgam)
   \nonumber\\
    &\overset{(d)}{\geq}
    \sum_{\vgam \in \Gamma_{M_{3}+1,n}(\mtheta)} 
           \mtheta^{ (M_{3}+1) \cdot \vgam }
           \cdot
           \CBgen{M_{3}+1}{n}(\vgam)
           \nonumber\\
    &\overset{(e)}{=}
    \bigl( \permbM{M_{3} +1}(\mtheta) \bigr)^{\! M_{3}+1}, 
    \label{eqn: a special case of Navin conjecture:3}
  \end{align} 
  where step $(a)$ follows from the definition of the permanent and the expression in~\eqref{sec:1:eqn:190} for $ M = M_{3} $,
  where step $(b)$ follows from the definition of $ \vgam_{\sigma_1} $ in Definition~\ref{sec:1:def:10:part:2:add} for $ M = M_{3} + 1 $,
  where step $(c)$ follows from Lemma~\ref{sec:1:lem:15:copy:2} for $ M = M_{3} + 1 $,
  where step $(d)$ follows from Lemma~\ref{sec:1:lem:17},
  and where step $(e)$ follows from the expression in~\eqref{sec:1:eqn:190} for $ M = M_{3} + 1 $.

  Because the inequality in~\eqref{eqn: a special case of Navin conjecture:3} holds for any $ M_{3} \in \sZ_{\geq 1} $, we have
  \begin{align*}
    \bigl( \permbM{M_{1}+M_{2}}(\mtheta) \bigr)^{\! M_{1}+M_{2}}
    &\leq \perm(\mtheta) \cdot 
    \bigl( \permbM{M_{1}+M_{2}-1}(\mtheta) \bigr)^{\! M_{1}+M_{2}-1}
    \nonumber\\
    &\leq \bigl( \perm(\mtheta) \bigr)^{\! 2} 
    \cdot \bigl( \permbM{M_{1}+M_{2}-2}(\mtheta) \bigr)^{\! M_{1}+M_{2}-2}
    \nonumber\\
    &\ \cdots
    \nonumber\\
    &\leq \bigl( \perm(\mtheta) \bigr)^{\! M_{1}} 
    \cdot \bigl( \permbM{M_{2}}(\mtheta) \bigr)^{\! M_{2}}.
  \end{align*}

  A similar line of reasoning can be used to prove the upper bound in~\eqref{eqn: a special case of Navin conjecture:2}.
  The details are omitted.
\end{proof}
%----------------------------------------------------------------------------
The inequalities~\eqref{SEC:1:EQN:147} can now be proven by observing that they are special cases of Lemma~\ref{prop: a special case of Navin conjecture}. Namely, they follow from~\eqref{eqn: a special case of Navin conjecture:2} for $M_{1} = M-1$ and $M_{2} = 1$, along with using $ \perm(\mtheta) = \permbM{1}(\mtheta) $.

\ifx\withclearpagecommands\x
\clearpage
\fi

\section[The Alternative Proof of Proposition~\ref{PROP:RATIO:PERM:PERMBM:2:1}]{Details of the Alternative Proof of 
               Proposition~\ref{PROP:RATIO:PERM:PERMBM:2:1}}
\label{app:alt:proof:prop:ratio:perm:permbM:2:1}

%**************************************************************************

\medskip

Toward establishing \eqref{eq:CBtwo:result:1} and~\eqref{eq:Ctwo:result:1}, we 
start by studying the matrix
\begin{align*}
  \frac{1}{2} 
  \bigl( 
  \mP_{\sigma_1} \! + \! \mP_{\sigma_2} 
  \bigr)
\end{align*}
for arbitrary $\sigma_1, \sigma_2 \in \setS_{[n]}$.

%***************************************************************************

\begin{lemma}
  \label{lemma:vgam:graph:properties:1}

  Let $\sigma_1, \sigma_2 \in \setS_{[n]}$ be fixed. Define the matrix
  \begin{align*}
    \vgam
      &\defeq
         \frac{1}{2} 
         \bigl( 
           \mP_{\sigma_1} \! + \! \mP_{\sigma_2} 
          \bigr).
  \end{align*}
  Let the graph $\sfG(\sigma_1,\sigma_2)$ be a bipartite graph with two times
  $n$ vertices: $n$ vertices on the LHS and $n$ vertices on the RHS, each of
  them labeled by $[n]$. Let there be an edge $(i,j)$ connecting the $i$-th
  vertex on the LHS with the $j$-th vertex on the RHS if
  $\gamma(i,j) = \frac{1}{2}$. Let $c\bigl( \sfG(\sigma_1,\sigma_2) \bigr)$ be
  the number of cycles in $\sfG(\sigma_1,\sigma_2)$.

  The matrix $\vgam$ has the following properties:
  \begin{enumerate}
    
  \item $\vgam \in \Gamtwon$.

  \item For every $i \in [n]$, the $i$-th row $\vgam(i,:)$ of $\vgam$ has one
    of the following two possible compositions:
    \begin{itemize}
  
    \item One entry is equal to $1$; $n \! - \! 1$ entries are equal to~$0$.
  
    \item Two entries are equal to $\frac{1}{2}$; $n \! - \! 2$ entries are
      equal to~$0$.

    \end{itemize}

  \end{enumerate}
  The graph $\sfG(\sigma_1,\sigma_2)$ has the following properties:
  \begin{enumerate}
    
  \item The edge set of $\sfG(\sigma_1,\sigma_2)$ forms a collection of
    disjoint cycles of even length.

  \item Every cycle in $\sfG(\sigma_1,\sigma_2)$ of length $2L$ corresponds to
    a cycle of length $L$ in the cycle notation of
    $\sigma_1 \circ \sigma_2^{-1}$.

  \item $c\bigl( \sfG(\sigma_1,\sigma_2) \bigr) = c(\sigma_1,\sigma_2)$.

  \end{enumerate}

\end{lemma}

%***************************************************************************

\begin{proof}
  The derivation of the properties of $\vgam$ is rather straightforward and is
  therefore omitted. The properities of $\sfG(\sigma_1,\sigma_2)$ follow from the following
  observations:
  \begin{itemize}
    
  \item Because $\sfG(\sigma_1,\sigma_2)$ is a bipartite graph, all cycles
    must have even length.

  \item Let $i \in [n]$. If $\sigma_1(i) = \sigma_2(i)$, then the $i$-th
    vertex on the LHS has no indicent edges.
    
  \item Let $i \in [n]$. If $\sigma_1(i) \neq \sigma_2(i)$, then the $i$-th
    vertex on the LHS has degree two and is part of a cycle.

  \item Let $j \in [n]$. If $\sigma_1^{-1}(j) = \sigma_2^{-1}(j)$, or,
    equivalently, if $\sigma_1\bigl( \sigma_2^{-1}(j) \bigr) = j$, then the
    $j$-th vertex on the RHS has no indicent edges.
    
  \item Let $j \in [n]$. If $\sigma_1^{-1}(j) \neq \sigma_2^{-1}(j)$, or,
    equivalently, if $\sigma_1\bigl( \sigma_2^{-1}(j) \bigr) \neq j$, then the
    $j$-th vertex on the RHS has degree two and is part of a cycle.

  \end{itemize}
\end{proof}

%***************************************************************************

\begin{figure*}[t]
  \begin{centering}
    \subfloat[]{\input{figures/example_sfG_1.tex}}
    \hspace{0.75cm}
    \subfloat[]{\input{figures/example_sfG_2.tex}}

    \subfloat[]{\input{figures/example_sfG_3.tex}}
    \hspace{0.75cm}
    \subfloat[]{\input{figures/example_sfG_4.tex}}
    \medskip
    \caption[Graphs discussed in Examples~\ref{example:sfG:1}
      and~\ref{example:sfG:2}.]{Graphs discussed in Examples~\ref{example:sfG:1}
      and~\ref{example:sfG:2}. (The red color is used for edges based on
      $\sigma_1$; the blue color is used for edges based on $\sigma_2$.)}
    \label{fig:example:sfG:1}
  \end{centering}
\end{figure*}
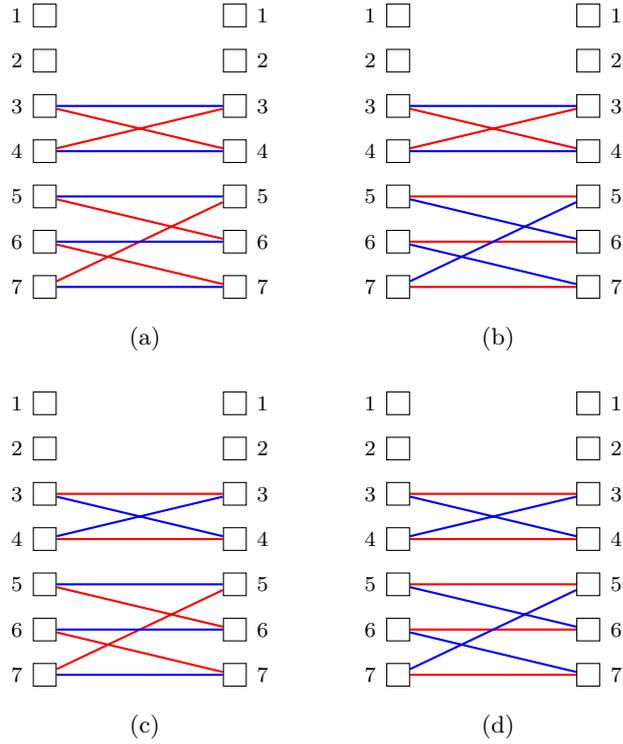

%***************************************************************************

\begin{example}
  \label{example:sfG:1}

  Let $\sigma_1: [7] \to [7]$ defined by $\sigma_1(1) = 1$, $\sigma_1(2) = 2$,
  $\sigma_1(3) = 4$, $\sigma_1(4) = 3$, $\sigma_1(5) = 6$, $\sigma_1(6) = 7$,
  $\sigma_1(7) = 5$. Let $\sigma_2: [7] \to [7]$ defined by $\sigma_2(i) = i$
  for all $i \in [7]$. Note that $\sigma_1 \circ \sigma_2^{-1}$ equals the
  permuation $\sigma$ in Footnote~\ref{footnote:cycle:notation:example:1}. Let
  \begin{align}
    \vgam
      &\defeq
         \frac{1}{2} 
         \bigl( 
           \mP_{\sigma_1} \! + \! \mP_{\sigma_2} 
          \bigr)
       = \frac{1}{2}
           \begin{pmatrix}
             2 & 0     & 0   & 0   & 0   & 0   & 0   \\
             0      & 2 & 0   & 0   & 0   & 0   & 0   \\
             0      & 0     & 1 & 1 & 0   & 0   & 0   \\
             0      & 0     & 1 & 1 & 0   & 0   & 0   \\
             0      & 0     & 0   & 0   & 1 & 1 & 0   \\
             0      & 0     & 0   & 0   & 0   & 1 & 1 \\
             0      & 0     & 0   & 0   & 1 & 0   & 1
           \end{pmatrix}
           \label{eq:example:sfG:1:gamma:1}
  \end{align}
  and define the graph $\sfG(\sigma_1,\sigma_2)$ as in
  Lemma~\ref{lemma:vgam:graph:properties:1}. The resulting graph
  $\sfG(\sigma_1,\sigma_2)$ is shown in Fig.~\ref{fig:example:sfG:1}(a). Note
  that the graph $\sfG(\sigma_1,\sigma_2)$ has two cycles, \ie,
  $c\bigl( \sfG(\sigma_1,\sigma_2) \bigr) = 2$.
  \eexample
\end{example}

%***************************************************************************

\begin{lemma}
  \label{lemma:sfGtwo:properties:1}

  Let $\vgam \in \Gamtwon$ be fixed. Define the graph
  $\sfG(\sigma_1,\sigma_2)$ as in
  Lemma~\ref{lemma:vgam:graph:properties:1}. The number of pairs
  $(\sigma_1, \sigma_2) \in (\setS_{[n]})^2$ such that\footnote{Note that the
    ordering in $(\sigma_1, \sigma_2)$ matters, \ie, $(\sigma_1, \sigma_2)$
    and $(\sigma_2, \sigma_1)$ are considered to be distinct pairs if
    $\sigma_1 \neq \sigma_2$.}
  \begin{align*}
    \vgam
      &= \frac{1}{2} 
         \bigl( 
           \mP_{\sigma_1} \! + \! \mP_{\sigma_2} 
          \bigr)
  \end{align*}
  equals $2^{c(\sfG(\sigma_1,\sigma_2))}$.
\end{lemma}

%***************************************************************************

\begin{proof}
  For every cycle in $\sfG(\sigma_1,\sigma_2)$, there are two ways of choosing
  the corresponding function values of $\sigma_1$ and $\sigma_2$. Overall,
  because there are $c\bigl( \sfG(\sigma_1,\sigma_2) \bigr)$ cycles in
  $\sfG(\sigma_1,\sigma_2)$, there are $2^{c(\sfG(\sigma_1,\sigma_2))}$ ways
  of choosing $\sigma_1$ and $\sigma_2$.
\end{proof}

%***************************************************************************

\begin{example}
  \label{example:sfG:2}

  Consider the matrix $\vgam$ in~\eqref{eq:example:sfG:1:gamma:1}. Because
  $c\bigl( \sfG(\sigma_1,\sigma_2) \bigr) = 2$, there are $2^2$ choices for
  $(\sigma_1, \sigma_2) \in (\setS_{[n]})^2$ such that
  $\vgam = \frac{1}{2} \bigl( \mP_{\sigma_1} \! + \! \mP_{\sigma_2}
  \bigr)$. They are shown in Figs.~\ref{fig:example:sfG:1}(a)--(d). \\
  \mbox{} \eexample
\end{example}

%***************************************************************************

\begin{lemma}
  \label{lemma:CMtwo:properties:1}

  It holds that
  \begin{align*}
    \CM{n}
      \Bigl(
        \frac{1}{2} 
        \bigl( 
          \mP_{\sigma_1} \! + \! \mP_{\sigma_2} 
        \bigr)
      \Bigr)
      &= 2^{c(\sigma_1,\sigma_2)}, 
           \qquad \sigma_1, \sigma_2 \in \setS_{[n]}.
  \end{align*}
\end{lemma}

%***************************************************************************

\begin{proof}
  Let
  $\vgam \defeq \frac{1}{2} \bigl( \mP_{\sigma_1} \! + \! \mP_{\sigma_2}
  \bigr) \in \Gamtwon$ and define the graph $\sfG(\sigma_1,\sigma_2)$ as in
  Lemma~\ref{lemma:vgam:graph:properties:1}. We obtain
  \begin{align*}
    \CM{n}(\vgam)
      &\overset{(a)}{=}
         \sum_{ (\sigma'_1, \sigma'_2) \in (\setS_{[n]})^2 }
           \biggl[ 
             \vgam = \frac{1}{2} \bigl( 
                                   \mP_{\sigma'_1} \! + \! \mP_{\sigma'_2} 
                                 \bigr)
           \biggr] \\
      &\overset{(b)}{=}
         2^{c(\sfG(\sigma_1,\sigma_2))} \\
      &\overset{(c)}{=}
         2^{c(\sigma_1,\sigma_2)},
  \end{align*}
  where step~$(a)$ follows from~\eqref{sec:1:eqn:56}, where step~$(b)$ follows
  from Lemma~\ref{lemma:sfGtwo:properties:1} and where step~$(c)$ follows from
  Lemma~\ref{lemma:vgam:graph:properties:1}.
\end{proof}

%***************************************************************************

\begin{lemma}
  \label{lemma:CBMtwo:properties:1}

  It holds that
  \begin{align*}
    \CBM{n}
      \biggl(
        \frac{1}{2} 
        \bigl( 
          \mP_{\sigma_1} \! + \! \mP_{\sigma_2} 
        \bigr)
      \biggr)
      &= 1, 
           \qquad \sigma_1, \sigma_2 \in \setS_{[n]}.
  \end{align*}
\end{lemma}

%***************************************************************************

\begin{proof}
  Let
  $\vgam \defeq \frac{1}{2} \bigl( \mP_{\sigma_1} \! + \! \mP_{\sigma_2}
  \bigr) \in \Gamtwon$. We have
  \begin{align*}
    \CBtwo{n}(\vgam) 
      &\overset{(a)}{=}
         (2!)^{ 2n -  n^2} 
           \cdot
           \prod_{i,j}
             \frac{ \bigl(2 - 2 \gamma(i,j) \bigr)! }
                  { \bigl(2 \gamma(i,j) \bigr)! } \\
      &= \prod_i
           \underbrace
           {
             \Biggl(
               (2!)^{ 2 -  n} 
               \cdot
               \prod_j
                 \frac{ \bigl(2 - 2 \gamma(i,j) \bigr)! }
                      { \bigl(2 \gamma(i,j) \bigr)! }
             \Biggr)
           }_{\overset{(b)}{=} \ 1} \\
      &= 1,
  \end{align*}
  where step~$(a)$ follows from Definition~\ref{sec:1:lem:43} and where
  step~$(b)$ follows from observing that for every $i \in [n]$, the $i$-th row
  $\vgam(i,:)$ of $\vgam$ has one of the following two possible compositions:
  \begin{itemize}
  
  \item one entry is equal to $1$ and $n - 1$ entries are equal to $0$;
  
  \item two entries are equal to $\frac{1}{2}$ and $n - 2$ entries are equal
    to $0$.

  \end{itemize}
  (See also Lemma~\ref{lemma:vgam:graph:properties:1}.) In both cases, one can
  verify that the expression in big parentheses evaluates to $1$.
\end{proof}

  \ifx\sectionheaderonnewpage\x
  \clearpage
  \fi

  \chapter{Proofs in Chapter~\ref{CHAPT: LCT}}

  \section{Details of the LCT for S-NFGs}
  \label{apx:LCT SNFGs} 
  \input{appendices/lct_snfg.tex}

  %------------------------------------------------------------------------

  \section{Proof of Proposition~\ref{SEC:LCT:PROP:1}}
  \label{apx:property of SNFGs}
  \input{appendices/props_lct_snfg.tex}

  %------------------------------------------------------------------------

  \section{Proof of Proposition~\ref{PROP:DENFG:LCT:1}}
  \label{apx:LCT DENFGs}
  \input{appendices/lct_denfg.tex}

  \chapter{Proofs in Chapter~\ref{CHAPT:SST}}

  \section{Proof of Lemma~\ref{SEC:SST:LEM:6}}
  \label{apx:alternative expression of ZBM by Pe}

  The partition function of $ \hgraphNavg $, as defined in Definition~\ref{sec:SST:def:2}, satisfies
  %----------------------------------------------------------------------------
  \begin{align*}
    Z(\hgraphNavg)
    &\overset{(a)}{=} \sum_{\vsigma}
    p_{\vSigma}(\vsigma) \cdot Z(\hgraphN_{M,\vsigma})
    \nonumber\\
    &\overset{(b)}{=}  \sum_{\vsigma}
    p_{\vSigma}(\vsigma) \cdot Z(\hgraphPsig)
    \nonumber\\
    &\overset{(c)}{=} 
    \sum_{\vsigma}
    p_{\vSigma}(\vsigma) \cdot
    \sum_{\vx_{[M]}}
    \Biggl( \prod_{m \in [M]} \prod_f f(\vx_{\setpff,m}) \Biggr)
    \cdot \prod_{e} P_{e,\sigma_e}(\vx_{\efi,[M]}, \vx_{\efj,[M]})
    \nonumber\\
    &\overset{(d)}{=} 
    \sum_{\vx_{[M]}}
    \Biggl( \prod_{m \in [M]} \prod_f f(\vx_{\setpff,m}) \Biggr)
    \cdot \prod_{e} \Biggl( 
      \sum_{\sigma_e} p_{\Sigmae}(\sigma_e) 
      \cdot P_{e,\sigma_e}(\vx_{\efi,[M]}, \vx_{\efj,[M]})
    \Biggr)
    \nonumber\\
    &\overset{(e)}{=} 
    \sum_{\vx_{[M]}}
    \Biggl( \prod_{m \in [M]} \prod_f f(\vx_{\setpff,m}) \Biggr)
    \cdot \prod_{e} P_{e}(\vx_{\efi,[M]}, \vx_{\efj,[M]}),
  \end{align*}
  %----------------------------------------------------------------------------
  where step $(a)$ follows from Lemma~\ref{sec:SST:lem:1},
  where step $(b)$ follows from~\eqref{sec:SST:eqn:33}, where step $(c)$ follows from
  the definition of $ \hgraphPsig $ in Definition~\ref{sec:SST:def:1}, \ie, 
  the partition function of $ \hgraphPsig $ is given by
  %------------------------------------------------------------------------
  \begin{align*}
    Z(\hgraphPsig) =  \sum_{\vx_{[M]}}
    \Biggl( \prod_{m \in [M]} \prod_f f(\vx_{\setpff,m}) \Biggr)
    \cdot \prod_{e} P_{e,\sigma_e}(\vx_{\efi,[M]}, \vx_{\efj,[M]}),
  \end{align*}
  %------------------------------------------------------------------------
  where step $(d)$ follows from the decomposition of $ p_{\vSigma}(\vsigma) = \prod_{e}p_{\Sigmae}(\sigma_e) $, as shown in~\eqref{sec:SST:eqn:31}, and $ \vsigma = (\sigma_e)_{e \in \setEfull} $ as defined in Definition~\ref{def:graph:cover:construction:1}, and where step $(e)$ follows from the definition of $ P_{e} $ in Definition~\ref{sec:SST:def:4}.

  \ifx\sectionheaderonnewpage\x
  \clearpage
  \fi

  \ifx\sectionheaderonnewpage\x
  \clearpage
  \fi

  %------------------------------------------------------------------------
  \section{Proof of Lemma~\ref{SEC:SST:PROP:2}}
  \label{apx:SST}  
  \input{appendices/sst.tex}
  %------------------------------------------------------------------------

  \ifx\sectionheaderonnewpage\x
  \clearpage
  \fi

  %------------------------------------------------------------------------
  \section{Proof of Proposition~\ref{SEC:SST:PROP:1}}
  \label{apx:property of SST}
  \input{appendices/props_sst.tex}
  %------------------------------------------------------------------------

  \ifx\sectionheaderonnewpage\x
  \clearpage
  \fi

  \chapter{Proof of Theorem~\ref{SEC:CHECKCON:THM:1} in Chapter~\ref{CHAPT: GCT DENFG}}
  \label{apx:check_graph_thm}
\input{appendices/check_graph_thm}

\end{appendices}

%% file: appendices/Props_ZN_DENFG.tex
We first consider the case where $\graphN$ is a strict-sense DE-NFG,
afterwards the case where $\graphN$ is a weak-sense DE-NFG.

%***************************************************************************
Let $\graphN$ be a strict-sense DE-NFG. Fix some $f \in \setF$. Because
$\matr{C}_f$ is a PSD matrix, the eigenvalue decomposition of $ \matr{C}_f  $ in~\eqref{sec:DENFG:eqn:9} implies that the eigenvalue $ \lambda_{f}(\ell_f) $ associated with the right eigenvector $ \vect{u}_{f,\ell_f} = \bigl( u_{f,\ell_f}(\vx_{\setpf}) \bigr)_{ \! \vx_{\setpf} }  $ satisfies
%---------------------------------------------------------------------------
\begin{align*}
   \lambda_{f}(\ell_f) \in \sR_{\geq 0}, \qquad f \in \setF,\,\ell_f \in \set{L}_f.
\end{align*}
%---------------------------------------------------------------------------
By the decomposition of $ f $ in~\eqref{sec:DENFG:eqn:5}, we get
\begin{align*}
  Z(\graphN) 
    &= \sum_{\tvx} 
         \prod_{f \in \setF}
           f(\tvx_{f}) \\
    &= \sum_{\tvx}
         \prod_{f \in \setF}
           \left( 
            \sum_{\ell_f \in \set{L}_f} 
              \lambda_{f}(\ell_f) \cdot
             u_{f,\ell_f}(\vx_{\setpf}) 
             \cdot
             \overline{ u_{f,\ell_f}(\vx_{\setpf}') }
           \right) \\
    &=  \sum_{\ell_{f_1} \in \set{L}_{f_1}} 
        \cdots
        \sum_{\ell_{f_{|\setF|} \in \set{L}_{|\setF|}}}
        \Biggl( \prod_{f \in \setF} \lambda_{f}(\ell_f) \Biggr)
        \cdot 
        \Biggl( 
            \sum_{ \vx_{\setEfull} \in \setx{\setEfull} }
            \prod_{f \in \setF} u_{f,\ell_f}(\vx_{\setpf}) 
        \Biggr)
        \nonumber\\
        &\hspace{3.6 cm}\cdot 
        \Biggl( 
            \sum_{ \vx_{\setEfull}' \in \setx{\setEfull} }
            \prod_{f \in \setF} 
            \overline{u_{f,\ell_f}(\vx_{\setpf}')}
        \Biggr)
        \\
    &=  \sum_{\ell_{f_1} \in \set{L}_{f_1}} 
        \cdots
        \sum_{\ell_{f_{|\setF|} \in \set{L}_{|\setF|}}}
        \Biggl( 
        \underbrace{
        \prod_{f \in \setF} 
            \lambda_{f}(\ell_f)
        }_{\geq 0} \Biggr)
        \cdot 
        \underbrace{
        \Biggl| 
            \sum_{ \vx_{\setEfull} \in \setx{\setEfull} }
            \prod_{f \in \setF} u_{f,\ell_f}(\vx_{\setpf}) 
        \Biggr|^{2}
        }_{\geq 0}
        \nonumber\\
        &\geq 0,
\end{align*}
where $ \vx_{\setEfull} \defeq ( \xe )_{e \in \setEfull} $, 
$ \vx_{\setEfull}' \defeq  ( \xe' )_{e \in  \setEfull}   $,
and $ \setx{\setEfull} \defeq \prod_{e} \setx{e} $.
% which shows that $Z(\graphN) \in \sRp$.

Let $\graphN$ be a weak-sense DE-NFG. Fix some $f \in \setF$. Because
$\matr{C}_f$ is a Hermitian matrix, the eigenvalue decomposition of $ \matr{C}_f  $ in~\eqref{sec:DENFG:eqn:9} implies that
%---------------------------------------------------------------------------
\begin{align*}
   \lambda_{f}(\ell_f) \in \sR,  \qquad f \in \setF,\,\ell_f \in \set{L}_f.
\end{align*}
%---------------------------------------------------------------------------
Similar calculations as above lead to
\begin{align*}
  Z(\graphN) 
    &= \sum_{\ell_{f_1} \in \set{L}_{f_1}} 
        \cdots
        \sum_{\ell_{f_{|\setF|} \in \set{L}_{|\setF|}}}
        \Biggl( 
        \underbrace{
        \prod_{f \in \setF} 
            \lambda_{f}(\ell_f)
        }_{\in \sR}
        \Biggr)
        \cdot 
        \underbrace{
        \Biggl| 
            \sum_{ \vx_{\setEfull} \in \setx{\setEfull} }
            \prod_{f \in \setF} u_{f,\ell_f}(\vx_{\setpf}) 
        \Biggr|^{2}
        }_{\geq 0}
        \in \sR.
\end{align*}

%***************************************************************************
%***************************************************************************

\ifx\sectionheaderonnewpage\x
\clearpage
\fi

%% file: appendices/sufficient_condition_pdmessages.tex
The proof can be viewed as a generalization of the main idea in the proof~\cite[Proposition 4]{Yedidia2005}.

For each $ e \in \setpf $ and $ f \in \setF $,
the fixed-point message vector $ \vmu_{\etof} $ 
satisfies the following properties.
%----------------------------------------------------------------------------
\begin{itemize}
  \item By Lemma~\ref{sec:DENFG:lem:1}, we know that 
  $ \matr{C}_{\mu_{\etof}} \in \setPSD{\setxe} $ and it has an eigenvalue decomposition:
  \begin{align}
    \matr{C}_{\mu_{\etof}}
    = \sum_{ \ell_{e} \in \set{L}_{e} } 
    \matr{u}_{ \mu_{\etof}, \ell_{e} } 
    \cdot \lambda_{ \mu_{\etof} }( \ell_{e} )
    \cdot \bigl( \matr{u}_{ \mu_{\etof}, \ell_{e} } \bigr)^{\Herm},
    \label{eqn: decomposition of the Choi matrix for SPA message}
  \end{align}
  where $ \set{L}_{e} $ is a finite set with $ |\set{L}_{e}| = |\setxe| $,
  where for each $ \ell_{e} \in \set{L}_{e} $, the column vector
  $ \matr{u}_{ \mu_{\etof}, \ell_{e} }
  = \bigl( u_{ \mu_{\etof}, \ell_{e} }( \xe ) 
  \bigr)_{\! \xe \in \setxe } \in \sC^{|\setxe|} $ 
  is the right-eigenvector associated with the eigenvalue  
  $ \lambda_{ \mu_{\etof} }( \ell_{e} ) \in \sR_{\geq 0} $.
  In particular, the right-eigenvectors of $ \matr{C}_{\mu_{\etof}} $ satisfy
  \begin{align*}
    \bigl( \matr{u}_{ \mu_{\etof}, \ell_{e}' } \bigr)^{\!\Herm}
    \cdot 
    \matr{u}_{ \mu_{\etof}, \ell_{e} }
    = \bigl[ \ell_{e}' \! = \! \ell_{e} \bigr], \qquad 
    \ell_{e},\ell_{e}' \in \set{L}_{e}.
  \end{align*}

  \item By the SPA update rules in Definition~\ref{sec:DENFG:def:2},
  we know that 
  \begin{align}
    1 = \sum_{ \txe } \mu_{\etof}(\txe) 
    = \bm{1}^{\tran} \cdot \matr{C}_{\etof} \cdot \bm{1},
    \label{eqn: property of choi matrix wrt normalization}
  \end{align}
  where $ \bm{1} $ is the all-one column vector of length $ |\setxe| $.
  These equalities imply that $ \matr{C}_{\etof} $ has at least one positive-valued eigenvalue, 
  \ie, there exists an argument $ \ell_{e,\mathrm{p}} \in \set{L}_{e} $ 
  such that  
  \begin{align}
    \lambda_{ \mu_{\etof} }( \ell_{e,\mathrm{p}} )
    \in \sR_{>0}. \label{eqn: property of ell e p}
  \end{align}
  If such $ \ell_{e,\mathrm{p}} $ does not exists, we have $ \matr{C}_{\etof} = \bm{0} $, which contradicts to the equalities in~\eqref{eqn: property of choi matrix wrt normalization}.

  \item Suppose that $ e = (f,f') $. 
  Because of $ \matr{C}_{\mu_{\etof}} \in \setPSD{\setxe} $, 
  we know that the following normalization constant
  \begin{align*}
    \kappa_{\ef}
    = 
    \sum_{\tvx_{\setpf'}}
    f'\bigl( \tvx_{\setpf'} \bigr)
    \cdot \prod_{e' \in \setpf' \setminus e}
    \mu_{e', f'}(\tx_{e'}),
  \end{align*}
  as defined in Definition~\ref{sec:DENFG:def:2}, is nonnegative-valued.
  Based on the definition of the SPA fixed-point message in Definition~\ref{sec:DENFG:def:2},
  we know that $ \kappa_{\ef} \neq 0 $, which implies 
  \begin{align}
     \kappa_{\ef} \in \sR_{>0}.\label{eqn: positive of normalization constant}
  \end{align}
\end{itemize}
%----------------------------------------------------------------------------
% implies
% \begin{align*}
%   \matr{C}_{ \mu_{\efi} }
%   = 
%   \frac{ 
%     \sum_{\tvx_{\setpf}:}
%     f_{i}\bigl( \tvx_{\setpfi} \bigr) 
%     \cdot \prod_{e' \in \setpfi \setminus e}
%     \mu_{\epfi}^{(t-1)}(\tx_{e'})
%   }{ 
%     \sum_{\tvx_{\setpf}}
%     f_{i}\bigl( \tvx_{\setpfi} \bigr) 
%     \cdot \prod_{e' \in \setpfi \setminus e}
%     \mu_{\epfi}^{(t-1)}(\tx_{e'})
%   }
%   \sum_{\ell_f \in \set{L}_f} 
%   \lambda_{f}(\ell_{f}) \cdot \vect{u}_{f}(\ell_f) 
%   \cdot \bigl( \vect{u}_{f}(\ell_f) \bigr)^{\!\Herm}
% \end{align*}
Without loss of generality, we consider the following setup.
%----------------------------------------------------------------------------
\begin{itemize}
  \item The edge $ 1 $ connects function nodes $ f_{1} $ and $ f_{2} $.

  \item $ \setpf_{1} = \{1,2,\ldots,|\setpf_{1}|\} $;

  \item $ \vx_{\setpf_{1}} = ( x_{1}, x_{2},\ldots,x_{|\setpf_{1}|} ) $;
\end{itemize}
%----------------------------------------------------------------------------
We first prove that $ \matr{C}_{\mu_{\onetoftwo}} $ is a positive definite matrix by contradiction.
Suppose that for the matrix $ \matr{C}_{\mu_{\onetoftwo}} $, 
there exists an argument $ \ell_{1,0} \in \set{L}_{1} $
such that 
\begin{align*}
  \lambda_{ \mu_{\onetoftwo} }( \ell_{1,0} ) = 0.
  % \label{eqn: zero of eigenvalue}
\end{align*}
Then we obtain a contradiction:
\begin{align}
  0
  &=
  \lambda_{ \mu_{\onetoftwo} }( \ell_{1,0} )
  \nonumber\\
  &\overset{(a)}{=}
  \bigl( \matr{u}_{ \mu_{\onetoftwo},\ell_{1,0} } \bigr)^{\Herm} 
  \cdot \matr{C}_{\mu_{\onetoftwo}}
  \cdot 
  \matr{u}_{ \mu_{\onetoftwo},\ell_{1,0} }
  % \nonumber\\
  % &=
  % \kappa_{\onetoftwo}
  % \cdot
  % \bigl( \matr{u}_{ \mu_{\onetoftwo},\ell_{1,0} } \bigr)^{\Herm}
  % \cdot 
  % \matr{u}_{ \mu_{\onetoftwo},\ell_{1,0} }
  \nonumber\\
  &\overset{(b)}{=}
  \kappa_{\onetoftwo}
  \cdot
  \sum_{ 
    \ell_{2},\ldots,\ell_{|\setpf_{1}|}
  }
  \Biggl( 
    \prod_{ e \setminus 1 }
    \underbrace{ \lambda_{\mu_{\etofone}}(\ell_{e}) }_{\overset{(c)}{\geq 0} }
  \Biggr)
  % \nonumber\\
  % &\qquad 
  \label{eqn: expanding SPA update rule:2}\\
  &\hspace{3 cm} \cdot 
  \underbrace{
    \Bigr(
      \bigl( \matr{u}_{ \mu_{\onetoftwo},\ell_{1,0} } \bigr)^{\Herm} 
      \mathop{\otimes}_{ e \setminus 1 }
      \bigl( \matr{u}_{\etofone,\ell_{e}} \bigr)^{\tran} 
    \Bigl)
    \cdot
    \matr{C}_{f_{1}}
    \cdot
    \Bigr(
      \matr{u}_{ \mu_{\onetoftwo},\ell_{1,0} } 
      \mathop{\otimes}_{ e \setminus 1 }
      \overline{ \matr{u}_{\etofone,\ell_{e}} }
    \Bigl)
  }_{\overset{(d)}{>0}}
  \label{eqn: expanding SPA update rule:1}\\
  &\overset{(e)}{\geq}
  \underbrace{ \kappa_{\onetoftwo} }_{\overset{(f)}{>}0}
  \cdot
  \Biggl( 
    \prod_{ e \setminus 1 }
    \underbrace{ \lambda_{\mu_{\etofone}}( \ell_{e,p} ) }_{\overset{(g)}{>0}}
  \Biggr)
  \cdot 
  \underbrace{
    \Bigr(
      \bigl( \matr{u}_{ \mu_{\onetoftwo},\ell_{1,0} } \bigr)^{\Herm} 
      \mathop{\otimes}_{ e \setminus 1 }
      \bigl( \matr{u}_{\etofone,\ell_{e,p}} \bigr)^{\tran} 
    \Bigl)
    \cdot
    \matr{C}_{f_{1}}
    \cdot
    \Bigr(
      \matr{u}_{ \mu_{\onetoftwo},\ell_{1,0} } 
      \mathop{\otimes}_{ e \setminus 1 }
      \overline{ \matr{u}_{\etofone,\ell_{e,p}} }
    \Bigl)
  }_{>0}
  \nonumber\\
  &>0, \nonumber
\end{align}
%----------------------------------------------------------------------------
\begin{itemize}
  \item where for simplicity, in the above derivations, 
  we use 
  $ \sum_{ 
    \ell_{2},\ldots,\ell_{|\setpf_{1}|}
  } $,
  $ \prod_{e \setminus 1} $, and $ \mathop{\otimes}_{e\setminus 1}$ 
  instead of
  $ \sum_{ 
    \ell_{2} \in \set{L}_{2},
    \ldots,\ell_{|\setpf_{1}|} \in \set{L}_{|\setpf_{1}|}
  } $,
  $ \prod_{e \in \{2,\ldots,|\setpf_{1}|\}} $,
  and
  $ \otimes_{ e \in \{2,\ldots,|\setpf_{1}|\} } $,

  \item where $ \mathop{\otimes}_{ e \setminus 1 } \matr{u}_{\etofone,\ell_{e,p}} $
  is the Kronecker product of the following columns vectors:
  \begin{align*}
    \matr{u}_{2, f_{1},\ell_{2,p}},\ldots, 
    \matr{u}_{|\setpf_{1}|, f_{1}, \ell_{|\setpf_{1}|,p}} 
  \end{align*}
  \ie,
  \begin{align*}
    \mathop{\otimes}_{ e \setminus 1 } \matr{u}_{\etofone,\ell_{e,p}}
    = \matr{u}_{2, f_{1},\ell_{2,p}}\otimes
    \matr{u}_{3, f_{1},\ell_{3,p}} \otimes \cdots \otimes
    \matr{u}_{|\setpf_{1}|, f_{1}, \ell_{|\setpf_{1}|,p} },
  \end{align*}

  \item where step $(a)$ follows from the fact that $ \lambda_{ \mu_{\onetoftwo} }( \ell_{1,0} ) = 0 $ is the eigenvalue of $ \matr{C}_{\mu_{\onetoftwo}} $ associated with the right-eigenvector $ \matr{u}_{ \mu_{\onetoftwo},\ell_{1,0} } $,

  \item where step $(b)$ follows from combining the definition of the Choi matrix 
  $ \matr{C}_{f} $ in~\eqref{eqn: def of Choi matrix representation} and
  the following derivations for the entries in the matrix 
  $ \matr{C}_{\mu_{\onetoftwo}} = 
  \bigl( \mu_{\onetoftwo}( z_{1}, z_{1}' ) \bigr)_{ \tz_{1} = ( z_{1}, z_{1}') \in \set{X}_{1} } $
  as follows:
  \begin{align*}
      \hspace{-0.5cm}&\hspace{-0.5cm}
      \mu_{\onetoftwo}( z_{1}, z_{1}' )
      \nonumber\\
      &=\kappa_{\onetoftwo}
      \cdot
      \sum_{\tvx_{\setpf_{1}}: \, \tx_{1} = \tz_{1}}
      f_{1}\bigl( \tvx_{\setpf_{1}} \bigr)
      \cdot \prod_{e \in \setpf_{1} \setminus 1}
      \mu_{e, f_{1}}(x_{e}, x_{e}')
      \nonumber\\
      &= 
      \kappa_{\onetoftwo}
      \cdot
      \sum_{\tvx_{\setpf_{1}}: \, \tx_{1} = \tz_{1}}
      f_{1}\bigl( \tvx_{\setpf_{1}} \bigr)
      \cdot \prod_{e \in \setpf_{1} \setminus 1}
      % \biggl( 
        \sum_{ \ell_{e} \in \set{L}_{e} } 
        u_{ \mu_{e, f_{1}}, \ell_{e} }( x_{e} )
        \cdot \lambda_{ \mu_{e,f_{1}} }( \ell_{e} )
        \cdot \overline{ u_{ \mu_{e,f_{1}}, \ell_{e} }(x_{e}') }
      \nonumber\\
      &=
      \kappa_{\onetoftwo}
      \cdot
      \sum_{ 
        \ell_{2},\ldots,\ell_{|\setpf_{1}|}
      }
      \Biggl( 
        \prod_{e \in \setpf_{1} \setminus 1} 
        \lambda_{ \mu_{e,f_{1}} }( \ell_{e} )
      \Biggr)
      \nonumber\\
      &\hspace{3 cm} \cdot 
      \sum_{\tvx_{\setpf_{1}}: \, \tx_{1} = \tz_{1}}
      f_{1}\bigl( \tvx_{\setpf_{1}} \bigr)
      \cdot \prod_{e \in \setpf_{1} \setminus 1}
      u_{ \mu_{e, f_{1}}, \ell_{e} }( x_{e} )
      \cdot \overline{ u_{ \mu_{e,f_{1}}, \ell_{e} }(x_{e}') },
      % \biggr)
  \end{align*}
  where $ \kappa_{\onetoftwo} $ is the normalization constant ensuring that 
  \begin{align*}
    \sum_{\tz_{1} = ( z_{1}, z_{1}') \in \set{X}_{1}} 
    \mu_{\onetoftwo}( z_{1}, z_{1}' )
    = 1,
  \end{align*}
  where the first the equality follows from the definition of the SPA fixed-point message vector in~\eqref{sec:DENFG:eqn:15},
  where the second equality follows from the decomposition in~\eqref{eqn: decomposition of the Choi matrix for SPA message},
  where the third equality follows from rearranging the terms on the RHS of the second equality,

  \item where step $(c)$ follows from the fact $ \matr{C}_{\mu_{\etof}} \in \setPSD{\setxe} $ and the eigenvalue decomposition of $ \matr{C}_{\mu_{\etof}} $ in~\eqref{eqn: decomposition of the Choi matrix for SPA message}, which implies
  $ \lambda_{\mu_{\etofone}}(\ell_{e}) \in \sR_{\geq 0} $ for all 
  $ \ell_{e} \in \set{L}_{e} $ and $ e \in \{2,\ldots,|\setpf_{1}|\} $,

  \item where step $(d)$ follows from the condition that the Choi matrix 
  $ \matr{C}_{f_{1}} $ is a positive definite matrix, as stated in the proposition statement,

  \item where step $(e)$ follows from the fact that the terms on the RHS of the summation operator 
  $ \sum_{ 
        \ell_{2},\ldots,\ell_{|\setpf_{1}|}
  } $, as shown in~\eqref{eqn: expanding SPA update rule:2} 
  and~\eqref{eqn: expanding SPA update rule:1}, are all non-negative real-valued,

  \item where step $(f)$ follows from~\eqref{eqn: positive of normalization constant},

  \item where step $(g)$ follows from the property of 
  $ \ell_{e, \mathrm{p}} $ in~\eqref{eqn: property of ell e p}.
\end{itemize}
%----------------------------------------------------------------------------

The proof for the property that the Choi-matrix representations of other SPA fixed-point messages are all positive definite matrices, is similar and thus it is omitted here.

%% file: figures/example_sfG_1.tex
\begin{tikzpicture}[node distance=2.2cm, on grid,auto]
  \tikzstyle{state}=[shape=rectangle,fill=white,draw,minimum size=0.3cm]
  \pgfmathsetmacro{\ldis}{2.5} 
  \pgfmathsetmacro{\sdis}{0.5} 
  \begin{pgfonlayer}{glass}
    \foreach \i in {1,...,7}
    {
      \node[state] (f\i) at (0,-0.6*\i) [label=left: \scriptsize$\i$] {};
      \node[state] (g\i) at (\ldis,-0.6*\i) [label=right: \scriptsize$\i$] {};
    }
    \end{pgfonlayer}
     \begin{pgfonlayer}{background}
         \draw[thick,red]
            (f3) -- (g4)      
            (f4) -- (g3)
            (f5) -- (g6)
            (f6) -- (g7)
            (f7) -- (g5)
            ;
        \draw[thick,blue]
            (f3) -- (g3)      
            (f4) -- (g4)
            (f5) -- (g5)
            (f6) -- (g6)
            (f7) -- (g7)
            ;
    \end{pgfonlayer}
\end{tikzpicture}

%% file: figures/example_sfG_2.tex
\begin{tikzpicture}[node distance=2.2cm, on grid,auto]
  \tikzstyle{state}=[shape=rectangle,fill=white,draw,minimum size=0.3cm]
  \pgfmathsetmacro{\ldis}{2.5} 
  \pgfmathsetmacro{\sdis}{0.5} 
  \begin{pgfonlayer}{glass}
    \foreach \i in {1,...,7}
    {
      \node[state] (f\i) at (0,-0.6*\i) [label=left: \scriptsize$\i$] {};
      \node[state] (g\i) at (\ldis,-0.6*\i) [label=right: \scriptsize$\i$] {};
    }
    \end{pgfonlayer}
    \begin{pgfonlayer}{background}
         \draw[thick,red]
            (f3) -- (g4)      
            (f4) -- (g3)
            (f5) -- (g5)
            (f6) -- (g6)
            (f7) -- (g7)
            ;
        \draw[thick,blue]
            (f3) -- (g3)      
            (f4) -- (g4)
            (f5) -- (g6)
            (f6) -- (g7)
            (f7) -- (g5)
            ;
  \end{pgfonlayer}
  \vspace{0.5cm}
\end{tikzpicture}

%% file: figures/example_sfG_3.tex
\begin{tikzpicture}[node distance=2.2cm, on grid,auto]
  \tikzstyle{state}=[shape=rectangle,fill=white,draw,minimum size=0.3cm]
  \pgfmathsetmacro{\ldis}{2.5} 
  \pgfmathsetmacro{\sdis}{0.5} 
  \begin{pgfonlayer}{glass}
    \foreach \i in {1,...,7}
    {
      \node[state] (f\i) at (0,-0.6*\i) [label=left: \scriptsize$\i$] {};
      \node[state] (g\i) at (\ldis,-0.6*\i) [label=right: \scriptsize$\i$] {};
    }
    \end{pgfonlayer}
    \begin{pgfonlayer}{background}
         \draw[thick,red]
            (f3) -- (g3)      
            (f4) -- (g4)
            (f5) -- (g6)
            (f6) -- (g7)
            (f7) -- (g5)
            ;
        \draw[thick,blue]
            (f3) -- (g4)      
            (f4) -- (g3)
            (f5) -- (g5)
            (f6) -- (g6)
            (f7) -- (g7)
            ;
  \end{pgfonlayer}
  \vspace{0.5cm}
\end{tikzpicture}

%% file: figures/example_sfG_4.tex
\begin{tikzpicture}[node distance=2.2cm, on grid,auto]
  \tikzstyle{state}=[shape=rectangle,fill=white,draw,minimum size=0.3cm]
  \pgfmathsetmacro{\ldis}{2.5} 
  \pgfmathsetmacro{\sdis}{0.5} 
  \begin{pgfonlayer}{glass}
    \foreach \i in {1,...,7}
    {
      \node[state] (f\i) at (0,-0.6*\i) [label=left: \scriptsize$\i$] {};
      \node[state] (g\i) at (\ldis,-0.6*\i) [label=right: \scriptsize$\i$] {};
    }
    \end{pgfonlayer}
    \begin{pgfonlayer}{background}
         \draw[thick,red]
            (f3) -- (g3)      
            (f4) -- (g4)
            (f5) -- (g5)
            (f6) -- (g6)
            (f7) -- (g7)
            ;
        \draw[thick,blue]
            (f3) -- (g4)      
            (f4) -- (g3)
            (f5) -- (g6)
            (f6) -- (g7)
            (f7) -- (g5)
            ;
  \end{pgfonlayer}
\end{tikzpicture}

%% file: appendices/LCT_SNFG.tex
Given $ |\set{X}_e| \in \sZpp $, without loss of generality, we suppose that
$\set{X}_e = \LCTset{X}_e = \{ 0, 1, \ldots, |\set{X}_e| \! - \! 1 \}$.  We have
to verify that the the functions $M_{\efi}$ and $M_{\efj}$
satisfy~\eqref{sec:LCT:exp:1}, \textit{i.e.},
\begin{align*}
  \sum_{\LCT{x}_e}
    M_{\efi}(x_{\efi}, \LCT{x}_e)
    \cdot
    M_{\efj}(x_{\efj}, \LCT{x}_e) 
    &= [x_{\efi} \! = \!  x_{\efj}],
         \qquad x_{\efi}, x_{\efj} \in \set{X}_e.
\end{align*}
We consider the following four cases:
\begin{enumerate}

\item For $x_{\efi} = 0$, $x_{\efj} = 0$, we obtain the constraint
  \begin{align*}
    &\zeta_{\efi} 
      \cdot
      \mu_{\efi}(0)
    \cdot
    \zeta_{\efj}
      \cdot
      \mu_{\efj}(0)
    \nonumber\\
    &\hspace{2cm}+
    \sum_{\LCT{x}_e: \, \LCT{x}_e \neq 0}
      \zeta_{\efi}
        \cdot
        \chi_{\efi} 
        \cdot
        \bigl( - \mu_{\efj}(\LCT{x}_e) \bigr)
      \cdot
      \zeta_{\efj}
        \cdot
        \chi_{\efj}
        \cdot
        \bigl( - \mu_{\efi}(\LCT{x}_e) \bigr)
       = 1.
  \end{align*}
  Using~\eqref{sec:LCT:exp:5}
  and~\eqref{sec:LCT:exp:6}, and multiplying both size by
  $Z_e$, we obtain the equivalent constraint
  \begin{align*}
    \mu_{\efi}(0)
    \cdot
    \mu_{\efj}(0)
    +
    \sum_{\LCT{x}_e: \, \LCT{x}_e \neq 0}
      \mu_{\efj}(\LCT{x}_e)
      \cdot
      \mu_{\efi}(\LCT{x}_e)
       &= Z_e.
  \end{align*}
  % \textit{i.e.},
  % \begin{align*}
  %   \sum_{\LCT{x}_e}
  %     \mu_{\efi}(\LCT{x}_e)
  %     \cdot
  %     \mu_{\efj}(\LCT{x}_e)
  %      &= Z_e.
  % \end{align*}
  Clearly, this constraint is satisfied since the left-hand side is just the
  definition of $Z_e$.
  
\item For $x_{\efi} = 0$, $x_{\efj} \in \set{X}_e \setminus \{ 0 \}$, we obtain
  the constraint
  \begin{align*}
    \zeta_{\efi} &
      \cdot
      \mu_{\efi}(0)
    \cdot
    \zeta_{\efj}
      \cdot
      \mu_{\efj}(x_{\efj}) 
    +
    \sum_{\LCT{x}_e: \, \LCT{x}_e \neq 0}
      \zeta_{\efi}
        \cdot
        \chi_{\efi} 
        \cdot
        \bigl( - \mu_{\efj}(\LCT{x}_e) \bigr)
    \nonumber\\
    & \cdot
      \Bigl(
        \zeta_{\efj}
          \cdot
          \chi_{\efj}
          \cdot
          \bigl(
            \delta_{\efj} 
                      \cdot
                      [x_{\efj} \!=\! \LCT{x}_e] 
                    +
                    \epsilon_{\efj}
                      \cdot
                      \mu_{\efj}(x_{\efj}) 
                      \cdot 
                      \mu_{\efi}(\LCT{x}_e)
          \bigr)     
      \Bigr)
        = 0.
  \end{align*}
  Using~\eqref{sec:LCT:exp:5}
  and~\eqref{sec:LCT:exp:6}, we obtain the equivalent
  constraint
  \begin{align*}
    &\hspace{-0.25cm}\mu_{\efi}(0)
    \cdot
    \mu_{\efj}(x_{\efj})
    \nonumber\\
    &-
    \sum_{\LCT{x}_e: \, \LCT{x}_e \neq 0}
      \mu_{\efj}(\LCT{x}_e)
      \cdot
      \Bigl(
        \delta_{\efj} 
          \cdot
          [x_{\efj} \!=\! \LCT{x}_e] 
        +
        \epsilon_{\efj}
          \cdot
          \mu_{\efj}(x_{\efj}) 
          \cdot 
          \mu_{\efi}(\LCT{x}_e)
      \Bigr)
      = 0.
  \end{align*}
  Simplifying the sum in the above expression, this constraint is equivalent
  to the constraint
  \begin{align*}
    \mu_{\efi}(0)
    \cdot
    \mu_{\efj}(x_{\efj})
    &-
    \mu_{\efj}(x_{\efj})
      \cdot
      \delta_{\efj}
    \nonumber\\
    &-
    \epsilon_{\efj}
      \cdot
      \mu_{\efj}(x_{\efj}) 
      \cdot
      \sum_{\LCT{x}_e: \, \LCT{x}_e \neq 0}
        \mu_{\efj}(\LCT{x}_e)
          \cdot
          \mu_{\efi}(\LCT{x}_e)      
        = 0.
  \end{align*}
  The sum appearing in the above expression equals
  $Z_e \cdot \bigl( 1 - \beli_e(0) \bigr)$. Therefore, this constraint can be
  rewritten as
  \begin{align*}
    \mu_{\efj}(x_{\efj})
      \cdot
      \Bigl(
        \mu_{\efi}(0)
        -
        \delta_{\efj}
        -
        \epsilon_{\efj}
          \cdot
          Z_e
          \cdot
          \bigl( 1 - \beli_e(0) \bigr)
      \Bigr)
      &= 0.
  \end{align*}
  The expression in the parentheses equals zero because
  of~\eqref{sec:LCT:exp:9}, and so this constraint is
  satisfied.

\item For $x_{\efi} \in \set{X}_e \setminus \{ 0 \}$, $x_{\efj} = 0$, we obtain
  similar expressions as for the second case, simply with the roles of $f_{i}$ and
  $f_{j}$ exchanged.

\item For $x_{\efi}, x_{\efj} \in \set{X}_e \setminus \{ 0 \}$, we obtain the
  constraint (we have already used~\eqref{sec:LCT:exp:5}
  and~\eqref{sec:LCT:exp:6} to simplify the expression)
  \begin{align*}
    &\hspace{-0.2 cm}\mu_{\efi}(x_{\efi})
      \cdot
      \mu_{\efj}(x_{\efj})
    +
    \sum_{\LCT{x}_e: \, \LCT{x}_e \neq 0}
      \Bigl(
        \delta_{\efi} 
          \cdot
          [x_{\efi} \!=\! \LCT{x}_e] 
        +
        \epsilon_{\efi}
          \cdot
          \mu_{\efi}(x_{\efi}) 
          \cdot 
          \mu_{\efj}(\LCT{x}_e)
      \Bigr) \\
      &
      \cdot
      \Bigl(
        \delta_{\efj} 
          \cdot
          [x_{\efj} \!=\! \LCT{x}_e] 
        +
        \epsilon_{\efj}
          \cdot
          \mu_{\efj}(x_{\efj}) 
          \cdot 
          \mu_{\efi}(\LCT{x}_e)
      \Bigr)
        = Z_e 
            \cdot
            [x_{\efi} \!=\! x_{\efj}].
  \end{align*}
  This constraint is equivalent to the constraint
  \begin{align*}
    &\hspace{-0.5cm}
    \mu_{\efi}(x_{\efi})
      \cdot
      \mu_{\efj}(x_{\efj})
    +
    \delta_{\efi}
      \cdot
      \delta_{\efj} 
      \cdot
      [x_{\efi} \!=\! x_{\efj}] \\
    &
    +
    \delta_{\efi} 
      \cdot
      \epsilon_{\efj}
      \cdot
      \mu_{\efj}(x_{\efj}) 
      \cdot
      \mu_{\efi}(x_{\efi})
    +
    \delta_{\efj} 
      \cdot
      \epsilon_{\efi}
      \cdot
      \mu_{\efi}(x_{\efi}) 
      \cdot 
      \mu_{\efj}(x_{\efj}) \\
    &
    +
    \epsilon_{\efi}
      \!\cdot\!
      \epsilon_{\efj}
      \!\cdot\!
      \mu_{\efi}(x_{\efi})
      \!\cdot\!
      \mu_{\efj}(x_{\efj})
      \!\cdot
      \!\!\!\!
      \sum_{\LCT{x}_e: \, \LCT{x}_e \neq 0}
      \!\!\!\!
        \mu_{\efj}(\LCT{x}_e)
        \!\cdot\!
        \mu_{\efi}(\LCT{x}_e)
      = Z_e
        \!
          \cdot
        \!
          [x_{\efi} \!=\! x_{\efj}].
  \end{align*}
  The sum appearing in the above expression equals
  $Z_e \cdot \bigl( 1 - \beli_e(0) \bigr)$. Therefore, after rewriting this
  constraint, we get the constraint
  \begin{align*}
    &[x_{\efi} \!=\! x_{\efj}]
     \cdot
      \bigl(
        \delta_{\efi}
          \cdot
          \delta_{\efj} 
        -
        Z_e
      \bigr) 
    +
    \mu_{\efi}(x_{\efi})
      \cdot
      \mu_{\efj}(x_{\efj})
    \\
    &\hspace{2cm}
      \cdot
      \Bigl(
        1
        +
        \delta_{\efi} 
          \cdot
          \epsilon_{\efj}
        +
        \delta_{\efj} 
          \cdot
          \epsilon_{\efi}
        +
        \epsilon_{\efi}
          \cdot
          \epsilon_{\efj}
          \cdot
          Z_e
          \cdot 
          \bigl( 1 - \beli_e(0) \bigr)
      \Bigr)
      = 0.
  \end{align*}
  We claim that this constraint is satisfied. 
  This follows from the above two
   expressions being equal to zero. For the first parenthetical
  expression, this observation follows
  from~\eqref{sec:LCT:exp:7}. For the second parenthetical
  expression, this observation follows from the following
  considerations. 
  If $ \beli_e(0) = 1$, the second parenthetical
  expression equals
  \begin{align*}
    1
    +
    \delta_{\efi} 
      \cdot
      \epsilon_{\efj}
    +
    \delta_{\efj} 
      \cdot
      \epsilon_{\efi}
    = 0,
  \end{align*}
  where the last equality follows from~\eqref{extra constraint on epsilon for beli = [x = 0]}.
  Now we consider $ \beli_e(0) \neq 1$.
  Because the product of the left-hand sides
  of~\eqref{sec:LCT:exp:8}
  and~\eqref{sec:LCT:exp:9} must equal the product of the
  right-hand sides of~\eqref{sec:LCT:exp:8}
  and~\eqref{sec:LCT:exp:9}, we get
  \begin{align*}
    &\delta_{\efi}
      \cdot
      \delta_{\efj}
    +
    \delta_{\efi}
      \cdot
      Z_e
      \cdot 
      \bigl( 1 - \beli_e(0) \bigr)
      \cdot
      \epsilon_{\efj}
    +
    \delta_{\efj} 
      \cdot
      Z_e
        \cdot 
        \bigl( 1 - \beli_e(0) \bigr)
        \cdot
        \epsilon_{\efi}
      \\  & \hspace{4cm}
         +
    (Z_e)^2
      \cdot 
      \bigl( 1 - \beli_e(0) \bigr)^2
      \cdot
      \epsilon_{\efi}
      \cdot
      \epsilon_{\efj}  
      = \mu_{\efi}(0)
             \cdot
             \mu_{\efj}(0).
  \end{align*}
  Using~\eqref{sec:LCT:exp:7} and
  $\mu_{\efi}(0) \cdot \mu_{\efj}(0) = Z_e \cdot \beli_e(0)$, we obtain
  \begin{align*}
    Z_e
    +
    \delta_{\efi}
      \cdot
      Z_e
      \cdot 
      \bigl( 1 - \beli_e(0) \bigr)
      \cdot&
      \epsilon_{\efj}
    +
    \delta_{\efj} 
      \cdot
      Z_e
      \cdot 
      \bigl( 1 - \beli_e(0) \bigr)
      \cdot
      \epsilon_{\efi}
    \nonumber \\
      &+
    (Z_e)^2
      \cdot 
      \bigl( 1 - \beli_e(0) \bigr)^2
      \cdot
      \epsilon_{\efi}
      \cdot
      \epsilon_{\efj}= Z_e \cdot \beli_e(0).
  \end{align*}
  Subtracting $Z_e \cdot \beli_e(0)$ and then dividing $ Z_e \cdot \bigl( 1 - \beli_e(0) \bigr) $ from both sides, results in
  \begin{align*}
    1
    +
    \delta_{\efi}
      \cdot
      \epsilon_{\efj}
    +
    \delta_{\efj} 
      \cdot
      \epsilon_{\efi}
    +
    Z_e
      \cdot 
      \bigl( 1 - \beli_e(0) \bigr)
      \cdot
      \epsilon_{\efi}
      \cdot
      \epsilon_{\efj}
    &= 0.
  \end{align*}
  This is the promised result.

\end{enumerate}

%***************************************************************************
%***************************************************************************

\ifx\sectionheaderonnewpage\x
\clearpage
\fi

%% file: appendices/Props_LCT_SNFG.tex
%***************************************************************************

\begin{enumerate}

\item This follows from the fact that in every step of the LCT the partition
  function is unchanged.

\item This follows from
  \begin{align*}
    \LCT{g}(\vect{0})
      &= \prod_f
           \LCT{f}(\vect{0}) \\
      &\overset{(a)}{=} \prod_f
           \sum_{\vx_{\setpf}}
            % \Biggl( 
              f(\vx_{\setpf})
              \cdot
              \prod_{e \in \setpf}
                M_{\etof}(x_{e}, 0)
            % \Biggr) 
        \\
      &\overset{(b)}{=} \prod_f
           \sum_{\vx_{\setpf}}
            % \Biggl( 
             f(\vx_{\setpf})
               \cdot
               \prod_{e \in \setpf}
                 \bigl(
                   \zeta_{\etof}
                     \cdot
                     \mu_{\etof}(x_{e})
                 \bigr) 
            % \Biggr)
      \\
      &= \left(
           \prod_f
             \sum_{\vx_{\setpf}}
               f(\vx_{\setpf})
                 \cdot
                 \prod_{e \in \setpf}
                   \mu_{\etof}(x_{e})
         \right)
         \cdot
         \left(
           \prod_{e = (f_{i},f_{j})}
             \bigl(
               \zeta_{\efj}
                 \cdot
                 \zeta_{\efi}
             \bigr)
         \right) \\
      &\overset{(c)}{=} 
      \Biggl(
           \prod_f
             Z_f( \vmu)
         \Biggr)
         \cdot
         \left(
           \prod_{e}
             \bigl( Z_e( \vmu) \bigr)^{\!\!-1}
         \right) \\
      &\overset{(d)}{=} \ZBSPA(\vmu),
  \end{align*}
  where step $(a)$ follows from~\eqref{sec:LCT:exp:10}, where
  step $(b)$ follows from~\eqref{sec:LCT:exp:3}
  and~\eqref{sec:LCT:exp:4}, where step $(c)$ follows
  from~\eqref{sec:LCT:exp:5}, and where step $(d)$
  follows from~\eqref{sec:SNFG:eqn:5}.

\item Consider $e = (f_{i}, f_{j}) \in \setEfull $ and 
  $\LCTv{x}_{\setpfi} \in \LCTset{X}_{\setpfi}$ with exactly one nonzero component
  such that $\LCT{x}_e \neq 0$ and $\LCT{x}_{e'} = 0$ for all
  $e' \in \setpfi \setminus \{ e \}$. For this setup, we get
  \begin{align*}
    \LCT{f_{i}}(\LCTvxfi)
      &\overset{(a)}{=} \sum_{x_{e}}
           M_{\efi}(x_{e}, \LCT{x}_{e})
           \cdot
             \sum_{\vx_{\setpfi} \setminus x_{e}}
               f_{i}(\vx_{\setpfi})
                 \cdot
                 \prod_{e' \in \setpfi \setminus \{ e \}}
                   M_{\epfi}(x_{e'}, \LCT{x}_{e'}) \\
      &\overset{(b)}{=} \sum_{x_{e}}
           M_{\efi}(x_{e}, \LCT{x}_e)
           \cdot
             \sum_{\vx_{\setpfi} \setminus x_{e}}
               f_{i}(\vx_{\setpfi})
                 \cdot \prod_{e' \in \setpfi \setminus \{ e \}}
                   \bigl(
                     \zeta_{\epfi} 
                       \cdot 
                       \mu_{\epfi}(x_{e'})
                   \bigr) \\
      &\overset{(c)}{=} \left(
           \prod_{e' \in \setpfi \setminus \{ e \}}
             \zeta_{\epfi} 
         \right)
         \cdot
         \sum_{x_{e}}
           M_{\efi}(x_{e}, \LCT{x}_e)
           \cdot
             \sum_{\vx_{\setpfi} \setminus x_{e}}
               f_{i}(\vx_{\setpfi})
                 \nonumber \\
                 &\hspace{8 cm}
                 \cdot \prod_{e' \in \setpfi \setminus \{ e \}}
                   \mu_{\epfi}(x_{e'}) \\
      &\overset{(d)}{=} \left(
           \prod_{e' \in \setpfi \setminus \{ e \}}
             \zeta_{\epfi} 
         \right)
         \cdot
         \sum_{x_{e}}
           M_{\efi}(x_{e}, \LCT{x}_e)
           \cdot
             \kappa_{\efj}
             \cdot
             \mu_{\efj}(x_{e}) \\
      &\overset{(e)}{=} \left(
           \prod_{e' \in \setpf \setminus \{ e \}}
             \zeta_{\epfi} 
         \right)
         \cdot
         \sum_{x_{e}}
           M_{\efi}(x_{e}, \LCT{x}_e)
           \cdot
             \kappa_{\efj}
             \cdot
             \zeta_{\efj}^{-1}
             \cdot
             M_{\efj}(x_{e},0) \\
      &\overset{(f)}{=} \left(
           \prod_{e' \in \setpf \setminus \{ e \}}
             \zeta_{\epfi} 
         \right)
         \cdot
         \kappa_{\efj}
         \cdot
         \zeta_{\efj}^{-1}
         \cdot
         [\LCT{x}_e \! = \! 0] \\
      &\overset{(g)}{=} 0,
  \end{align*}
  where $ \sum_{\vx_{\setpfi} \setminus x_{e}} $ represents the operation that takes the summation of the collection of the variables 
  $ ( x_{e'} )_{ e' \in \setpfi \setminus \{ e \} } $
  over the alphabet $ \prod_{ e' \in \setpfi \setminus \{ e \} } \setx{e'} $,
  where step $(a)$ follows from~\eqref{sec:LCT:exp:10}, where
  step $(b)$ follows from $\LCT{x}_{e'} = 0$ for
  $e' \in \setpf \setminus \{ e \}$ and
  equations~\eqref{sec:LCT:exp:3} and~\eqref{sec:LCT:exp:4},
  where step $(c)$ follows from reordering terms, where step $(d)$ follows from the SPA message update
  rules in~\eqref{sec:SNFG:eqn:8}, where step $(e)$
  follows
  from~\eqref{sec:LCT:exp:4},
  and where step $(f)$ follows
  from~\eqref{sec:LCT:exp:2}, and where step $(g)$
  follows from $\LCT{x}_e \neq 0$.

\item Assume that $ \bigl( \setF,\setE'(\LCTv{x}) \bigr) $ is not a generalized loop for some
  $ \LCTv{x} \in \prod_{e} \LCTset{X}_{e} $. Based on the third
  property and the definition of a generalized loop, we can conclude that 
  for this $ \LCTv{x} $,
  there is at least one $f \in \setF$ such that $ \wh(\LCTvxf) = 1 $ and $\LCT{f}(\LCTvxf) = 0$, which
  implies $\LCT{g}(\LCTv{x}) = 0$.

\item Follows immediately from the fourth property.

\item From the end of the first step in the construction of the
  LCT $\LCT{\graphN}$ of $\graphN$, it follows that the SPA fixed-point message
  vector $\vmu$ for $\graphN$ induces an SPA fixed-point message vector
  $\LCTv{\mu} $ for
  $\LCT{\graphN}$ via
  \begin{align*}
    \mu_{\etof}(x_e)
      &= ( \zeta_{\etof} )^{-1}
         \cdot
         \sum_{\LCT{x}_e}
           M_{\etof}(x_e, \LCT{x}_e)
           \cdot
           \LCT{\mu}_{\etof}(\LCT{x}_e),
  \end{align*}
  where $\kappa_{\etof}$ is some nonzero constant. We claim that
  $\LCT{\mu}_{\etof}(\LCT{x}_e) = [\LCT{x}_e \! = \! 0]$ for all
  $\LCT{x}_e \in \LCTset{X}_e$, is the unique (up to rescaling) solution to
  this expression. The claim follows from
  $\LCT{\mu}_{\etof}(\LCT{x}_e) = [\LCT{x}_e \! = \! 0]$ for all
  $\LCT{x}_e \in \LCTset{X}_e$, satisfying the above expression, along with
  the matrix $ \matr{M}_{\etof} $ being
  invertible thanks to~\eqref{sec:LCT:exp:1}.

\item
  \begin{enumerate}

  \item This follows immediately from the expressions in
    Definition~\ref{sec:LCT:def:1}.

  \item From the assumptions it follows that
    $\epsilon_{\efi} = \epsilon_{\efj}$. Because of the symmetry in
    Definition~\ref{sec:LCT:def:1} and because the parameters in every
    parameter pair are the same, the two matrices
    $\matr{M}_{\efi}$ and
    $\matr{M}_{\efj}$ must be
    equal. Moreover, because of~\eqref{sec:LCT:exp:1},
    or~\eqref{sec:LCT:exp:2}, these two matrices must be
    orthogonal.

  \end{enumerate}

\end{enumerate}

%***************************************************************************
%***************************************************************************

\ifx\sectionheaderonnewpage\x
\clearpage
\fi

%% file: appendices/LCT_DENFG.tex
The proofs of Properties 1--7 are straightforward generalizations of the
corresponding proofs of Properties 1--7 in Proposition~\ref{sec:LCT:prop:1}. 
\begin{enumerate}

\item[8)] For each $ e \in \setpf $ and $ f \in \setF $, we get the functions 
$ M_{\ef} $ by applying the substitutions introduced in Definition~\ref{def:DENFG:LCT:1} to the definition of the LCT for S-NFG in Definition~\ref{sec:LCT:def:1}.
Because the SPA fixed-point message vector $ \vmu $ satisfies 
$ \matr{C}_{ \mu_{\ef} } \in \setPSD{\set{X}_{e}} $, as proven in Lemma~\ref{sec:DENFG:lem:1}, and the fact that the constants $\zeta_{\efj}$, $\chi_{\efj}$, $\delta_{\ef}$, and $\epsilon_{\ef}$ are real-valued, as stated in Definition~\ref{def:DENFG:LCT:1}, one can verify that
\begin{align*}
  M_{\ef}\bigl( (\xe, \xe') ,(\LCT{x}_{e}, \LCT{x}_{e}') \bigr) 
  = \overline{ M_{\ef}\bigl( (\xe', \xe) ,(\LCT{x}_{e}', \LCT{x}_{e}) \bigr) }
\end{align*}
for all $ 
  \txe = (\xe, \xe') \in \tset{X}_{e}$
  and $\LCTt{x}_{e} = ( \LCT{x}_{e}, \LCT{x}_{e}' ) \in \LCTtset{X}_{e} $.
Therefore, we know that the matrices $ \matr{C}_{ M_{\efi} }$ and 
$\matr{C}_{ M_{\efj} }$ are Hermitian matrices.

\item[9)] For each $ f \in \setF $, we have 
\begin{align*}
  \LCT{f}(\LCTv{x}_{\setpf}', \LCTv{x}_{\setpf})
  &\overset{(a)}{=}
  \sum_{ \tvx_{\setpf} }
  f( \vx_{\setpf}, \vx_{\setpf}' )
  \cdot
  \prod_{e \in \setpf}
  M_{\ef}\bigl( (x_{e}, x_{e}'), (\LCT{x}_e', \LCT{x}_e) \bigr)
  \nonumber\\
  &\overset{(b)}{=}
  \sum_{ \tvx_{\setpf} }
  f( \vx_{\setpf}', \vx_{\setpf} )
  \cdot
  \prod_{e \in \setpf}
  M_{\ef}\bigl( (x_{e}', x_{e}), (\LCT{x}_e', \LCT{x}_e) \bigr)
  \nonumber\\
  &\overset{(c)}{=}
  \sum_{ \tvx_{\setpf} }
  \overline{f( \vx_{\setpf}, \vx_{\setpf}' )}
  \cdot
  \prod_{e \in \setpf}
  \overline{ 
    M_{\ef}\bigl( (x_{e}, x_{e}'), (\LCT{x}_e, \LCT{x}_e') \bigr)
  }
  \nonumber\\
  &\overset{(d)}{=} \overline{ \LCT{f}(\LCTv{x}_{\setpf}, \LCTv{x}_{\setpf}') },
  \qquad \LCTtv{x}_{\setpf} = (\LCTv{x}_{\setpf}', \LCTv{x}_{\setpf}) 
  \in \LCTtset{X}_{\setpf},
\end{align*}
where step $(a)$ follows from the definition of $ \LCT{f} $ in~\eqref{sec:LCT:eqn:1},
where step $(b)$ follows from rearranging the arguments in the associated functions,
where step $(c)$ follows from the fact that the Choi-matrix representation 
$ \matr{C}_{f} $ in~\eqref{sec:DENFG:eqn:6} is a PSD matrix,
and Property~\ref{prop:DENFG:LCT:1:item:1}: the Choi-matrix representations
$ \matr{C}_{ M_{\efi} }$ and $\matr{C}_{ M_{\efj} }$ are Hermitian matrices,
and where step $(d)$ again follows from the definition of $ \LCT{f} $ in~\eqref{sec:LCT:eqn:1}.

\item[10)] This follows from 
Property~\ref{prop:DENFG:LCT:1:item:2} and the definition of 
the weak-sense DE-NFG in Definition~\ref{sec:DENFG:def:4}.

\end{enumerate}

%***************************************************************************
%***************************************************************************

\ifx\sectionheaderonnewpage\x
\clearpage
\fi

%% file: sst.tex
In this chapter, we introduce the symmetric-subspace transform (SST) for both S-NFG and DE-NFG. The SST provides a different perspective to understand the $M$-covers of S-NFGs or DE-NFGs. While this transformation has been applied in quantum physics~\cite{Wood2015,Harrow2013}, to the best of our knowledge, it is the first time to introduce the SST in the factor-graph literature.

The developments in this chapter were motivated by the work of Wood \etal~\cite{Wood2015}, where, in terms of the language of the present paper, they transformed a certain integral into the average partition function of some NFGs, with the average taken over double covers. (Although~\cite{Wood2015} focus on double covers, it is clear that their results can be extended to general $M$-covers.) However, in contrast to~\cite{Wood2015}, our approach takes the opposite direction. We express the average partition function of a given NFG in terms of some integral, where the average is over all $M$-covers of the considered NFG.

%***************************************************************************

The developments in this chapter were partially also motivated by the results presented in~\cite{Vontobel2016}.

Let $\graphN$ be some S-NFG or DE-NFG. In
Definition~\ref{sec:GraCov:def:2}, the degree-$M$ Bethe
partition function of $\graphN$ is 
\begin{align*}
  \ZBM(\graphN) 
    = \sqrt[M]{
         \Bigl\langle
           Z\bigl( \hgraphN \bigr)
         \Bigr\rangle_{ \hgraphN \in \hat{\set{N}}_{M}}
       }.
\end{align*}
(See also
Definition~\ref{sec:GraCov:def:1}.) Our
particular interest lies in the limit superior of this quantity as
$M \to \infty$, as stated in Conjecture~\ref{sec:GraCov:conj:1}. Toward proving
Conjecture~\ref{sec:GraCov:conj:1}, we aim to reformulate
$\bigl( \ZBM(\graphN) \bigr)^M$ (or equivalently
$\bigl\langle Z\bigl( \hgraphN \bigr) \bigr\rangle_{ \hgraphN \in \hat{\set{N}}_{M}}$) as
an integral that can be analyzed in 
the limit $M \to \infty$. This approach allows us to study the behavior of the Bethe partition function for large values of $M$.

%***************************************************************************

%----------------------------------------------------------------------------

In order to introduce the approach, we consider a specific example of an S-NFG. In this example, the mathematical formalism that we will present is general, but we will illustrate
it under a particularly simple setup. Specifically, 
we will focus a part of the S-NFG looks like
Fig.~\ref{sec:SST:fig:3} and set $\set{X}_e = \{ 0, 1 \}$ and $M = 2$.

%***************************************************************************
%***************************************************************************

\section{Specifying an \texorpdfstring{$M$}{}-Cover of an S-NFG}

%***************************************************************************
The following construction process defines an $M$-cover of the original S-NFG $\graphN$.

\begin{definition}
  \label{def:graph:cover:construction:1}

  The construction process for an $M$-cover of $\graphN$ is outlined as follows (see also
  Definition~\ref{sec:GraCov:def:1}):
  \begin{enumerate}
  
  \item For each edge $e \in \setEfull$, we specify a permutation
    $\sigma_e \in \set{S}_{[M]}$. We collect all these permutations into a vector $\vsigma \defeq (\sigma_e)_{e \in \setEfull} \in \set{S}_{[M]}^{|\setEfull|}$.
  
  \item\label{def:graph:cover:construction:1:item:1} For each function node $f \in \set{F}$, we draw $M$ copies of $f$. Each copy
    of function node $f$ is associated with a collection of sockets and variables. Specifically,
    if $e \in \setpf$, then the $m$-th copy of $f$ has a socket with
    associated variable denoted by $x_{e,f,m}$.
  
  \item For every edge $e = (f_i, f_{j}) \in \setEfull$, we specify how the
    sockets corresponding to $(x_{\efi,m})_{m \in [M]}$, are connected to the
    sockets corresponding to $(x_{\efj,m})_{m \in [M]}$. 
    In this thesis, the connection is determined based on the chosen permutation $\sigma_e \in \set{S}_{[M]}$. Specifically, the socket corresponding to $x_{\efi,m}$ is connected to the socket corresponding to $x_{\efj,\sigma_{e}(m)}$ for $m \in [M]$. The possible cases for this connection are illustrated in Figs.~\ref{sec:SST:fig:4} and~\ref{sec:SST:fig:5}, considering that $|\set{S}_{[M]}| = 2! = 2$ for $M = 2$.

    % we do
    % it as follows: if the permutation $\sigma_e \in \set{S}_{[M]}$ is chosen, then
    % the socket corresponding to $x_{\efi,m}$ is connected to the socket
    % corresponding to $x_{\efj,\sigma_{e}(m)}$ for $ m \in [M] $.

    % Figs.~\ref{sec:SST:fig:4} and~\ref{sec:SST:fig:6} show the two possible cases
    % that can happen in the illustrated S-NFG. (Note that $|\set{S}_{[M]}| = 2! = 2$
    % for $M = 2$.)
  
  \item The resulting $M$-cover is denoted by $\hgraphN_{M,\vsigma}$.

  \end{enumerate}
  \edefinition
\end{definition}

%***************************************************************************

In order to simplify the notation, let $\vSigma_{M} \defeq \vSigma \defeq \bigl( \Sigmae \bigr)_{\! e}$ be the random vector consisting of
independent and uniformly distributed (i.u.d.) random permutations from
$\set{S}_{[M]}^{|\setEfull|}$.\footnote{We use $\Sigmae$ instead of $\Sigma_e$ for these random
  variables in order to distinguish them from the symbol used to denote a sum
  over all edges.} The probability distributions of $ \Sigmae $ and $ \vSigma $ are given by
%----------------------------------------------------------------------------
\begin{align}
  p_{\Sigmae}(\sigma_e) 
  &= \frac{1}{|\set{S}_{[M]}|} = \frac{1}{M!}, \qquad e \in \setEfull,\,
  \sigma_e \in \set{S}_{[M]},
  \label{sec:SST:eqn:pdf of SST for each edge}\\
  p_{\vSigma}(\vsigma) &= \prod_{e}p_{\Sigmae}(\sigma_e)
  = \frac{1}{(M!)^{|\setEfull|}}, \qquad 
  \vsigma \in \set{S}_{[M]}^{|\setEfull|} 
  .\label{sec:SST:eqn:31}
\end{align}
%----------------------------------------------------------------------------

With this, $\hgraphN_{M,\vsigma}$ represents a uniformly
sampled $M$-cover of $\graphN$. We can rewrite the quantity of interest, $\bigl( \ZBM(\graphN) \bigr)^M$, as follows:
\begin{align}
  \bigl( \ZBM(\graphN) \bigr)^M
    &= \Bigl\langle
         Z\bigl( \hgraphN \bigr)
       \Bigr\rangle_{ \hgraphN \in \hat{\set{N}}_{M}}
     = \sum_{\vsigma}
         p_{\vSigma}(\vsigma)
         \cdot
         Z\bigl( \hgraphN_{M,\vsigma} \bigr), \label{sec:SST:eqn:15}
\end{align}
where $ \sum_{\vsigma} $ denotes $ \sum_{\vsigma \in \set{S}_{[M]}^{|\setEfull|}} $.
%***************************************************************************
%***************************************************************************

\section{Replacement of \texorpdfstring{$M$}{}-Cover S-NFG 
                       by another S-NFG}

%***************************************************************************

In order to proceed, it is convenient to replace each $M$-cover
$\hgraphN_{M,\vsigma}$ with an S-NFG denoted as $\hgraphPsig$,
which satisfies the property 
%----------------------------------------------------------------------------
\begin{align}
  Z\bigl( \hgraphN_{M,\vsigma} \bigr)
  = Z\bigl( \hgraphPsig \bigr). \label{sec:SST:eqn:33}
\end{align}
%----------------------------------------------------------------------------
Before providing the general definition of $\hgraphPsig$, we briefly
discuss this replacement in terms of the illustrated example. Instead
of working with the S-NFG in Fig.~\ref{sec:SST:fig:4}, we will
work with the S-NFG in Fig.~\ref{sec:SST:fig:6}. Similarly,
instead of working with the S-NFG in
Fig.~\ref{sec:SST:fig:5}, we will work with the S-NFG in
Fig.~\ref{sec:SST:fig:7}.

%***************************************************************************

\begin{definition}
  \label{sec:SST:def:1}

  For a given $\vsigma \in \set{S}_{[M]}^{\setEfull}$, the S-NFG
  $\hgraphPsig$ is defined as follows:
  \begin{enumerate}
  
  \item \label{step: for f covers} For every $f \in \set{F}$, we draw $M$ copies of $f$. 
   Each copy of $f$ is associated with a collection of sockets and variables. Specifically, if $e \in \setpf$, the $m$-th copy of $f$ has a socket with an associated variable $x_{\ef,m}$. (This step is the same as Step~\ref{def:graph:cover:construction:1:item:1} in Definition~\ref{def:graph:cover:construction:1}.)
  
  \item For every edge $e = (f_{i}, f_{j}) \in \setEfull$, we draw a function node
    $P_{e,\sigma_e}$. This node is connected to the sockets associated with the collection of
    variables $(x_{\efi,m})_{m \in [M]}$ and to the sockets associated with the collection of variables
    $(x_{\efj,m})_{m \in [M]}$. The local function $P_{e,\sigma_e}: \set{X}_e^{M} \times \set{X}_e^{M}  \to  \{0,1\} $ is
    defined to be
    \begin{align*}
      P_{e,\sigma_e}\bigl( \vx_{\efi,[M]}, \vx_{\efj,[M]} \bigr)
        &\defeq
           \prod_{m \in [M]}
     \bigl[
       x_{\efi,m} \! = \! x_{\efj,\sigma_e(m)}
     \bigr],
       \quad \vx_{\efi,[M]}, \vx_{\efj,[M]} \in \set{X}_e^M,
    \end{align*}
    where 
    \begin{align*}
      \vx_{\ef,[M]}
      &\defeq
       (x_{\ef,1}, \ldots, x_{\ef,M}) \in  \set{X}_e^M, \qquad 
       f \in \{f_{i}, f_{j}\}.
    \end{align*}
    (Note that $P_{e,\sigma_e}$ is an indicator
    function, as it takes only the values zero and one.)
  
  \end{enumerate}
  \edefinition
\end{definition}

%***************************************************************************

In the case of Fig.~\ref{sec:SST:fig:4}, the permutation
$\sigma_e$ satisfies $\sigma_e(1) = 1$ and $\sigma_e(2) = 2$. Therefore, the corresponding local function $P_{e,\sigma_e}$ with $M = 2$ is given by:
\begin{align*}
  P_{e,\sigma_e}\bigl( \vx_{\efi,[M]}, \vx_{\efj,[M]} \bigr)
  &=
     [x_{\efi,1} \! = \! x_{\efj,1}]
     \cdot
     [x_{\efi,2} \! = \! x_{\efj,2}].
\end{align*}
Representing this function in matrix form, with rows indexed by
$(x_{\efi,1}, x_{\efi,2}) $ and columns
indexed by
$(x_{\efj,1}, x_{\efj,2})$, which take the values following the order $ (0,0), \, (0,1), \, (1,0), \, (1,1) $, we have
\begin{align*}
  \matr{P}_{e,\sigma_e}
    &= \begin{pmatrix}
      1 & 0 & 0 & 0 \\
      0 & 1 & 0 & 0 \\
      0 & 0 & 1 & 0 \\
      0 & 0 & 0 & 1  
     \end{pmatrix}.
\end{align*}
The resulting S-NFG is shown in Fig.~\ref{sec:SST:fig:6}.

%***************************************************************************

On the other hand, in the case of Fig.~\ref{sec:SST:fig:5}, the
permutation $\sigma_e$ satisfies $\sigma_e(1) = 2$ and $\sigma_e(2) = 1$.  Thus, the corresponding local function $P_{e,\sigma_e}$ is given by:
\begin{align*}
  P_{e,\sigma_e}\bigl( \vx_{\efi,[M]}, \vx_{\efj,[M]} \bigr)
  &=
     [x_{\efi,1} \! = \! x_{\efj,2}]
     \cdot
     [x_{\efi,2} \! = \! x_{\efj,1}].
\end{align*}
Representing this function in matrix form (with the same row and column indexing as
above), we get
\begin{align*}
  \matr{P}_{e,\sigma_e}
  &= \begin{pmatrix}
    1 & 0 & 0 & 0 \\
    0 & 0 & 1 & 0 \\
    0 & 1 & 0 & 0 \\
    0 & 0 & 0 & 1  
   \end{pmatrix}.
\end{align*}
The resulting S-NFG is shown in Fig.~\ref{sec:SST:fig:7}.

%***************************************************************************
%***************************************************************************

\section[An NFG Representing the Average \texorpdfstring{$M$}{}-Cover]{An NFG Representing the Average Degree \texorpdfstring{$M$}{}-Cover 
                       of an S-NFG}
\label{sec:SST:NFG:average:cover:1}

%***************************************************************************

The next step is to construct an NFG $\hgraphNavg$, whose partition sum equals
$\sum_{\vsigma} p_{\vSigma}(\vsigma) \cdot Z\bigl( \hgraphN_{M,\vsigma} \bigr)$. We do this
as follows. (See Fig.~\ref{sec:SST:fig:8} for an illustration.)

%***************************************************************************

\begin{definition}
  \label{sec:SST:def:2}

  % The construction of $\hgraphNavg$ is essentially the same as the
  % construction of $\hgraphPsig$ for any $\vsigma$, but with the
  % following modifications and additions.

  The construction of $\hgraphNavg$ is essentially the same as the
  construction of $\hgraphPsig$ for any $\vsigma$ in Definition~\ref{sec:SST:def:2}, but with the
  following modifications and additions. 
  Namely, for every $e \in \setEfull$, the
  permutation $\sigma_e$ is not fixed, but a variable. With this, the function
  $P_{e,\sigma_e}\bigl( \vx_{\efi,[M]}, \vx_{\efj,[M]} \bigr)$ is now considered to be a function
  not only of $\vx_{\efi,[M]}$ and $\vx_{\efj,[M]}$, but also of
  $\sigma_e$. Graphically, this is done as follows:
  \begin{itemize}

  \item For every $e \in \setEfull$, we draw a function node $p_{\Sigmae}$.

  \item For every $e \in \setEfull$, we draw an edge that connects
    $P_{e,\sigma_e}$ and $p_{\Sigmae}$ and that represents the variable
    $\sigma_e$. Recall $ p_{\Sigmae}( \sigma_{e} ) = (M!)^{-1} $ for all $ e \in \setEfull $ and $ \sigma_{e} \in \set{S}_{[M]} $, as defined in~\eqref{sec:SST:eqn:pdf of SST for each edge}.

  \end{itemize}
  % We define $ p_{\vSigma}(\vsigma) 
  % \defeq \prod_{e \in \setEfull} p_{\Sigmae}( \sigma_{e} ) $.
  \edefinition
\end{definition}

%***************************************************************************

\begin{lemma}
  \label{sec:SST:lem:1}

  The partition function of $\hgraphNavg$ satisfies
  \begin{align*}
    Z\bigl( \hgraphNavg \bigr)
      &= \sum_{\vsigma}
           p_{\vSigma}(\vsigma) 
           \cdot Z\bigl( \hgraphN_{M,\vsigma} \bigr)
      = \Bigl\langle
           Z\bigl( \hgraphN \bigr)
         \Bigr\rangle_{ \hgraphN \in \hat{\set{N}}_{M}}
      = \bigl( \ZBM(\graphN) \bigr)^M.
  \end{align*}
\end{lemma}

%***************************************************************************

\begin{proof}
  The first equality  follows directly from the definition of $\hgraphNavg$ in Definition~\ref{sec:SST:def:2}, and the last two equalities follow from the equalities in~\eqref{sec:SST:eqn:15}.
\end{proof}

%***************************************************************************

In the subsequent discussion, we fix an arbitrary edge $e = (f_{i}, f_{j}) \in \setEfull$.

%***************************************************************************

%---------------------------------------------------------------------------
\begin{definition}\label{sec:SST:def:4}
  We define
  %---------------------------------------------------------------------------
  \begin{align*}
    P_e\bigl( \vx_{\efi,[M]}, \vx_{\efj,[M]} \bigr)
      &\defeq
         \sum_{\sigma_e \in \set{S}_{[M]}}
           p_{\Sigmae}(\sigma_e)
           \cdot
             P_{e,\sigma_e}\bigl( \vx_{\efi,[M]}, \vx_{\efj,[M]} \bigr) \\
      &= \frac{1}{M!}
           \cdot
           \sum_{\sigma_e \in \set{S}_{[M]}}
             P_{e,\sigma_e}\bigl( \vx_{\efi,[M]}, \vx_{\efj,[M]} \bigr),
  \end{align*}
  where the last equality follows from $p_{\Sigmae}(\sigma_e) = 1/M!$ for all
  $\sigma_e \in \set{S}_{[M]}$. Also, we define the following vectors:
  %---------------------------------------------------------------------------
  \begin{align*}
    \vx_{[M]} &\defeq ( x_{e,f,m} )_{e \in \setpf, f \in \setF, m \in [M]}
    \in \prod_{e \in \setpf, f \in \setF} \setx{e}^{M},
    \nonumber\\
    \vx_{\setpff,m} &\defeq ( x_{e,f,m} )_{e \in \setpf}
    \in \setx{\setpf},
    \quad f \in \setF, \ m \in [M]. 
  \end{align*}
  \edefinition
\end{definition}
%---------------------------------------------------------------------------

%***************************************************************************
In the case of Fig.~\ref{sec:SST:fig:9}, the function $P_e$ corresponds to the exterior function of the dashed box in Fig.~\ref{sec:SST:fig:8}. This function is obtained through the closing-the-box operation~\cite{Loeliger2004}. For the specific example in Fig.~\ref{sec:SST:fig:8}, the
matrix representation of $ P_e $ is given by (the same row and column indexing as
$\matr{P}_{e,\sigma_e}$)
\begin{align}
  \matr{P}_e
    &\defeq
       \sum_{\sigma_{e} \in \set{S}_2}
         \frac{1}{2!}
         \cdot
         \matr{P}_{e,\sigma_e}
     = \begin{pmatrix}
         1 & 0   & 0   & 0 \\
         0 & 1/2 & 1/2 & 0 \\
         0 & 1/2 & 1/2 & 0 \\
         0 & 0   & 0   & 1
       \end{pmatrix}
         \label{sec:SST:eqn:1}.
\end{align}

%***************************************************************************

%----------------------------------------------------------------------------
\begin{lemma}\label{sec:SST:lem:6}
  \label{SEC:SST:LEM:6}
  The partition function $ Z\bigl( \hgraphNavg \bigr) $ can be expressed as
  %---------------------------------------------------------------------------
  \begin{align*}
    Z\bigl( \hgraphNavg \bigr) = \sum_{\vx_{[M]}}
    \Biggl( \prod_{m \in [M]} \prod_f f(\vx_{\setpff,m}) \Biggr)
    \cdot \prod_{e} P_e\bigl( \vx_{\efi,[M]}, \vx_{\efj,[M]} \bigr).
  \end{align*}
  %---------------------------------------------------------------------------
\end{lemma}
%----------------------------------------------------------------------------
%----------------------------------------------------------------------------
\begin{proof}
  See Appendix~\ref{apx:alternative expression of ZBM by Pe}.
\end{proof}
%----------------------------------------------------------------------------

In the subsequent analysis, we use the language of the method of types\footnote{For a detailed exposition of the method of types, we refer to the book by Cover and Thomas~\cite{T.M.Cover2006}, although we use different notation in our discussion.} to characterize the function $P_e$. 

% In the following, we use the language of the method of types to characterize the
% function $P_e$. 

%***************************************************************************

\begin{definition}\label{sec:SST:def:6}
  We introduce the following objects:\footnote{Note that, for simplicity, we use $x$ instead of the more precise notation $\xe$.}
  \begin{itemize}

  \item The type
    $\vt_e(\vv_e) \defeq \bigl( t_{e,x}(\vv_e) \bigr)_{x \in \set{X}_e} \in \Pi_{\setxe}$ of a
    vector $\vv_e \defeq (v_{e,m})_{m \in [M]} \in \set{X}_e^M$ is defined to be
    \begin{align*}
      t_{e,x}(\vv_e)
        &\defeq
           \frac{1}{M}
           \cdot
           \bigl| \hskip0.5mm 
             \left\{
               m \in [M]
             \ \middle| \ 
               v_{e,m} = x
             \right\} \hskip0.5mm 
           \bigr|,
             \qquad x \in \set{X}_e.
    \end{align*}

  \item Let $\set{B}_{\set{X}_e^M}$ be the set of possible types of vectors of
    length $M$ over $\set{X}_e$, \ie,
    \begin{align*}
      \set{B}_{\set{X}_e^M}
        &\defeq
           \bigl\{
             \vt_e \in \Pi_{\setxe}
           \bigm|
             \text{there exists $\vv_e \in \set{X}_e^M$ 
                   such that $\vt(\vv_e) = \vt_e$}
           \bigr\}.
    \end{align*}
  Let $ \set{B}_{\set{X}^M} $ be the Cartesian product $ \prod_{e} \set{B}_{\set{X}_e^M} $.
    
  \item Let $\vt_e \in \set{B}_{\set{X}_e^M}$. Then the type class of $\vt_e$ is defined to be the set
    \begin{align*}
      \set{T}_{e,\vt_e}
      &\defeq
         \bigl\{
           \vv_e \in \set{X}_e^M
         \bigm|
           \vt_e(\vv_e) = \vt_e
         \bigr\}.
    \end{align*}
  
  \end{itemize}
  \edefinition
\end{definition}

%***************************************************************************

%---------------------------------------------------------------------------
\begin{lemma}
  \label{sec:SST:lem:2}

  It holds that
  \begin{align*}
    \bigl| \set{B}_{\set{X}_e^M} \bigr|
      &= \binom{|\set{X}_e| + M - 1}{M}, \\
    |\set{T}_{e,\vt_e}|
      &= \frac{M!}{
        \prod_{x \in \set{X}_e} 
        \bigl( (M \cdot t_{e,x})! \bigr)
      },
      \quad \vt_e \in \set{B}_{\set{X}_e^M}.
  \end{align*}
\end{lemma}

\begin{proof}
  These expressions can be derived from standard combinatorial results.
\end{proof}

%***************************************************************************

For the illustrated example, the number of possible types is
$\bigl| \set{B}_{\set{X}_e^M} \bigr| = \binom{2 + 2 - 1}{2} = 3$ and the set 
$\set{B}_{\set{X}_e^M} $ is given by $ \{ (1,0), \, (1/2,1/2), \, (0,1) \} $. The corresponding type classes have the following
sizes: $|\set{T}_{e,(1,0)}| = 1$, $|\set{T}_{e,(1/2,1/2)}| = 2$, and
$|\set{T}_{e,(0,1)}| = 1$.

%***************************************************************************

\begin{lemma}\label{sec:SST:lem:4}
  The function $P_e$ for $ e = (f_{i}, f_{j}) $ satisfies
  \begin{align*}
    P_e\bigl( \vx_{\efi,[M]}, \vx_{\efj,[M]} \bigr)
      &= \begin{cases}
           |\set{T}_{e,\vt_e}|^{-1}
             &  
             \vt_e = \vt_e\bigl(\vx_{\efi,[M]}\bigr) 
             = \vt_e\bigl( \vx_{\efj,[M]} \bigr) \\
           0
             & \text{otherwise}
         \end{cases}.
  \end{align*}
\end{lemma}

%***************************************************************************

\begin{proof}
  This result follows from the definition of the function $P_{e,\sigma_e}$ in
  Definition~\ref{sec:SST:def:1}, along with the properties of the uniform distribution $p_{\Sigmae}$ and the symmetry of $ P_e $, as shown in Definition~\ref{sec:SST:def:4}.
\end{proof}

%***************************************************************************

For the illustrated example, the expression in the above lemma is
corroborated by the matrix in~\eqref{sec:SST:eqn:1}. Here, the matrix is reproduced with horizontal and vertical lines to highlight the three type classes of the row and column indices, respectively:
\begin{align*}
  \matr{P}_e
    &= \left(
         \begin{array}{c|cc|c}
           1 & 0   & 0   & 0 \\
         \hline
           0 & 1/2 & 1/2 & 0 \\
           0 & 1/2 & 1/2 & 0 \\
         \hline
           0 & 0   & 0   & 1
         \end{array}
       \right).
\end{align*}

%***************************************************************************

The matrix $\matr{P}_e$ associated with $P_e$ is known as the
symmetric-subspace projection operator (see, \eg,~\cite{Harrow2013}).

Now we can express the $M$-th power of the degree-$M$ Bethe partition function as the partition function of the average degree-$M$ cover.
%----------------------------------------------------------------------------
\begin{theorem} 
  \label{thm: expression of Z for degree M Bethe partition function}
  The degree-$M$ Bethe partition function satisfies
  \begin{align*}
    \bigl( \ZBM(\graphN) \bigr)^M
    = 
    \sum_{\vx_{[M]}}
    \Biggl( \prod_{m \in [M]} \prod_f f(\vx_{\setpff,m}) \Biggr)
    \cdot \prod_{e} P_e\bigl( \vx_{\efi,[M]}, \vx_{\efj,[M]} \bigr).
  \end{align*}
  where
  \begin{align*}
    \prod_{e} P_e\bigl( \vx_{\efi,[M]}, \vx_{\efj,[M]} \bigr)
    = \begin{cases}
      \prod\limits_{e} \frac{1}{|\set{T}_{e,\vt_e}|}
       &  
       \vt_e = \vt_e\bigl(\vx_{\efi,[M]}\bigr) 
       = \vt_e\bigl( \vx_{\efj,[M]} \bigr),\, 
       \forall e \in \setEfull \\
      0
       & \text{otherwise}
    \end{cases}.
  \end{align*}
\end{theorem}
%----------------------------------------------------------------------------
%----------------------------------------------------------------------------
\begin{proof}
  This follows from Lemmas~\ref{sec:SST:lem:1},~\ref{sec:SST:lem:6},
  and~\ref{sec:SST:lem:4}.
\end{proof}
%----------------------------------------------------------------------------

\section[Reformulation of the function \texorpdfstring{$P_e$}{}]{Reformulation of the function \texorpdfstring{$P_e$}{} as an Integral}
\label{sec:SST:Pe:reformulation:1}

% \begin{figure}[t]
%   % \subfloat[]{
%     % \begin{minipage}[t]{0.3\textwidth}
%       \centering
%       \begin{tikzpicture}
%         \input{figures/head_files_figs.tex}
%         \input{figures/sst_examples/snfg/length_large.tex}
%         \input{figures/sst_examples/snfg/background_nodes_lines.tex}
%         \input{figures/sst_examples/snfg/lines_not_permuted.tex}
%         \input{figures/sst_examples/snfg/otb_pe.tex}
%         \node[] (var1) at (-0.18*\ldis,-1*\ldis) [label=right: $\cvpsi_{e}$] {};
%         \node[state] (Psige) at (0,-1.5*\ldis) [label=below: $ \bigl| \set{B}_{\set{X}_e^M}\! \bigr| \cdot \muFSsimple$] {};
%         \input{figures/sst_examples/snfg/otb_back_nodes_lines.tex}
%         \begin{pgfonlayer}{background}
%           \node[state_dash6] (db1) at (0,0.1*\ldis) [label=below:$P_{e}$] {};
%         \end{pgfonlayer}
%       \end{tikzpicture}
%     % \end{minipage}
%     \caption{The NFG for illustrating the SST for S-NFG.\label{sec:SST:fig:10}}
%   % }
% \end{figure}

%***************************************************************************

In this section, we explore an approach to express the function $P_e$ as an  integral. This integral representation can be utilized to modify the NFG in
Fig.~\ref{sec:SST:fig:9}, resulting in the NFG shown in
Fig.~\ref{sec:SST:fig:10}.

%***************************************************************************

%***************************************************************************

\begin{definition}\label{sec:SST:def:3}
  For the following considerations, we consider an arbitrary edge
  $e = (f_{i}, f_{j}) \in \setEfull$ and a positive integer $M$. We introduce
  the following notation:
  \begin{itemize}

    \item The $2$-norm of a function $\psi_e: \set{X}_e \to \sC$ is defined to be
      $\lVert \psi_e \rVert_{2} \defeq \sqrt{\sum_{\xe} |\psi_e(\xe)|^2}$.

    \item The 2-norm of a vector
      $\vpsi_e = \bigl( \psi_e(\xe) \bigr)_{\xe \in \set{X}_e} \in
      \sC^{|\set{X}_e|}$ is defined to be
      $\lVert \vpsi_e \rVert_{2} \defeq \sqrt{\sum_{\xe} |\psi_e(\xe)|^2} =
      \sqrt{\vpsi_e^\Herm \cdot \vpsi_e}$.

    % \item The vector containing the absolute values of the entries in the vector $ \cvpsi_e $ is defined to be $ |\cvpsi_e| = \bigl( |\cpsi_e(\xe)| \bigr)_{\xe \in \set{X}_e} \in
    %   \sR_{\geq 0}^{|\set{X}_e|} $.

  \end{itemize}

  Now, let $\cpsi_e: \set{X}_e \to \sC$ be a function with $2$-norm equals one, \ie,
  $\lVert \cpsi_e \rVert_{2} = 1$. Based on $\cpsi_e$, we define the function
  $\funcFS_{e,\cpsi_e}: \set{X}_e^{M} \times \set{X}_e^{M} \to \sC$ as follows:
  \begin{align*}
    \funcFS_{e,\cpsi_e}\bigl( \vx_{\efi,[M]}, \vx_{\efj,[M]} \bigr)
    &\defeq
      \bigl| \set{B}_{\set{X}_e^M} \bigr|
      \cdot
      \Biggl(
         \prod_{m \in [M]}
           \cpsi_e(x_{\efi,m})
      \Biggr)
      \cdot
      \Biggl(
         \prod_{m \in [M]}
           \overline{ \cpsi_e(x_{\efj,m}) }
      \Biggr),
  \end{align*}
  for all $ \vx_{\efi,[M]}, \vx_{\efj,[M]} \in \set{X}_e^M $.
  \edefinition
\end{definition}

%***************************************************************************
If we associate the function $\cpsi_e$ with the column vector
%----------------------------------------------------------------------------
\begin{align}
  \cvpsi_e \defeq \bigl( \cpsi_e(\xe) \bigr)_{\xe \in \set{X}_e}
  \in \sC^{|\setxe|},
  \label{sec:SST:eqn:39}
\end{align}
%----------------------------------------------------------------------------
and the
function $\funcFS_{e,\cpsi_e}$ with the matrix $\matrFS_{e,\cpsi_e}$ (with rows
indexed by $\vx_{\efi,[M]}$ and columns indexed by $\vx_{\efj,[M]}$), then we have
\begin{align*}
  \matrFS_{e,\cpsi_e}
    &= \bigl| \set{B}_{\set{X}_e^M} \bigr|
    \cdot
    \bigl( \cvpsi_e^{\otimes M} \bigr)
    \cdot
    \bigl( \cvpsi_e^{\otimes M} \bigr)^{\! \Herm}
    \in \sC^{|\setxe|^{M} \times |\setxe|^{M}},
\end{align*}
where $\cvpsi_e^{\otimes M}$ represents the $M$-fold Kronecker product of
$\cvpsi_e$ with itself.

%***************************************************************************

% Alternatively, with the introduction of some suitable locally compact
% topological group, it is a Haar measure over that group.

In the following, for every $ e \in \setEfull $, we introduce the so-called Fubini-Study measure in dimension $|\set{X}_e|$ which can be viewed as a, 
in a suitable sense, uniform measure $\muFSsimple$ on functions
$\cpsi_e: \set{X}_e \to \sC_{\geq 0}$ with a $ 2 $-norm of one. This measure is also a Haar measure.
For the purposes of this thesis, it is sufficient
to know that a sample $\cvpsi_{e} \in \sC^{|\set{X}_e|}$ from $\muFSsimple$ can be
generated as follows:
\begin{itemize}

\item Let
  \begin{align*}
    \vW_e 
      &\defeq
         \bigl(
           (W_{e,\xe,0}, W_{e,\xe,1})
         \bigr)_{\xe \in \set{X}_e}
           \in \sR^{|\set{X}_e|} \times \sR^{|\set{X}_e|}
  \end{align*}
  be a length-$2 |\set{X}_e|$ random vector with i.i.d. entries distributed
  according to a normal distribution with mean zero and variance one.

  \item Let $\vw_e$ be a realization of $\vW_e$. Let\footnote{Recall that
    ``$\imagunit$'' denotes the imaginary unit.}
  \begin{align*}
    \cvw_e
      &\defeq
         \vw_e / \lVert \vw_e \rVert_{2}, \nonumber \\
    \cpsi_e(\xe)
      &=
         \cw_{e,\xe,0} + \imagunit \cw_{e,\xe,1}. 
   \nonumber
  \end{align*}
   Because of the bijection between $\cvw$ and $\cvpsi$, we will use, with some slight abuse of notation, the notation $\muFSsimple$ not only for the distribution of $\cvpsi$, but also for the distribution of $\cvw$. In this case, the 2-norm of the function $ \cpsi_e $ is still one, \ie,
   %------------------------------------------------------------------------
   \begin{align*}
    \sum_{\xe} \bigl| \cpsi_e(\xe) \bigr|^{2} = 1. 
   \end{align*}
   %------------------------------------------------------------------------
\end{itemize}

%***************************************************************************

%---------------------------------------------------------------------------
\begin{lemma}\label{sec:SST:prop:2}
  \label{SEC:SST:PROP:2}
  For all $ \vx_{\efi,[M]}, \vx_{\efj,[M]} \in \set{X}_e^M $, it holds that 
  \begin{align*}
    \int
      \funcFS_{e,\cpsi_e}\bigl( \vx_{\efi,[M]}, \vx_{\efj,[M]} \bigr)
    \dd{\muFSsimple\bigl( \cvpsi_e \bigr)} 
    &= P_e\bigl( \vx_{\efi,[M]}, \vx_{\efj,[M]} \bigr).
  \end{align*}
\end{lemma}
%---------------------------------------------------------------------------
\begin{proof}
  See Appendix~\ref{apx:SST}.
\end{proof}

\section[Another NFG Representing the Average \texorpdfstring{$M$}{}-Cover]{An Alternative NFG Representing 
the Average \texorpdfstring{$M$}{}-Cover of an S-NFG\label{sec:SST:sunbsec:2}}

%***************************************************************************

In Section~\ref{sec:SST:NFG:average:cover:1}, we introduced the NFG
$\hgraphNavg$, whose partition function equals
$\bigl( \ZBM(\graphN) \bigr)^M$. Now, based on the
results from Section~\ref{sec:SST:Pe:reformulation:1}, we formulate an
alternative NFG, denoted by $\hgraphNavgalt$. The partition function of $\hgraphNavgalt$ also equals
$\bigl( \ZBM(\graphN) \bigr)^M$.

%***************************************************************************

Recall that $P_e\bigl( \vx_{\efi,[M]}, \vx_{\efj,[M]} \bigr)$ represents the exterior
function of the dashed box in Fig.~\ref{sec:SST:fig:8}; in other
words, if we close the dashed box in Fig.~\ref{sec:SST:fig:8}, then we
obtain a single function node representing the function
$P_e\bigl( \vx_{\efi,[M]}, \vx_{\efj,[M]} \bigr)$, as shown in Fig.~\ref{sec:SST:fig:9}.

%***************************************************************************

The NFG $\hgraphNavgalt$ in Fig.~\ref{sec:SST:fig:10} is defined
such that it is the same as the NFG $\hgraphNavg$ in
Fig.~\ref{sec:SST:fig:9}, except for the content of the function $ P_{e} $
for every $e \in \setEfull$. However, the content of the dashed box in
Fig.~\ref{sec:SST:fig:10} is set up in such a way that its exterior function
equals the function $ P_{e} $ in
Fig.~\ref{sec:SST:fig:9} as well as the exterior function of the dashed box in Fig.~\ref{sec:SST:fig:8}. 

By closing the dashed box in Fig.~\ref{sec:SST:fig:10}, a single function node representing the function $P_e$ is formed, as shown in Fig.~\ref{sec:SST:fig:9}. This ensures that the partition function of the NFG in Fig.~\ref{sec:SST:fig:10} equals the partition functions of the NFGs in Figs.~\ref{sec:SST:fig:8} and~\ref{sec:SST:fig:9}. In other words, it can be stated that $Z\bigl( \hgraphNavgalt \bigr) = Z\bigl( \hgraphNavg \bigr) = \bigl( \ZBM(\graphN) \bigr)^M$.\footnote{In terms of NFG language, the
  NFG in Fig.~\ref{sec:SST:fig:10} can be obtained
  by, first applying a closing-the-box operation and
  then an opening-the-box operation~\cite{Loeliger2004} for every $e \in \setEfull$. Note that
  closing-the-box and opening-the-box operations maintain the partition
  function of the NFG.}

% With that, closing the dashed box in
% Fig.~\ref{sec:SST:fig:10} results in a single function node
% representing the function $P_e$, as shown in Fig.~\ref{sec:SST:fig:9}. 

% This guarantees
% that the partition function of the NFG in Fig.~\ref{sec:SST:fig:10}
% equals the partition functions of the NFGs in
% Figs.~\ref{sec:SST:fig:8} and~\ref{sec:SST:fig:9}, \ie,
% $Z\bigl( \hgraphNavgalt \bigr) = Z\bigl( \hgraphNavg \bigr) = \bigl( \ZBM(\graphN) \bigr)^M$.\footnote{In terms of NFG language, the
%   NFG in Fig.~\ref{sec:SST:fig:10} can be obtained
%   by, first applying a closing-the-box operation and
%   then an opening-the-box operation~\cite{Loeliger2004} for every $e \in \setEfull$. Note that
%   closing-the-box and opening-the-box operations maintain the partition
%   function of the NFG.}

%***************************************************************************

In the NFG shown in Fig.~\ref{sec:SST:fig:10}, compared to the NFG in Fig.~\ref{sec:SST:fig:8}, the following new elements are introduced for every $ e \in \setEfull $:
\begin{itemize}

\item An edge is added with an associated variable (vector)
  $\cvpsi_e \in \sC^{|\set{X}_e|}$.

\item A function node is added, representing the function
  \begin{align}
    p_e\bigl( \cvpsi_e \bigr)
      &\defeq
         \bigl| \set{B}_{\set{X}_e^M} \bigr|
         \cdot
         \muFSsimple\bigl( \cvpsi_e \bigr).
     \label{eqn: def of pe}
  \end{align}

\item $M$ function nodes representing the function $\cpsi_e$ are added, and $M$ function nodes representing the function $\overline{\cpsi_e}$ are added.  The function $\cpsi_e$ is specified by $\cvpsi_e$ in a straightforward manner.

\end{itemize}

%***************************************************************************

%---------------------------------------------------------------------------
\begin{definition}
  \label{sec:SST:def:5}
  \index{SST!for S-NFG}
  We make the following definitions for $ \avgalt{\graphN} $.
  %-----------------------------------------------------------------------
  \begin{enumerate}

    \item We define a set of variables
    \begin{align*}
      \vxavgalt \defeq ( x_{\ef,m} )_{e \in \setpf, \, f \in \setF, \, m \in [M]} 
      \in \setx{}^{2M}.
    \end{align*}

    \item We define a set of vectors 
    \begin{align*}
      \cvpsiavgalt \defeq ( \cvpsi_e )_{e \in \setEfull},
    \end{align*}
    where $ \cvpsi_{e} $ is defined in~\eqref{sec:SST:eqn:39} for all $ e \in \setEfull $.

    \item We define the measure $ \muFSsimple\bigl( \cvpsiavgalt \bigr) $ to be the product of Fubini-Study measures: 
    \begin{align*}
      \muFSsimple\bigl( \cvpsiavgalt \bigr) 
      \defeq \prod_{e} \muFSsimple\bigl( \cvpsi_e \bigr).
    \end{align*}

    \item For each $ e = (f_{i}, f_{j}) \in \setEfull $ such that $ i<j $, we define 
    \begin{align*}
      \cpsi_{\efi}(\xe)
      \defeq \cpsi_{e}(\xe), \qquad
      \cpsi_{\efj}(\xe)
      \defeq \overline{\cpsi_{e}(\xe)}, \qquad
      \xe \in \setxe.
    \end{align*}
  \end{enumerate}

  The construction of $\avgalt{\graphN}$ is given as follows.
  %----------------------------------------------------------------------------
  \begin{enumerate}
    \item For every $f \in \set{F}$, we follow the same step as in Step~\ref{step: for f covers} in Definition~\ref{sec:SST:def:1}.

    \item For every $ e = (f_{i}, f_{j}) \in \setEfull $, we add the following elements:
    \begin{itemize}

    \item an edge associated with variable (vector) 
    $\cvpsi_e \in \sC^{|\set{X}_e|}$;

    \item a function node representing the function $ p_e\bigl( \cvpsi_e \bigr) $ as defined in~\eqref{eqn: def of pe};

    \item $M$ function nodes representing the function $\cpsi_e$;

    \item $M$ function nodes representing the function $\overline{\cpsi_e}$. 

    \end{itemize}
    Then we connect each of the $M$ sockets $ (x_{e,f_{i},m})_{m \in [M]} $ to distinct function nodes representing the function $ \cpsi_e $. Similarly, we connect each of the $M$ sockets $ (x_{e,f_{j},m})_{m \in [M]} $ to distinct function nodes representing the function $ \overline{\cpsi_e} $. 

  \end{enumerate}
  %----------------------------------------------------------------------------

  \edefinition
\end{definition}
%---------------------------------------------------------------------------

With this, the global function of $\hgraphNavgalt$ is
\begin{align}
  g\bigl( \vxavgalt, \cvpsiavgalt \bigr)
    &= \left(
       \prod_f
         \prod_{m \in [M]} f(\vx_{\setpff,m})
     \right)
     \nonumber\\
     &\quad \cdot
     \left(
       \prod_{e=(f_{i},f_{j})}
          \Biggl( 
           \bigl| \set{B}_{\set{X}_e^M} \bigr|
           \cdot
           \prod_{m \in [M]}
            \bigl( 
              \cpsi_{\efi}(x_{\efi,m})
              \cdot
              \cpsi_{\efj}(x_{\efj,m})
            \bigr)  
          \Biggr)
     \right) \label{sec:SST:eqn:26} \\
     &= \Biggl( \prod_f \prod_{m \in [M]} f(\vx_{\setpff,m}) \Biggr)
     \cdot \prod_{e} 
     \funcFS_{e,\cpsi_e}\bigl( \vx_{\efi,[M]}, \vx_{\efj,[M]} \bigr), 
     \label{sec:SST:eqn:10},
\end{align}
%***************************************************************************
for all $ \vx_{\efi,[M]}, \vx_{\efj,[M]} \in \set{X}_e^M, $
where the last equality follows from the definition of $ \funcFS $ in Definition~\ref{sec:SST:def:3}. We define $ \ZSSTM $ to be the partition function of $\hgraphNavgalt$ for fixed $ \cvpsiavgalt $:
%----------------------------------------------------------------------------
\begin{align}
  \ZSSTM\bigl( \avgalt{\graphN}, \cvpsiavgalt \bigr) \defeq 
  \sum_{\vxavgalt} g\bigl( \vxavgalt, \cvpsiavgalt \bigr), 
  \label{sec:SST:eqn:43}
\end{align}
%----------------------------------------------------------------------------
where $ \sum_{\vxavgalt} $ represents $ \sum_{ \vxavgalt \in \setx{}^{2M} } $.

%---------------------------------------------------------------------------
\begin{proposition}
  \label{sec:SST:prop:1}
  \label{SEC:SST:PROP:1}

  The partition function of $\hgraphNavgalt$ satisfies
  \begin{align*}
    Z\bigl( \hgraphNavgalt \bigr)
    &= \int \ZSSTM\bigl( \avgalt{\graphN}, \cvpsiavgalt \bigr)
    \dd{\muFSsimple\bigl( \cvpsiavgalt \bigr)}
    \nonumber\\
    &= 
    \left( \prod_e\bigl| \set{B}_{\set{X}_e^M} \bigr| \right)
      \cdot  
      \int \prod_f \Bigl( \ZSSTf\bigl( \cvpsi_{\setpff} \bigr) \Bigr)^{ \! M}
      \dd{\muFSsimple\bigl( \cvpsiavgalt \bigr)} 
      \\
    &= Z\bigl( \hgraphNavg \bigr),
  \end{align*}
  where for each function node $ f \in \setF $, we define
  %----------------------------------------------------------------------------
  \begin{align}
    \ZSSTf\bigl( \cvpsi_{\setpff} \bigr) 
    \defeq \sum_{\vx_{\setpf}}
      f(\vx_{\setpf}) 
      \cdot
      \prod_{e \in \setpf} 
      \cpsi_{e,f}(x_{e}), \label{sec:SST:eqn:42}
  \end{align}
  %----------------------------------------------------------------------------
  and $ \cvpsi_{\setpff} \defeq ( \cvpsi_{\ef} )_{e \in \setpf} $.
\end{proposition}
%***************************************************************************

\begin{proof}
  See Appendix~\ref{apx:property of SST}.
\end{proof}

%***************************************************************************
Consider conditioning on $\cvpsiavgalt$. The simplification $ \prod_f \Bigl( \ZSSTf\bigl( \cvpsi_{\setpff} \bigr) \Bigr)^{ \! M} $ obtained in Proposition~\ref{sec:SST:prop:1} is made possible by decomposing the NFG into $|\setF| \cdot M$ tree-structured components. This decomposition allows us to factorize and separate the partition function into individual components, each corresponding to a tree-structured factor graph consisting of a function node. 

% The simplification that was achieved in the proof of
% Proposition~\ref{sec:SST:prop:1} was thanks to the fact that,
% conditioned on $\cvpsiavgalt$, the NFG decomposes into an NFG with
% $|\setF| \cdot M$ components, where each component is tree-structured.

%***************************************************************************
%***************************************************************************

\section{Combining the Results of this Section}
\label{sec:SST:combining:results:1}
%***************************************************************************

Combining the main results of this section, \ie,
\begin{alignat}{2}
  \text{(Introduction of Chapter~\ref{chapt:SST})} \quad &&
  \ZBM(\graphN)
    &=
       \sqrt[M]{
         \Bigl\langle 
           Z\bigl( \hgraphN \bigr)
         \Bigr\rangle_{\hgraphN \in \hat{\set{N}}_{M}} }, 
  \label{sec:SST:eqn:expresssion of ZBM} \\
  \text{(Lemma~\ref{sec:SST:lem:1})} \quad &&
  Z\bigl( \hgraphNavg \bigr)
    &= \bigl( \ZBM(\graphN) \bigr)^{\! M}, 
  \label{sec:SST:eqn:23} \\
  \text{(Proposition~\ref{sec:SST:prop:1})} \quad &&
  Z\bigl( \hgraphNavgalt \bigr)
    &= Z\bigl( \hgraphNavg \bigr), 
  \label{sec:SST:eqn:24} \\
  \text{(Proposition~\ref{sec:SST:prop:1})} \quad &&
  Z\bigl( \hgraphNavgalt \bigr)
    &=\left( \prod_e\bigl| \set{B}_{\set{X}_e^M} \bigr| \right)
    \nonumber \\
    &&&\quad \cdot  
    \int \prod_f \Bigl( \ZSSTf\bigl( \cvpsi_{\setpff} \bigr) \Bigr)^{ \! M}
    \dd{\muFSsimple\bigl( \cvpsiavgalt \bigr)}
\end{alignat}
we can obtain the expression for $\ZBM(\graphN)$ as follows:
\begin{align}
  \ZBM(\graphN)
  &=  
  \Biggl( \prod_e \bigl| \set{B}_{\set{X}_e^M} \bigr| \Biggr)^{\!\!\! 1/M}
  \cdot \Biggl(
  \int \prod_f \Bigl( \ZSSTf\bigl( \cvpsi_{\setpff} \bigr) \Bigr)^{ \! M}
  \dd{\muFSsimple\bigl( \cvpsiavgalt \bigr)}
  \Biggr)^{\!\!\! 1/M}. \label{alternative expression of ZBM in snfg}
\end{align}

\section{SST for DE-NFGs}
\label{sec:SST:DENFG:1}
\index{SST!for DE-NFG}
So far, this section was about the SST for an S-NFGs. The extension to the SST
of a DE-NFGs is, by and large, straightforward. Comparing with
Fig.~\ref{sec:SST:fig:1} for S-NFGs, we note that the only major change is that for
every edge $e \in \setEfull$, the role of the alphabet $\set{X}_e$ is taken over by
the alphabet $\tset{X}_e = \set{X}_e \times \set{X}_e$.

With this, 
the value of
$\ZBM(\graphN)$ for a DE-NFG $\graphN$ can be written as (compare
with Theorem~\ref{thm: expression of Z for degree M Bethe partition function})
\begin{align}
    \bigl( \ZBM(\graphN) \bigr)^M
    = 
    \sum_{\vx_{[M]}}
    \Biggl( \prod_{m \in [M]} \prod_f f(\tvx_{\setpff,m}) \Biggr)
    \cdot \prod_{e} P_e\bigl( \tvx_{\efi,[M]}, \tvx_{\efj,[M]} \bigr),
  \label{eqn: Z for degree M average cover of DENFG}
\end{align}
where
\begin{align*}
  P_e\bigl( \tvx_{\efi,[M]}, \tvx_{\efj,[M]} \bigr)
  = \begin{cases}
    \frac{1}{|\set{T}_{e,\vt_e}|}
     &  
     \vt_e = \vt_e\bigl(\tvx_{\efi,[M]}\bigr) 
     = \vt_e\bigl( \tvx_{\efj,[M]} \bigr) \\
   0
     & \text{otherwise}
  \end{cases}.
\end{align*}
Following the similar idea in Section~\ref{sec:SST:combining:results:1}, we also rewrite $\ZBM(\graphN)$ as 
\begin{align}
    \bigl( \ZBM(\graphN) \bigr)^{\!M}
    &= \biggl( \prod_e |\set{B}_{\tset{X}_e^M}| \biggr)^{\!\!\! 1/M}
    \!\!\!\!\cdot\! \left( \int \prod_f 
      \Biggl( \sum_{\tvx_f}
      f(\tvx_f) \prod_{e \in \setpf} 
      \cpsi_{e,f}(\tx_{e,f})
      \Biggr)^{\!\!\!M}
      \!\!\!
      \dd{\muFSsimple(\cvpsiavgalt)}
    \right)^{\!\!\!\! 1/M}
    \!\!\!\!,
\end{align}
where $\cvpsiavgalt \defeq ( \cvpsi_e )_{e \in \setEfull}$, 
and where for every
$e \in \setEfull$, the norm-one vector $\cvpsi_e$ is distributed according to the
Fubini-Study measure in dimension $|\tset{X}_e|$. 

%% file: appendices/Props_SST.tex
It holds that
%---------------------------------------------------------------------------
\begin{align*}
    Z(\hgraphNavgalt)
    &= \int  
    \ZSSTM(\avgalt{\graphN},\cvpsiavgalt)
    \dd{\muFSsimple(\cvpsiavgalt)} \nonumber \\
    &\overset{(a)}{=}
    \int  
    \sum_{\vxavgalt} g\bigl( \vxavgalt, \cvpsiavgalt \bigr)
    \dd{\muFSsimple(\cvpsiavgalt)}
    \nonumber\\
    &\overset{(b)}{=}  
        \Biggl( \prod_{e}|\set{B}_{\set{X}_e^M}| \Biggr)
        \cdot
        \int
             \sum_{\vxavgalt}
                \prod_f
                \prod_{m \in [M]}
                f(\vx_{\setpf,m})
                \cdot 
                % \Biggl( 
                \prod_{e \in \setpf} 
                \cpsi_{e,f}(x_{e,m})
                % \Biggr) 
         \, \dd{\muFSsimple(\cvpsiavgalt)} \\
    &= 
        \Biggl( \prod_{e}|\set{B}_{\set{X}_e^M}| \Biggr)
        \cdot
        \int
            \prod_f
                \prod_{m \in [M]}
                    \Biggl(
                        \sum_{\vx_{\setpf,m}}
                        f(\vx_{\setpf,m})
                        \cdot
                        % \Biggl( 
                            \prod_{e \in \setpf}
                            \cpsi_{e,f}(x_{e,m})
                        % \Biggr)
                     \Biggr)
         \, \dd{\muFSsimple(\cvpsiavgalt)} \\
    &=  
        \Biggl( \prod_{e}|\set{B}_{\set{X}_e^M}| \Biggr)
        \cdot
        \int
            \prod_f
                \Biggl(
                    \sum_{\xf}
                            f(\xf)
                            \prod_{e \in \setpf}
                                \cpsi_{e,f}(x_{e})
                \Biggr)^{\!\!\! M} \!\!
         \, \dd{\muFSsimple(\cvpsiavgalt)} \nonumber\\
    &\overset{(c)}{=}  
        \Biggl( \prod_{e}|\set{B}_{\set{X}_e^M}| \Biggr)
        \cdot
        \int
            \prod_f
                \bigl(
                    \ZSSTf( \cvpsi_{\setpff} )
                \bigr)^{\! M} 
         \, \dd{\muFSsimple(\cvpsiavgalt)},
\end{align*}
%---------------------------------------------------------------------------
where step $(a)$ follows from the definition of $ \ZSSTM $ in~\eqref{sec:SST:eqn:43},
where step $(b)$ follows from the expression of $ g\bigl( \vxavgalt, \cvpsiavgalt \bigr) $ in~\eqref{sec:SST:eqn:26},
and where step $(c)$ follows from the definition of $ \ZSSTf $ in~\eqref{sec:SST:eqn:42}.
Also, we have
%------------------------------------------------------------------------
\begin{align*}
    \hspace{0.5cm}&\hspace{-0.5cm}\int  
    \sum_{\vxavgalt} g\bigl( \vxavgalt, \cvpsiavgalt \bigr)
    \dd{\muFSsimple(\cvpsiavgalt)} 
    \nonumber\\
    &\overset{(a)}{=} \int 
    \sum_{\vxavgalt} \Biggl( \prod_f \prod_{m \in [M]} f(\vx_{\setpff,m}) \Biggr)
    \cdot \prod_{e} \funcFS_{e,\cpsi_e}\bigl( \vx_{\efi,[M]}, \vx_{\efj,[M]} \bigr) 
    \dd{\muFSsimple(\cvpsiavgalt)} 
    \nonumber\\
    &\overset{(b)}{=}  \sum_{\vxavgalt} 
    \Biggl( \prod_f \prod_{m \in [M]} f(\vx_{\setpff,m}) \Biggr)
    \cdot \prod_{e} \int 
    \funcFS_{e,\cpsi_e}\bigl( \vx_{\efi,[M]}, \vx_{\efj,[M]} \bigr) 
    \dd{\muFSsimple(\cvpsi_e)} 
    \nonumber\\
    &\overset{(c)}{=} \sum_{\vxavgalt} 
    \Biggl( \prod_f \prod_{m \in [M]} f(\vx_{\setpff,m}) \Biggr)
    \cdot \prod_{e} P_e\bigl( \vx_{\efi,[M]}, \vx_{\efj,[M]} \bigr) \nonumber\\
    &\overset{(d)}{=} Z\bigl( \hgraphNavg \bigr),
\end{align*}
%------------------------------------------------------------------------
where step $(a)$ follows from the expression of $ g\bigl( \vxavgalt, \cvpsiavgalt \bigr) $ in~\eqref{sec:SST:eqn:10},
where step $(b)$ follows from the fact $ \dd{\muFSsimple(\cvpsiavgalt)} = \prod_{e} \dd{\muFSsimple(\cvpsi_{e})}  $ as defined in Definition~\ref{sec:SST:def:5},
where step $(c)$ follows from the property of $ P_{e} $ in Lemma~\ref{sec:SST:prop:2}, and where step $(d)$ follows from the expression of $ Z\bigl( \hgraphNavg \bigr) $ in Lemma~\ref{sec:SST:lem:6}.

%% file: appendices/check_graph_thm.tex
The following definitions and remark provide the groundwork for the subsequent derivations in this appendix.

\begin{definition}\label{sec:CheckCon:def:1}
  We make the following definitions.
  %---------------------------------------------------------------------------
  \begin{enumerate}

    \item Define $ \LCTtset{X} \defeq \prod\limits_{e} \LCTtset{X}_e$.

    \item For each $ e = (f_{i},f_{j}) \in \setEfull $ with $ i < j $, two collections of variables are defined:
    %-----------------------------------------------------------------------
    \begin{align*}
      \LCTvxs \defeq 
      ( \LCTt{x}_{e,f_{i}}  
      )_{ e = (f_{i},f_{j}) \in \setEfull:\, i<j }
      \in \LCTtset{X},
      \qquad 
      \LCTvxl \defeq ( \LCTt{x}_{e,f_{j}} 
      )_{ e = (f_{i},f_{j}) \in \setEfull:\, i<j }
      \in \LCTtset{X}.
    \end{align*}
    %-----------------------------------------------------------------------
    Here, $\LCTvxs$ represents the collection of variables corresponding to the function nodes with ``small'' indices (\ie, $f_i$ for each $ e = (f_{i},f_{j}) \in \setEfull$), and $\LCTvxl$ represents the collection of variables corresponding to the function nodes with the ``large'' indices (\ie, $f_j$ for each $ e = (f_{i},f_{j}) \in \setEfull$).
    Recall that in Definition~\ref{sec:DENFG:def:4}, the order of the function nodes is fixed by $ \set{F} = \{ f_{1},\ldots,f_{|\set{F}|} \}. $.

    \item The functions $ \gzero $ and $ \gone $ are defined as follows:
    %-----------------------------------------------------------------------
    \begin{align}
      \gzero
      (\LCTvxs, \LCTvxl) &\defeq
      \Biggl( \prod_{f}\LCT{f}( \LCTtv{x}_{\setpf,f}) \Biggr) \cdot
      [\LCTvxs\!=\!\LCTvxl \!=\! \tv{0}],  \qquad 
      \LCTvxs, \LCTvxl \in \LCTtset{X},
      \label{sec:CheckCon:eqn:60}\\
      \gone(\LCTvxs, \LCTvxl) &\defeq
      \Biggl( \prod_{f}\LCT{f}( \LCTtv{x}_{\setpf,f}) \Biggr) \cdot
      [ \wh(\LCTvxs) + \wh( \LCTvxl) \!\geq\! 1 ], \qquad 
      \LCTvxs, \LCTvxl \in \LCTtset{X},
      \label{sec:CheckCon:eqn:59}
    \end{align}
    %-----------------------------------------------------------------------
    where
    \begin{align*}
      \wh( \LCTvxs )
      =
      \sum_{ e }
      \wh( \LCTt{x}_{e,f_{i}}  ), \qquad 
      \wh( \LCTvxl )
      =
      \sum_{ e }
      \wh( \LCTt{x}_{e,f_{j}}  ),
    \end{align*}
    and where $ \wh(\cdot) $ denotes the Hamming weight 
    (the number of nonzero elements) of a vector:
    \begin{align*}
      \wh( \LCTt{x}_{e} )
      &\defeq \bigl[ \LCTt{x}_{e} \!\neq\! \tilde{0} \bigr], \qquad 
      \LCTt{x}_{e} \in \LCTtset{X}_{e}.
    \end{align*}

    \item A product over an empty index set evaluates to $1$. For example, if 
    $ \set{I} = \emptyset $, then we have 
    $ \prod_{ m \in  \set{I} } (\ldots) = 1 $.

  \end{enumerate}
  %---------------------------------------------------------------------------
  \edefinition
\end{definition}
%---------------------------------------------------------------------------

For simplicity, we use $ ( \cdot )_{m} $ and $ \sum_{ \LCTvxs, \LCTvxl } $, instead of $ (\cdot)_{m \in [M]} $ and $ \sum_{ \LCTvxs, \LCTvxl \in \LCTtset{X} } $, if there is no ambiguity.

\begin{remark}
    For each $ e = (f_{i},f_{j}) \in \setEfull $, 
    the function
    $P_{e}\bigl( \LCTtv{x}_{\efi,[M]}, \LCTtv{x}_{\efj,[M]} \bigr)$ is obtained by
    replacing $( \vx_{\efi,[M]}, \vx_{\efj,[M]} ) \in \setxe^{M} \times \setxe^{M} $ in $P_{e}( \vx_{\efi,[M]}, \vx_{\efj,[M]} )$ (as defined in Definition~\ref{sec:SST:def:4})
    with $ ( \LCTtv{x}_{\efi,[M]}, \LCTtv{x}_{\efj,[M]} ) \in \LCTtset{X}_e^{M} \times \LCTtset{X}_e^{M} $, where $ \LCTtv{x}_{\efi,[M]} \defeq ( \LCTt{x}_{\efi,m} )_{m}$, and $\LCTtv{x}_{\efj,[M]} \defeq ( \LCTt{x}_{\efj,m} )_{m}$.
    The properties of the function $P_{e}$ proven earlier hold after applying this replacement. Furthermore, proving these properties in the new setup is a straightforward extension of the previous proof.
    \eremark
\end{remark}

\begin{lemma}\label{sec:CheckCon:prop:properties of N(0)}
  The functions $ \gzero $ and $ \gone $ possess the following properties.
  \begin{enumerate}

    \item It holds that
    \begin{align}
        \gzero(\tv{0}, \tv{0}) = 
        \sum_{ \LCTvxs, \LCTvxl }
        \gzero( \LCTvxs, \LCTvxl ) 
        &=\ZBSPA^{*}(\graphN). \label{sec:CheckCon:eqn:73}
    \end{align}

    \item It holds that
      \begin{align*}
        \sum_{ \LCTvxs, \LCTvxl }
        \gone( \LCTvxs, \LCTvxl )
        &= \prod_{f}
         \left( 
            \sum_{\LCTtv{x}_{\setpf}}
            \LCT{f}( \LCTtv{x}_{\setpf} ) 
        \right)
        -
        \ZBSPA^{*}(\graphN).
      \end{align*}
    Also, we have $ \gone( \LCTvxs, \LCTvxl ) = 0 $ if $ \wh(\LCTvxs) + \wh( \LCTvxl) = 1$.
    
    % \begin{align*}
    %     \prod_{m =1}^{0}
    %     \gzero( \LCTvxsm, \LCTvxlm ) 
    %     = \prod_{m =M+1}^{M}
    %       \gone( \LCTvxsm, \LCTvxlm ) 
    %     = 1, \qquad \LCTvxsm, \LCTvxlm \in \prod_{e} \LCTtset{X}_e,
    % \end{align*}
    % which allows us to express $ \bigl( \ZBM(\graphN) \bigr)^{\! M} $ as follows:
    \item The function $ \bigl( \ZBM(\graphN) \bigr)^{\! M} $ can be decomposed as follows:
    \begin{align}
      \hspace{0.5 cm}&\hspace{-0.5 cm}
      \bigl( \ZBM(\graphN) \bigr)^{\! M}
      \nonumber\\
      &= 
      \sum_{k = 0}^{M}
      \binom{M}{k}
        \cdot 
        \sum_{ \LCTvxsM, \LCTvxlM \in \LCTtset{X}^{M} }
        \left( 
          \prod_{ m = 1 }^{k}
          \gzero( \LCTvxsm, \LCTvxlm ) 
        \right)
        \cdot 
        \left( 
          \prod_{ m = k+1 }^{M}
          \gone( \LCTvxsm, \LCTvxlm ) 
        \right)
      \nonumber\\
      &\hspace{5.2 cm}
      \cdot \prod_{e} 
      P_{e}\bigl( \LCTtv{x}_{\efi,[M]}, \LCTtv{x}_{\efj,[M]} \bigr),
      \label{sec:CheckCon:eqn:64}
    \end{align}
    where
    \begin{alignat*}{4}
        \LCTtv{x}_{\efi,[M]}
        &\defeq ( \LCTt{x}_{\efi,m} )_{m} \in 
        \LCTtset{X}_{e}^{M},  \qquad
        &\LCTtv{x}_{\efj,[M]}
        &\defeq ( \LCTt{x}_{\efj,m} )_{m} \in 
        \LCTtset{X}_{e}^{M}, 
        \nonumber\\
        \LCTvxsM 
        &\defeq (\LCTvxsm)_{m} \in \LCTtset{X}^{M},
        \qquad
        &\LCTvxlM 
        &\defeq (\LCTvxlm)_{m} \in \LCTtset{X}^{M},
        \nonumber\\
        \LCTvxsm 
        &\defeq
        ( \LCTt{x}_{\efi,m} )_{ e = (f_{i},f_{j}) \in \setEfull:\, i <j }
        \in \LCTtset{X}, 
        \nonumber\\
        \LCTvxlm 
        &\defeq
        ( \LCTt{x}_{\efj,m} )_{ e = (f_{i},f_{j}) \in \setEfull:\, i <j }
        \in \LCTtset{X}.
    \end{alignat*}
    and the arguments $\LCTt{x}_{\efi,m}$ and $\LCTt{x}_{\efj,m}$ take values  
    in $\LCTtset{X}_e$ for all $e = (f_{i}, f_{j}) \in \setEfull$ and $ m \in [M] $.
    By the definition of a product over an empty set in Definition~\ref{sec:CheckCon:def:1}, 
    we have $ \prod_{m=1}^{0}(\ldots) = \prod_{m=M+1}^{M}(\ldots) = 1 $.
    For simplicity, if there is no ambiguity, we use $ \sum_{\LCTvxsM} $ and $ \sum_{\LCTvxlM} $ instead of 
    $ \sum_{\LCTvxsM \in \LCTtset{X}^{M} } $ and 
    $ \sum_{\LCTvxlM \in \LCTtset{X}^{M} } $, respectively.

  \end{enumerate}
\end{lemma}

\begin{proof}
  The proof of each property is listed as follows.
  %---------------------------------------------------------------------------
  \begin{enumerate}

    \item This follows from the definition of $ \gzero $ in~\eqref{sec:CheckCon:eqn:60}  and the properties of the LCT proven in Proposition~\ref{prop:DENFG:LCT:1}.

    \item This follows from the definition of $ \gone $ in~\eqref{sec:CheckCon:eqn:59} and the properties of the LCT proven in Proposition~\ref{prop:DENFG:LCT:1}.

    \item It holds that
    \begin{align*}
      \hspace{0.5 cm}&\hspace{-0.5 cm}\bigl( \ZBM(\graphN) \bigr)^{\! M} 
      \nonumber\\
      &\overset{(a)}{=}
      \sum_{ \LCTvxsM, \LCTvxlM }
      \Biggl( 
        \prod_{m \in [M]} \prod_{f} 
        \LCT{f}( \LCTtv{x}_{\setpf,f,m}) 
      \Biggr) 
      \cdot \prod_{e = (f_{i},f_{j})} 
      P_{e}(\LCTtv{x}_{\efi,[M]}, \LCTtv{x}_{\efj,[M]}) \nonumber\\
      &\overset{(b)}{=} \sum_{ \LCTvxsM, \LCTvxlM }
      \Biggl( 
        \prod_{m \in [M]} 
        \Bigl( \gzero( \LCTvxsm, \LCTvxlm ) 
        + \gone( \LCTvxsm, \LCTvxlm ) 
        \Bigr)
      \Biggr) 
      \nonumber\\
      &\hspace{2.4 cm} \cdot \prod_{e = (f_{i},f_{j})} P_{e}\bigl( \LCTtv{x}_{\efi,[M]}, \LCTtv{x}_{\efj,[M]} \bigr)
      \nonumber\\
      &\overset{(c)}{=}
      \sum_{k = 0}^{M}
      \binom{M}{k}
      \cdot \sum_{ \LCTvxsM, \LCTvxlM }
        \left( 
          \prod_{m = 1}^{k}
          \gzero( \LCTvxsm, \LCTvxlm ) 
        \right)
        \cdot 
        \left( 
          \prod_{m = k+1 }^{M}
          \gone( \LCTvxsm, \LCTvxlm ) 
        \right) 
      \nonumber\\
      &\hspace{4.5 cm}
      \cdot \prod_{e = (f_{i},f_{j})} 
      P_{e}\bigl( \LCTtv{x}_{\efi,[M]}, \LCTtv{x}_{\efj,[M]} \bigr),
    \end{align*}
    \begin{itemize}
      \item where step $(a)$ follows from 
      replacing $ \tx_{e} $, $ \tset{X}_e $, and
      $ f(\tvx_{\setpf}) $, $ f(\tvx_{\setpff,m}) $ 
      with $ \LCTt{x}_{e} $, $ \LCTtset{X}_{e} $, $ \LCT{f}( \LCTtv{x}_{\setpf,f}) $,
      and $ \LCT{f}( \LCTtv{x}_{\setpf,f,m}) $, respectively, in~\eqref{eqn: Z for degree M average cover of DENFG},
      (Eqn.~\eqref{eqn: Z for degree M average cover of DENFG} still holds after applying these substitutions.)

      % \item where step $(b)$ follows from applying the same substitutions as in step $(a)$ to the expression of $ Z\bigl( \hgraphNavg \bigr) $ in Lemma~\ref{sec:SST:lem:6},
      % (Lemma~\ref{sec:SST:lem:6} still hold after applying these substitutions)
      % \begin{align*}
      %   \hspace{0.25cm}&\hspace{-0.25cm}
      %   \bigl( \ZBM(\graphN) \bigr)^{\! M} 
      %   = Z\bigl( \hgraphNavg \bigr)
      %   = \sum_{\vx_{[M]}}
      %   \Biggl( \prod_{m \in [M]} \prod_f f(\vx_{\setpff,m}) \Biggr)
      %   \cdot \prod_{e} P_e\bigl( \vx_{\efi,[M]}, \vx_{\efj,[M]} \bigr),
      % \end{align*}
      % \begin{itemize}

      \item where step $(b)$ follows from the definitions of $ \gzero $ and $ \gone $ in Definition~\ref{sec:CheckCon:def:1},

      \item where step $(c)$ follows from the following derivations:
      for any $ \sigma \in \set{S}_{[M]} $,  where $ \set{S}_{[M]} $ is the set of all permutations of elements 
        in $[M] = \{1,\ldots,M\}$, the following equalities hold:
        \begin{align*}
          &\vt_e\bigl(\LCTtv{x}_{\efi,[M]}\bigr) 
          = \vt_e\bigl(\LCTtv{x}_{\efi,\sigma([M])}\bigr), \qquad  
          \vt_e\bigl(\LCTtv{x}_{\efj,[M]}\bigr) 
          = \vt_e\bigl(\LCTtv{x}_{\efj,\sigma([M])}\bigr), 
        \end{align*}
        for all 
        $ \LCTtv{x}_{\efi,[M]},\LCTtv{x}_{\efj,[M]} \in \LCTtset{X}_{e}^{M}$ 
        and 
        $ e = (f_{i}, f_{j}) \in \setEfull $,
        where 
        \begin{align*}
          \LCTtv{x}_{\efi,\sigma([M])} 
          \defeq ( \LCTt{x}_{\efi,\sigma(m)} )_{m} \in \LCTtset{X}_{e}^{M}, \qquad
          \LCTtv{x}_{\efj,\sigma([M])} \defeq (  \LCTt{x}_{\efj,\sigma(m)} )_{m}
          \in \LCTtset{X}_{e}^{M};
        \end{align*}
        then, by the property of $ P_{e} $ Lemma~\ref{sec:SST:lem:4}, we obtain
        \begin{align*}
          P_{e}\bigl( \LCTtv{x}_{\efi,[M]}, \LCTtv{x}_{\efj,[M]} \bigr)
          = P_{e}\bigl( \LCTtv{x}_{\efi,\sigma([M])}, 
          \LCTtv{x}_{\efj,\sigma([M])} \bigr);
        \end{align*}
      %----------------------------------------------------------------------------
      finally, the above equality for $ P_{e} $ implies
        %-----------------------------------------------------------------------
        \begin{align*}
          &\hspace{-0.25cm}
          \left( 
            \prod_{m = 1}^{k}
            \gzero( \LCTvxsm, \LCTvxlm ) 
          \right)
          \cdot 
          \left( 
            \prod_{m = k+1}^{M}
            \gone( \LCTvxsm, \LCTvxlm ) 
          \right) 
          \cdot \prod_{e} P_{e}\bigl( \LCTtv{x}_{\efi,[M]}, \LCTtv{x}_{\efj,[M]} \bigr)
          \nonumber\\
          &= 
          \left( 
            \prod_{m = 1}^{k}
            \gzero\bigl( \LCTvxssigm, 
            \LCTvxlsigm \bigr) 
          \right)
          \cdot 
          \left( 
            \prod_{m = k+1}^{M}
            \gone\bigl( \LCTvxssigm, 
            \LCTvxlsigm \bigr) 
          \right)
          \nonumber\\
          &\quad\ \cdot 
          \prod_{e} P_{e}\bigl( \LCTtv{x}_{\efi,[M]}, \LCTtv{x}_{\efj,[M]} \bigr).
        \end{align*}
    \end{itemize}
  \end{enumerate}
\end{proof}

%----------------------------------------------------------------------------
 \begin{lemma}\label{sec:CheckCon:lem:2}
    The following properties for binomial coefficients hold.
    %----------------------------------------------------------------------------
    \begin{enumerate}
      \item Consider the case where
      \begin{align*}
        k_{i} \in \sZpp, \quad i \in [n],\quad
        M \in \sZpp, \quad
        \sum\limits_{i \in [n]} k_{i} \leq M.
      \end{align*}
      The following inequality holds
      \begin{align}
        \prod\limits_{i \in [n]}\binom{M}{ k_{i} }
        &\geq 
        \binom{M}{ \sum\limits_{i \in [n]} k_{i} }, 
        \label{sec:CheckCon:eqn:81}\\
        \frac{M!}{ 
          \Biggl( 
            \biggl( M-\sum\limits_{i \in [n]} k_{i} \biggr)! 
          \Biggr) 
          \cdot \prod\limits_{i \in [n]}(k_{i}!) 
        }
        &\geq \binom{M}{ \sum\limits_{i \in [n]} k_{i} }. \label{sec:CheckCon:eqn:82}
      \end{align}
      %------------------------------------------------------------------------
      The inequalities~\eqref{sec:CheckCon:eqn:81} and~\eqref{sec:CheckCon:eqn:82} hold with equalities if and only if there exists an integer $ i' \in [n] $ such that $ k_{i} = k_{i'} $ for $ i = i' $ and $ k_{i} = 0 $  for $ i \in [n] \in \setminus \{i'\} $.

      \item The binomial coefficient $ \binom{M}{k} $ increases w.r.t. $ k \in \sZp $ when $ k \leq M/2 $ and decreases w.r.t. $ k \in \sZp $ when $ k \geq M/2 $.
    \end{enumerate}
    %----------------------------------------------------------------------------
    
 \end{lemma}
 %----------------------------------------------------------------------------
 %----------------------------------------------------------------------------
 \begin{proof}
  % These properties follow from standard combinatorial results.
  The inequality~\eqref{sec:CheckCon:eqn:81} follows from a straightforward generalization of the following standard combinatorial results:
  %------------------------------------------------------------------------
  \begin{align*}
      &\hspace{-0.25cm} \binom{M}{ k_{1} }
      \cdot 
      \binom{M}{ k_{2} }
      \cdot \binom{M}{k_{1} + k_{2}}^{-1}
      \nonumber\\
      &= 
      \underbrace{
        \frac{ (k_{1} + k_{2})! 
        }{ ( k_{1} ! ) \cdot ( k_{2} ! ) }  
      }_{ \geq 1 }
      \cdot 
      \underbrace{
        \frac{ \bigl( M \cdot ( M - 1 ) \cdots
        (M-k_{1}+1) \bigr)
        \cdot 
        \bigl( M \cdot ( M - 1 ) \cdots
        (M-k_{2}+1) \bigr)
        }{ M \cdot ( M - 1 ) \cdots
        (M-k_{1} - k_{2}+1) } 
      }_{ \geq 1 }
      \nonumber\\
      &\geq 1,
  \end{align*}
  %------------------------------------------------------------------------
  where the inequality holds with equality if and only if at least one of the entries in $ (k_{1}, k_{2}) $ is zero-valued.
  
  There is a combinatorial proof for $ \binom{M}{ k_{1} } \cdot  \binom{M}{ k_{2} }
  \geq  \binom{M}{k_{1} + k_{2}} $. 
  The number $ \binom{M}{ k_{1} } \cdot  \binom{M}{ k_{2} } $ equals the number of ways to finish the following procedure:
  first, choose $ k_{1} $ different elements from $ M $ different elements; second, put these $ k_{1} $ elements back to the $M$ elements; last, choose $ k_{2} $ different elements from these $ M $ different elements. Therefore, compared with the ways to choose $ k_{1} + k_{2} $ different elements from $ M $ different elements, we have more ways to finish the former procedure.

  The proof for the inequality~\eqref{sec:CheckCon:eqn:82} and the monotonicity of the binomial coefficients as stated in the second property in the lemma statement is straightforward and thus it is omitted here.
 \end{proof}
 %----------------------------------------------------------------------------

% %-----------------------------------------------------------------------
% \begin{align}
% &\liminf_{M \to \infty}
% \bigl| \bigl( \ZBM(\graphN) \bigr)^{\! M} \bigr|^{1/M}
% \nonumber\\
% &=\liminf_{M \to \infty}
% \Biggl| \sum_{k=0}^{M}
%   \binom{M}{k}
%   \sum_{ \LCTvxsM, \LCTvxlM }
%   \Biggl( 
%     \Bigl( 
%       \prod_{m =1}^{k}
%       \gzero( \LCTvxsm, \LCTvxlm ) 
%     \Bigr)^{k}
%     \cdot 
%     \Bigl( 
%       \prod_{m =k+1}^{M}
%       \gone( \LCTvxsm, \LCTvxlm ) 
%     \Bigr)^{M-k} 
%   \Biggr) 
% \nonumber\\
% &\qquad \quad \ \cdot
% \prod_{e} 
% \Bigl( 
%     P_{e}\bigl( \LCTtv{x}_{\efi,[M]}, \LCTtv{x}_{\efj,[M]} \bigr) 
% \Bigr) 
% \Biggr|^{1/M}
% \label{sec:CheckCon:eqn:76}\\
% &\geq \ZBSPA^{*}(\graphN). \nonumber
% \end{align}
% %-----------------------------------------------------------------------
%----------------------------------------------------------------------------
\begin{lemma} \label{sec:CheckCon:lem: lower bound of Pe}
    Fix an  integer $ M \in \sZpp $, an integer $ k $ such that $ 0 \leq k \leq M-1 $, and two collections of variables $ \LCTtv{x}_{\efi,[M]}, \LCTtv{x}_{\efj,[M]} \in \prod_{e} \LCTtset{X}_{e}^{M} $. If the following condition holds
    %------------------------------------------------------------------------
    \begin{align}\label{sec:CheckCon:eqn:nonzero condition for G fun}
        \Biggl( 
          \prod_{m =1}^{k}
          \gzero( \LCTvxsm, \LCTvxlm ) 
        \Biggr)
        \cdot 
        \Biggl( 
          \prod_{m =k+1}^{M}
          \gone( \LCTvxsm, \LCTvxlm ) 
        \Biggr)
        \cdot \prod_{e} 
        P_{e}\bigl( \LCTtv{x}_{\efi,[M]}, \LCTtv{x}_{\efj,[M]} \bigr) \neq 0,
    \end{align}
    %------------------------------------------------------------------------
    then we have
    %-----------------------------------------------------------------------
    \begin{align}
        \prod_{e} 
        P_{e}\bigl( \LCTtv{x}_{\efi,[M]}, \LCTtv{x}_{\efj,[M]} \bigr) 
        \leq \binom{M}{k}^{\!\! -1},
        \label{sec:CheckCon:eqn:67}
    \end{align}
    %-----------------------------------------------------------------------
   where $ t_{ e, \LCTt{x}_{e} }\bigl( 
   \LCTtv{x}_{e,f_{i},[M]}
  \bigr) $ and the type 
  $ 
    \vt_e =\bigl( t_{ e, \LCTt{x}_{e} }( 
     \LCTtv{x}_{e,f_{i},[M]}
    ) \bigr)_{ \LCTt{x}_{e} \in \LCTtset{X}_{e} } 
  $ 
  are obtained by replacing $ \xe $ and 
  $ \setx{e} $ with $ \LCTt{x}_{e} $ and $ \LCTtset{X}_{e} $, respectively, in Definition~\ref{sec:SST:def:6}.

  Note that the inequality~\eqref{sec:CheckCon:eqn:67} hold with equality under the case where $ k =0 $ and for each $ e \in \setEfull $, there exists 
    $ \LCTt{x}_{e} \neq \tzero $ such that 
  \begin{align*}
    t_{ e, \LCTt{x}_{e} }\bigl( 
     \LCTtv{x}_{e,f_{i},[M]}
    \bigr) = t_{ e, \LCTt{x}_{e} }\bigl( 
     \LCTtv{x}_{e,f_{j},[M]}
    \bigr) = 1.
  \end{align*}
\end{lemma}
%----------------------------------------------------------------------------
%----------------------------------------------------------------------------
\begin{proof}
    % In this proof, we suppose that the condition~\eqref{sec:CheckCon:eqn:nonzero condition for G fun} holds. 

    If~\eqref{sec:CheckCon:eqn:nonzero condition for G fun} is satisfied, we know that each term on the right-hand side of the equality in~\eqref{sec:CheckCon:eqn:nonzero condition for G fun} is non-zero, then
    the following three properties hold.
    %----------------------------------------------------------------------------
    \begin{enumerate}
        \item By the definitions of $ \gzero $, the condition 
        $ \gzero( \LCTvxsm, \LCTvxlm ) \neq 0 $ for all $ 1 \leq m \leq k $ implies
        %--------------------------------------------------------------------
        \begin{align*}
            &\LCTvxsm = \LCTvxlm = \tv{0}, \qquad 1 \leq m \leq k.
            % \label{sec:CheckCon:eqn:65}
        \end{align*}
        %--------------------------------------------------------------------
        Then we have
        %--------------------------------------------------------------------
        \begin{align}
            \wh\bigl(\LCTtv{x}_{e,f_{i},[M]} \bigr)=
            M\cdot \sum_{ \LCTt{x}_{e}:\, \LCTt{x}_{e} \neq 0 } t_{ e, \LCTt{x}_{e} }
            \bigl( 
             \LCTtv{x}_{e,f_{i},[M]}
            \bigr) &\leq M-k, \quad
            e = (f_{i},f_{j}) \in \setEfull, \, i < j.
            \label{sec:CheckCon:eqn:83} 
        \end{align}
        %--------------------------------------------------------------------

        \item According to the property $ g_{1} $ proven in Lemma~\ref{sec:CheckCon:prop:properties of N(0)}, which is obtained based on the property of the LCT, we know that the condition 
        $ \gone( \LCTvxsm, \LCTvxlm ) \neq 0 $ for all $ k+1 \leq m \leq M $
        implies that 
        %--------------------------------------------------------------------
        \begin{align*}
            \wh(\LCTvxsm) + \wh(\LCTvxlm) \geq 2, \qquad k+1 \leq m \leq M.
        \end{align*}
        %--------------------------------------------------------------------
        Then we have
        \begin{align}
          \wh(\LCTvxsM) + \wh(\LCTvxlM)
          = 
          \sum_{m \in [M]}
          \bigl( \wh(\LCTvxsm) + \wh(\LCTvxlm) \bigr) 
          \geq 2( M - k).
          \label{sec:CheckCon:eqn:cond wh greater than 2}
        \end{align}
        % { \color{red} Here we use the property of LCT, \ie, $ \wh(\vx) = 1 $ implies $ g(\vx) = 0 $. }

        \item Following a similar idea as in Lemma~\ref{sec:SST:lem:4}, we know that the condition $ 
        P_{e}\bigl( \LCTtv{x}_{\efi,[M]}, \LCTtv{x}_{\efj,[M]} \bigr) \neq 0 $ 
        for all $ e \in \setEfull $ implies the following equalities:
        \begin{align}
          \prod_{e} 
          P_{e}\bigl( \LCTtv{x}_{\efi,[M]}, \LCTtv{x}_{\efj,[M]} \bigr)
          &= \prod_{e}
          \frac{ \prod\limits_{\LCTt{x}_{e}}
          \bigl( (M \cdot t_{e,\LCTt{x}_{e}})! \bigr) }{ M! },
          \label{eqn: expression of Pe}\\
          \vt_{e,\LCTt{x}_{e}} 
          &=
          \vt_e(\LCTtv{x}_{\efi,[M]}) 
          = \vt_e(\LCTtv{x}_{\efj,[M]}), 
          \qquad e = (f_{i}, f_{j}) \in \setEfull.
          \nonumber
          % \label{property of pe to be nonzero}
        \end{align}
        Then we have
        %--------------------------------------------------------------------
        \begin{align}
            \wh(\LCTvxsM) = \wh(\LCTvxlM)=
            M\cdot \sum_{e}
            \ \sum_{ \LCTt{x}_{e}:\, \LCTt{x}_{e} \neq 0 } 
            t_{ e, \LCTt{x}_{e} }
            &\geq M-k,
            \label{existence of ke}
        \end{align}
        %--------------------------------------------------------------------
        where the inequality follows from~\eqref{sec:CheckCon:eqn:cond wh greater than 2}.
        
        % Also, because $ k \leq M-1 $, the property $ g_{1} $ proven in Lemma~\ref{sec:CheckCon:prop:properties of N(0)} implies that 
        % there exists 
        % $ e_{\ell} = (f_{i_{\ell}}, f_{j_{\ell}})
        % \in \setEfull $ and $ \LCTt{x}_{e_{\ell}, f_{i_{\ell}},M}, \LCTt{x}_{e_{\ell}, f_{j_{\ell}}, M} \in \LCTtset{X}_{e_{\ell}} $ 
        % with $ \ell \in [2] $ such that $ e_{1} \neq e_{2} $ and
        % \begin{align*}
        %   \LCTt{x}_{e_{\ell}, f_{i_{\ell}},M} \neq \tzero,\qquad
        %   \LCTt{x}_{e_{\ell}, f_{j_{\ell}},M} \neq \tzero, \qquad 
        %   \ell \in [2].
        % \end{align*}
        % Then we have 
        % \begin{align}
        %   \sum_{ \LCTt{x}_{e_{\ell}} \neq 0 } t_{e_{\ell},\LCTt{x}_{e_{\ell}}} 
        %   \geq 1, \qquad 
        %   \ell \in [2].
        %   \label{eqn: property implied by g1 at M}
        % \end{align}
    \end{enumerate}
    %----------------------------------------------------------------------------
    % Suppose that 
    % %-----------------------------------------------------------------------
    % \begin{align}   
    %     \begin{aligned}
    %         &|\setF| \geq 3, \\
    %         & 0 \leq k \leq M - 1, \\
    %         &\LCTt{x}_{e_{m},f_{1},m} \neq 0,\, 
    %         \LCTt{x}_{e_{m}',f_{1},m} \neq 0, \qquad 
    %         e_{m} = (f_{1}, f_{2}), \, 
    %         e_{m}' = (f_{1},f_{3}),
    %         \qquad 
    %         k+1 \leq m \leq M, 
    %         % \LCTtv{x}_{\efi,[M]}, \LCTtv{x}_{\efj,[M]}
    %         % \in \prod \LCTtset{X}_{e}^{M}
    %     \end{aligned} \label{sec:CheckCon:eqn:86}
    % \end{align}
    % %-----------------------------------------------------------------------
    % which satisfies~\eqref{sec:CheckCon:eqn:65} and~\eqref{sec:CheckCon:eqn:cond wh greater than 2}.
    % The proof for other cases satisfying~\eqref{sec:CheckCon:eqn:65} and~\eqref{sec:CheckCon:eqn:cond wh greater than 2} is similar and thus it is omitted here. 
    Based on these properties, we discuss two cases in the following.
    %----------------------------------------------------------------------------
    \begin{enumerate}
        \item Consider the case where there exists an edge 
        $ e_{1} \in \setEfull $ such that 
        %--------------------------------------------------------------------
        \begin{align}
            M \cdot 
            \sum_{ \LCTt{x}_{e_{1}}:\, \LCTt{x}_{e_{1}} \neq 0 }
            t_{e_{1},\LCTt{x}_{e_{1}}}
            \geq \frac{M}{2}. \label{sec:CheckCon:eqn:97}
        \end{align}
        %--------------------------------------------------------------------
        In this case, we have
        %------------------------------------------------------------------------
        \begin{align}
        \frac{ M! }{ 
            \prod\limits_{\LCTt{x}_{e_{1}}} \Bigl( 
                \bigl( M 
                    \cdot t_{e_{1},\LCTt{x}_{e_{1}}}
                \bigr)!
            \Bigr) 
        }\overset{(a)}{\geq}
        \binom{ M 
        }{ 
          M \cdot  \sum\limits_{ \LCTt{x}_{e_{1}}:\, \LCTt{x}_{e_{1}} \neq 0 }
          t_{e_{1},\LCTt{x}_{e_{1}}} 
        }
        \overset{(b)}{\geq}
         \binom{ M }{ M-k }=
        \binom{ M }{ k }, 
        \label{eqn: bound of type for e prime}
        \end{align}
        %------------------------------------------------------------------------
        where step $(a)$ follows from~\eqref{sec:CheckCon:eqn:82} in Lemma~\ref{sec:CheckCon:lem:2},
        and 
        where step $(b)$ follows from the inequalities 
        in~\eqref{sec:CheckCon:eqn:97} and~\eqref{sec:CheckCon:eqn:83}, \ie,
        \begin{align*}
          M-k
          \geq
          M \cdot 
          \sum_{ \LCTt{x}_{e_{1}}:\, \LCTt{x}_{e_{1}} \neq 0 }
          t_{e_{1},\LCTt{x}_{e_{1}}} 
          \geq 
          \frac{M}{2}, 
          % \label{eqn: range of the sum of non-zero types for e}
        \end{align*} 
        and Lemma~\ref{sec:CheckCon:lem:2}, \ie, 
        $ \binom{ M }{ M - k -c } \geq \binom{ M }{ M - k } $ 
        for any integer $ c \in \sZp $ such that $  M - k - c \geq M/2 $.
        % Note that the inequality at step $(b)$ in~\eqref{eqn: bound of type for e prime} holds with equality if and only if 
        % the second equality in~\eqref{eqn: range of the sum of non-zero types for e} holds with equality.

        In this case, we obtain an upper bound of the expression in~\eqref{eqn: expression of Pe} as follows:
        \begin{align*}
          \prod_{e}
          \frac{ \prod\limits_{\LCTt{x}_{e}}
          \bigl( (M \cdot t_{e_{1},\LCTt{x}_{e}})! \bigr) }{ M! }
          =
          \underbrace{
            \frac{ \prod\limits_{\LCTt{x}_{e_{1}}}
            \bigl( (M \cdot t_{e,\LCTt{x}_{e_{1}}})! \bigr) }{ M! }
          }_{ \overset{(a)}{\leq}  \binom{M}{k} }
          \prod_{e \in \setEfull \setminus \{e_{1}\}}
          \underbrace{ \frac{ \prod\limits_{\LCTt{x}_{e}}
          \bigl( (M \cdot t_{e,\LCTt{x}_{e}})! \bigr) }{ M! }
          }_{ \overset{(b)}{\leq} 1 }
          \leq \binom{M}{k}.
          % \label{eqn: upper bound of product of Pe: case 1}
        \end{align*}
        where step $(a)$ follows from~\eqref{eqn: bound of type for e prime},
        and where step $(b)$ follows from
        \begin{align*}
          \sum_{ \LCTt{x}_{e} } 
          M \cdot t_{e,\LCTt{x}_{e}} &= M, \nonumber\\
          M \cdot t_{e,\LCTt{x}_{e}} &\in \sZp, \qquad 
          \LCTt{x}_{e} \in \LCTtset{X}_{e},\, 
          e \in \setEfull.
        \end{align*}
        % the property of $ g_{1} $ proven in Lemma~\ref{sec:CheckCon:prop:properties of N(0)}: there exists at least one edge in $ e_{2} \in \setEfull \setminus \{e_{1}\} $ such that
        % \begin{align}
        %   1 \leq M \cdot t_{e_{2},\LCTt{x}_{e_{2}}} \leq M-k, \qquad 
        %   \LCTt{x}_{e_{2}} \in \LCTtset{X}_{e_{2}} \setminus \{0\},
        %   \label{eqn: property of the type at edge e}
        % \end{align}
        % which implies
        % \begin{align*}
        %   \frac{ \prod\limits_{\LCTt{x}_{e_{2}}}
        %   \bigl( (M \cdot t_{e_{2},\LCTt{x}_{e_{2}}})! \bigr) }{ M! }
        %   \leq 1.
        % \end{align*}
        Therefore, the inequality~\eqref{sec:CheckCon:eqn:67} holds in this case.
        % Note that if for each $ e \in \setEfull $, there exists 
        % $ \LCTt{x}_{e} \neq \tzero $ such that $ t_{e,\LCTt{x}_{e}} = 1 $, then the inequality~\eqref{sec:CheckCon:eqn:67} hold with equality.

        \item Consider the case where there does not exist an edge $ e' \in \setEfull $ satisfying the inequality in~\eqref{sec:CheckCon:eqn:97}, \ie,
        \begin{align}
          M \cdot 
          \sum_{ \LCTt{x}_{e}:\, \LCTt{x}_{e} \neq 0 }
          t_{ e, \LCTt{x}_{e} }
          < \frac{M}{2}, \qquad e \in \setEfull.
          \label{eqn: upper bound of the types for each edge}
        \end{align}
        In this case, we obtain an upper bound of the expression in~\eqref{eqn: expression of Pe} as follows:
        %------------------------------------------------------------------------
        \begin{align*}
          \begin{aligned}
            \prod_{ e \in \setEfull }
            \frac{ M! }{ 
                \prod_{\LCTt{x}_{e}} \biggl( 
                    \Bigl( M 
                        \cdot t_{ e, \LCTt{x}_{e} }
                    \Bigr)!
                \biggr) 
            }& 
            \overset{(a)}{\geq} 
            \prod_{ e \in \setEfull } \binom{ M }{ 
            M \cdot  \sum\limits_{ \LCTt{x}_{e}:\, \LCTt{x}_{e} \neq 0 }
            t_{ e, \LCTt{x}_{e} } }  
            \\
            &\overset{(b)}{\geq} 
            \prod_{ e \in \setEfull }
            \binom{ M }{ 
                M \cdot 
                \sum\limits_{ \LCTt{x}_{e}:\, \LCTt{x}_{e} \neq 0 }
                t_{ e, \LCTt{x}_{e} }
                - k_{e} 
            }\\
            &\overset{(c)}{\geq}
            \binom{ M }{ 
                \sum\limits_{e}
                \sum\limits_{ \LCTt{x}_{e}:\, \LCTt{x}_{e} \neq 0 }
                M \cdot t_{ e, \LCTt{x}_{e} }
                - \sum\limits_{e}k_{e} 
            }\\
            &\overset{(d)}{=} \binom{ M }{ k }, 
          \end{aligned}
          % \label{sec:CheckCon:eqn:90}
        \end{align*}
        %------------------------------------------------------------------------
        %----------------------------------------------------------------------------
        \begin{itemize}
          \item where step $(a)$ follows from~\eqref{sec:CheckCon:eqn:82} in Lemma~\ref{sec:CheckCon:lem:2}, 

          \item where step $(b)$ follows from the properties and the definition of $ k_{e} $ given as follows:
          %----------------------------------------------------------------------------
          \begin{itemize}

              \item the inequality~\eqref{eqn: upper bound of the types for each edge} holds;

              \item the coefficient $ \binom{ M }{ c } $ is an increasing function w.r.t. $ c \in \sZ $ in the range $ 0 \leq c \leq  M/2 $;

              \item based on~\eqref{existence of ke}, we can choose an arbitrary collection of integers 
              $ (k_e)_{e \in \setEfull} \in \sZ^{|\setEfull|} $ such that:
              \begin{align}
                & 0 \leq k_{e} \leq 
                \sum_{ \LCTt{x}_{e}:\, \LCTt{x}_{e} \neq 0 }
                M \cdot t_{ e, \LCTt{x}_{e} }, \qquad 
                  e \in \setEfull,
                 \nonumber \\
                &\sum_{e}
                  \sum_{ \LCTt{x}_{e}:\, \LCTt{x}_{e} \neq 0 } 
                    M \cdot t_{ e, \LCTt{x}_{e} }
                    - \sum_{e} k_{e} = M-k,
                  \label{sec:CheckCon:eqn:91}
              \end{align}
              (these conditions ensure that $k_e$ is a valid choice within the range of possible values)

          \end{itemize}
          %----------------------------------------------------------------------------

          \item  where step $(c)$ follows the inequality~\eqref{sec:CheckCon:eqn:81} in Lemma~\ref{sec:CheckCon:lem:2},

          \item where step $(d)$ follows from~\eqref{sec:CheckCon:eqn:91}.
        \end{itemize}
        %----------------------------------------------------------------------------
        Therefore, the inequality~\eqref{sec:CheckCon:eqn:67} also holds in this case.
    \end{enumerate}

\end{proof}
%----------------------------------------------------------------------------
% and where step $(c)$ follows from the expression of $ \bigl( \ZBM(\graphN) \bigr)^{\! M} $ in~\eqref{sec:CheckCon:eqn:64}.

By the strict inequality~\eqref{sec:CheckCon:eqn:70} in the theorem statement, 
we obtain
\begin{align*}
  \ZBSPA^{*}(\graphN) > 0.
\end{align*}
Then we define $ \alpha_{\graphN} $ as follows
%------------------------------------------------------------------------
\begin{align}
    \alpha_{\graphN}
    &\defeq 
    \frac{ 
      \sum\limits_{ \LCTvxs, \LCTvxl }
      \bigl|\gone( \LCTvxs, \LCTvxl )\bigr| 
    }{ 
      \ZBSPA^{*}(\graphN) }.
    \label{sec:CheckCon:eqn:71}
\end{align}
%------------------------------------------------------------------------

% \ie,
% %------------------------------------------------------------------------
% \begin{align*}
%     \prod_{f}
%     \Biggl( 
%         \sum_{\LCTtv{x}_{\setpf}}
%         \bigl| \LCT{f}( \LCTtv{x}_{\setpf} ) \bigr|
%     \Biggr) - \ZBSPA^{*}(\graphN)
%     < \frac{1}{2}\cdot \ZBSPA^{*}(\graphN).
% \end{align*}
%------------------------------------------------------------------------
%----------------------------------------------------------------------------
\begin{lemma}\label{sec:CheckCon:lem:liminf Z N0 greater than ZSPA} 
    If the strict inequality in~\eqref{sec:CheckCon:eqn:70} holds, then we have
    %------------------------------------------------------------------------
    \begin{align*}
        \liminf_{M \to \infty}
        \ZBM(\graphN)
        \geq \ZBSPA^{*}(\graphN).
    \end{align*}
    %------------------------------------------------------------------------
\end{lemma}
%----------------------------------------------------------------------------
%----------------------------------------------------------------------------
\begin{proof}
  It holds that
  %----------------------------------------------------------------------------
  \begin{align}
    0 \leq 
      \alpha_{\graphN} 
      =\frac{ \prod_{f}
      \Biggl( 
          \sum\limits_{\LCTtv{x}_{\setpf}}
          \bigl| \LCT{f}( \LCTtv{x}_{\setpf} ) \bigr|
      \Biggr) - \ZBSPA^{*}(\graphN) 
    }{ \ZBSPA^{*}(\graphN) }
    < \frac{1}{2},
    \label{sec:CheckCon:property of alpha}
  \end{align}
  %----------------------------------------------------------------------------
  where the equality follows from the property of $ g_{1} $ as stated in Lemma~\ref{sec:CheckCon:prop:properties of N(0)},
  and
  the second strict inequality follows from the assumption stated in~\eqref{sec:CheckCon:eqn:70}. 
  Then we have
    %------------------------------------------------------------------------
    \begin{align}
        \bigl( \ZBM(\graphN) \bigr)^{\! M} 
        &\overset{(a)}{=} \Biggl| \sum_{k=0}^{M}
          \binom{M}{k}
          \cdot 
          \!\!\!\!
          \sum_{ \LCTvxsM, \LCTvxlM }
          \!\!
          % \Biggl( 
            \Biggl( 
              \prod_{m =1}^{k}
              \gzero( \LCTvxsm, \LCTvxlm ) 
            \Biggr)
          % \Biggr) 
        % \Biggr)
        \cdot 
        \Biggl( 
          \prod_{m =k+1}^{M}
          \gone( \LCTvxsm, \LCTvxlm ) 
        \Biggr)
        \nonumber\\
        &\hspace{4 cm}
        \cdot 
        \prod_{e} 
        P_{e}\bigl( \LCTtv{x}_{\efi,[M]}, \LCTtv{x}_{\efj,[M]} \bigr) \Biggr|
        \nonumber\\
        %--------------------------------------------------------------------
        &\overset{(b)}{\geq}
        \underbrace{
          \sum_{ \LCTvxsM, \LCTvxlM }
          \Biggl( 
              \prod_{m \in [M]}
              \gzero( \LCTvxsm, \LCTvxlm ) 
          \Biggr)
          \cdot 
          \prod_{e} 
          P_{e}\bigl( \LCTtv{x}_{\efi,[M]}, \LCTtv{x}_{\efj,[M]} \bigr)
        }_{ \overset{(c)}{=}  \bigl( \ZBSPA^{*}(\graphN) \bigr)^{\! M} }
        \nonumber\\
        %--------------------------------------------------------------------
        &\quad -
        \sum_{k =0}^{M-1}
        \binom{M}{k}
        \sum_{ \LCTvxsM, \LCTvxlM }
        \prod_{m =1}^{k}
        \gzero( \LCTvxsm, \LCTvxlm ) 
        \cdot 
        \Biggl|
          \prod_{m =k+1}^{M}
          \gone( \LCTvxsm, \LCTvxlm ) 
        \Biggr|
        \nonumber\\
        &\quad \cdot 
        \prod_{e} 
        P_{e}\bigl( \LCTtv{x}_{\efi,[M]}, \LCTtv{x}_{\efj,[M]} \bigr)
        \nonumber\\
        %--------------------------------------------------------------------
        &\overset{(f)}{>}
        \bigl( \ZBSPA^{*}(\graphN) \bigr)^{M}
          \cdot 
        \Biggl( 1 - 
          \sum_{k =0}^{M-1} \alpha_{\graphN}^{M-k}
        \Biggr)
        \nonumber\\
        &\overset{(i)}{\geq}
        \bigl( \ZBSPA^{*}(\graphN) \bigr)^{\! M}
        \cdot 
        \Biggl(  1 - 
            \alpha_{\graphN}
            \cdot
            \frac{ 1 - \alpha_{\graphN}^{M}
            }{ 1-\alpha_{\graphN} }
        \Biggr)
        \nonumber\\
        &\overset{(j)}{\geq} 
        \bigl( \ZBSPA^{*}(\graphN) \bigr)^{\! M}
        \cdot 
        \Biggl(  
        \underbrace{1 - 
            \frac{ \alpha_{\graphN} 
            }{ 1-\alpha_{\graphN} }
        }_{ \overset{(k)}{>} 0}
        \Biggr),
        \label{sec:CheckCon:eqn:75}
    \end{align}
    %------------------------------------------------------------------------
    %----------------------------------------------------------------------------
    \begin{itemize}
      \item where step $(a)$ follows from the expression of $ \bigl( \ZBM(\graphN) \bigr)^{\! M} $ in~\eqref{sec:CheckCon:eqn:64},

      \item where step $(b)$ follows from the triangle inequality, 

      \item where step $(c)$ follows from the following derivations:
      %-----------------------------------------------------------------------
      \begin{align*}
          \sum_{ \LCTvxsM, \LCTvxlM }
          &\Biggl( 
              \prod_{m \in [M]}
              \gzero( \LCTvxsm, \LCTvxlm ) 
          \Biggr)
          \cdot 
          \prod_{e} P_{e}
          \bigl( \LCTtv{x}_{\efi,[M]}, \LCTtv{x}_{\efj,[M]} \bigr)
          \nonumber\\
          &\overset{(d)}{=}
          \bigl( \gzero( \tv{0}) \bigr)^{\! M}
          % \cdot \sum_{ \LCTvxsM, \LCTvxlM }
          % [ \LCTvxsM \!=\! \LCTvxlM \!=\! \tv{0} ]
          \cdot 
          \prod_{e} P_{e}( \tv{0}, \tv{0}) 
          \nonumber\\
          % &\overset{(e)}{=}
          % \bigl( \gzero( \tv{0}) \bigr)^{\! M}
          % \cdot 
          % \sum_{ \LCTvxsM, \LCTvxlM }
          % [ \LCTvxsM \!=\! \LCTvxlM \!=\! \tv{0} ]
          % \nonumber\\
          &\overset{(e)}{=}
          \bigl( \ZBSPA^{*}(\graphN) \bigr)^{\! M},
          % \nonumber\\
          % &\overset{(f)}{\in} \sR_{> 0},
      \end{align*}
      %-----------------------------------------------------------------------
      %----------------------------------------------------------------------------
      % \begin{itemize}
        \item where step $(d)$ follows from the fact that in order to have $ \gzero( \LCTvxsm, \LCTvxlm ) \neq 0 $ for $ 1 \leq m \leq M $, the equality $ \LCTvxsM = \LCTvxsM = \tv{0} $ has to be satisfied,

        \item where step $(e)$ follows from  the property of $ \gzero $ in~\eqref{sec:CheckCon:eqn:73} and the property of $ P_{e} $ in Lemma~\ref{sec:SST:lem:4}: 
        \begin{align*}
            P_{e}\bigl( \tv{0}, \tv{0} \bigr)  = 1, \qquad
            e \in \setEfull,
        \end{align*}
        
        % \item where step $(f)$ follows from Remark~\ref{sec:DENFG:remk:1},
      % \end{itemize}
      %----------------------------------------------------------------------------

      \item where step $(f)$ follows from the following derivations:
      %-------------------------------------------------------------------
      \begin{align*}
          \hspace{-3cm}\binom{M}{k}
          &\sum_{ \LCTvxsM, \LCTvxlM }
          \Biggl|
            \prod_{m =1}^{k}
            \gzero( \LCTvxsm, \LCTvxlm ) 
          \Biggr|
          \cdot 
          \Biggl|
            \prod_{m =k+1}^{M}
            \gone( \LCTvxsm, \LCTvxlm ) 
          \Biggr|
          \cdot 
          \prod_{e} 
          P_{e}\bigl( \LCTtv{x}_{\efi,[M]}, \LCTtv{x}_{\efj,[M]} \bigr)
          \nonumber\\
          %-------------------------------------------------------------------
          &\overset{(g)}{\leq}
          \sum_{ \LCTvxsM, \LCTvxlM }
          \Biggl|
            \prod_{m =1}^{k}
            \gzero( \LCTvxsm, \LCTvxlm ) 
          \Biggr|
          \cdot 
          \Biggl|
            \prod_{m =k+1}^{M}
            \gone( \LCTvxsm, \LCTvxlm ) 
          \Biggr|
          \nonumber\\
          &=
          \Biggl( 
              \sum_{\LCTvxs, \LCTvxl} 
              \bigl| \gzero( \LCTvxs, \LCTvxl ) \bigr|
          \Biggr)^{\!\!\! k}
          \cdot
          \Biggl( 
              \sum_{\LCTvxs, \LCTvxl} 
              \bigl| \gone( \LCTvxs, \LCTvxl ) \bigr|
          \Biggr)^{\!\!\! M-k}
          \nonumber\\
          &\overset{(h)}{=}
          \bigl( \ZBSPA^{*}(\graphN) \bigr)^{M}
          \cdot 
          \alpha_{\graphN}^{M-k},
      \end{align*}
      %------------------------------------------------------------------- 
      %----------------------------------------------------------------------------
      % \begin{itemize}
      \item where step $(g)$ follows from Lemma~\ref{sec:CheckCon:lem: lower bound of Pe},

      \item where step $(h)$ follows from the properties of $ \gzero $ and $ \gone $ in Lemma~\ref{sec:CheckCon:prop:properties of N(0)} and the definition of $ \alpha_{\graphN} $ in~\eqref{sec:CheckCon:eqn:71},
      % \end{itemize}
      %----------------------------------------------------------------------------

      \item where step $(i)$ follows from
      \begin{align*}
          \sum_{k=0}^{M-1}\alpha_{\graphN}^{M-k}
          = \alpha_{\graphN} \cdot 
          \frac{1-\alpha_{\graphN}^{M}}{1 - \alpha_{\graphN}},
      \end{align*} 

      \item where step $(j)$ follows from $ \alpha_{\graphN} \in \sR_{>0} $ as proven in~\eqref{sec:CheckCon:property of alpha} and the following equality:
      %-----------------------------------------------------------------------
      \begin{align*}
          1 - 
          \alpha_{\graphN}
          \cdot
          \frac{ 1 - \alpha_{\graphN}^{M}
          }{ 1-\alpha_{\graphN} }
          = 1 - 
          \frac{ \alpha_{\graphN}
          }{ 1-\alpha_{\graphN} }
          + 
          \underbrace{
              \frac{ \alpha_{\graphN}^{M+1} 
              }{ 1-\alpha_{\graphN} }
          }_{\geq 0},
      \end{align*}
      %-----------------------------------------------------------------------
      
      \item where step $(k)$ follows from $ \alpha_{\graphN} < 1/2 $ as proven in~\eqref{sec:CheckCon:property of alpha}.
    \end{itemize}
    %--------------------------------------------------------------------------
    Finally, we obtain
    %-----------------------------------------------------------------------
    \begin{align*}
        \liminf_{M \to \infty}
        \ZBM(\graphN)
        &\overset{(a)}{\geq}
        \ZBSPA^{*}(\graphN)
        \cdot 
        \liminf_{M \to \infty}
        \Biggl(  
            1 -
            \frac{ \alpha_{\graphN}
            }{ 1-\alpha_{\graphN} }
        \Biggr)^{\! \!\! 1/M}
        \nonumber\\
        &\overset{(b)}{=} \ZBSPA^{*}(\graphN), 
    \end{align*}
    %-----------------------------------------------------------------------
    where step $(a)$ follows from the inequalities in~\eqref{sec:CheckCon:eqn:75},
    and where step $(b)$ follows from the fact that for any positive-valued number $ c $, we have
    %----------------------------------------------------------------------------
    \begin{align*}
      \lim_{M \to \infty} c^{1/M} = 1.
    \end{align*}
    %----------------------------------------------------------------------------
\end{proof}
%----------------------------------------------------------------------------

%----------------------------------------------------------------------------
\begin{lemma}\label{sec:CheckCon:lem:limsup Z N0 smaller than ZSPA} 
    If the strict inequality in~\eqref{sec:CheckCon:eqn:70} holds, then we have
    %------------------------------------------------------------------------
    \begin{align*}
        \limsup_{M \to \infty}
        \ZBM(\graphN)
        \leq \ZBSPA^{*}(\graphN).
    \end{align*}
    %------------------------------------------------------------------------
\end{lemma}
%----------------------------------------------------------------------------
%----------------------------------------------------------------------------
\begin{proof}
  The proof follow from a similar idea as in the proof of Lemma~\ref{sec:CheckCon:lem:liminf Z N0 greater than ZSPA}. 
  % we get
  % \begin{align*}
  %   \bigl( \ZBM(\graphN) \bigr)^{\! M} 
  %   %--------------------------------------------------------------------
  %   \leq 
  %   \bigl( \ZBSPA^{*}(\graphN) \bigr)^{\! M}
  %   \cdot 
  %   \Biggl(  1 + 
  %       \frac{ \alpha_{\graphN} 
  %       }{ 1-\alpha_{\graphN} }
  %   \Biggr)
  % \end{align*}
  % which implies
  % \begin{align*}
  %   \limsup_{M \to \infty}
  %   \ZBM(\graphN)
  %   &\leq \ZBSPA^{*}(\graphN).
  % \end{align*}
\end{proof}
%----------------------------------------------------------------------------

By Lemmas~\ref{sec:CheckCon:lem:liminf Z N0 greater than ZSPA} and~\ref{sec:CheckCon:lem:limsup Z N0 smaller than ZSPA}, we prove Theorem~\ref{sec:CheckCon:thm:1}.